\documentclass[11pt]{article}
\usepackage{amsmath, amssymb, amsthm}
\usepackage{physics}
\usepackage{braket}
\usepackage{xcolor}
\numberwithin{equation}{section}
\usepackage{makecell}
\usepackage{arydshln}
\usepackage{float}
\usepackage{graphicx}
\usepackage{subcaption}
\usepackage{comment}
\usepackage{blindtext}
\usepackage{titlesec}
\usepackage{hyperref}
\usepackage[style=numeric,sorting=none]{biblatex}
\usepackage{paralist}
\usepackage{braket}
\usepackage{pdflscape}
\usepackage{enumitem}
\usepackage{mathtools}
\usepackage{tocloft}
\usepackage{titling}
\makeatletter
\renewenvironment{abstract}
  {\small
   \begin{center}
     \bfseries \abstractname
   \end{center}
   \vspace{-1.5em}
   \list{}{\leftmargin=20pt\rightmargin=20pt}%
   \item\relax}
  {\endlist\vspace{-1em}}
\makeatother

\DeclareMathOperator*{\argmax}{arg max}
\DeclareMathOperator*{\sign}{sign}
\DeclareMathOperator*{\Span}{Span}

\newtheorem{definition}{Definition}
\newtheorem{theorem}{Theorem}
\newtheorem{proposition}[theorem]{Proposition}

\newtheorem{conjecture}{Conjecture}
\newtheorem{corollary}{Corollary}

\newcommand{\mH}{\mathcal{H}}
\newcommand{\mS}{\mathcal{S}}
\newcommand{\mSS}{\mathcal{SS}}
\newcommand{\mQ}{\mathcal{Q}}
\newcommand{\mT}{\mathcal{T}}
\newcommand{\mM}{\mathcal{M}}
\newcommand{\mL}{\mathcal{L}}
\newcommand{\mW}{\mathcal{W}}
\newcommand{\mP}{\mathcal{P}}
\newcommand{\mC}{\mathcal{C}}
\newcommand{\mD}{\mathcal{D}}
\newcommand{\mV}{\mathcal{V}}
\newcommand{\mB}{\mathcal{B}}
\newcommand{\mR}{\mathcal{R}}
\newcommand{\mF}{\mathcal{F}}

\makeatletter
\renewenvironment{proof}[1][\proofname]{\par
  \pushQED{\qed}
  \normalfont\topsep6pt \trivlist
  \item[\hskip\labelsep\itshape #1.]\mbox{}\\*  
  \ignorespaces
}{
  \popQED\endtrivlist\@endpefalse
}
\makeatother

\title{Extremizing Measures of Magic on Pure States\\by Clifford-stabilizer States}

\author
{Muhammad Erew, Moshe Goldstein
\\
\\
\normalsize{Raymond and Beverly Sackler School of Physics and Astronomy,}\\
\normalsize{Tel-Aviv University, Tel-Aviv 6997801, Israel.}\\
\\
\normalsize{erew@tauex.tau.ac.il , mgoldstein@tauex.tau.ac.il}
}

\date{}

\begin{document}

\maketitle
 
\begin{abstract}
Magic states are essential resources enabling universal, fault-tolerant quantum computation within the stabilizer framework.
They were originally termed ``magic'' not only because their non-stabilizerness provides the extra resource needed for stabilizer codes, constrained by the Eastin-Knill theorem, to achieve universality, but also because they can be fault-tolerantly distilled in a manner compatible with stabilizer-based error correction.
Since then, various discrete and continuous measures have been introduced to quantify non-stabilizerness (``magic''), for both pure and mixed states.
However, not every state with nonzero magic, as defined by these measures, has been shown to be distillable by a stabilizer code, and many of currently known distillable states arise as special cases of \emph{Clifford-stabilizer states}, defined as pure states uniquely stabilized by finite subgroups of the Clifford group.
In this work, we develop a general framework for \emph{group-covariant functionals} on the real manifold of Hermitian operators.
We formalize the notions of $G$-stabilizer spaces, states and codes for arbitrary finite subgroups $G \subset \mathrm{U}(\mH)$, the group of unitary operators acting on the Hilbert space $\mH$, and introduce analytic families of $G$-covariant functionals on $\mathrm{Herm}(\mH)$, the real vector space of Hermitian operators on $\mH$.
Our main theorem shows that any $G$-invariant pure state is an extremal point for a broad class of derived functionals, including symmetric combinations, max-type functionals, and $\alpha$-Rényi-type sums, whenever the underlying family is $G$-covariant.
This extremality holds when variations are restricted to directions lying in the orthogonal complement of the stabilized subspace associated with the group, while the state remains pure.
Specializing to the Pauli and Clifford groups, this framework unifies the extremality structure of several canonical magic measures (mana, stabilizer Rényi entropies, generalized stabilizer entropies, stabilizer fidelity) within a single group-theoretic and geometric picture.
In particular, Clifford-stabilizer states extremize these three measures of magic for pure states.
We identify a distinguished class of Clifford-stabilizer states, the non-degenerate eigenstates of Clifford operations.
We find all Clifford-inequivalent states of those for qubits, qutrits, ququints, and two-qubit systems, discovering new candidates for magic distillation protocols.
Notably, we propose an inefficient distillation protocol for a two-qubit magic state with higher stabilizer fidelity than $\ket{TT}$ and $\ket{TH}$.
Finally, we conjecture, based on numerical results, that stabilized subspaces of finite subgroups of the Clifford group whose Wigner functions are everywhere non-vanishing maximize the mana of magic under variations in the orthogonal complement. Moreover, guided by both numerical observations and structural considerations, we conjecture that the so-called SIC-POVM fiducial states are Clifford-stabilizer states.
\end{abstract}

\newpage

\tableofcontents

\newpage

\section{Introduction}

Quantum computing promises to solve certain computational problems exponentially faster than classical computers~\cite{NielsenChuang}.
However, building a fault-tolerant quantum computer faces formidable challenges, primarily due to the fragile nature of quantum information.
Quantum error-correcting codes (QECCs) are indispensable for addressing this fragility~\cite{GottesmanThesis}. These codes protect quantum information by encoding logical qubits into larger physicsl spaces, allowing for the detection and correction of errors without collapsing the quantum state.

A fundamental limitation, the Gottesman-Knill theorem, asserts that operations composed solely of stabilizer elements (preparations of stabilizer states, Clifford unitaries, and Pauli measurements) can be efficiently simulated on a classical computer~\cite{GottesmanHeisenberg,AaronsonGottesman}.
Hence, they are insufficient for universal quantum computation.
Further compounding the challenge, the Eastin-Knill theorem shows that no QECC can implement a universal set of transversal gates, that is, fault-tolerant gates that act independently on each qudit~\cite{EastinKnill,Campbell2017}.
These results underscore the necessity of resources beyond stabilizer operations to achieve fault-tolerant universality.

One elegant response to these limitations is \emph{magic state distillation}.
Rather than circumventing these restrictions, magic state distillation supplements them by preparing special non-stabilizer states known as magic states~\cite{KitaevBravyi,ReichardtMagicCSS,HowardContextuality,Campbell2017,Seddon2019,Bravyi2019simulationofquantum}.
These states, when combined with stabilizer operations, enable universal quantum computation.
Magic state distillation protocols purify noisy non-stabilizer states into high-fidelity magic states using only stabilizer operations, thus remaining compatible with fault-tolerant error correction.

After the seminal work of Bravyi and Kitaev, and motivated by the Gottesman-Knill theorem, the literature began to use the term magic to refer to non-stabilizerness. This usage, however, slightly departs from the original sense in which Bravyi and Kitaev introduced the term.
In their work, magic captured a dual role: first, the non-stabilizer states that, when combined with Clifford operations and Pauli measurements, complete the set of resources required for universal quantum computation; and second, the fact that these same non-stabilizer states can be distilled by stabilizer codes to achieve fault-tolerant universality.
In this work, we will adopt the modern convention: whenever we say \emph{magic}, we refer to non-stabilizerness itself, or to the degree of deviation from the stabilizer set in a quantum state.

Quantifying magic is crucial for understanding which states serve as resources for quantum computation and how powerful they are.
To this end, several resource-theoretic measures of magic have been proposed~\cite{Veitch_2014,HowardContextuality,WinterYangCoherence,HowardCampell2017ResourceTheory,WangWilde2019thauma,WangThaumaBounds}.
For pure states, quantities such as stabilizer fidelity and the mana of magic play central roles.
Stabilizer fidelity captures how close a given state is to the set of stabilizer states, while mana quantifies the extent to which the discrete Wigner function of a state deviates from a genuine classical probability distribution on the phase space~\cite{GibbonsFiniteFieldWigner,Gross2007Hudson,HowardContextuality, Veitch2012,Veitch_2014,Pashayan2015}.
These measures, along with others like the min- and max-relative entropies of magic, offer a structured framework for studying and comparing quantum states from the resource perspective.

Various magic state distillation protocols have been developed across different system dimensions.
Seminal protocols include those proposed by Bravyi and Kitaev for qubits~\cite{KitaevBravyi}, and extensions to higher prime dimensions~\cite{Anwar_2012,Distillation_in_All_Prime_Dimensions,FaultTolerantinD,HowardCampell2017ResourceTheory}, as well as tight distillation schemes for qutrits~\cite{qutritTight}.
Beyond single-qudit settings, recent work has explored the distillation of magic states associated with multiqubit non-Clifford gates, most notably the three-qubit CCZ and Toffoli states~\cite{Toffoli1,Toffoli2,Gupta2024}.
A common thread running through these works is the appearance of highly symmetric, Clifford-related states and structures, hints that deeper principles may govern which states are especially suited for magic distillation.

Motivated by these observations, this work investigates magic measures in relation to stabilized spaces of finite subgroups of the Clifford group.
We consider a system of qudits with dimension $ d $, where $ d $ is prime, and denote by $ \mathcal{C}_{n,d} $ the Clifford group acting on $ n $ such qudits. 
Assume the existence of a stabilizer code that encodes a single logical qudit into $ N $ physical qudits, and that a gate $ M \in \mathcal{C}_{1,d}$ acts transversally; that is, the logical operation $ M_L $ is implemented by applying $ M $ independently to each physical qudit.
Initializing each qudit in some mixed state twirled over the finite group generated by $M$, $\langle M \rangle$, guarantees that, after projection onto the code space by measuring the trivial syndrome, the logical state remains in a twirled state over $\langle M_L \rangle$.
This follows from the fact that $M$ is transversal and the logical $M$, $M_L$, commutes with the code projector $\Pi$.
Consequently, consider a mixed input state that is close to a magic, non-degenerate eigenstate of $ M $, say the eigenstate corresponding to eigenvalue $ 1 $. 
After measuring and obtaining the trivial syndrome, the resulting logical state is a mixed state near the encoded logical same magic eigenstate. 
The central question is whether this output state is closer to the target eigenstate than the input was. If so, the process constitutes a form of distillation.

This motivates our investigation of eigenspaces of Clifford operators, both for single and multi-qudit systems, and more generally, of spaces stabilized by finite subgroups of the Clifford group.
We examine their behavior under magic measures, particularly stabilizer fidelity and mana, and show that pure states in such spaces extremize these measures under perturbations orthogonal to the stabilized subspace.
While such structures are trivial for qubits beyond one-dimensional spaces, richer stabilized subspaces emerge in qutrits and higher prime dimensions.
A particularly striking example is the code introduced in~\cite{qutritTight}, which distills a continuously parameterized family of states lying within a degenerate eigenspace of the squared Hadamard operator $H^2$.

These extremal properties reflect deep geometric features of the resource landscape and suggest that Clifford-stabilizer states, states that are uniquely stabilized by some finite subgroup of the Clifford group, play a central role in characterizing potentially distillable non-stabilizer states.
To support this claim, we provide explicit constructions and analyses of such states in various systems, including qutrits, ququints, and two-qubit configurations.
Our findings show that non-degenerate Clifford eigenstates with non-vanishing Wigner functions often correspond to smooth local maxima of mana, while degenerate Clifford eigenstates or states with vanishing Wigner function at some pase-space points, despite more intricate behavior, still exhibit local extremality within constrained subspaces.
These results motivate a conjecture: any Clifford-stabilizer state with a non-vanishing Wigner function locally maximizes the mana of magic under perturbations orthogonal to the stabilized subspace.

\noindent
\textbf{Our contribution.}
The central technical contribution of this work is a general theory of \emph{group-covariant functionals} on the space of Hermitian operators acting on a finite Hilbert space $\mH$.
We view $\mathrm{Herm}(\mH)$ as a finite-dimensional real manifold and consider smooth and analytic real- or complex-valued maps on it that transform covariantly under conjugation by a finite subgroup $G \subset \mathrm{U}(\mH)$.
Within this setting, we formalize the notions of stabilization by operators and by groups, introduce $G$-stabilizer spaces, $G$-stabilizer states, $G$-stabilizer codes, and the corresponding invariant subspace $\mS_G$.
We then define $G$-covariant families of functions $\{F_v\}_{v\in\mV}$ on $\mathrm{Herm}(\mH)$, where $\mV$ is a finite set, and prove a general extremality theorem: for any $G$-invariant pure state $\psi = \ket{\psi}\bra{\psi}$, the point $\psi$ is extremal on the submanifold
$
    \varrho_{\psi;G}
    =
    \{\rho \in \varrho(\mH)
      \; \mid \; \operatorname{supp}(\rho) \subset
      \Span\{\ket{\psi}\} \oplus \mS_G^{\perp}\}
$
for three large classes of functionals built from $\{F_v\}_{v\in\mV}$:
(i) symmetric compositions of the form $\Sigma(O) = S(F_{\mathrm{v}(1)}(O),\ldots,F_{\mathrm{v}(|\mV|)}(O))$,
(ii) maximized envelopes $\mF(O) = \max_{v\in\mV} F_v(O)$, and
(iii) semi-Rényi sums $\mR_\alpha(O) = \sum_{v\in\mV} |F_v(O)|^\alpha$.
We further introduce a componentwise version of covariance for product label sets $\mV = \prod_i \mV_i$ and prove that this structure is stable under natural “contraction” operations that sum over one of the components.
When instantiated for the Pauli group and its normalizer, the Clifford group, this general formalism produces a unified treatment of several standard magic measures: It shows that Clifford-invariant states are extremal for stabilizer fidelity, for mana of magic, for stabilizer $\alpha$-Rényi entropies, and for generalized stabilizers entropies, when the variations considered are restricted to directions orthogonal to the stabilized subspace of that group, while remaining within the (projective Hilbert) manifold of pure states.
We also extend this result for the stabilizer fidelity to arbitrary finite groups of unitaries and their normalizers, thereby generalizing the familiar Pauli-Clifford setting.
We further:
\begin{itemize}
    \item establish a general framework to fully characterize the type of criticality (or extremality) for the three functions; and
    \begin{itemize}
        \item completely characterize the critical behavior of the stabilizer fidelity and stabilizer Rényi entropy for all Clifford-inequivalent, non-degenerate Clifford eigenstates of single qubits;
        \item completely characterize the critical behavior of both the stabilizer fidelity and the mana of magic for all Clifford-inequivalent, non-degenerate Clifford eigenstates of single qutrits, and cover most of the analysis needed for the stabilizer Rényi entropy;
        \item completely characterize the critical behavior of the mana of magic for all but one Clifford-inequivalent, non-degenerate Clifford eigenstate of single ququints, and leave the behavior of the stabilizer fidelity and the and stabilizer Rényi entropy at that remaining case for the reader.
    \end{itemize}
    \item observe that for single ququints, there exist states with stabilizer fidelity lower than all the states mentioned above, indicating the presence of a global minimum at another state;
    \item identify all Clifford-inequivalent, non-degenerate Clifford eigenstates for two qubits, three of which are new. Among these, we find a two-qubit magic state that is (Clifford-equivalent to) a state recently discovered independently in parallel work, where it was shown to maximize the Rényi entropy of magic~\cite{Liu2025}. We show this maximization through our tools too.
\end{itemize}
Beyond theoretical insights, we also propose an explicit (albeit inefficient) distillation protocol for a new two-qubit magic state whose stabilizer fidelity surpasses previously known candidates such as $\ket{TT}$ and $\ket{TH}$.
This provides practical evidence that Clifford-stabilizer states serve not only as extremal points in the magic landscape but also as operationally relevant resources that complete the computation towards universality fault-tolerantly.

Moreover, we connect our results to the Zauner conjecture and to Flammia’s results on fiducial states \cite{Zauner1999,Flammia2006,BengtssonApplebyFlammia2018}, and formulate a new conjecture of our own. We further propose two novel measures of magic for future investigation, and introduce in Appendix~\ref{app:group-stabilizer extent} the notion of \emph{group-stabilizer extent}, a generalization of the stabilizer extent that captures the resource content of states stabilized by arbitrary finite unitary subgroups.
This extends the resource-theoretic framework beyond the standard Pauli and Clifford settings.

\textbf{Outline.} This work is organized as follows.
In Section~\ref{sec:Background and Preliminaries}, we review the necessary background on the generalized Pauli group, Clifford operations, the discrete Wigner function, and resource theories of magic.
Section~\ref{sec:Group-Covariant Functionals} introduces the notion of group-stabilized spaces and formalizes the extremization properties of magic measures.
Section~\ref{sec:Examples of Non-degenerate Clifford Eigenstates of Qudits} provides explicit examples of critical extremizing states for single qudits and discusses connections to known distillation protocols.
Section~\ref{sec:Non-Degenerate Eigenstates of Cliffords for Two Qubits} does the same for two-qubit systems by finding all Clifford-inequivalent non-degenerate eigenstates of Clifford operations on two qubits.
Section~\ref{sec:Three-Qubit States} presents notable three-qubit states that extremize stabilizer fidelity, explores their Clifford equivalence, and highlights their potential relevance to magic state distillation.
Section~\ref{sec:An Inefficient Distillation Protocol Demonstrating Magic} proposes an inefficient distillation scheme for a new two-qubit magic state, suggesting that it may be a strong candidate for magic distillation.
Section~\ref{sec: Candidate Measures to investigate} introduces two novel families of candidate measures of magic, which we propose as promising directions for future research.
Finally, we conclude in Section~\ref{sec:Conclusion} with relevant platforms, open questions, and conjectures regarding the role of group-stabilized structures in magic distillation.
Technical details are relegated to the Appendices.

\section{Background and Preliminaries}
\label{sec:Background and Preliminaries}

To ground our analysis, this section reviews the algebraic and phase-space structures underlying stabilizer quantum formalism and the resource theory of magic. We begin with the generalized Pauli group for $N$ qudits of prime dimension $d$ and its role in defining the discrete phase space for odd $d$. We then outline the associated Clifford operations, stabilizer states, and Wigner-function formalism. Finally, we recall standard resource-theoretic quantities- such as stabilizer fidelity and mana- that quantify deviations from classicality and will serve as the main tools in the subsequent sections.

\subsection{The Generalized Pauli Group for $N$-Qudits}

For a single qudit with prime dimension $d$, othe generalized Pauli operators are defined by first introducing two fundamental unitary operators: the shift operator $X$ and the phase operator $Z$.
Their actions on the computational basis states $\{\ket{j}\}_{j=0}^{d-1}$ are given by:
\begin{align}
  X\ket{j} &= \ket{j+1   \, \, (\mathrm{mod} \, \,  d)}, \\
  Z\ket{j} &= \omega^j \ket{j},
\end{align}
where $\omega = e^{2\pi i/d}$ is the primitive $d$th root of unity. These operators satisfy the commutation relation
\begin{equation}
  ZX = \omega XZ,
\end{equation}
which generalizes the well-known relation among the qubit Pauli matrices.
The \emph{generalized Pauli group} for a single qudit is then defined as
\begin{equation}
  \mP_{1,d} = \Big\{ (-1)^x \zeta^k X^a Z^b \; \mid \; a,b,k \in \mathbb{Z}_d \ , \ x \in \mathbb{Z}_2 \Big\},
\end{equation}
where $\zeta = e^{\pi i/d}$.

For an $N$-qudit system, the Hilbert space is given by $\mH = (\mathbb{C}^d)^{\otimes N}$.
The generalized Pauli group extends naturally via tensor products:
\begin{equation}
  \mP_{N,d} = \Big\{ (-1)^x \zeta^k   X^{\mathbf{a}} Z^{\mathbf{b}} \; \mid \; \mathbf{a}, \mathbf{b} \in \mathbb{Z}_d^N,\ k \in \mathbb{Z}_d \ , \ x \in \mathbb{Z}_2 \Big\},
\end{equation}
where
\begin{equation}
    X^{\mathbf{a}} = X^{a_1} \otimes X^{a_2} \otimes \cdots \otimes X^{a_N}, \quad Z^{\mathbf{b}} = Z^{b_1} \otimes Z^{b_2} \otimes \cdots \otimes Z^{b_N}.
\end{equation}

This group is non-Abelian and serves as a fundamental structure for describing symmetries, constructing quantum error-correcting codes, and forming the basis for the discrete phase space formalism.
In our convention, the group has order $2 d^{2N+1}$.
This differs from the choice made in Ref.~\cite{Gross2007Hudson} and other references, where the generalized Pauli group, \emph{the Weyl-Heisenberg group}, is defined as
\[
   \bigl\{ \omega^{k} X^{\mathbf{a}} Z^{\mathbf{b}}
   \; \big| \;
   \mathbf{a},\mathbf{b} \in \mathbb{Z}_d^{N},
   k \in \mathbb{Z}_d
   \bigr\},
\]
which incorporates fewer phase factors.
Our definition has the advantage of aligning with the standard qubit case, where the Pauli group explicitly contains all four phases
$\{\pm 1, \pm i\}$.

In the resource-theoretic study of magic and departures from stabilizerness, global phases have no operational significance.
It is therefore natural to introduce the \emph{phase-reduced Pauli group}
\[
    \tilde{\mP}_{N,d}
    \coloneqq
    \mP_{N,d} \big/ \braket{\zeta^{k}\mathbb{I}},
\]
whose elements are Pauli operators modulo global phases.
This quotient has cardinality $d^{2N}$ and will be the canonical object underlying some the constructions and analysis that follow.

\subsection{Heisenberg-Weyl Displacement Operators}

The algebraic structure of the generalized Pauli group lends itself naturally to a phase space formulation for finite-dimensional systems. One introduces \emph{Heisenberg-Weyl operators} (or displacement operators) associated with points in the discrete phase space $\mathbb{V}_{N,d}\equiv \mathbb{Z}_d^N\times\mathbb{Z}_d^N$.
A typical definition is:
\begin{equation}
  T_{\boldsymbol{\chi}} = \tau^{\mathbf{p} \cdot \mathbf{q}}   X^{\mathbf{p}} Z^{\mathbf{q}}, \quad \text{with } \boldsymbol{\chi}=(\mathbf{p},\mathbf{q}) \in \mathbb{V}_{N,d},
\end{equation}
where $\tau$ is a fixed phase factor often chosen as
$\tau =\omega^{2^{-1}}=\omega^{(d+1)/2}$, with $2^{-1}$ denoting the inverse of 2 in $\mathbb{Z}_d$.
The explicit forms of these operators are:
\begin{equation}
  T_{\boldsymbol{\chi}} = \tau^{\mathbf{p} \cdot \mathbf{q}} \sum_{\mathbf{j}\in\mathbb{Z}_d^N} \omega^{\mathbf{q}\cdot\mathbf{j}}  \ket{\mathbf{p}+\mathbf{j}}\bra{\mathbf{j}}, \quad \text{with } \boldsymbol{\chi}=(\mathbf{p},\mathbf{q}) \in \mathbb{V}_{N,d}.
\end{equation}

The Heisenberg-Weyl displacement operators $T_{\boldsymbol{\chi}}$ satisfy several important identities:
\begin{itemize}

\item \textbf{Trace:}
  \begin{equation}
    \Tr\Big( T_{\boldsymbol{\chi}} \Big) = d^N  \delta_{\boldsymbol{\chi},\boldsymbol{0}}.
  \end{equation}
  
  \item \textbf{Composition Rule:} For any two phase space points $\boldsymbol{\chi}$ and $\boldsymbol{\chi^\prime}$, the product of the corresponding Weyl operators obeys:
  \begin{equation}
    T_{\boldsymbol{\chi}}  T_{\boldsymbol{\chi^\prime}} = \tau^{-\langle \boldsymbol{\chi}, \boldsymbol{\chi^\prime} \rangle}   T_{\boldsymbol{\chi}+\boldsymbol{\chi^\prime}},
  \end{equation}
  where the addition $\boldsymbol{\chi}+\boldsymbol{\chi^\prime}$ is performed modulo $d$, and $\langle \cdot,\cdot \rangle$ denotes a symplectic inner product defined on $\mathbb{Z}_d^{N}$.
  A common choice for the symplectic product is:
\begin{equation}
\langle (\mathbf{p},\mathbf{q}), (\mathbf{p}',\mathbf{q}') \rangle = \mathbf{p} \cdot \mathbf{q}' - \mathbf{q} \cdot \mathbf{p}'.
\end{equation}
This symplectic structure ensures that the phase factors arising in the operator products capture the essential non-commutativity of the underlying quantum operators, and this rule of composition encapsulates the non-commutative geometry of the discrete phase space.
  
  \item \textbf{Adjoint and Inversion:} The unitarity of $T_{\boldsymbol{\chi}}$ guarantees that:
  \begin{equation}
    T^\dagger_{\boldsymbol{\chi}} = T_{-\boldsymbol{\chi}},
  \end{equation}
  which directly follows from the composition rule.

  \item \textbf{Orthogonality:}
  \begin{equation}
    \Tr\Big( T_{\boldsymbol{\chi}}  T^\dagger_{\boldsymbol{\chi^\prime}} \Big) = d^N  \delta_{\boldsymbol{\chi},\boldsymbol{\chi^\prime}},
  \end{equation}
  ensuring that $\{T_{\boldsymbol{\chi}}\}$ forms a complete orthogonal basis for operators on $\mH$.
  
  \item \textbf{Commutation Relations:} By exchanging the order of the operators, one finds:
  \begin{equation}
    T_{\boldsymbol{\chi}}  T_{\boldsymbol{\chi^\prime}} = \tau^{-2\langle \boldsymbol{\chi}, \boldsymbol{\chi^\prime} \rangle}   T_{\boldsymbol{\chi^\prime}}  T_{\boldsymbol{\chi}},
  \end{equation}
  where the phase factor $\tau^{-2\langle \boldsymbol{\chi}, \boldsymbol{\chi^\prime} \rangle}=\omega^{-\langle \boldsymbol{\chi}, \boldsymbol{\chi^\prime} \rangle}$ quantifies the non-commutativity.
\end{itemize}

These identities are not only mathematically elegant but also provide the backbone for practical computations in quantum state tomography, error correction, and the analysis of quantum dynamics in finite-dimensional systems.

\subsection{Phase Space Operators}

In addition to the Heisenberg-Weyl displacement operators, an alternative operator basis for the space of linear operators on $\mH$ is provided by the \emph{phase space operators} $\{ A_{\boldsymbol{\chi}} \}$.
These operators, sometimes called \emph{phase point operators}, are constructed via a symplectic Fourier transform of the displacement operators.
One common definition is:
\begin{equation}
  A_{\boldsymbol{\chi}} = \frac{1}{d^N} \sum_{\boldsymbol{\chi^\prime}\in\mathbb{V}_{N,d}} \omega^{-\langle \boldsymbol{\chi}, \boldsymbol{\chi^\prime} \rangle}   T_{\boldsymbol{\chi^\prime}}.
  \label{eq:phase_point_operator}
\end{equation}
The explicit forms of these operators are:
\begin{equation}
  A_{\boldsymbol{\chi}} = \sum_{\mathbf{j}\in\mathbb{Z}_d^N} \omega^{2\mathbf{q} \cdot (\mathbf{p}-\mathbf{j})}  \ket{2\mathbf{p}-\mathbf{j}}\bra{\mathbf{j}}, \quad \text{with } \boldsymbol{\chi}=(\mathbf{p},\mathbf{q}) \in \mathbb{V}_{N,d}.
\end{equation}

These phase space operators satisfy several key identities:

\begin{itemize}
  \item \textbf{Hermiticity:}  
  Since the displacement operators obey $T^\dagger_{\boldsymbol{\chi}} = T_{-\boldsymbol{\chi}}$ and the symplectic inner product is antisymmetric, one may show that
  \begin{equation}
    A^\dagger_{\boldsymbol{\chi}} = A_{\boldsymbol{\chi}}.
  \end{equation}
  Thus, the $A_{\boldsymbol{\chi}}$ are Hermitian operators.
  
  \item \textbf{Completeness:}  
  They provide a resolution of the identity in the operator space:
  \begin{equation}
    \sum_{\boldsymbol{\chi}\in \mathbb{V}_{N,d}} A_{\boldsymbol{\chi}} = d^N  I,
    \label{eq:completeness_of_phase_point_operators}
  \end{equation}
  where $I$ denotes the identity operator on $\mH$.

  \item \textbf{Orthogonality:}  
  The phase space operators form an orthogonal basis for the space of operators on $\mH$:
  \begin{equation}
    \Tr\Big[A_{\boldsymbol{\chi}}  A_{\boldsymbol{\chi^\prime}}\Big] = d^N  \delta_{\boldsymbol{\chi},\boldsymbol{\chi^\prime}}.
  \end{equation}
  
  \item \textbf{Covariance:}  
  Under displacements, the phase space operators transform in a covariant manner.
  That is, for any displacement operator $T_{\boldsymbol{\chi}}$ one has
  \begin{equation}
    T_{\boldsymbol{\chi}}  A_{\boldsymbol{\chi^\prime}}  T^\dagger_{\boldsymbol{\chi}} = A_{\boldsymbol{\chi} + \boldsymbol{\chi^\prime}}.
  \end{equation}
  So one can use an equivalent defenition of the phase-space operators:
    \begin{equation}
    A_{\boldsymbol{\chi}} = T_{\boldsymbol{\chi}}  A_{\boldsymbol{0}}  T^\dagger_{\boldsymbol{\chi}},
  \end{equation}
  where $ A_{\boldsymbol{0}} $ is defined as
  \begin{equation}
    A_{\boldsymbol{0}} = \frac{1}{d^N} \sum_{\boldsymbol{\chi}\in \mathbb{V}_{N,d}} T_{\boldsymbol{\chi}}.
  \end{equation}
  
  \item \textbf{Additional properties:}
  They satisfy also these relations:

  \begin{equation}
    A_{\boldsymbol{\chi}} = T_{2\boldsymbol{\chi}} A_{\boldsymbol{0}} = A_{\boldsymbol{0}} T^\dagger_{2\boldsymbol{\chi}},
  \end{equation}
    \begin{equation}
    A_{\boldsymbol{\chi_1}}A_{\boldsymbol{\chi_2}} = \omega^{2\langle \boldsymbol{\chi_1},\boldsymbol{\chi_2} \rangle} T_{2(\boldsymbol{\chi_1}-\boldsymbol{\chi_2})},
  \end{equation}
      \begin{equation}
    A_{\boldsymbol{\chi_1}}A_{\boldsymbol{\chi_2}}A_{\boldsymbol{\chi_3}} = \omega^{2(\langle \boldsymbol{\chi_1},\boldsymbol{\chi_2} \rangle+\langle \boldsymbol{\chi_2},\boldsymbol{\chi_3} \rangle+\langle \boldsymbol{\chi_3},\boldsymbol{\chi_1} \rangle)} A_{\boldsymbol{\chi_1}+\boldsymbol{\chi_2}-\boldsymbol{\chi_3}},
  \end{equation}
  
\end{itemize}

\subsection{Stabilizer States and Clifford Operations}

A significant subset of quantum states in the finite-dimensional setting are the \emph{stabilizer states} \cite{StabilizerForHigh}.
These are defined as simultaneous eigenstates associated to $+1$ eigenvalue of a maximal Abelian subgroup $\mS$ of the generalized Pauli group $\mP_{N,d}$, being maximal when $|\mS|=d^N$.
In prime dimensions, stabilizer states have the notable property that their discrete Wigner function is non-negative, which is intimately connected to their efficient classical simulability.

More formally~\cite{Gross2007Hudson}, a subset $ M \subset \mathbb{V}_{N,d} $ is called \textit{isotropic} if $ \langle \boldsymbol{\chi}_1,\boldsymbol{\chi}_2 \rangle = 0 $ for all $ \boldsymbol{\chi}_1,\boldsymbol{\chi}_2 \in M $.
The subset $ M $ is called \textit{maximally isotropic} if it is isotropic and has cardinality $ d^N $.
For a maximally isotropic subspace $ M \subset \mathbb{V}_{N,d} $, and for a vector $ \boldsymbol{\chi} \in \mathbb{V}_{N,d} $, up to a global phase, there is a unique state vector $ |M, \boldsymbol{\chi} \rangle $ that satisfies the eigenvalue equation \[ \omega^{\langle \boldsymbol{\chi}, \boldsymbol{\chi'} \rangle}   T(m)   |M, \boldsymbol{\chi} \rangle = |M, \boldsymbol{\chi} \rangle \quad \forall \boldsymbol{\chi'} \in M. \]
The state vector $ |M, \boldsymbol{\chi} \rangle $ is called the stabilizer state associated to $ M$ and $ \boldsymbol{\chi} $.

The set of all pure stabilizer states of $ N $ qudits is denoted by $ \mSS_{N,d} $. 
The convex hull of these states, denoted $ \mathrm{STAB}_{N,d} $, consists of all mixed stabilizer (density) states, whose extreme points are exactly the pure stabilizer states. 
Geometrically, $ \mathrm{STAB}_{N,d} $ forms a highly symmetric convex polytope- often referred to as the \emph{Wigner polytope}- whose vertices correspond to pure stabilizer states. 
For a single qubit, this reduces to the familiar octahedron of stabilizer states, and in general, the Clifford group acts multiply transitively on the polytope’s vertices.

\subsubsection*{Clifford Operations and Symplectic Permutations}

\emph{Clifford operations} are unitary transformations $ C \in \mathrm{U}(\mH) $ that map the generalized Pauli group onto itself under conjugation, i.e.,
\begin{equation}
    C P C^\dagger \in \mP_{N,d}, \quad \forall  P \in \mP_{N,d}.
\end{equation}
The set of all such operations forms a group, denoted in this work by $ \mC_{N,d} $, known as the \emph{Clifford group} on $ N $ qudits.

Since overall phase factors $e^{i\theta} I$ act trivially on all physical states and on the elements of $\mP_{N,d}$, it is natural to quotient them out.
The resulting quotient group is called the \emph{reduced Clifford group}:
\begin{equation}
  \overline{\mC}_{N,d} = \mC_{N,d} / U(1),
\end{equation}
where two Clifford unitaries that differ only by a global phase are identified.

Another convenient convention for defining a finite version of the Clifford group is to regard it as a subgroup of $\mathrm{SU}(\mH)$ rather than of $\mathrm{U}(\mH)$. 
We denote this subgroup by $\tilde{\mC}_{N,d}$.
In this work, we define the finite extension of the special Clifford group as
\begin{equation}
    \mC'_{N,d}
    = \Bigl\{  \lambda  C  \; \big| \; 
    C \in \tilde{\mC}_{N,d}, 
    \lambda \in \Lambda(\tilde{\mC}_{N,d})
    \Bigr\},
\end{equation}
where $\Lambda(\tilde{\mC}_{N,d})$ denotes the set of all eigenvalues of all elements in $\tilde{\mC}_{N,d}$.
We refer to the group $\mC'_{N,d}$ as the \emph{eigenphase-extended Clifford group}, obtained by adjoining to the special Clifford group $\tilde{\mC}_{N,d} \subset \mathrm{SU}(\mH)$ all possible scalar multiples corresponding to eigenvalues of its elements.
This construction guarantees closure under scalar multiplication by all phases appearing in the Clifford spectrum, while preserving finiteness as an extension of $\tilde{\mC}_{N,d}$.
Moreover, it simplifies subsequent analysis by enabling us to consider only the stabilized subspaces associated with its subgroups.

A key geometrical insight into the structure of Clifford operations is that they induce \emph{symplectic permutations} on the discrete phase space.
More precisely, for any Clifford unitary $C$, there exists an associated symplectic matrix $S_C$ (with entries in $\mathbb{Z}_d$) and a vector $\boldsymbol{a}_C \in \mathbb{V}_{N,d}$ such that:
\begin{equation}
  C  T_{\boldsymbol{\chi}}  C^\dagger = T_{\boldsymbol{a}_C} T_{S_C\boldsymbol{\chi}}T^\dagger_{\boldsymbol{a}_C}=\omega^{-\langle \boldsymbol{a}_C,S_C\boldsymbol{\chi} \rangle} T_{S_C\boldsymbol{\chi}}.
\end{equation}
\noindent
A matrix $ S $ is said to be \emph{symplectic} if it satisfies the condition
\begin{equation}
  S^T J S = J,
\end{equation}
with $J$ being the standard symplectic form on $\mathbb{Z}_d^{2N}$. This condition ensures that the symplectic inner product is preserved, meaning that the geometric structure of the phase space is maintained under Clifford operations.
Another way to state the above is to say that a unitary operation $C$ is Clifford iff there exists an invertible symplectic matrix $S_C$ (with entries in $\mathbb{Z}_d$) and a vector $\boldsymbol{a}_C \in \mathbb{V}_{N,d}$ such that:
\begin{equation}
  C A_{\boldsymbol{\chi}} C^\dagger = A_{\boldsymbol{a}_C+S_C\boldsymbol{\chi}}.
\end{equation}
Actually, $\boldsymbol{a}$ and $S$ define the Clifford operation up to a global phase.
Therefore, there are $|\mathbb{V}_{N,d}|\cdot |\text{Sp}(\mathbb{Z}_d^{2N})| = d^{2N} \cdot d^{N^2} \prod_{i=1}^{N} (d^{2i} - 1)$ different Clifford operations for $N$ qudits of $d$-level systems, where $d$ is prime.

Equivalently, the Clifford group can be defined as the set of unitary operations that map stabilizer states onto stabilizer states. We refer the reader to standard references for its formalism, tools, and further insights~\cite{GibbonsFiniteFieldWigner,Gross2007Hudson,Veitch2012,Veitch_2014,Pashayan2015}.

\subsection{The Discrete Wigner Function}

One of the powerful applications of the phase space formalism is the representation of quantum density states via the discrete Wigner function \cite{GibbonsFiniteFieldWigner}.
For a state described by the density operator $\rho$, the Wigner function is defined as
\begin{equation}
  W_{\boldsymbol{\chi}}(\rho) = \frac{1}{d^N} \Tr\Big( A_{\boldsymbol{\chi}}  \rho \Big),
\end{equation}
where $A_{\boldsymbol{\chi}}$ are the phase point operators (defined via a symplectic Fourier transform of the displacement operators as in Eq.~(\ref{eq:phase_point_operator})) and $\boldsymbol{\chi} \in \mathbb{V}_{N,d}$ labels a point in the discrete phase space.
The discrete Wigner function $W_{\boldsymbol{\chi}}(\rho)$ acts as a quasi-probability distribution over the finite phase space, allowing one to examine both quantum interference effects and non-classicality.
Although it is analogous to a probability distribution, its possible negative values are a signature of uniquely quantum mechanical phenomena.

\subsubsection*{Properties and Identities}

The discrete Wigner function satisfies several important properties:

\begin{itemize}

  \item \textbf{Reality:}  
    Since the phase point operators $A_{\boldsymbol{\chi}}$ are Hermitian, the Wigner function is real-valued:
    \begin{equation}
      W_{\boldsymbol{\chi}}(\rho) \in \mathbb{R}.
    \end{equation}

  \item \textbf{Normalization:}  
    The Wigner function is normalized over the discrete phase space:
    \begin{equation}
      \sum_{\boldsymbol{\chi} \in \mathbb{V}_{N,d}} W_{\boldsymbol{\chi}}(\rho) = 1.
    \end{equation}
    This follows from the completeness relation of the phase point operators (Eq.~(\ref{eq:completeness_of_phase_point_operators})).

  \item \textbf{Trace:}  
    The expectation value of a Hermitian observable $O$ can be expressed as:
    \begin{equation}
      \Tr(\rho O) = d^N \sum_{\boldsymbol{\chi} \in \mathbb{V}_{N,d}} W_{\boldsymbol{\chi}}(\rho)  W_{\boldsymbol{\chi}}(O),
    \end{equation}
    where $O_W(\boldsymbol{\chi})$ is also defined by
      $W_{\boldsymbol{\chi}}(O) \equiv \Tr(O A_{\boldsymbol{\chi}})$.

  \item \textbf{State Reconstruction:}  
    The quantum state $\rho$ can be reconstructed uniquely from its Wigner function via:
    \begin{equation}
      \rho = \sum_{\boldsymbol{\chi} \in \mathbb{V}_{N,d}} W_{\boldsymbol{\chi}}(\rho)  A_{\boldsymbol{\chi}}.
    \end{equation}

  \item \textbf{Covariance under Translations:}  
    The Wigner function transforms covariantly under phase-space translations. If $\rho' = T_{\boldsymbol{\chi}_0}  \rho  T_{\boldsymbol{\chi}_0}^\dagger$ is the displaced state, where $T(\boldsymbol{\chi}_0)$ is a Heisenberg-Weyl operator, then
    \begin{equation}
      W_{\boldsymbol{\chi}}(\rho') = W_{\boldsymbol{\chi} - \boldsymbol{\chi}_0}(\rho).
    \end{equation}

  \item \textbf{Purity:}  
    The purity of a quantum state is related to the squared $L_2$-norm of its Wigner function:
    \begin{equation}
      \Tr(\rho^2) = d^N \sum_{\boldsymbol{\chi} \in \mathbb{V}_{N,d}} \left[ W_{\boldsymbol{\chi}}(\rho) \right]^2.
    \end{equation}
    For pure states, $\Tr(\rho^2) = 1$, which imposes a constraint on the spread of $W(\rho)$ over phase space.
\end{itemize}

These properties highlight the role of the discrete Wigner function as a bridge between operator-based and quasi-probabilistic formulations of quantum mechanics.
Importantly, the appearance of negative values in $W_{\boldsymbol{\chi}}(\rho)$ signals quantumness and nonclassicality and has been directly linked to the presence of non-stabilizersness and magic, a key computational resource in quantum information processing.

\subsection{The Reduced Clifford Group for Single Qudit}

\subsubsection{Generators}

The single-qudit Clifford group $\mC_{1,d}$ is defined as the set of operators that map the Heisenberg-Weyl displacement group to itself under conjugation:
\begin{equation}
    \mC_{1,d} = \{C \in \text{U}(d) \; \mid \; C \mP_{1,d} C^\dagger = \mP_{1,d}\}.
\end{equation}
As discussed before, for convenience, one can consider the Clifford group as a subgroup of $\text{SU}(d)$ by defining the generators such that they have determinant 1 (see \cite{Prakash2021}).
It can be shown that this group is generated by the unitaries $S$ and $H$, defined for $d > 2$ as:
\begin{equation}
    S = \sum_{j=0}^{d-1} \tau^{j(j+1)} |j\rangle \langle j|,
\end{equation}
\begin{equation}
    H = \frac{\delta_d}{\sqrt{d}} \sum_{j,k=0}^{d-1} \omega^{jk} |j\rangle \langle k|,
\end{equation}
where $\tau$ is as before: $\tau =\omega^{2^{-1}}=\omega^{(d+1)/2}$, with $2^{-1}$ denoting the inverse of 2 in $\mathbb{Z}_d$.
The phase factor $\delta_d$ is given by
\footnote{
Using known results in number theory for qudratic Gauss sums, an alternative expression for $\delta_d$ is given by
\begin{equation}
    \delta_d = \frac{1}{\sqrt{d}} \sum_{k=0}^{d-1} \omega^{2k^2} =\left( \frac{2}{d} \right)_L \epsilon_d,
\end{equation}
where $\left( \frac{2}{d} \right)_L$ denotes the Legendre symbol from number theory
\footnote{
$\left( \frac{a}{d}\right)_L=+1$ if $a$ is a quadratic residue modulo $d$, $\left( \frac{a}{d}\right)_L=-1$ if $a$ is not a quadratic residue modulo $d$, and $\left( \frac{a}{d}\right)_L=0$ if $d$ divides $a$.
}
.
This choice of overall phase ensures that $\det H = 1$.
Furthermore, for $d \geq 5$, one can verify that $\det S = 1$.
}
\begin{equation}
    \delta_d = \begin{cases}
    1, & d \equiv 1 \mod 8, \\
    -i, & d \equiv 3 \mod 8, \\
    -1, & d \equiv 5 \mod 8, \\
    i, & d \equiv 7 \mod 8. \\
    \end{cases}
\end{equation}
Observe that if $\det C = 1$, then $\det \omega^n C = 1$ for any $n$. Therefore, if we consider the Clifford group as a subgroup of $\text{SU}(d)$, we should also include the generator $\omega$ in the set of generators.

\subsubsection{$\mathbb{Z}_d^2 \rtimes \text{SL}(2,\mathbb{Z}_d)$ and Explicit Forms}

The reduced Clifford group is isomorphic to the semi-direct product of $\mathbb{Z}_d^2$ and $\text{SL}(2,\mathbb{Z}_d)$, as shown explicitly in \cite{SingleQuditCliffordIsomorphism}. Up to an overall phase, any Clifford element can be uniquely expressed as:
\begin{equation}
    C \sim T_{\boldsymbol{\chi}} V_{\hat{F}},
\end{equation}
where the equivalence relation $\sim$ denotes equality modulo a global phase, $\boldsymbol{\chi}\in\mathbb{Z}_d^2$ and $\hat{F} \in \text{SL}(2,\mathbb{Z}_d)$.
The special linear group $\text{SL}(2,\mathbb{Z}_d)$ consists of all $2\times2$ matrices over $\mathbb{Z}_d$ with determinant congruent to $1$ modulo $d$, and it is generated by:
\begin{equation}
    \hat{S} = \begin{bmatrix} 1 & 0 \\ 1 & 1 \end{bmatrix}, \quad \hat{H} = \begin{bmatrix} 0 & -1 \\ 1 & 0 \end{bmatrix}.
\end{equation}
$V:\text{SL}(2,\mathbb{Z}_d)\rightarrow \text{SU}(d)$ is the group homomorphism defined as follows (see \cite{Prakash2021}).
\begin{equation}
    V_{\begin{bmatrix} a & b \\ c & d \end{bmatrix}} =
    \begin{cases} 
        \left( \frac{-2b}{d} \right)_L \epsilon_d \frac{1}{\sqrt{d}} \sum_{j,k=0}^{d-1} \omega^{\frac{a k^2 - 2 j k + d j^2}{2b}} \ket{j} \bra{k}, & b \neq 0, \\
        \left( \frac{a}{d} \right)_L \sum_{k=0}^{d-1} \omega^{\frac{a c k^2}{2}} \ket{ak} \bra{k}, & b = 0.
    \end{cases}
\end{equation}
It is straightforward to verify that $H=V_{\hat{H}}$ and $S=T_{(0,2^{-1})}V_{\hat{S}}$.

This decomposition of Clifford operations satisfies the relation:
\begin{equation}
    T_{\boldsymbol{\chi}_1} V_{\hat{F}_1} T_{\boldsymbol{\chi}_2} V_{\hat{F}_2} \sim T_{\boldsymbol{\chi}_1 + \hat{F}_1 \boldsymbol{\chi}_2} V_{\hat{F}_1 \hat{F}_2}.
\end{equation}
Thus, up to an overall phase, any element of the Clifford group can be written in terms of these maps, allowing for efficient representations of Clifford transformations.

\subsection{Quantifying Magic}
\label{subsec:magic-measures}

As outlined in the Introduction, stabilizer states together with Clifford operations and Pauli measurements admit efficient classical simulation and thus fall short of universality.
The relevant computational resource is precisely a state’s deviation from the stabilizer framework- colloquially, its magic.
Formally, for $N$ qudits of prime dimension $d$, let $\mSS_{N,d}$ denote the set of pure stabilizer states and $\mathrm{STAB}_{N,d} := \operatorname{conv}  \big(\mSS_{N,d}\big)$ the stabilizer polytope (the convex hull of stabilizer states) in the space of density operators.
A state is magic iff it lies outside $\mathrm{STAB}_{N,d}$.

To quantify the amount of this resource, one introduces \emph{magic monotones}: functions that (i) vanish exactly on $\mathrm{STAB}_{N,d}$, (ii) are nonincreasing under stabilizer operations (Cliffords, Pauli measurements with classical feedforward, and stabilizer-state injections), and often (iii) satisfy additional regularity such as convexity.
Several such measures have been proposed, differing in operational meaning, computability, and geometric sensitivity to the boundary of $\mathrm{STAB}_{N,d}$~\cite{Veitch_2014,HowardContextuality,WinterYangCoherence,HowardCampell2017ResourceTheory,WangWilde2019thauma,WangThaumaBounds}. 

In what follows we adopt the notational conventions introduced above (namely, $\mathrm{STAB}_{N,d}$ and $\mSS_{N,d}$) and briefly recall the principal magic monotones employed in this work.
For pure states, several quantities play a central role: the \emph{stabilizer fidelity}, which quantifies the distance to the stabilizer set; the \emph{$\alpha$-stabilizer Rényi entropy}, which probes the spread of the state over the Pauli (or phase-space) degrees of freedom and refines the distinction between stabilizer and non-stabilizer states across Rényi orders; the \emph{generalized stabilizer entropies}, which extend the integer-order stabilizer Rényi entropies to more general qudit settings and scalable many-body contexts; and the \emph{mana of magic} (Wigner negativity), which characterizes nonclassicality through the appearance of negative values in the discrete Wigner function~\cite{GibbonsFiniteFieldWigner,Gross2007Hudson,Veitch2012,Veitch_2014,Pashayan2015}.

A common structural feature of these measures is that each can be expressed as a finite composition of elementary operations, namely finite sums, absolute values, and real powers, applied to finite families of $C^{\infty}$-smooth functions.
This analytic form, combined with the key axioms of magic measures (monotonicity under free operations, i.e.\ $\mathrm{STAB}_{N,d}$-preserving maps including Clifford unitaries, and faithfulness), underpins the arguments developed in this paper.
Analogous techniques may extend to other magic monotones, though such considerations fall outside the scope of this paper.

\subsubsection{Wigner negativity and mana}

The mana of magic is~\cite{Veitch2012,Veitch_2014,Pashayan2015}
\begin{equation}
    \mathcal{M}(\rho)
    := 
    \log \left(
        \sum_{\boldsymbol{\chi} \in \mathbb{V}_{N,d}}
        \left| W_{\boldsymbol{\chi}}(\rho) \right|
    \right),
\end{equation}
that is, the logarithm of the $1$-norm of the discrete Wigner representation of~$\rho$.
In odd prime or prime-power dimensions, the discrete Wigner function 
$W_{\boldsymbol{\chi}}(\rho)$ provides a quasiprobability representation of the quantum state~$\rho$ with the key property that
stabilizer states are exactly those whose Wigner function is 
non-negative on all phase-space points
$\boldsymbol{\chi}\in\mathbb{V}_{N,d}$.
Thus, for stabilizer states the Wigner function behaves like a classical probability distribution on the finite phase space.

For general states, negativity in the Wigner function is a hallmark of nonclassicality.
A basic quantity capturing this is the \emph{Wigner trace norm},
\begin{equation}
    \| \rho \|_W
    :=
    \sum_{\boldsymbol{\chi}\in\mathbb{V}_{N,d}}
    \left| W_{\boldsymbol{\chi}}(\rho) \right|.
\end{equation}
This norm equals $1$ precisely for stabilizer states.
Whenever negative values occur, the absolute-value sum exceeds~$1$, and the excess defines the
\emph{Wigner negativity},
\begin{equation}
    \mathcal{N}(\rho)
    :=
    \frac{1}{2}\left( \| \rho \|_W - 1 \right).
\end{equation}
Thus $\mathcal{N}(\rho)=0$ if and only if $\rho$ is a stabilizer state or a convex mixture of stabilizer states, and $\mathcal{N}(\rho)>0$ quantifies the ``amount'' of non-stabilizerness present.

The mana is then defined as the logarithmic refinement of this quantity:
\begin{equation}
    \mathcal{M}(\rho)
    =
    \log\| \rho \|_W.
\end{equation}
It is one of the simplest and most widely used magic monotones, possessing several desirable properties:
\begin{itemize}
    \item \textbf{Faithfulness:} 
    $\mathcal{M}(\rho)=0$ if and only if $\rho$ has a non-negative Wigner function, i.e., if and only if $\rho$ is a stabilizer state or mixture thereof.
    \item \textbf{Additivity:}
    $\mathcal{M}(\rho\otimes\sigma)
    =
    \mathcal{M}(\rho)+\mathcal{M}(\sigma)$.
    \item \textbf{Monotonicity under stabilizer operations:}
    Clifford unitaries, Pauli measurements, and stabilizer-preserving channels cannot increase mana.
    \item \textbf{Operational meaning:}
    Mana directly characterizes the sample complexity of quasiprobability Monte-Carlo simulation of quantum circuits with magic-state inputs~\cite{Pashayan2015}: 
    the variance of such simulations scales as $\| \rho \|_W^{2}$.
\end{itemize}

\noindent
In this way, mana provides a computable, geometrically transparent measure of quantum magic.
It quantifies precisely how far the Wigner representation of a state deviates from a classical probability distribution, through the total amount of Wigner negativity encoded in~$W(\rho)$.

\subsubsection{Stabilizer fidelity and the min-relative entropy of magic}

The stabilizer fidelity of a pure state~$\ket{\psi}$ is defined~\cite{Bravyi2019simulationofquantum} as
\begin{equation}
    F(\ket{\psi})
    :=
    \max_{\ket{s} \in \mSS_{N,d}}
        \bigl| \langle \psi | s \rangle \bigr|^{2},
\end{equation}
where $\mSS_{N,d}$ denotes the finite set of pure stabilizer states in $\mH_{N,d}$.
This quantity measures the closeness of $\ket{\psi}$ to the stabilizer polytope:
It is the squared overlap with the nearest stabilizer state.
Thus $F(\ket{\psi})=1$ if and only if $\ket{\psi}$ is itself a stabilizer state, and it becomes strictly smaller than~$1$ for any genuine magic state.

Geometrically, stabilizer fidelity quantifies how far $\ket{\psi}$ lies outside the convex hull of stabilizer states.
It can be interpreted as the cosine squared of the smallest angle between $\ket{\psi}$ and the discrete set of stabilizer rays.
From an operational viewpoint, large stabilizer fidelity indicates that $\ket{\psi}$ can be well-approximated by a stabilizer state, whereas small fidelity signals a highly non-stabilizer (high-magic) state.

Stabilizer fidelity is closely related to the \emph{min-relative entropy of magic}.
For pure states, one has the identity
\begin{equation}
    \mD_{\min}(\ket{\psi})
    :=
    - \log F(\ket{\psi}),
\end{equation}
so the min-relative entropy of magic reduces to the negative logarithm of the stabilizer fidelity.
This makes $-\log F(\ket{\psi})$ a valid magic monotone: it is faithful, additive, and monotonic under stabilizer operations.

For general mixed states $\rho$, the min-relative entropy of magic is defined~\cite{LiuBuTakagi2019} by
\begin{equation}
    \mD_{\mathrm{min}}(\rho)
    :=
    \min_{\sigma \in \mathrm{STAB}_{N,d}}
        D_{\min}(\rho \| \sigma),
\end{equation}
where $\mathrm{STAB}_{N,d}$ is the convex set of stabilizer states and
\begin{equation}
    D_{\min}(\rho \| \sigma)
    :=
    - \log \operatorname{tr}\left( \Pi_{\rho} \sigma \right),
\end{equation}
with $\Pi_\rho$ the projector onto the support of~$\rho$.
The min-relative entropy is a one-shot analogue of the Umegaki relative entropy, and plays an important role in one-shot magic-state distillation, where it characterizes the number of high-fidelity magic states that can be extracted from a finite number of noisy copies.

The stabilizer fidelity $F(\ket{\psi})$ thus serves as a simple, computable, and operationally meaningful measure of magic.  
It directly quantifies the ``distance'' of $\ket{\psi}$ from the stabilizer manifold, connects neatly to relative-entropy measures, and appears naturally in tasks such as:
\begin{itemize}
    \item bounding classical simulation cost of Clifford+$\ket{\psi}$ circuits,
    \item analyzing robustness under Clifford mixing,
    \item understanding one-shot distillation protocols,
    \item and characterizing faithfulness of magic monotones.
\end{itemize}
In particular, states with small stabilizer fidelity (such as high-level magic states) are highly resourceful for universal quantum computation, whereas states with fidelity close to~$1$ behave nearly classically with respect to stabilizer operations.

\subsubsection{Stabilizer Rényi Entropy}

A particularly useful quantifier of nonstabilizerness is the stabilizer Rényi entropy (SRE) \cite{SRE1,SRE2,SRE3,Liu2025}. 
It assigns Rényi entropies to the probability distribution formed by a state's overlaps with Pauli or, more generally, Weyl-Heisenberg operators. 
The SRE is faithful, additive, and nonincreasing under stabilizer operations, providing a computable and experimentally accessible measure of how far a state departs from the stabilizer manifold. 
We now present its formal definition.

For a given $N$-qudit pure state $\ket{\psi}\in\mH_{N,d}$, define
\begin{equation}
\label{eq:P_chi}
    P_{\boldsymbol{\chi}}  \left(\ket{\psi}\right)
    \equiv
    \frac{1}{d^N}
    \left|
        \bra{\psi} T_{\boldsymbol{\chi}} \ket{\psi}
    \right|^2
    =
    \frac{1}{d^N}
    \bra{\psi} T_{\boldsymbol{\chi}} \ket{\psi}
    \bra{\psi} T^\dagger_{\boldsymbol{\chi}} \ket{\psi},
    \quad 
    \boldsymbol{\chi}\in \mathbb{V}_{N,d},
\end{equation}
which forms a valid probability distribution over the discrete phase space.
It is clear from the definition that $P_{\boldsymbol{\chi}}  \left(\ket{\psi}\right)=P_{-\boldsymbol{\chi}}  \left(\ket{\psi}\right)$.
Using this distribution, one can define the $\alpha$-SRE as the $\alpha$-Rényi entropy of the probability vector:
\begin{equation}
    M_\alpha(\ket{\psi})
    \equiv
    \frac{1}{1-\alpha}
    \log  \left[
        \frac{1}{d^{(1-\alpha)N}}
        \sum_{\boldsymbol{\chi}} 
        P_{\boldsymbol{\chi}}^\alpha(\ket{\psi})
    \right]
    =
    \frac{1}{1-\alpha}
    \log  \left[
        \sum_{\boldsymbol{\chi}}
        P_{\boldsymbol{\chi}}^\alpha(\ket{\psi})
    \right]
    - N\log d,
\end{equation}
where $\alpha \ge 2$. The offset term $-N\log d$ ensures that $M_\alpha(\ket{\psi}) \ge 0$. It is customary to introduce the shorthand
\begin{equation}
    \Xi_\alpha(\ket{\psi})\equiv\sum_{\boldsymbol{\chi}}
        P_{\boldsymbol{\chi}}^\alpha(\ket{\psi}) ,
\end{equation}
so that the SRE can be expressed as $M_\alpha(\ket{\psi})= \frac{1}{1-\alpha} \log \Xi_\alpha (\ket{\psi}) - N\log d$.

\noindent
\emph{Remark:} To accommodate the qubit case in a uniform manner, it is convenient to replace the label of phase-space points by elements of the phase-reduced Pauli group $\tilde{\mP}_{N,d}$ and work with the quantities
\[
    P_{p}(\ket{\psi})
    \coloneqq
    \frac{1}{d^{N}}
    \bigl|\bra{\psi} p \ket{\psi}\bigr|^{2}
    =
    \frac{1}{d^{N}}
    \bra{\psi} p \ket{\psi}
    \bra{\psi} p^{\dagger} \ket{\psi},
    \qquad
    p \in \tilde{\mP}_{N,d}.
\]
These coefficients play the role of a ``Pauli-basis probability
distribution'' associated with the pure state~$\ket{\psi}$, and have the advantage that they are defined uniformly for all local dimensions $d$, naturally including $d=2$, without any special construction of the phase-space for this case.
With this notation, the SRE of order~$\alpha$ takes the form
\[
    M_{\alpha}(\ket{\psi})
    \coloneqq
    \frac{1}{1-\alpha}
    \log
    \left[
        \frac{1}{d^{(1-\alpha)N}}
        \sum_{p \in \tilde{\mP}_{N,d}}
        P_{p}(\ket{\psi})^{\alpha}
    \right].
\]
Equivalently,
\[
    M_{\alpha}(\ket{\psi})
    =
    \frac{1}{1-\alpha}
    \log
    \left[
        \sum_{p \in \tilde{\mP}_{N,d}}
        P_{p}(\ket{\psi})^{\alpha}
    \right]
    - N \log d,
\]
which highlights that the shift by $N\log d$ simply ensures that stabilizer states have vanishing SRE for all~$\alpha$.

These stabilizer $\alpha$-Rényi entropies, defined for pure states, satisfy key properties that make them meaningful from a resource-theoretic standpoint: they are faithful, invariant under Clifford operations, and additive under tensor products.  
A central result established in~\cite{SRE1,SRE2,SRE3} is that, for any state $\ket{\psi}\in\mH_{N,d}$ and any $\alpha\geq2$, the $\alpha$-SRE is bounded by
\begin{equation}
    M_\alpha(\ket{\psi})
    \leq
    \frac{1}{1-\alpha}
    \log  \left[
        \frac{1 + (d^N-1)(d^N+1)^{1-\alpha}}{d^N}
    \right],
    \label{eq: Bound on SRE}
\end{equation}
with equality if and only if $\ket{\psi}$ belongs to a Weyl-Heisenberg (WH) covariant symmetric informationally complete (SIC) positive operator-valued measure POVM~\cite{SRE3}.

For mixed states, the $2$-SRE can be generalized as~\cite{SRE1,SRE2}:
\begin{equation}
    \tilde{M}_2(\rho)
    =
    -\log  \left[
        \frac{
            \sum_{\boldsymbol{\chi}} 
            \left| \Tr  \left( T_{\boldsymbol{\chi}} \rho \right) \right|^4
        }{
            \sum_{\boldsymbol{\chi}} 
            \left| \Tr  \left( T_{\boldsymbol{\chi}} \rho \right) \right|^2
        }
    \right].
\end{equation}
This mixed-state extension preserves the same desirable features: faithfulness, Clifford invariance, and additivity- now extended to arbitrary density operators~\cite{SRE1,SRE2}.

\noindent
\emph{Remark.}
The $\alpha$-SREs are closely related to the $L^p$-norms of characteristic functions introduced in Refs.~\cite{Dai2022,Feng2022}.
More precisely, for a pure state $\ket{\psi}$ one has
\begin{equation}
\label{eq:SREandLpNorm}
    M_\alpha(\ket{\psi})
    =
    \frac{1}{1-\alpha}
    \log   \left[
        d^{-N}\, L_{2\alpha}^{\,2\alpha}
          \left(\ket{\psi}\bra{\psi}\right)
    \right],
\end{equation}
where the $L^p$-norm of the characteristic function of a quantum state $\rho$ is defined as
\begin{equation}
\label{eq:LpNorm}
    L_p(\rho)
    :=
    \left(
        \sum_{\boldsymbol{\chi}\in\mathbb{V}_{N,d}}
        \left|
            \Tr   \left(
                T_{\boldsymbol{\chi}}\,\rho
            \right)
        \right|^p
    \right)^{1/p}.
\end{equation}
For pure states, this expression admits the equivalent form
\begin{equation}
    L_p(\ket{\psi}\bra{\psi})
    =
    d^{N/2}
    \left(
        \sum_{\boldsymbol{\chi}\in\mathbb{V}_{N,d}}
        \left|
            P_{\boldsymbol{\chi}}(\ket{\psi})
        \right|^{p/2}
    \right)^{1/p},
\end{equation}
where $P_{\boldsymbol{\chi}}(\ket{\psi})$ is defined in Eq.~\eqref{eq:P_chi}.
As is evident from the above expressions, $M_\alpha$ is a strictly monotone increasing function of $L_{2\alpha}$.
Consequently, questions concerning extremality, ordering, or optimization of the $\alpha$-SRE are entirely equivalent to those for the corresponding $L^{2\alpha}$-norms.
In particular, for the purposes of the present work it suffices to study extremal properties of the stabilizer R\'enyi entropies, without loss of generality.

\subsubsection{Generalized Stabilizer R\'enyi Entropies}

The SREs of \textbf{integer} order $\alpha \ge 2$ form a family of smooth, Clifford-invariant quantifiers of non-stabilizerness on the projective Hilbert space.
For pure states, they admit the representation
\begin{equation}
    M_\alpha (\ket{\psi})
    =
    \frac{1}{1-\alpha}
    \log \left[
        \Tr \left(
            Q_{2\alpha}
            \bigl( \ket{\psi} \bra{\psi} \bigr)^{\otimes 2\alpha}
        \right)
    \right],
\end{equation}
where
\begin{equation}
    Q_{2\alpha}
    \equiv
    \frac{1}{d^N}
    \sum_{p \in \tilde{\mathcal{P}}_{N,d}}
    \left(
        p \otimes p^\dagger
    \right)^{\otimes \alpha}.
\end{equation}
The operators $Q_{2\alpha}$ belong to the commutant of the Clifford group acting on $2\alpha$ replicas,
i.e., the set of operators commuting with $U^{\otimes 2\alpha}$ for all Clifford unitaries $U$,
and satisfy
\[
    \Tr \left(
        Q_{2\alpha}\, \rho^{\otimes 2\alpha}
    \right)
    \le 1,
\]
with equality if and only if $\rho$ is a stabilizer state.

While the smooth integer-order SREs constitute a natural and physically motivated hierarchy, they represent only a particular instance of a more general construction.
Recently, a systematic framework has been developed in which magic measures arise from the
\emph{intrinsic Clifford commutant} acting on multiple replicas of the system~\cite{Turkeshi2025A,Turkeshi2025B}.
Concretely, let $k$ be an integer such that the Clifford commutant $\mathrm{Comm}_k(\mathcal{C}_{N,d})$ strictly contains the permutation algebra $\mathrm{Comm}_k(\mathcal{U}(d^N))$.
Elements in the difference
\begin{equation}
\label{eq:IntComm}
    \overline{\mathrm{Comm}}_k(\mathcal{C}_{N,d})
    :=
    \mathrm{Comm}_k(\mathcal{C}_{N,d})
    \setminus
    \mathrm{Comm}_k(\mathcal{U}(d^N))
\end{equation}
give rise to genuinely Clifford-sensitive invariants.

For any operator $W \in \overline{\mathrm{Comm}}_k(\mathcal{C}_{N,d})$, the associated \emph{generalized stabilizer purity} is defined by
\begin{equation}
\label{eq:GSP}
    \zeta_W(\ket{\psi})
    :=
    \Tr  \left[
        W\,
        \bigl(\ket{\psi}\bra{\psi}\bigr)^{\otimes k}
    \right],
\end{equation}
and the corresponding \emph{generalized stabilizer entropy (GSE)} by
\begin{equation}
\label{eq:GSE}
    M_W(\ket{\psi})
    :=
    -\ln \zeta_W(\ket{\psi}).
\end{equation}
By construction, $\zeta_W(\ket{\psi}) \le 1$, with equality if and only if $\ket{\psi}$ is a stabilizer state.
Consequently, $M_W(\ket{\psi}) \ge 0$, vanishing precisely on stabilizer states.
Moreover, each GSE satisfies:
(i) invariance under Clifford unitaries, and
(ii) additivity under tensor products,
thereby constituting a valid measure of magic for pure many-body states.

Within this general framework, the stabilizer R\'enyi entropies arise from the specific choice
$W = Q_{2\alpha}$, yielding
\begin{equation}
    M_\alpha(\ket{\psi})
    \;\propto\;
    M_{Q_{2\alpha}}(\ket{\psi}).
\end{equation}
In this sense, the SREs form a distinguished one-parameter subfamily of GSEs, characterized by Pauli-moment operators acting on an even number of replicas.
The broader GSE framework, however, admits additional Clifford-sensitive observables not captured by the SRE hierarchy, including measures defined on fewer replicas (e.g.\ $k=3$ for odd-prime qudits) and exhibiting qualitatively different sensitivity to non-stabilizer structure.

It is important to emphasize that SREs and GSEs are not related by set inclusion.
In particular, GSEs do not subsume the full family of SREs; they contain only those SREs corresponding to integer Rényi orders $\alpha$.
Equivalently, the classes of GSEs and SREs have a nonempty intersection, while each also contains elements not present in the other.

Finally, we note that the class of GSEs falls within the scope of our general extremality results.
In particular, for any choice of $W \in \overline{\mathrm{Comm}}_k(\mathcal{C}_{N,d})$, the associated GSE constitutes a valid instance of the functionals covered by our theorem.
Nevertheless, in the present work we do not attempt a systematic analysis of extremality across the full GSE family.
Instead, we restrict explicit calculations to representative and physically motivated examples, namely the stabilizer fidelity, the mana of magic, and the $2$-SRE, for which analytic control and physical interpretation are most transparent.

\section{Group-Covariant Functionals}
\label{sec:Group-Covariant Functionals}

\subsection{Mathematical Preliminaries and Notation}
\label{subsec:MathPreliminaries}

Throughout this section we fix all notations used in the remainder of the paper.
We work over a finite-dimensional Hilbert space~$\mH$ defined over the complex
field~$\mathbb{C}$, and collect here the basic constructions, conventions,
and differential-geometric tools required for our analysis.

\subsubsection{Functions on Hermitian Operators}

Let $ \mH $ be a finite-dimensional Hilbert space over the field of complex numbers~$\mathbb{C}$, with dimension $\dim \mH = D_\mH$ and orthonormal basis~$\mB_\mH$.
We denote by
$
    \mathrm{Herm}(\mH)
$
the real vector space of Hermitian operators acting on~$ \mH $, and by
\[
    \mathrm{End}(\mH)
    \coloneqq
    \{ O : \mH \rightarrow \mH \; \mid \; O \text{ is complex-linear} \}
\]
the complex vector space of all linear endomorphisms of~$ \mH $.
The group of unitary operators acting on~$\mH$ is denoted by $ \mathrm{U}(\mH) $.

The space $ \mathrm{Herm}(\mH) $ is naturally identified as a real, flat,
finite-dimensional manifold of dimension $ D_\mH^2 $.
A real- or complex-valued $(0,n)$-tensor on~$\mathrm{Herm}(\mH)$ is a multilinear
map of the form
\[
    F^{(n)} : \mathrm{Herm}(\mH)^n \longrightarrow \mathbb{R}
    \quad \text{or} \quad \mathbb{C}.
\]
Every such map admits an expansion in matrix elements:
\begin{equation}
\label{eq:TensorExpansionImproved}
    F^{(n)}(O_1,\ldots,O_n)
    =
    \sum_{\substack{
        \ket{\psi_1},\ket{\phi_1},\ldots, \\
        \ket{\psi_n},\ket{\phi_n}\in\mB_\mH
    }}
    f^{(n)}_{\ket{\psi_1},\ket{\phi_1},\ldots,\ket{\psi_n},\ket{\phi_n}}
    \prod_{k=1}^{n} \bra{\psi_k}O_k\ket{\phi_k},
\end{equation}
where the coefficients
$
    f^{(n)}_{\ket{\psi_1},\ket{\phi_1},\ldots,\ket{\psi_n},\ket{\phi_n}}
$
are fixed complex numbers determined by the choice of the orthonormal basis~$\mB_\mH$.
They are chosen so that $ F^{(n)} $ is real- or complex-valued according
to its declared type.  
We denote the resulting vector space of such tensors by
$
    \mathrm{T}^{(0,n)}\bigl(\mathrm{Herm}(\mH)\bigr).
$

\paragraph{Symmetric maps and tensors.}
A map 
$
    F^{(n)} : A^n \to B
$
is symmetric if it is invariant under all permutations of its arguments:
\[
    F^{(n)}(a_{\pi(1)},\ldots,a_{\pi(n)})
    =
    F^{(n)}(a_1,\ldots,a_n)
    \qquad
    \forall \pi \in S_n.
\]
The set of all symmetric $(0,n)$-tensors on $\mathrm{Herm}(\mH)$ is a linear subspace,
denoted
$
    \mathrm{T}^{(0,n)}_{\mathrm{sym}}\bigl(\mathrm{Herm}(\mH)\bigr).
$

\paragraph{Smoothness and analyticity.}
Smoothness and analyticity on~$\mathrm{Herm}(\mH)$ are understood in the standard sense of differential geometry.
A function $ f : \mathrm{Herm}(\mH) \to \mathbb{R} $ (or $ \mathbb{C} $) is said to be of class~$C^k$ if all its partial derivatives up to order~$k$ exist and are continuous in any local coordinate chart.
The space of infinitely differentiable functions is denoted by $ C^{\infty}(\mathrm{Herm}(\mH)) $.
A smooth function $ F \in C^{\infty}(\mathrm{Herm}(\mH)) $ is called analytic if it admits, around every point $ O \in \mathrm{Herm}(\mH) $, a convergent Taylor expansion of the form
\[
    F(O + \Delta)
    =
    \sum_{n=0}^{\infty}
    \frac{1}{n!}
    F^{(n)}_O(\Delta, \ldots, \Delta),
\]
where each $ F^{(n)}_O $ is a symmetric $(0,n)$-tensor representing the $n$th differential of $f$ at~$O$.

\subsubsection{Stabilization, Stabilizer Spaces, and Twirling}

We now record precise terminology related to stabilization by operators and by groups.
Although these notions are standard, we include them here for completeness and to fix notation.

\begin{definition}[Stabilization]
Let $ O \in \mathrm{End}(\mH_{N,d}) $.
A vector $ \ket{\psi} \in \mH_{N,d} $ is said to be \textbf{stabilized} by $O$ if
\[
    O\ket{\psi} = \ket{\psi}.
\]
Equivalently, we say that $O$ \emph{stabilizes}~$\ket{\psi}$.
\end{definition}

\begin{definition}[Stabilized Space of an Operator]
For any $ O \in \mathrm{End}(\mH_{N,d}) $, the \textbf{stabilized space} of~$O$ is
\[
    \mS(O)
    \coloneqq
    \bigl\{
        \ket{\psi}\in\mH_{N,d}
        \; \big| \;
        O\ket{\psi} = \ket{\psi}
    \bigr\}.
\]
Equivalently, $ \mS(O) $ is the eigenvalue-$1$ eigenspace of $O$.
\end{definition}

\begin{definition}[{\boldmath$G$}-Invariant Subspace]
Let $ G \subset \mathrm{U}(\mH_{N,d}) $ be a finite subgroup.
The \textbf{$G$-invariant subspace} (or \textbf{stabilized space of $G$}) is
\[
    \mS_G
    \coloneqq
    \bigl\{
        \ket{\psi} \in \mH_{N,d}
        \; \big| \;
        g\ket{\psi} = \ket{\psi}
         \forall g \in G
    \bigr\}
    =
    \bigcap_{g\in G} \mS(g).
\]
\end{definition}
\noindent
This subspace plays a central role in the structural analysis of many
group-covariant quantities.

Let $ P_{V} $ denote the orthogonal projector onto a subspace
$ V \subseteq \mH_{N,d} $.
For any finite subgroup $ G \subset \mathrm{U}(\mH_{N,d}) $, it is a
standard fact that the orthogonal projector onto the stabilized subspace
$
    \mS_{G}
$
is obtained by averaging over the group elements:
\footnote{Let $M_G = |G|^{-1}\sum_{g\in G} g$.  
One verifies:
\begin{inparaenum}[(i)]
\item $M_G^\dagger = M_G$ (Hermiticity);
\item $M_G^2 = M_G$ (idempotence);
\item $M_G\ket{\psi} = \ket{\psi}$ for all $\ket{\psi}\in\mS_G$;
\item for all $g\in G$, $g(M_G\ket{\psi}) = M_G\ket{\psi}$ (global $G$-invariance);
\item $M_G\ket{\psi} =0$ for all $\ket{\psi}\in\mS_G^\perp$.
\end{inparaenum}
Hence $M_G$ is the projector onto~$\mS_G$.}
\begin{equation}
\label{eq:ProjectorSG}
    P_{\mS_{G}}
    =
    \frac{1}{|G|}
    \sum_{g \in G} g .
\end{equation}

\begin{definition}[{\boldmath$G$}-Stabilizer Code]
Let $ G \subset \mathrm{U}(\mH_{N,d}) $ be a finite subgroup.
A \textbf{$G$-stabilizer code} is a stabilized subspace of a subgroup $ H \subset G $.
\end{definition}

\begin{definition}[{\boldmath$G$}-Stabilizer State]
Let $ G \subset \mathrm{U}(\mH_{N,d}) $ be a finite subgroup.
A \textbf{$G$-stabilizer state} is a normalized vector $ \ket{\psi}\in\mH_{N,d} $ whose density operator is the rank-one projector onto a one-dimensional $G$-stabilizer code.  
Equivalently, it is the unique (up to global phase) generator of a one-dimensional $G$-stabilizer code.
\end{definition}
\noindent
We denote by $ \mSS_G $ the finite set of all such states, modulo global phase.

The normalizer and centralizer of a subgroup $H \subset G$ are
\[
    \mathrm{N}_G(H)
    =
    \{ g\in G \; \mid \; gHg^{-1} = H \},
    \qquad
    \mathrm{C}_G(H)
    =
    \{ g\in G \; \mid \; gh = hg \, \, \, \forall h\in H \}.
\]
At the level of operators on $\mH_{N,d}$, the \emph{commutant}
of $G$ is defined as
\[
\mathrm{Comm}(G)
:=
\{\,X\in \mathrm{End}(\mH_{N,d}) : [X,g]=0 \ \forall g\in G\,\}.
\]
Note that the action of the unitary normalizer preserves the set of
$G$-stabilizer states:
\begin{equation}
\label{eq:NormalizerAction}
    C \cdot \mSS_G = \mSS_G
    \qquad
    \forall C \in \mathrm{N}_{\mathrm{U}(\mH_{N,d})}(G).
\end{equation}

\begin{definition}[Twirling]
\textbf{Twirling} an operator $O\in\mathrm{End}(\mH_{N,d})$ over a finite subgroup $ G \subset \mathrm{U}(\mH_{N,d}) $ is the process of averaging its conjugation over all elements of $G$:
    \begin{equation}
        O   \to   \mT_G(O)=\frac{1}{|G|} \sum_{g \in G} g O g^{-1}   .
    \end{equation}
\end{definition}
\noindent
It is a standard consequence of Schur’s lemma that the twirling map $\mT_G$ is the \emph{unique Hilbert-Schmidt orthogonal projection} from $\mathrm{End}(\mH_{N,d})$ onto the commutant $\mathrm{Comm}(G)$.

\subsection*{A special case: the Pauli group and its normalizer, the Clifford group}

There exist many examples of finite groups of unitaries together with their corresponding stabilizer states. 
A simple illustrative case is the single-qubit group 
\[
\langle e^{i\pi X / 3},  e^{i\pi Z / 2} \rangle,
\]
where $X$ and $Z$ are the standard Pauli operators. 
This group has order $12$, and its stabilizer states (up to an overall global phase) are
\[
\ket{0},  \ket{1},  \ket{\pm},  
\frac{\ket{0} \pm i\sqrt{3}\ket{1}}{2},  
\frac{\sqrt{3}\ket{0} \pm i\ket{1}}{2}.
\]
The general results derived in this work apply to this example as well.

Although our results hold for arbitrary finite groups, it is instructive to connect them now, in advance, to the framework of quantum magic. 
When the finite group $G$ is taken to be the Pauli group, the corresponding $G$-stabilizer states (or \emph{Pauli-stabilizer states}) coincide with the familiar stabilizer states for $N$ qudits commonly discussed in the literature. 
Moreover, any subgroup of the normalizer of the Pauli group, that is, of the Clifford group $\mC_{N,d}$, is, in our convention, a subgroup of the finite eigenphase-extended Clifford group $\mC^\prime_{N,d}$ introduced before.
Stabilizer states of this group are termed \emph{Clifford stabilizer states}.

\subsection{Group-Covariant Functionals}

The Clifford invariance of the magic measures and the ability to project operators by twirling them with a subgroup of the Clifford group will play crucial roles, serving as fundamental cornerstones of our arguments.
It turns out that the results can be generalized for finite subgroups of $\mathrm{U}(\mH)$. Therefore, we start with general definitions and facts.

\begin{definition}[Group Covariance]
\label{def:GroupCovariance}
Let $ G \subset \mathrm{U}(\mH) $ be a subgroup, let $ \mathbb{F} $ be a target set, and let $ \mV $ be a finite indexing set.
A family of functions
\[
    F^{(n)}_{v} :
    \mathrm{Herm}(\mH)^n \longrightarrow \mathbb{F},
\]
indexed by labels $v \in \mV$, is said to be \emph{$G$-covariant under the induced action on~$\mV$} if, for every $ g \in G $, there exists a invertible mapping $ \pi_g : \mV \to \mV $ such that
\begin{equation}
\label{eq:GroupCovariance}
    F^{(n)}_{v}\left(g^\dagger O_1 g, \ldots, g^\dagger O_n g\right)
    =
    F^{(n)}_{\pi_g(v)}(O_1, \ldots, O_n),
\end{equation}
for all Hermitian operators $O_1, \ldots, O_n \in \mathrm{Herm}(\mH)$.
The correspondence $ g \mapsto \pi_g $ is referred to as the \emph{$G$-covariance rule on~$\mV$} associated with the family~$F^{(n)}$.
\end{definition}
\noindent
In fact, when $\mathbb{F}$ is a field, the collection
$\{ F^{(n)}_{v} \}_{v \in \mV}$ may be viewed as the components of a
single vector-valued map
\[
    F^{(n)} : \mathrm{Herm}(\mH)^n \longrightarrow \mathbb{F}^{|\mV|}.
\]
Under this identification, the covariance condition~\eqref{eq:GroupCovariance} takes the compact form
\begin{equation}
    F^{(n)} \left(g^\dagger O_1 g,\ldots,g^\dagger O_n g\right)
    =
    P(g)\, F^{(n)}(O_1,\ldots,O_n),
\end{equation}
where $P : G \to \mathrm{GL}(|\mV|,\mathbb{R})$ is a \emph{permutation representation} of~$G$ on~$\mathbb{F}^{|\mV|}$ (associated with the action $g \mapsto \pi_g$).
For the purposes of this work, however, it is more convenient to retain the formulation in Definition~\ref{def:GroupCovariance}.

For the sake of clarity and logical consistency, we now formalize the notion of an \emph{extremal point}, so that the reader is not left with any ambiguity regarding the conditions under which a point is considered extremal.
\begin{definition}[Extremal Point]
\label{def:extremal_point}
Let $ f : \mathbb{R} \to \mathbb{R} $ be a continuous function.  
A point $ x^{*} \in \mathbb{R} $ is called an \emph{extremal point} of $ f $ if and only if at least one of the following holds:
\begin{enumerate}
    \item $ f $ is differentiable at $ x^{*} $ and $ f'(x^{*}) = 0 $;
    \item There exists $ \epsilon > 0 $ such that
    \[
        f(x) \ge f(x^{*}) \quad \text{for all } x \in (x^{*} - \epsilon, x^{*} + \epsilon),
    \]
    in which case $ x^{*} $ is a \emph{local maximum};
    \item There exists $ \epsilon > 0 $ such that
    \[
        f(x) \le f(x^{*}) \quad \text{for all } x \in (x^{*} - \epsilon, x^{*} + \epsilon),
    \]
    in which case $ x^{*} $ is a \emph{local minimum};
    \item There exists $ \epsilon > 0 $ such that
    \[
        f(x) = f(x^{*}) \quad \text{for all } x \in (x^{*} - \epsilon, x^{*} + \epsilon),
    \]
    in which case $ x^{*} $ is a \emph{locally constant point}.
\end{enumerate}
\end{definition}

Now we are ready to state the main theorem.
\begin{theorem}
\label{theo:MainTheorem}
Let $ G \subset \mathrm{U}(\mH) $ be a finite subgroup, and denote its stabilized subspace by $ \mS_{G} $, with orthogonal complement $ \mS_{G}^{\perp} $.
Fix a normalized vector $ \ket{\psi} \in \mS_{G} $, and write
$
    \psi \coloneqq \ket{\psi}\bra{\psi}
$
for its associated rank-one projector.  
Consider the submanifold of rank-one, trace-one operators whose support is contained in
$
    \Span\{\ket{\psi}\} \oplus \mS_{G}^{\perp} \subset \mH,
$
namely,
\[
    \varrho_{\psi;G}
    \coloneqq
    \Bigl\{
        \rho \in \varrho(\mH)
        \; \big| \;
        \operatorname{supp}(\rho)
        \subset
        \Span\{\ket{\psi}\} \oplus \mS_{G}^{\perp}
    \Bigr\},
\]
where $ \varrho(\mH) $ denotes the manifold of pure-state density operators on~$\mH$.

Let
\[
    F_{v} :
    \mathrm{Herm}(\mH) \longrightarrow \mathbb{C},
    \qquad v \in \mV,
\]
be a family of complex-valued functions that is $G$-covariant under the induced action of $G$ on the finite index set~$\mV$.
Let 
\[
    \mathrm{v} : \{1,\ldots,|\mV|\} \longrightarrow \mV
\]
be a fixed bijective enumeration of~$\mV$.  

Then the following hold.

\begin{enumerate}[label=(\roman*)]

\item
\textbf{Extremality for symmetric functionals.}
If each $F_{v}$ is of class~$C^{1}$, then
\[
    \Sigma(O)
    \coloneqq
    S\bigl( F_{\mathrm{v}(1)}(O),\ldots,F_{\mathrm{v}(|\mV|)}(O) \bigr),
\]
with $ S:\mathbb{C}^{|\mV|}\to\mathbb{C} $ any symmetric $C^{1}$ function, has $ \psi $ as an extremal point when restricted to the manifold $ \varrho_{\psi;G} $.

\item
\textbf{Extremality for the maximized functional.}
If each $F_{v}$ is real-valued and of class~$C^{\infty}$, then the maximally lifted functional
\[
    \mF(O)
    \coloneqq
    \max_{v \in \mV} F_{v}(O)
\]
also has $ \psi $ as an extremal point when restricted to $ \varrho_{\psi;G} $.

\item
\textbf{Extremality for the semi-Rényi functional.}
If each $F_{v}$ is of class~$C^{\infty}$, then for every $ \alpha>0 $ the semi-Rényi functional
\[
    \mR_{\alpha}(O)
    \coloneqq
    \sum_{v\in\mV} \bigl|F_{v}(O)\bigr|^{\alpha}
\]
admits $ \psi $ as an extremal point when restricted to $ \varrho_{\psi;G} $.

\end{enumerate}
\end{theorem}
\noindent
\emph{Remark.}
Note that whenever $ F_v(\phi) \ge 0 $ for all $v \in \mV$ and all $\phi$ in the domain, $\mF(O)$ coincides with the limiting functional
\[
    \lim_{\alpha \to \infty} \sqrt[\alpha]{\mR_{\alpha}}
    =
    \max_{v \in \mV} |F_v|
    =
    \max_{v \in \mV} F_v = \mF.
\]

\begin{proof}[Proof of Theorem~\ref{theo:MainTheorem}]
We work on the manifold of rank-one Hermitian operators with unit trace, 
parametrized by projectors $ \phi = \ket{\phi}\bra{\phi} $,
which form the projective Hilbert manifold.  
Consider a smooth variation of $ \ket{\psi} $ within the subspace 
$\Span\{\ket{\psi}\} \oplus \mS_{G}^{\perp}$
toward a normalized orthogonal state $ \ket{\varphi} \perp \ket{\psi} $.
Any variation containing a component along~$\ket{\psi}$ can, after appropriate rescaling and rephasing, 
be reduced to such an orthogonal variation.
Hence we may write
\begin{subequations}
\begin{gather}
    \ket{\psi(\epsilon)} = 
    \frac{\ket{\psi} + \epsilon \ket{\varphi}}{\sqrt{1+\epsilon^2}}, \\ \psi(\epsilon)
    = \ket{\psi(\epsilon)}\bra{\psi(\epsilon)}
    = \psi + \frac{\epsilon}{1+\epsilon^2} \sigma
      + \frac{\epsilon^2}{1+\epsilon^2} \mu,
\end{gather}
\label{eq:stateWithVariation}
\end{subequations}
where
\begin{subequations}
\begin{gather}
    \sigma = \ket{\varphi}\bra{\psi} + \ket{\psi}\bra{\varphi},\\
    \mu     = \varphi - \psi = \ket{\psi} \bra{\psi} - \ket{\varphi} \bra{\varphi}.
\end{gather}
\end{subequations}

\textbf{Proof of (i).}
For every $O\in\mathrm{Herm}(\mH)$, denote by $ \Sigma^{(n)}_O $ the symmetric $(0,n)$-tensor representing the $n$th differential of $\Sigma$ at $O$.
We need to show that
\[
\left.\frac{d\Sigma(\psi(\epsilon))}{d\epsilon}\right|_{\epsilon=0}
= \Sigma^{(1)}_\psi(\sigma) = 0.
\]

Since $ \Sigma(O) \equiv S(F_{\mathrm{v}(1)}, \ldots, F_{\mathrm{v}(|\mV|)}) $ 
is a symmetric function of its arguments, and each component $ F_{v(i)} $ transforms covariantly under conjugation $ O \mapsto g^\dagger O g $ for any $ g \in G $, 
the arguments of $ S $ merely undergo a permutation induced by $ g $.  
Hence, the overall value of $ \Sigma $ remains unchanged under conjugation by any element of~$G$.
Explicitly, for every $ g \in G $,
\begin{equation}
\begin{split}
    \Sigma(g^\dagger O g)
    &= S\bigl(F_{\mathrm{v}(1)}(g^\dagger O g), \ldots, F_{\mathrm{v}(|\mV|)}(g^\dagger O g)\bigr) \\[3pt]
    &= S\bigl(F_{\pi_g(\mathrm{v}(1))}(O), \ldots, F_{\pi_g(\mathrm{v}(|\mV|))}(O)\bigr) \\[3pt]
    &= S\bigl(F_{\mathrm{v}(\tilde{\pi}_g(1))}(O), \ldots, F_{\mathrm{v}(\tilde{\pi}_g(|\mV|))}(O)\bigr) \\[3pt]
    &= \Sigma(O),
\end{split}
\end{equation}
where $ g \mapsto \pi_g $ denotes the $G$-covariance rule of the family $F$ acting on~$\mV$,
and $ \tilde{\pi}_g : \{1,\ldots,|\mV|\} \to \{1,\ldots,|\mV|\} $ is the permutation induced by~$\pi_g$ 
on the index set labeling the elements of~$\mV$,
explicitly given by $ \tilde{\pi}_g \coloneqq \mathrm{v}^{-1} \circ \pi_g \circ \mathrm{v}. $

Because $ g^\dagger \psi g = \psi $ for every $ g \in G $,
it follows that
\begin{equation}
    \Sigma(\psi + \Delta) = \Sigma(\psi + g^\dagger \Delta g)
    \qquad \forall \Delta \in \mathrm{Herm}(\mH),
\end{equation}
and hence, by differentiating at $ \psi $,
\begin{equation}
    \Sigma^{(n)}_\psi(\Delta) = \Sigma^{(n)}_\psi(g^\dagger \Delta g)
    \qquad \forall n \in \mathbb{N}, \Delta \in \mathrm{Herm}(\mH).
\end{equation}
Averaging over the finite group~$G$ gives
\begin{equation}
    \Sigma^{(n)}_\psi(\Delta)
    = \frac{1}{|G|} \sum_{g \in G} \Sigma^{(n)}_\psi(g^\dagger \Delta g).
\end{equation}

For $ n=1 $, the first differential $ \Sigma^{(1)}_\psi $ is a $(0,1)$-tensor, so this relation implies
\begin{equation}
    \Sigma^{(1)}_\psi(\Delta)
    = \Sigma^{(1)}_\psi\left( \frac{1}{|G|} \sum_{g \in G} g^\dagger \Delta g \right).
\end{equation}
When $ \Delta = \sigma $, the twirling operation
$
\frac{1}{|G|} \sum_{g \in G} g^\dagger \sigma g
$
vanishes. Indeed,
\begin{equation}
    \frac{1}{|G|} \sum_{g \in G} g^\dagger \ket{\psi}\bra{\varphi} g = \frac{1}{|G|} \sum_{g \in G} \ket{\psi}\bra{\varphi} g =  \ket{\psi}\bra{\varphi} P_{\mS_G}=0 ,
\end{equation}
where $P_{\mS_G}$ is the projector onto the stabilized subspace $\mS_{G}$ of $G$.
Since $\ket{\varphi}$ is orthogonal to $\mS_G$, the expression vanishes.
Consequently,
\[
\Sigma^{(1)}_\psi(\sigma) = 0,
\]
which establishes the claim.
\hfill$\square$

\textbf{Proof of (ii).}
Define, for any $ O \in \mathrm{Herm}(\mH) $,
\begin{subequations}
\begin{align}
    \mV_{0,O} &\coloneqq \{ v \in \mV \; \mid \; F_v(O) = 0 \}, \\[3pt]
    \mV_{\bar{0},O} &\coloneqq \{ v \in \mV \; \mid \; F_v(O) \neq 0 \}.
\end{align}
\end{subequations}
These two disjoint subsets form a partition of~$\mV$.  
It follows immediately that for every $ g \in G $ inducing an action
\[
    g \mapsto \pi_g
\]
on~$\mV$ via the $G$-covariance rule of the family~$\{F_v\}_{v\in\mV}$,
the permutation~$\pi_g$ acts separately on the two subsets:
\[
    \pi_g(\mV_{0,O}) = \mV_{0,O},
    \qquad
    \pi_g(\mV_{\bar{0},O}) = \mV_{\bar{0},O}.
\]

The functional $\mR_\alpha(\psi(\epsilon))$ can be decomposed as
\begin{equation}
\label{eq:R_alpha}
    \mR_\alpha(\psi(\epsilon)) =
    \sum_{v \in \mV_{\bar{0},\psi}}
        \bigl| F_v(\psi(\epsilon)) \bigr|^{\alpha}
    +
    \sum_{v \in \mV_{0,\psi}}
        \bigl| F_v(\psi(\epsilon)) \bigr|^{\alpha}.
\end{equation}
For sufficiently small values of $ |\epsilon| $, the first sum can be rewritten as
\[
    \sum_{v \in \mV_{\bar{0},\psi}}
        \bigl( \sign[F_v(\psi)]  F_v(\psi(\epsilon)) \bigr)^{\alpha},
\]
which defines a smooth and analytic function of~$\epsilon$.
This contribution satisfies the assumptions of part~(i), and therefore,
according to the result established there, its first-order term
in~$\epsilon$ vanishes.

Whenever $ F_v(\psi) = 0 $, the dependence of $ F_v(\psi(\epsilon)) $ on~$\epsilon$ can be written in the form
\[
    F_v(\psi(\epsilon))
    =
    \epsilon  G_v(\epsilon; \psi, \sigma, \varphi),
\]
where $ G_v(\epsilon; \psi, \sigma, \varphi) $ is a smooth and analytic function of~$\epsilon$.
Consequently, the functional $\mR_\alpha(\psi(\epsilon))$ can be expressed as the sum of  
a smooth analytic term whose first-order contribution vanishes,
and a residual part of the form
\begin{equation}
\label{eq:Falpha_residual}
    \left| \epsilon \right|^{\alpha}
    \sum_{v \in \mV_{0,\psi}}
        \bigl|
            G_v(\epsilon; \psi, \sigma, \varphi)
        \bigr|^{\alpha},
\end{equation}
where each $ G_v(\epsilon; \psi, \sigma, \varphi) $ is smooth and analytic in~$\epsilon$.
Hence, the point $\epsilon = 0$ is an extremum of~$\mR_\alpha(\psi(\epsilon))$,
and every normalized state $ \ket{\psi} \in \mS_G $ constitutes an extremal point of~$\mR_\alpha$
when restricted to the submanifold
\[
    \varrho\bigl(\Span\{\ket{\psi}\} \oplus \mS_G^{\perp}\bigr).
\]
\hfill$\square$

\textbf{Proof of (iii).}

We now show that the point $\epsilon = 0$ is an extremum of the functional
\[
    \mF(\psi(\epsilon))
    \coloneqq
    \max_{v \in \mV} F_v\bigl(\psi(\epsilon)\bigr).
\]
For any operator $ O \in \mathrm{Herm}(\mH) $, let $ M_O \subseteq \mV $ denote the set of maximizing indices,
\begin{equation}
\label{eq:def_MO}
    M_O
    \coloneqq
    \bigl\{
        v \in \mV
        \; \big| \;
        F_v(O)
        = \max_{v' \in \mV} F_{v'}(O)
    \bigr\}.
\end{equation}
For sufficiently small values of $ |\epsilon| $, the indices $v \in \mV$ that can maximize
$ F_v(\psi(\epsilon)) $ are necessarily contained in~$ M_\psi $.

Let $ v \in M_\psi $.  
Since $ g^\dagger \psi g = \psi $ for all $ g \in G $, the $G$-covariance rule implies
$ \pi_g(v) \in M_\psi $.
Denote by $ o_v \subseteq \mV $ the orbit of $v$ under the induced action
$ \pi_g $ for $ g \in G $.
It follows that $ o_v \subseteq M_\psi $, and that the induced permutation~$\pi_g$
acts separately on the disjoint subsets $ o_v $ and $ \mV \setminus o_v $.

Using the result established in part~(i), we deduce that for every $ v \in M_\psi $,
the average of $ F_{v'}(\psi(\epsilon)) $ over all $ v' \in o_v $ satisfies
\begin{equation}
\label{eq:orbit_average}
    \frac{1}{|G|}
    \sum_{g \in G}
        F_{\pi_g(v)}\bigl(\psi(\epsilon)\bigr)
    =
    F_v(\psi)
    + O(\epsilon^2).
\end{equation}
Since this holds for all $ v \in M_\psi $, the orbit-stabilizer theorem implies
\begin{equation}
\label{eq:Mpsi_average}
    \frac{1}{|M_\psi|}
    \sum_{v \in M_\psi}
        F_v\bigl(\psi(\epsilon)\bigr)
    =
    F_v(\psi)
    + O(\epsilon^2).
\end{equation}

Because the average of the linear term in~$\epsilon$ vanishes,
two possibilities arise:
either the linear term vanishes individually for every $ v \in M_\psi $,
or it takes positive coefficients for some $v$ and negative ones for others.
In the latter case, the point $ \epsilon = 0 $ corresponds to a
\emph{sharp local minimum}.
Otherwise, if the linear term vanishes for all $ v \in M_\psi $, then
\begin{equation}
\label{eq:local_expansion}
    F_v\bigl(\psi(\epsilon)\bigr)
    =
    F_v(\psi)
    + O(\epsilon^2).
\end{equation}
Hence, for sufficiently small~$|\epsilon|$,
the maximizing element of~$\mV$ may change between the cases
$\epsilon > 0$ and $\epsilon < 0$
(possibly remaining the same).
In either situation, the right and left derivatives of~$\mF(\psi(\epsilon))$
coincide and both vanish, implying that the point~$\epsilon = 0$
is \emph{extremal along the path}.
Although the function $\mF(\psi(\epsilon))$ need not be smooth at this point,
the extremality of~$\epsilon = 0$ remains guaranteed.
\end{proof}

\paragraph{Notes on the theorem.}
\begin{itemize}

\item[]
\textbf{(1) The full-manifold case.}
If $ \dim \mS_{G} = 1 $, then
\[
    \varrho_{\psi;G} = \varrho(\mH),
\]
i.e., the constraint on the support disappears and the admissible variations span the entire real manifold of rank-one, trace-one operators on~$\mH$, of real dimension $2D_{\mH}-2$.

\item[]
\textbf{(2) Extension to mixed states in $\mS_{G}$.}
The theorem extends directly to mixed states supported on $\mS_{G}$ of the form
\[
    \rho = \sum_{i=1}^{\dim \mS_{G}} p_i \ket{\psi_i}\bra{\psi_i},
\]
where $ \{\ket{\psi_i}\}_{i=1}^{\dim \mS_{G}} $ is an orthonormal basis of the invariant subspace $\mS_{G}$, and
$ \{p_i\}_{i=1}^{\dim \mS_{G}} $ is a classical probability distribution.

We restrict to variations of the form
\[
    \ket{\psi_i}
    \longrightarrow
    \ket{\psi_i(\vec{\epsilon})}
    =
    \frac{
        \ket{\psi_i} + \epsilon_i\ket{\tilde{\varphi}_i}
    }{
        \sqrt{1 + \epsilon_i^2 \braket{\tilde{\varphi}_i|\tilde{\varphi}_i}}
    },
\]
with variation directions
\[
    \ket{\tilde{\varphi}_i}
    =
    \sum_{j=1}^{\dim \mS_{G}^{\perp}}
        f_{i,j}\ket{\varphi_j},
\]
where $ \{\ket{\varphi_j}\}_{j=1}^{\dim \mS_{G}^{\perp}} $ is an orthonormal basis of the orthogonal complement $\mS_{G}^{\perp}$, and $f_{i,j}\in\mathbb{C}$ are arbitrary coefficients.
Under these variations, the same argument as in the pure-state case shows that $\rho$ remains a critical point of the functional.

\item[]
\textbf{Linearity and $\boldsymbol{G}$-stabilizer fidelity.}
Assume that the family $ \{F_v\}_{v\in\mV} $ consists of rank-$(0,1)$ tensors on $ \mathrm{Herm}(\mH) $.
Then one has the expansion
\begin{equation}
    \frac{1}{|M_{\psi}|}
    \sum_{v \in M_{\psi}}
    F_v\bigl(\psi(\epsilon)\bigr)
    =
    F_v(\psi)
    +
    \frac{\epsilon^2}{1+\epsilon^2} \frac{1}{|M_{\psi}|} \sum_{v \in M_{\psi}}F_v\bigl(\mT_G(\mu)\bigr),
\end{equation}
so that, if the linear term vanishes for all $v\in M_\psi$, one obtains
\begin{equation}
    F_v\bigl(\psi(\epsilon)\bigr)
    =
    F_v(\psi)
    +
    \frac{\epsilon^2}{1+\epsilon^2}F_v(\mu).
\end{equation}

Consequently, the nature of the extremum at $\epsilon=0$ is governed by the signs of the quantities $F_v(\mu)$:
\begin{enumerate}
    \item If $F_v(\mu) < 0$ for all $v\in M_\psi$, then $\epsilon=0$ is a local \emph{smooth maximum}.
    \item If $F_v(\mu) \le 0$ for all $v\in M_\psi$, with equality for at least one $v$, then $\epsilon=0$ is a \emph{flat extremum}, i.e., the function is constant along the entire variation path.
    \item If there exists $v\in M_\psi$ with $F_v(\mu) > 0$, then $\epsilon=0$ is a local \emph{smooth minimum}.
\end{enumerate}
A particularly important instance of such functionals is the following fidelity-type quantity.

\begin{definition}
    Let $ G \subset \mathrm{U}(\mH) $ be a finite subgroup.  
    The \textbf{$G$-stabilizer fidelity} of a state $ \ket{\psi}\in\mH $ is defined by
    \begin{equation}
        F_G(\ket{\psi})
        =
        \max_{\ket{s} \in \mSS_G}
        \left|\braket{s|\psi}\right|^{2},
    \end{equation}
    where $ \mSS_G $ denotes the (finite) set of $G$-stabilizer states, modulo global phase.
\end{definition}
\noindent
\emph{Remark.}
For the $G$-stabilizer fidelity (and for the $G$-stabilizer extent defined in Appendix~\ref{app:group-stabilizer extent}) we do not impose the condition $\Span \mSS_G = \mH$ or any related spanning assumptions.

As an immediate consequence of Theorem~\ref{theo:MainTheorem} one obtains:
\begin{corollary}
    Let $ G \subset \mathrm{U}(\mH) $ be a finite subgroup, and let $ \mQ \subset \mathrm{N}_{\mathrm{U}(\mH)}(G) $ be any finite subgroup of the normalizer of $G$.  
    Then every normalized vector $ \ket{\psi} \in \mS_{\mQ} $ is an extremum of the $G$-stabilizer fidelity $F_G$ on the subspace
    $
        \Span\{\ket{\psi}\} \oplus \mS_{G}^{\perp}.
    $
\end{corollary}

\item[]
\textbf{(4) Group-stabilizer extent.}
At this point in the flow of the paper, it is worth briefly mentioning that one can also define a notion of group-stabilizer extent. While this concept is not directly central to the main goals of this work, it arises naturally within the broader discussion. Similar to how we generalized the stabilizer fidelity, one can extend other measures defined on pure states and derive new, insightful properties and statements. A detailed discussion of the group-stabilizer extent is provided in Appendix~\ref{app:group-stabilizer extent}.

\end{itemize}

\subsection{Componentwise Group-Covariant Functionals}

In many constructions, the label set of a covariant family naturally factorizes
into several components.  This subsection formalizes covariance properties that
respect such a product structure and establishes a stability result under
componentwise contractions.

\begin{definition}[Componentwise Map]
\label{def:ComponentwiseMap}
Let
\[
    \mV = \prod_{i=1}^{k} \mV_i
\]
be the Cartesian product of finite sets $ \{\mV_i\}_{i=1}^{k} $.
A map
\[
    s : \mV \longrightarrow \mV
\]
is said to be \emph{componentwise} (with respect to the decomposition
$\mV = \prod_{i=1}^{k} \mV_i$) if it acts independently on each factor, i.e.,
if there exist invertible maps
\[
    s_i : \mV_i \longrightarrow \mV_i,
    \qquad i = 1,\ldots,k,
\]
such that
\begin{equation}
\label{eq:componentwise}
    s(v_1,\ldots,v_k)
    =
    \bigl(
        s_1(v_1),
        \ldots,
        s_k(v_k)
    \bigr)
    \qquad
    \forall (v_1,\ldots,v_k)\in\mV.
\end{equation}
We call $(s_1,\ldots,s_k)$ the \emph{component maps} of $s$.
\end{definition}

\begin{definition}[Componentwise $G$-Covariant Family]
\label{def:ComponentwiseCovariant}
Let $G\subset \mathrm{U}(\mH)$ be a subgroup and let
\[
    \mV = \prod_{i=1}^{k} \mV_i
\]
be the Cartesian product of finite sets $ \{\mV_i\}_{i=1}^{k} $.
Consider a $G$-covariant family of real-valued functions
\[
    F^{(n)}_{v} :
    \mathrm{Herm}(\mH)^n \longrightarrow \mathbb{R},
    \qquad v\in\mV,
\]
equipped with a covariance rule
\[
    g \longmapsto \pi_g : \mV \longrightarrow \mV.
\]
The family $ \{F^{(n)}_{v}\}_{v\in\mV} $ is said to be
\emph{componentwise $G$-covariant} (with respect to the decomposition
$\mV = \prod_{i=1}^{k}\mV_i$)
if every map $ \pi_g $ is componentwise in the sense of
Definition~\ref{def:ComponentwiseMap}.
\end{definition}

\begin{proposition}[Induced Covariance Under Componentwise Contraction]
\label{prop:InducedCovariance}
Let $G\subset \mathrm{U}(\mH)$ be a subgroup, let
\[
    \mV = \prod_{i=1}^{k} \mV_i
\]
be the Cartesian product of finite sets $ \{\mV_i\}_{i=1}^{k} $,
and let $ \{F^{(n)}_{v}\}_{v\in\mV} $ be a componentwise $G$-covariant family (with respect to the decomposition
$\mV = \prod_{i=1}^{k}\mV_i$).
Let
\[
    r : \mV \longrightarrow \mV
\]
be an invertible map for which the following holds:
for every $g\in G$, there exists an invertible componentwise map
\[
    \tau_g : \mV \longrightarrow \mV
\]
such that the \emph{intertwining condition}
\begin{equation}
\label{eq:intertwining}
    \pi_g \circ r = r \circ \tau_g
\end{equation}
is satisfied.

Define an induced family indexed by the reduced label set
\[
    \widetilde{\mV}
    \coloneqq
    \prod_{i=1}^{k-1} \mV_i
\]
via the contraction
\begin{equation}
\label{eq:induced_function}
    \widetilde{F}^{(n)}_{(v_1,\ldots,v_{k-1})}(O_1,\ldots,O_n)
    \coloneqq
    \sum_{v_k \in \mV_k}
        F^{(n)}_{r(v_1,\ldots,v_k)}(O_1,\ldots,O_n).
\end{equation}
Then the induced family
$
    \{\widetilde{F}^{(n)}_{(v_1,\ldots,v_{k-1})}\}_{(v_1,\ldots,v_{k-1})\in\widetilde{\mV}}
$
is itself componentwise $G$-covariant under the separable action
\[
    (v_1,\ldots,v_{k-1})
    \longmapsto
    \bigl(\tau_{g,1}(v_1),\ldots,\tau_{g,k-1}(v_{k-1})\bigr),
\]
where $\tau_{g,i}$ are the component maps of $\tau_g$.
\end{proposition}

\begin{proof}[Proof of Proposition~\ref{prop:InducedCovariance}]
Fix $g\in G$.
Write the componentwise decomposition of $\tau_g$ as
\begin{equation}
\label{eq:qg_components}
    \tau_g(v_1,\ldots,v_k)
    =
    \bigl(
        \tau_{g,1}(v_1),
        \ldots,
        \tau_{g,k}(v_k)
    \bigr),
    \qquad
    \tau_{g,i} : \mV_i \to \mV_i \text{ invertible}.
\end{equation}
Using $G$-covariance of $F^{(n)}$ and the intertwining relation
\eqref{eq:intertwining}, we compute:
\begin{align}
\tilde{F}^{(n)}_{(v_1,\ldots,v_{k-1})}(g^\dagger O_1 g, \ldots, g^\dagger O_n g)
    &=
    \sum_{v_k \in \mV_k}
        F^{(n)}_{r(v_1,\ldots,v_k)}(g^\dagger O_1 g, \ldots, g^\dagger O_n g)
        \notag \\
    &=
    \sum_{v_k \in \mV_k}
        F^{(n)}_{\pi_g(r(v_1,\ldots,v_k))}(O_1, \ldots, O_n)
        \notag \\
    &=
    \sum_{v_k \in \mV_k}
        F^{(n)}_{r(\tau_g(v_1,\ldots,v_k))}(O_1, \ldots, O_n)
        \notag \\
    &=
    \sum_{v_k \in \mV_k}
        F^{(n)}_{r(\tau_{g,1}(v_1), \tau_{g,2}(v_2), \ldots, \tau_{g,k}(v_k))}(O_1, \ldots, O_n)
        \notag \\
    &=
    \tilde{F}^{(n)}_{(\tau_{g,1}(v_1), \tau_{g,2}(v_2), \ldots, \tau_{g,{k-1}}(v_{k-1}))}(O_1, \ldots, O_n).
    \label{eq:covariance_result}
\end{align}
Thus the induced family $\widetilde{F}^{(n)}$ transforms according to the
componentwise rule
\[
    (v_1,\ldots,v_{k-1}) \longmapsto
    \bigl(
        \tau_{g,1}(v_1),\ldots,\tau_{g,k-1}(v_{k-1})
    \bigr),
\]
establishing componentwise $G$-covariance.
\end{proof}

\subsection{Clifford-Covariance and Magic Measures}

\subsubsection{Notation for the rest of the paper}

From now on in the paper, we work throughout with the Hilbert space
\[
    \mH_{N,d} = (\mathbb{C}^d)^{\otimes N},
\]
corresponding to an $N$-qudit system in which each local subsystem (qudit) has a prime dimension~$d$.
The generalized Pauli group acting on~$\mH_{N,d}$, also known as the Weyl-Heisenberg group, 
is denoted by~$\mP_{N,d}$.
The Clifford group on~$\mH_{N,d}$, defined as the normalizer of~$\mP_{N,d}$ 
within the unitary group, is denoted by~$\mC_{N,d}$.

The set of all stabilizer states on~$\mH_{N,d}$ is denoted by~$\mSS_{N,d}$,
and their convex hull by~$\mathrm{STAB}_{N,d}$.
The latter consists of all mixed stabilizer (density) states, whose extreme points are exactly the pure stabilizer states.

We denote by~$\mathrm{Herm}(\mH_{N,d})$ the real vector space of Hermitian operators acting on~$\mH_{N,d}$,
and by
\[
    \mathrm{End}(\mH_{N,d})
     \equiv 
    \{  O:\mH_{N,d} \rightarrow \mH_{N,d} \; \mid \; O \text{ is linear}  \}
\]
the set of all linear endomorphisms on~$\mH_{N,d}$.
The group of unitary linear operators on~$\mH_{N,d}$ is denoted by~$\mathrm{U}(\mH_{N,d})$.

Recall that for any Clifford operation $C \in \mC_{N,d}$, there exist a unique invertible symplectic matrix $S_C \in \mathrm{Sp}(2N,\mathbb{Z}_d)$ and a unique phase-space displacement vector $\boldsymbol{a}_C \in \mathbb{V}_{N,d}$, where
\[
    \mathbb{V}_{N,d} \coloneqq \mathbb{Z}_d^{N} \times \mathbb{Z}_d^{N}
\]
denotes the discrete phase space, such that
\begin{subequations}
\label{eq:Clifford_action_on_TA}
\begin{gather}
    C  T_{\boldsymbol{\chi}}  C^\dagger
    =
    T_{\boldsymbol{a}_C}  T_{S_C \boldsymbol{\chi}}  T_{\boldsymbol{a}_C}^\dagger
    =
    \omega^{-\langle\boldsymbol{a}_C, S_C \boldsymbol{\chi}\rangle}
     T_{S_C \boldsymbol{\chi}},
    \\[0.4em]
    C  A_{\boldsymbol{\chi}}  C^\dagger
    =
    A_{\boldsymbol{a}_C + S_C \boldsymbol{\chi}},
\end{gather}
\end{subequations}
where $ T_{\boldsymbol{\chi}} $ are the Weyl-Heisenberg displacement operators and
$ A_{\boldsymbol{\chi}} $ are the associated phase-point operators.

Consequently, for any operator $O \in \mathrm{Herm}(\mH_{N,d})$,
the corresponding Wigner function transforms as
\begin{equation}
\label{eq:Clifford_covariance_of_Wigners}
    W_{C^\dagger O C}(\boldsymbol{\chi})
     = 
    W_O(\boldsymbol{a}_C + S_C \boldsymbol{\chi}),
\end{equation}
demonstrating that Clifford operations act as
\emph{affine symplectic transformations} on the discrete phase space.  
Up to an overall phase of $C$,
this establishes a one-to-one correspondence between elements of the Clifford group
and affine symplectic transformations~\cite{Gross2007Hudson}.

\begin{definition}[Associated Symplectic Matrix and Displacement Vector]
\label{def:AssociatedSymplecticDisplacement}
For every Clifford operation $C \in \mC_{N,d}$,
the unique pair consisting of the symplectic matrix
$S_C \in \mathrm{Sp}(2N,\mathbb{Z}_d)$
and the displacement vector
$\boldsymbol{a}_C \in \mathbb{V}_{N,d}$
satisfying
\begin{equation}
    C  T_{\boldsymbol{\chi}}  C^\dagger
    = T_{\boldsymbol{a}_C}  T_{S_C \boldsymbol{\chi}}  T^\dagger_{\boldsymbol{a}_C}
    = \omega^{-\langle \boldsymbol{a}_C,  S_C \boldsymbol{\chi} \rangle}
      T_{S_C \boldsymbol{\chi}},
\end{equation}
is called, respectively, the \emph{associated symplectic matrix} and
the \emph{associated displacement vector} of~$C$.
\end{definition}

\subsubsection{Symplectic Clifford-Covariant Families on the Phase-Space}

In many phase-space-based constructions, one naturally encounters families of operator-valued or scalar-valued functions whose indices label discrete phase-space points.
When such families arise from pointwise symplectic Clifford-covariant structures, as will be the case in the next subsubsection (Section~\ref{sec:Weyl-Heisenberg Correlation Functions}), it is natural to also consider their behavior under symplectic Fourier transforms on phase space.
To formalize this behavior, we introduce the notion of \emph{symplectically componentwise Clifford-covariant} families.
This notion captures the idea that Clifford conjugation of the operator arguments induces, for each phase-space index separately, an affine symplectic transformation of the corresponding label.

\begin{definition}[Symplectically componentwise Clifford-covariant families]
\label{def:SymplecticCliifordCovariance}
Let
\[
F_{\boldsymbol{\chi}_1,\ldots,\boldsymbol{\chi}_n} :
\mathrm{Herm}^k(\mH) \longrightarrow \mathbb{C},
\qquad
\boldsymbol{\chi}_i \in \mathbb{V}_{N,d},
\]
be a componentwise Clifford-covariant family of functions indexed by phase-space points.
We say that the family
$F_{\boldsymbol{\chi}_1,\ldots,\boldsymbol{\chi}_n}$
is \emph{symplectically componentwise Clifford-covariant} if, for every Clifford unitary
$C \in \mC_{N,d}$, there exist
\begin{itemize}
    \item symplectic matrices $\{ (S_C)_i \}_{i=1}^n \subset \mathrm{Sp}(2N,\mathbb{Z}_d)$, and
    \item phase-space translation vectors $\{ (\boldsymbol{a}_C)_i \}_{i=1}^n \subset \mathbb{V}_{N,d}$,
\end{itemize}
such that, for all $(O_1,\ldots,O_k) \in \mathrm{Herm}^k(\mH)$,
\begin{equation}
\label{eq:symplectic_componentwise_covariance}
F_{\boldsymbol{\chi}_1,\ldots,\boldsymbol{\chi}_n}
 \left(
C^\dagger O_1 C,\ldots,C^\dagger O_k C
\right)
=
F_{(S_C)_1 \boldsymbol{\chi}_1 + (\boldsymbol{a}_C)_1,\;\ldots,\;
(S_C)_n \boldsymbol{\chi}_n + (\boldsymbol{a}_C)_n}
 \left(
O_1,\ldots,O_k
\right).
\end{equation}
\end{definition}

An important structural question is whether this notion of covariance is preserved under standard transformations on phase-space-indexed families.
In particular, symplectic Fourier transforms play a central role in discrete phase-space analysis, relating complementary representations and often revealing invariant or extremal structures.

The following proposition shows that symplectically componentwise Clifford covariance is stable under such transformations in a precise sense.
While a symplectic Fourier transform may introduce index-dependent phase factors, these phases are purely character-valued and disappear upon taking appropriate absolute values, yielding strictly covariant families.

\begin{proposition}[Stability under symplectic Fourier transforms]
\label{prop:SymplecticFourierStability}
Let
\[
F_{\boldsymbol{\chi}_1,\ldots,\boldsymbol{\chi}_n} :
\mathrm{Herm}^k(\mH) \longrightarrow \mathbb{C},
\qquad
\boldsymbol{\chi}_i \in \mathbb{V}_{N,d},
\]
be a symplectically componentwise Clifford-covariant family of functions indexed by phase-space points, in the sense of Definition~\ref{def:SymplecticCliifordCovariance}.
Then, for any subset of indices
\[
\mathcal{I} \subseteq \{1,\ldots,n\},
\]
the family obtained by applying a symplectic Fourier transform with respect to each index
$\boldsymbol{\chi}_i$ for $i \in \mathcal{I}$ (denoted by
$\widetilde{F}_{\boldsymbol{\eta}_1,\ldots,\boldsymbol{\eta}_n}$)
is symplectically componentwise Clifford-covariant \emph{up to an index-dependent phase factor}.
In particular, the family of pointwise squared absolute values,
\[
\bigl\{\, \lvert \widetilde{F}_{\boldsymbol{\eta}_1,\ldots,\boldsymbol{\eta}_n} \rvert^{2} \,\bigr\}_{(\boldsymbol{\eta}_1,\ldots,\boldsymbol{\eta}_n)\in\mathbb{V}_{N,d}^{\,n}},
\]
is a symplectically componentwise Clifford-covariant family of functions indexed by phase-space points and depends smoothly on the original family $F$.
\end{proposition}

\begin{proof}[Proof of Proposition~\ref{prop:SymplecticFourierStability}]
Fix a subset of indices $\mathcal{I}\subseteq\{1,\ldots,n\}$. 
For notational convenience, write $\boldsymbol{\chi}=(\boldsymbol{\chi}_1,\ldots,\boldsymbol{\chi}_n)$ and let
$\boldsymbol{\eta}=(\boldsymbol{\eta}_1,\ldots,\boldsymbol{\eta}_n)$ denote the new index variables, where
$\boldsymbol{\eta}_i\in\mathbb{V}_{N,d}$ for $i\in\mathcal{I}$ and $\boldsymbol{\eta}_i=\boldsymbol{\chi}_i$ for $i\notin\mathcal{I}$.
Define the \emph{partial symplectic Fourier transform} of $F$ over the indices in $\mathcal{I}$ by
\begin{equation}
\label{eq:partial_symplectic_FT_def}
\widetilde{F}_{\boldsymbol{\eta}_1,\ldots,\boldsymbol{\eta}_n}(O_1,\ldots,O_k)
:=
\frac{1}{d^{N|\mathcal{I}|}}
\sum_{\{\boldsymbol{\chi}_i\in\mathbb{V}_{N,d}\,:\, i\in\mathcal{I}\}}
\omega^{-\sum_{i\in\mathcal{I}}\langle \boldsymbol{\eta}_i,\boldsymbol{\chi}_i\rangle}\,
F_{\boldsymbol{\chi}_1,\ldots,\boldsymbol{\chi}_n}(O_1,\ldots,O_k),
\end{equation}
where for $i\notin\mathcal{I}$ the variables $\boldsymbol{\chi}_i$ are held fixed and equal to $\boldsymbol{\eta}_i$.
We must show that $\widetilde{F}$ satisfies the covariance condition of
Definition~\ref{def:SymplecticCliifordCovariance}.

Let $C\in\mC_{N,d}$ be arbitrary. By the assumed symplectic componentwise Clifford covariance of $F$,
there exist symplectic matrices $(S_C)_i\in\mathrm{Sp}(2N,\mathbb{Z}_d)$ and translations $(\boldsymbol{a}_C)_i\in\mathbb{V}_{N,d}$
such that for all $(O_1,\ldots,O_k)\in\mathrm{Herm}^k(\mH)$,
\begin{equation}
\label{eq:F_covariance_used_in_proof}
F_{\boldsymbol{\chi}_1,\ldots,\boldsymbol{\chi}_n}(C^\dagger O_1 C,\ldots,C^\dagger O_k C)
=
F_{(S_C)_1\boldsymbol{\chi}_1+(\boldsymbol{a}_C)_1,\ldots,(S_C)_n\boldsymbol{\chi}_n+(\boldsymbol{a}_C)_n}(O_1,\ldots,O_k).
\end{equation}
Apply the definition \eqref{eq:partial_symplectic_FT_def} to the conjugated operators to obtain
\begin{equation}
\label{eq:start_transform_proof}
\begin{split}
\widetilde{F}_{\boldsymbol{\eta}_1,\ldots,\boldsymbol{\eta}_n} & (C^\dagger O_1 C,\ldots,C^\dagger O_k C)
=
\frac{1}{d^{N|\mathcal{I}|}}
\sum_{\{\boldsymbol{\chi}_i\,:\, i\in\mathcal{I}\}}
\omega^{-\sum_{i\in\mathcal{I}}\langle \boldsymbol{\eta}_i,\boldsymbol{\chi}_i\rangle}\,
F_{\boldsymbol{\chi}_1,\ldots,\boldsymbol{\chi}_n}(C^\dagger O_1 C,\ldots,C^\dagger O_k C) \\
&=
\frac{1}{d^{N|\mathcal{I}|}}
\sum_{\{\boldsymbol{\chi}_i\,:\, i\in\mathcal{I}\}}
\omega^{-\sum_{i\in\mathcal{I}}\langle \boldsymbol{\eta}_i,\boldsymbol{\chi}_i\rangle}\,
F_{(S_C)_1\boldsymbol{\chi}_1+(\boldsymbol{a}_C)_1,\ldots,(S_C)_n\boldsymbol{\chi}_n+(\boldsymbol{a}_C)_n}(O_1,\ldots,O_k),
\end{split}
\end{equation}
where in the second line we used \eqref{eq:F_covariance_used_in_proof}.

Now perform, for each $i\in\mathcal{I}$, the change of variables
\[
\boldsymbol{\chi}_i' := (S_C)_i\boldsymbol{\chi}_i + (\boldsymbol{a}_C)_i.
\]
Since $(S_C)_i\in\mathrm{Sp}(2N,\mathbb{Z}_d)\subseteq \mathrm{GL}(2N,\mathbb{Z}_d)$ is invertible over $\mathbb{Z}_d$,
this is a bijection of $\mathbb{V}_{N,d}$ and hence preserves the uniform sum over $\mathbb{V}_{N,d}$.
Equivalently,
\[
\boldsymbol{\chi}_i = (S_C)_i^{-1}\bigl(\boldsymbol{\chi}_i' - (\boldsymbol{a}_C)_i\bigr),
\qquad i\in\mathcal{I}.
\]
For indices $i\notin\mathcal{I}$, we introduce the same notation
\[
\boldsymbol{\chi}_i' := (S_C)_i\,\boldsymbol{\chi}_i + (\boldsymbol{a}_C)_i,
\]
which constitutes a relabeling rather than a change of summation variables, as no sum is taken over these indices.
Substituting into \eqref{eq:start_transform_proof} yields
\begin{equation}
\begin{split}
\widetilde{F}_{\boldsymbol{\eta}_1,\ldots,\boldsymbol{\eta}_n} & (C^\dagger O_1 C,\ldots,C^\dagger O_k C) \\
&=
\frac{1}{d^{N|\mathcal{I}|}}
\sum_{\{\boldsymbol{\chi}_i'\,:\, i\in\mathcal{I}\}}
\omega^{-\sum_{i\in\mathcal{I}}\left\langle \boldsymbol{\eta}_i,\ (S_C)_i^{-1}(\boldsymbol{\chi}_i'-(\boldsymbol{a}_C)_i)\right\rangle}\,
F_{\boldsymbol{\chi}_1',\ldots,\boldsymbol{\chi}_n'}(O_1,\ldots,O_k) \\
&=
\frac{1}{d^{N|\mathcal{I}|}}
\sum_{\{\boldsymbol{\chi}_i'\,:\, i\in\mathcal{I}\}}
\omega^{-\sum_{i\in\mathcal{I}}\left\langle (S_C)_i\boldsymbol{\eta}_i,\ \boldsymbol{\chi}_i'\right\rangle}\,
\omega^{\sum_{i\in\mathcal{I}}\left\langle (S_C)_i\boldsymbol{\eta}_i,\ (\boldsymbol{a}_C)_i\right\rangle}\,
F_{\boldsymbol{\chi}_1',\ldots,\boldsymbol{\chi}_n'}(O_1,\ldots,O_k).
\end{split}
\end{equation}
In the second line we used the bilinearity of the pairing $\langle\cdot,\cdot\rangle$ and the invariance of the symplectic inner product under symplectic transformations, namely
\(
\langle \boldsymbol{\eta}, S^{-1}\boldsymbol{x}\rangle
=
\langle S\,\boldsymbol{\eta}, \boldsymbol{x}\rangle,
\)
which holds for any symplectic matrix $S\in\mathrm{Sp}(2N,\mathbb{Z}_d)$.

The phase factor
$\omega^{\sum_{i\in\mathcal{I}}\langle (S_C)_i\boldsymbol{\eta}_i,\ (\boldsymbol{a}_C)_i\rangle}$
is independent of the summation variables $\boldsymbol{\chi}_i'$ and can be pulled out. Hence
\begin{equation}
\begin{split}
\widetilde{F}_{\boldsymbol{\eta}_1,\ldots,\boldsymbol{\eta}_n} & (C^\dagger O_1 C,\ldots,C^\dagger O_k C)
\\ & =
\omega^{\sum_{i\in\mathcal{I}}\left\langle (S_C)_i\boldsymbol{\eta}_i,\ (\boldsymbol{a}_C)_i\right\rangle}\,
\frac{1}{d^{N|\mathcal{I}|}}
\sum_{\{\boldsymbol{\chi}_i'\,:\, i\in\mathcal{I}\}}
\omega^{-\sum_{i\in\mathcal{I}}\left\langle (S_C)_i\boldsymbol{\eta}_i,\ \boldsymbol{\chi}_i'\right\rangle}\,
F_{\boldsymbol{\chi}_1',\ldots,\boldsymbol{\chi}_n'}(O_1,\ldots,O_k).
\end{split}
\end{equation}
Recognizing the remaining sum as precisely the partial symplectic Fourier transform
\eqref{eq:partial_symplectic_FT_def}, we conclude that
\begin{equation}
\label{eq:covariance_for_Ftilde_general}
\widetilde{F}_{\boldsymbol{\eta}_1,\ldots,\boldsymbol{\eta}_n}(C^\dagger O_1 C,\ldots,C^\dagger O_k C)
=
\omega^{\sum_{i\in\mathcal{I}}\left\langle (S_C)_i\boldsymbol{\eta}_i,\ (\boldsymbol{a}_C)_i\right\rangle}\,
\widetilde{F}_{\boldsymbol{\eta}_1',\ldots,\boldsymbol{\eta}_n'}(O_1,\ldots,O_k),
\end{equation}
where
\[
\boldsymbol{\eta}_i'=
\begin{cases}
(S_C)_i\boldsymbol{\eta}_i, & i\in\mathcal{I},\\
(S_C)_i\boldsymbol{\eta}_i+(\boldsymbol{a}_C)_i, & i\notin\mathcal{I}.
\end{cases}
\]

Therefore, although the transformed family $\widetilde{F}_{\boldsymbol{\eta}_1,\ldots,\boldsymbol{\eta}_n}$ acquires an index-dependent phase character under Clifford conjugation, as displayed in \eqref{eq:covariance_for_Ftilde_general}, this character is entirely eliminated upon taking pointwise absolute values.
Consequently, the family
\[
\bigl\{\, \lvert \widetilde{F}_{\boldsymbol{\eta}_1,\ldots,\boldsymbol{\eta}_n} \rvert \,\bigr\}_{(\boldsymbol{\eta}_1,\ldots,\boldsymbol{\eta}_n)\in\mathbb{V}_{N,d}^{\,n}}
\]
is manifestly symplectically componentwise Clifford-covariant in the sense of Definition~\ref{def:SymplecticCliifordCovariance}.

Moreover, since the partial symplectic Fourier transform consists of finite linear combinations with fixed coefficients, and the map $z \mapsto |z|^2$ is a smooth operation on $\mathbb{C}$, the squared moduli
\[
\bigl\{\, \lvert \widetilde{F}_{\boldsymbol{\eta}_1,\ldots,\boldsymbol{\eta}_n} \rvert^{2} \,\bigr\}_{(\boldsymbol{\eta}_1,\ldots,\boldsymbol{\eta}_n)\in\mathbb{V}_{N,d}^{\,n}}
\]
depend smoothly on the original family
\[
\bigl\{\, F_{\boldsymbol{\chi}_1,\ldots,\boldsymbol{\chi}_n} \,\bigr\}_{(\boldsymbol{\chi}_1,\ldots,\boldsymbol{\chi}_n)\in\mathbb{V}_{N,d}^{\,n}} .
\]
This completes the proof.
\end{proof}

The proposition thus establishes that symplectically componentwise Clifford-covariant families are closed, up to harmless phase characters, under symplectic Fourier transforms.
In particular, the squared-modulus families obtained in this way define genuine Clifford-covariant phase-space objects, free of residual gauge phases.

This observation will be especially useful in later constructions, where Fourier-dual representations naturally arise and Clifford-invariant quantities must be extracted.
In such settings, the squared absolute values provide a canonical and smoothly related class of phase-space invariants derived from the original family.

\subsubsection{Weyl-Heisenberg Correlation Functions}
\label{sec:Weyl-Heisenberg Correlation Functions}

A simple and highly instructive example of a componentwise Clifford-covariant family is provided by the \emph{Weyl-Heisenberg $n$-point correlation functions}.  
For phase-space points
$
    \boldsymbol{\chi}_1,\ldots,\boldsymbol{\chi}_n \in \mathbb{V}_{N,d}
$
and operators
$
    O_1,\ldots,O_n \in \mathrm{Herm}(\mH_{N,d}),
$
they are defined by
\begin{equation}
    C^{\mathrm{WH}}_{\boldsymbol{\chi}_1,\ldots,\boldsymbol{\chi}_n}(O_1,\ldots,O_n)
    \coloneqq
    \prod_{i=1}^n W_{O_i}(\boldsymbol{\chi}_i),
\end{equation}
where $W_{O_i}$ denotes the discrete Wigner function of~$O_i$.

The family $ C^{\mathrm{WH}} $ satisfies the uniform componentwise Clifford-covariance rule:
\begin{equation}
    C^{\mathrm{WH}}_{\boldsymbol{\chi}_1,\ldots,\boldsymbol{\chi}_n}
    (C^\dagger O_1 C,\ldots,C^\dagger O_n C)
    =
    C^{\mathrm{WH}}_{S_C \boldsymbol{\chi}_1 + \boldsymbol{a}_C,
                    \ldots,
                    S_C \boldsymbol{\chi}_n + \boldsymbol{a}_C}
    (O_1,\ldots,O_n),
\end{equation}
for every Clifford operation $C \in \mC_{N,d}$ with associated
symplectic matrix $ S_C \in \mathrm{Sp}(2N,\mathbb{Z}_d) $
and phase-space translation vector
$ \boldsymbol{a}_C \in \mathbb{V}_{N,d} $.
The covariance property is preserved under multiplication by any function depending solely on the symplectic inner products
$
    \langle \boldsymbol{\chi}_i,\boldsymbol{\chi}_j\rangle
$,
that is, any functional of the form
\begin{equation}
    \tilde{C}^{\mathrm{WH}}_{\boldsymbol{\chi}_1,\ldots,\boldsymbol{\chi}_n}(O_1,\ldots,O_n)
    =
    f\left(\{\langle \boldsymbol{\chi}_i,\boldsymbol{\chi}_j\rangle\}_{i<j}\right)
    \prod_{i=1}^n W_{O_i}(\boldsymbol{\chi}_i),
\end{equation}
with $ f:\mathbb{R}^{n(n-1)/2}\to\mathbb{R} $ arbitrary, is again componentwise Clifford covariant.

Let
$
    F^{(k,n)}_{\boldsymbol{\chi}_1,\ldots,\boldsymbol{\chi}_k}(O_1,\ldots,O_n)
$
be a family that is componentwise Clifford covariant under affine invertible transformations of each phase-space entry, i.e.~for every Clifford operation $ C\in\mC_{N,d} $ there exist
invertible matrices
$
    \{S_{C,i}\}_{i=1}^k\subset \mathrm{Sp}(2N,\mathbb{Z}_d)
$
and displacement vectors
$
    \{\boldsymbol{b}_{C,i}\}_{i=1}^k \subset \mathbb{V}_{N,d}
$
such that
\begin{equation}
\label{eq:Uniform_Clifford_covariance_Fkn_improved}
    F^{(k,n)}_{\boldsymbol{\chi}_1,\ldots,\boldsymbol{\chi}_k}
    (C^\dagger O_1 C,\ldots,C^\dagger O_n C)
    =
    F^{(k,n)}_{S_{C,1}\boldsymbol{\chi}_1+\boldsymbol{b}_{C,1},\ldots,
               S_{C,k}\boldsymbol{\chi}_k+\boldsymbol{b}_{C,k}}
    (O_1,\ldots,O_n).
\end{equation}
From such a family one may construct lower-rank componentwise Clifford-covariant functionals by summation over one phase-space variable.  
Let $l = (l_{i,j})_{i,j=1}^k$ be an invertible matrix over $\mathbb{Z}_d$
and let $\{\boldsymbol{c}_i\}_{i=1}^k\subset\mathbb{V}_{N,d}$ be fixed displacements.
Define
\begin{equation}
\label{eq:Fk-1_definition_improved}
    \tilde{F}^{(k-1,n)}_{\boldsymbol{\chi}_1,\ldots,\boldsymbol{\chi}_{k-1}}
    (O_1,\ldots,O_n)
    \coloneqq
    \sum_{\boldsymbol{\chi}_k \in \mathbb{V}_{N,d}}
    F^{(k,n)}_{l_{1,j}\boldsymbol{\chi}_j+\boldsymbol{c}_1,
               \ldots,
               l_{k,j}\boldsymbol{\chi}_j+\boldsymbol{c}_k}
    (O_1,\ldots,O_n),
\end{equation}
where repeated indices are summed (Einstein convention).

This construction is a special instance of Proposition~\ref{prop:InducedCovariance}.
Nevertheless, one can verify the covariance property directly:
\begin{equation}
\begin{split}
\tilde{F}^{(k-1,n)}_{\boldsymbol{\chi}_1,\ldots,\boldsymbol{\chi}_{k-1}}
    &(C^\dagger O_1 C,\ldots,C^\dagger O_n C)
\\
    &= \sum_{\boldsymbol{\chi}_k \in {\mathbb{V}_{N,d}}}
       F^{(k,n)}_{l_{1,j}\boldsymbol{\chi}_j+\boldsymbol{c}_1, 
                  \ldots, 
                  l_{k,j}\boldsymbol{\chi}_j+\boldsymbol{c}_k}
       (C^\dagger O_1 C,\ldots,C^\dagger O_n C)
\\
    &= \sum_{\boldsymbol{\chi}_k \in {\mathbb{V}_{N,d}}}
       F^{(k,n)}_{S_{C,1}(l_{1,j}\boldsymbol{\chi}_j+\boldsymbol{c}_1)+\boldsymbol{b}_{C,1}, 
                  \ldots, 
                  S_{C,k}(l_{k,j}\boldsymbol{\chi}_j+\boldsymbol{c}_k)+\boldsymbol{b}_{C,k}}
       (O_1,\ldots,O_n)
\\
    &= \sum_{\boldsymbol{\chi}_k \in {\mathbb{V}_{N,d}}}
       F^{(k,n)}_{l_{1,j}(S_{C,1}\boldsymbol{\chi}_j)
                  + S_{C,1}\boldsymbol{c}_1 + \boldsymbol{b}_{C,1}, 
                  \ldots, 
                  l_{k,j}(S_{C,k}\boldsymbol{\chi}_j)
                  + S_{C,k}\boldsymbol{c}_k + \boldsymbol{b}_{C,k}}
       (O_1,\ldots,O_n)
\\
    &= \sum_{\boldsymbol{\chi}_k \in {\mathbb{V}_{N,d}}}
       F^{(k,n)}_{l_{1,j}(S_{C,1}\boldsymbol{\chi}_j+\boldsymbol{d}_j)+\boldsymbol{c}_1, 
                  \ldots, 
                  l_{k,j}(S_{C,k}\boldsymbol{\chi}_j+\boldsymbol{d}_j)+\boldsymbol{c}_k}
       (O_1,\ldots,O_n)
\\
    &= \tilde{F}^{(k-1,n)}_{S_{C,1}\boldsymbol{\chi}_1+\boldsymbol{d}_1, 
                    \ldots, 
                    S_{C,k-1}\boldsymbol{\chi}_{k-1}+\boldsymbol{d}_{k-1}}
       (O_1,\ldots,O_n),
\end{split}
\end{equation}
where the displacement vectors $\boldsymbol{d}_1,\ldots,\boldsymbol{d}_k \in \mathbb{V}_{N,d}$
are determined uniquely by the linear system
\begin{equation}
\label{eq:d_constraint_improved}
    \sum_{j=1}^k l_{i,j}\boldsymbol{d}_j + \boldsymbol{c}_i
    =
    S_{C,i}\boldsymbol{c}_i + \boldsymbol{b}_{C,i},
    \qquad i=1,\ldots,k,
\end{equation}
whose explicit solution is
\begin{equation}
\label{eq:d_solution_improved}
    \boldsymbol{d}_i
    =
    \sum_{j=1}^k (l^{-1})_{i,j}
    \bigl(
        S_{C,j}\boldsymbol{c}_j - \boldsymbol{c}_j + \boldsymbol{b}_{C,j}
    \bigr).
\end{equation}

Thus, the family
$
    \tilde{F}^{(k-1,n)}
$
inherits affine componentwise Clifford covariance from
$
    F^{(k,n)}
$.
The construction may be applied recursively, producing an entire hierarchy of
lower-order affine componentwise Clifford-covariant families.
In particular, Weyl-Heisenberg correlation functions constitute a natural
initial building block from which large classes of such covariant functionals can be systematically generated.

Finally, it is natural to consider the behavior of Weyl-Heisenberg correlation functions under symplectic Fourier transforms on one or more of their phase-space arguments.
Since the discrete Wigner function itself is defined via a symplectic Fourier transform of Weyl-Heisenberg displacement operators, applying additional symplectic Fourier transforms to the phase-space labels $\boldsymbol{\chi}_1,\ldots,\boldsymbol{\chi}_n$ produces dual families that mix complementary phase-space representations.
By Proposition~\ref{prop:SymplecticFourierStability}, such transformations preserve symplectic componentwise Clifford covariance up to index-dependent phase factors, and therefore give rise to symplecticallt Clifford-covariant families upon taking pointwise absolute values or squared moduli.
In this sense, symplectic Fourier transforms provide a systematic mechanism for generating new Clifford-covariant phase-space functionals from Weyl-Heisenberg correlation functions, enlarging the class of admissible covariant observables while retaining full control over their transformation properties.

\subsubsection{Magic Measures as Representative Examples}

We now present four representative examples of (affine, componentwise) Clifford-covariant families of functions.  
Each of these gives rise to a well-known magic measure, and each falls within the scope of Theorem~\ref{theo:MainTheorem}.  
The examples are organized in increasing structural complexity: starting from one-point correlations (mana), progressing to two-point kernels (stabilizer Rényi entropies), then to stabilizer-overlap functionals (stabilizer fidelity), and finally to generalized stabilizer entropies.

\paragraph{(i) Mana of Magic.}
The simplest Clifford-covariant family is given by the one-point Weyl-Heisenberg correlation function, namely the Wigner function:
\begin{equation}
    W_{\boldsymbol{\chi}}(O)
    \coloneqq
    \frac{1}{d^N}
    \Tr\bigl[ O A_{\boldsymbol{\chi}} \bigr],
\end{equation}
which obeys the fundamental covariance rule
\begin{equation}
\label{eq:Wigner_covariance}
    W_{\boldsymbol{\chi}}(C^{\dagger} O C)
    =
    W_{S_C\boldsymbol{\chi} + \boldsymbol{a}_C}(O),
\end{equation}
for every Clifford unitary $ C \in \mC_{N,d} $, with associated symplectic action
$S_C \in \mathrm{Sp}(2N,\mathbb{Z}_d)$ and displacement vector
$\boldsymbol{a}_C \in \mathbb{V}_{N,d}$.
From this one-point function one defines the Wigner trace norm
\begin{equation}
    \| O \|_{W}
    \coloneqq
    \sum_{\boldsymbol{\chi}\in\mathbb{V}_{N,d}}
    \bigl| W_{\boldsymbol{\chi}}(O) \bigr|.
\end{equation}
For a quantum state $\rho$, the mana of magic is then
\begin{equation}
    \mM(\rho)
    \coloneqq
    \log \| \rho \|_{W}.
\end{equation}

\paragraph{(ii) Weyl-Heisenberg Kernel and Stabilizer Rényi Entropies.}
A natural two-point generalization of the Wigner function is the
\emph{Weyl-Heisenberg kernel}:
\begin{equation}
    K_{\boldsymbol{\chi}}(O_1,O_2)
    \coloneqq
    \frac{1}{d^N}
    \Tr\bigl[
        O_1 T_{\boldsymbol{\chi}} O_2 T_{\boldsymbol{\chi}}^{\dagger}
    \bigr]
    =
    \sum_{\boldsymbol{\chi}'\in\mathbb{V}_{N,d}}
        W_{O_1}(\boldsymbol{\chi}')
        W_{O_2}(\boldsymbol{\chi}' - \boldsymbol{\chi}),
\end{equation}
which is a special case of Eq.~\eqref{eq:Fk-1_definition_improved} and of Proposition~\ref{prop:InducedCovariance}.  
It satisfies the affine Clifford-covariance law
\begin{equation}
\label{eq:WH_kernel_covariance}
    K_{\boldsymbol{\chi}}(C^\dagger O_1 C, C^\dagger O_2 C)
    =
    K_{S_C\boldsymbol{\chi}}(O_1,O_2).
\end{equation}
From this kernel one defines the $\alpha$-SRE:
\begin{equation}
    M_{\alpha}(O)
    =
    \frac{1}{1-\alpha}
    \log\left[
        \frac{1}{d^{(1-\alpha)N}}
        \sum_{\boldsymbol{\chi}}
        \bigl| P_{\boldsymbol{\chi}}(O) \bigr|^{\alpha}
    \right],
\end{equation}
where
\begin{equation}
    P_{\boldsymbol{\chi}}(O)
    \coloneqq
    \frac{K_{\boldsymbol{\chi}}(O,O)}{\bigl( \Tr O \bigr)^2}
\end{equation}
defines a normalized probability distribution on~$\mathbb{V}_{N,d}$ whenever $O\ge 0$ and $O\not\equiv 0$.

For later comparison, it is convenient to identify the index set $\mV$ from Theorem~\ref{theo:MainTheorem} either with the discrete phase space $ \mathbb{V}_{N,d} $ or with the generalized Pauli group, the Weyl-Heisenberg group $ \mP_{N,d} $.  
Under this identification, $ M_\alpha(O) $ may be rewritten as
\begin{equation}
\begin{split}
    M_\alpha(O)
    &=
    \frac{1}{1-\alpha}
    \log\left[
        \frac{1}{2 d^{N+1}}
        \sum_{g\in\mP_{N,d}}
        \left|
            \frac{
                \Tr(O g O g^\dagger)
            }{
                (\Tr O)^2
            }
        \right|^{\alpha}
    \right],
\end{split}
\end{equation}
which makes the qubit case $d=2$ manifest.

\noindent
\emph{Remark.}
Recall that the phase-point operators
$\{ A_{\boldsymbol{\chi}} \}_{\boldsymbol{\chi}\in\mathbb{V}_{N,d}}$
and the Weyl-Heisenberg displacement operators
$\{ T_{\boldsymbol{\chi}} \}_{\boldsymbol{\chi}\in\mathbb{V}_{N,d}}$
are related by the symplectic Fourier transform and its inverse.
In view of Proposition~\ref{prop:SymplecticFourierStability}, this relation immediately implies that Clifford covariance of the Wigner-function family
$\{ W_{\boldsymbol{\chi}} \}_{\boldsymbol{\chi}\in\mathbb{V}_{N,d}}$
is inherited by the associated classical stabilizer distribution on phase space,
$\{ P_{\boldsymbol{\chi}} \}_{\boldsymbol{\chi}\in\mathbb{V}_{N,d}}$,
without the need for a separate proof.

\paragraph{(iii) Stabilizer Fidelity.}
Let us now take the index set $\mV$ in Theorem~\ref{theo:MainTheorem} to be the set of stabilizer states $\mSS_{N,d}$.  
The stabilizer fidelity of an operator $O$ is
\begin{equation}
    F(O)
    \coloneqq
    \max_{\ket{s}\in\mSS_{N,d}}
    \bigl| \bra{s} O \ket{s} \bigr|.
\end{equation}

A structural and another way to express stabilizer states is via maximal isotropic subspaces of
$\mathbb{V}_{N,d}$.  
Let $\mathrm{M}$ denote the set of all such subspaces.  
For every
\[
    \mathcal{M}
    =
    \{\boldsymbol{\chi}_1,\ldots,\boldsymbol{\chi}_{d^N}\}
    \in \mathrm{M},
\qquad
\braket{\boldsymbol{\chi}_i,\boldsymbol{\chi}_j}=0,
\]
and every displacement $\boldsymbol{\chi}\in\mathbb{V}_{N,d}$, there exists a unique (up to phase) stabilizer state
$\ket{\mathcal{M};\boldsymbol{\chi}}$
satisfying the eigenvalue equations
\begin{equation}
\label{eq:stabilizer_eigenvalue_equation}
    \omega^{\braket{\boldsymbol{\chi},\boldsymbol{\chi}_i}}
      T_{\boldsymbol{\chi}_i}
      \ket{\mathcal{M};\boldsymbol{\chi}}
    =
    \ket{\mathcal{M};\boldsymbol{\chi}},
    \qquad i=1,\ldots,d^N.
\end{equation}
Its Wigner function is uniformly supported on the affine subspace
$\mathcal{M}+\boldsymbol{\chi}$:
\begin{equation}
\label{eq:stabilizer_Wigner}
    W_{\boldsymbol{\chi}'}\left(
        \ket{\mathcal{M};\boldsymbol{\chi}}
        \bra{\mathcal{M};\boldsymbol{\chi}}
    \right)
    =
    \frac{1}{d^N}
    \sum_{i=1}^{d^N}
        \delta_{\boldsymbol{\chi}',\boldsymbol{\chi}_i+\boldsymbol{\chi}}.
\end{equation}

From this we define the \emph{stabilizer-overlap functional}:
\begin{equation}
\label{eq:stabilizer_indicator_fidelity}
    \mathcal{F}_{\boldsymbol{\chi}_1,\ldots,\boldsymbol{\chi}_{d^N};\boldsymbol{\chi}}(O)
    \coloneqq
    \frac{1}{d^N}
    \left( \prod_{i\ne j} \delta_{\braket{\boldsymbol{\chi}_i,\boldsymbol{\chi}_j}} \right)
    \sum_{i=1}^{d^N}
        W_{\boldsymbol{\chi}_i+\boldsymbol{\chi}}(O),
\end{equation}
which vanishes unless the set $\mathcal{M}$ is a maximal isotropic subspace.  
Otherwise,
\[
    \mathcal{F}_{\mathcal{M};\boldsymbol{\chi}}(O)
    =
    \bra{\mathcal{M};\boldsymbol{\chi}} O \ket{\mathcal{M};\boldsymbol{\chi}}.
\]
This functional is uniformly Clifford-covariant:
\begin{equation}
    \mathcal{F}_{\boldsymbol{\chi}_1,\ldots,\boldsymbol{\chi}_{d^N};\boldsymbol{\chi}}
    (C^\dagger O C)
    =
    \mathcal{F}_{S_C\boldsymbol{\chi}_1,\ldots,S_C\boldsymbol{\chi}_{d^N};S_C\boldsymbol{\chi}+\boldsymbol{a}_C}(O).
\end{equation}
Taking the limit $\alpha\to\infty$ of the corresponding $\alpha$-Rényi functional yields
\begin{equation}
\label{eq:FidelityReps}
\begin{split}
    F(O)
    &=
    \max_{\mathcal{M}\in\mathrm{M}}
    \max_{\boldsymbol{\chi}\in\mathbb{V}_{N,d}}
    \bigl|
        \mathcal{F}_{\mathcal{M};\boldsymbol{\chi}}(O)
    \bigr|
    =
    \max_{\ket{s}\in\mSS_{N,d}}
    \bigl| \bra{s} O \ket{s} \bigr|.
\end{split}
\end{equation}
\noindent
Both representations in Eq.~\eqref{eq:FidelityReps} clearly demonstrate that the stabilizer fidelity falls within the class of quantities addressed by our main theorem.

\paragraph{(iv) GSEs.}

By construction, the GSEs $M_W(\rho)$ defined in Eqs.~\eqref{eq:IntComm}, \eqref{eq:GSP}, and \eqref{eq:GSE} are smooth, real-analytic functions on the space of Hermitian operators
$\mathrm{Herm}(\mH)$.
Moreover, they are invariant under conjugation by Clifford unitaries, in the sense that
\[
M_W\!\left(C\,\rho\,C^\dagger\right) = M_W(\rho)
\qquad \forall\, C \in \mC_{N,d}.
\]
Note that in this example the smooth Clifford-covariant family underlying the construction of $M_W$ is trivial, consisting of a single element- namely, $M_W$ itself.

Building on the smooth Clifford-covariant families exhibited in these four examples, and on the manner in which the corresponding magic measures are constructed from them, we obtain the following corollary as a direct consequence of Theorem~\ref{theo:MainTheorem}.
\paragraph{Corollary of Theorem~\ref{theo:MainTheorem}.}
\begin{corollary}
\label{cor:CliffordExtremality}
Let $ \mQ \subset \mC_{N,d} $ be a finite subgroup, and denote its stabilized subspace by $ \mS_{\mQ} $, with orthogonal complement $ \mS_{\mQ}^{\perp} $.  
Fix a normalized vector $ \ket{\psi} \in \mS_{\mQ} $, and write $ \psi \coloneqq \ket{\psi}\bra{\psi} $.  
Consider the submanifold of rank-one projectors whose support lies in
\[
    \Span\{\ket{\psi}\}
    \oplus
    \mS_{\mQ}^{\perp}
    \subset
    \mH,
\]
namely
\[
    \varrho_{\psi;\mQ}
    \coloneqq
    \Bigl\{
        \rho\in\varrho(\mH)
        \; \big| \;
        \operatorname{supp}(\rho)
        \subset
        \Span\{\ket{\psi}\}
        \oplus
        \mS_{\mQ}^{\perp}
    \Bigr\},
\]
where $\varrho(\mH)$ denotes the manifold of pure-state density operators.

Then the operator $\psi$ is an extremal point of each of the following magic measures, when restricted to the manifold $\varrho_{\psi;\mQ}$:
\begin{itemize}
    \item the mana of magic,
    \item the stabilizer $\alpha$-Rényi entropies (and the $L^p$-norms of characteristic functions),
    \item the stabilizer fidelity,
    \item the generalized stabilizer entropies.
\end{itemize}
\end{corollary}

\subsection{Extremality Analysis for the Three Representative Measures}

In this subsection we specialize the general expansion of Theorem~\ref{theo:MainTheorem} to the first three magic measures appearing in Corollary~\ref{cor:CliffordExtremality}.  
For each measure we analyze the behavior of the perturbed state
$\ket{\psi(\epsilon)}$, defined previously in
Eq.~\eqref{eq:stateWithVariation},
\[
    \ket{\psi(\epsilon)}
    =\frac{\psi + \epsilon \sigma + \epsilon^{2}\varphi}{1+\epsilon^{2}} ,
\]
in a neighbourhood of the reference state~$\ket{\psi}$, and determine the nature of the extremum at~$\epsilon=0$.

\paragraph{(i) The mana of magic.}
For the Wigner trace norm one obtains the expansion
\begin{equation}
\label{eq: variation of Wigner norm}
\begin{split}
\| \psi(\epsilon) \|_W  =  \| \psi \|_W & + \frac{|\epsilon|}{1+\epsilon^2} \cdot \sum_{\substack{\boldsymbol{\chi}\in\mathbb{V}_{N,d} \\ W_{\boldsymbol{\chi}}(\psi)=0}} |W_{\boldsymbol{\chi}}(\sigma)| + 
\\ & +\frac{\epsilon^2}{1+\epsilon^2} \cdot  
\sum_{\boldsymbol{\chi}\in\mathbb{V}_{N,d}} \sign[W_{\boldsymbol{\chi}}(\psi)]\cdot W_{\boldsymbol{\chi}}(\mT_\mQ(\mu))
\\ & +\frac{\epsilon^2}{1+\epsilon^2} \cdot  \sum_{\substack{\boldsymbol{\chi}\in\mathbb{V}_{N,d} \\ W_{\boldsymbol{\chi}}(\psi)=0}} \sign[W_{\boldsymbol{\chi}}(\epsilon \sigma)] \cdot W_{\boldsymbol{\chi}}(\mu)
\\ & +\frac{\epsilon^2}{1+\epsilon^2} \cdot \sum_{\substack{\boldsymbol{\chi}\in\mathbb{V}_{N,d} \\ W_{\boldsymbol{\chi}}(\psi)=0 \\ W_{\boldsymbol{\chi}}(\sigma)=0}} |W_{\boldsymbol{\chi}}(\mu)|.
\end{split}
\end{equation}
Here the function $\mathrm{sign}(x)$ is defined to be zero at $x=0$.
Therefore,
\begin{enumerate}
    \item If $\exists \boldsymbol{\chi} \in \mathbb{V}_{N,d}$ such that $W_{\boldsymbol{\chi}}(\psi)=0\neq W_{\boldsymbol{\chi}}(\sigma)$, there will be a positive dominant leading order of $|\epsilon|$, and the point $\epsilon=0$ is a sharp minima of the path.
    \item Otherwise, we are left with
        \begin{equation}
        \begin{split}
            \| \psi(\epsilon) \|_W  = & \| \psi \|_W +\frac{\epsilon^2}{1+\epsilon^2} \sum_{\substack{\boldsymbol{\chi}\in\mathbb{V}_{N,d} \\ W_{\boldsymbol{\chi}}(\psi)=0 \\ W_{\boldsymbol{\chi}}(\sigma)=0}} |W_{\boldsymbol{\chi}}(\varphi)|
            \\ &  +\frac{\epsilon^2}{1+\epsilon^2} \cdot  
            \sum_{\boldsymbol{\chi}\in\mathbb{V}_{N,d}} \sign[W_{\boldsymbol{\chi}}(\psi)]\cdot W_{\boldsymbol{\chi}}(\mT_\mQ(\mu))   ,
        \end{split}
        \label{eq:varied Wigner trace norm when the linear term vanishes}
        \end{equation}
and, hence, at $\epsilon=0$ we get a smooth extremum.
\end{enumerate}

\paragraph{(ii) The stabilizer fidelity.}
Let
\[
    S_{\mathrm{nearest}}(\psi)
    =
    \bigl\{
        \ket{s}\in\mSS_G :
        |\braket{s|\psi}|^{2}
        = F(\psi)
    \bigr\}
\]
be the set of $G$-stabilizer states achieving the stabilizer fidelity of $\ket{\psi}$.  
For sufficiently small $|\epsilon|$ only these states can maximize the overlap with $\ket{\psi(\epsilon)}$.

As in the second note after the proof of Theorem~\ref{theo:MainTheorem}, averaging over this finite set gives
\begin{equation}
\label{eq:averageOfOverlaps}
\begin{split}
    \frac{1}{|S_{\text{nearest}}(|\psi\rangle)|}
    &\sum_{\ket{s}\in S_{\text{nearest}}(|\psi\rangle)}
    \left| \braket{s | \psi (\epsilon)} \right|^2
    \\
    & =
    F(\psi)
    +
    \frac{\epsilon^2}{1+\epsilon^2}\frac{1}{|S_{\text{nearest}}(|\psi\rangle)|}
    \sum_{\ket{s}\in S_{\text{nearest}}(|\psi\rangle)}\braket{s|\mT_G(\mu)|s}.
\end{split}
\end{equation}
Because the average of the linear term in~$\epsilon$ vanishes,
two possibilities arise:
either the linear term vanishes individually for every $\ket{s}\in S_{\text{nearest}}(|\psi\rangle)$,
or it takes positive coefficients for some $\ket{s}$ and negative ones for others.
In the latter case, the point $ \epsilon = 0 $ corresponds to a
\emph{sharp local minimum}.
Otherwise, if the linear term of $|\braket{s'|\psi(\epsilon)}|^2$ in $\epsilon$ vanishes for all $\ket{s}\in S_{\text{nearest}}(|\psi\rangle)$, one obtains for every $\ket{s}\in S_{\text{nearest}}(|\psi\rangle)$
\begin{equation}
    |\braket{s|\psi(\epsilon)}|^2
    =
    F(\psi)
    +
    \frac{\epsilon^2}{1+\epsilon^2}\braket{s|\mu|s}.
    \label{eq:variation in fidelity when the linear sharp term vanishes}
\end{equation}

Consequently, the nature of the extremum at $\epsilon=0$ is governed by the signs of the quantities $\braket{s|\mu|s}$:
\begin{enumerate}
    \item If $\braket{s|\mu|s} < 0$ for all $\ket{s}\in S_{\text{nearest}}(|\psi\rangle)$, then $\epsilon=0$ is a local \emph{smooth maximum}.
    \item If $\braket{s|\mu|s} \le 0$ for all $\ket{s}\in S_{\text{nearest}}(|\psi\rangle)$, with equality for at least one $\ket{s}$, then $\epsilon=0$ is a \emph{flat extremum}, i.e., the function is constant along the entire variation path.
    \item If there exists $\ket{s}\in S_{\text{nearest}}(|\psi\rangle)$ with $\braket{s|\mu|s} > 0$, then $\epsilon=0$ is a local \emph{smooth minimum}.
\end{enumerate}

\subsubsection*{A note for future calculations}

For every state $ \ket{s} \in S_{\mathrm{nearest}}(\ket{\psi}) $, the quantity appearing in the linear coefficient of the stabilizer-fidelity expansion satisfies
\begin{equation}
\label{eq:L matrix}
    \braket{s|\sigma(\psi,\varphi)|s}
    =
    2\operatorname{Re}\left( \braket{\varphi|s}\braket{s|\psi} \right)
    =:
    2\operatorname{Re}\ell_{\ket{\psi}}\left(\ket{s}; \ket{\varphi}\right).
\end{equation}
The map
\[
    \ell_{\ket{\psi}}\left(\ket{s}; \cdot \right)
    : \mH_{N,d} \longrightarrow \mathbb{C},
\qquad
    \ket{\varphi} \longmapsto \braket{\varphi|s}\braket{s|\psi},
\]
is linear in its second argument.  
Thus, to understand the full behavior of the linear term it is sufficient to evaluate
$\ell_{\ket{\psi}}(\ket{s};\ket{\varphi})$ on any basis of the orthogonal complement
$\mathrm{Span}^{\perp}\{\ket{\psi}\}$.  
This observation will be used in the numerical investigations of Section~\ref{sec:Examples of Non-degenerate Clifford Eigenstates of Qudits}.

\paragraph{(iii) The $2$-SRE.} Since $\Xi_\alpha(\ket{\psi})$ is a smooth function on the projective Hilbert manifold, its behavior in the vicinity of critical states can be systematically analyzed using known differential methods. 
For integer values of~$\alpha$, this analysis can be carried out in a guaranteed finite number of steps. 
In what follows, we focus on the case $\alpha=2$, which is of particular interest because it naturally generalizes to include mixed states.

Recalling that
\[
\psi(\epsilon)
=\frac{\psi + \epsilon \sigma
  + \epsilon^2\varphi}{1+\epsilon^2},
\]
we can express
\begin{equation}
    \tilde{P}_{\boldsymbol{\chi}}(\ket{\psi(\epsilon)}) \equiv \left(1+\epsilon^2\right)^2 P_{\boldsymbol{\chi}}(\ket{\psi(\epsilon)}) 
    =\left(1+\epsilon^2\right)^4 K_{\boldsymbol{\chi}}(\psi(\epsilon),\psi(\epsilon))= \sum_{i=0}^4 \epsilon^i \tilde{P}^{(i)}_{\boldsymbol{\chi}},
\end{equation}
where
\begin{subequations}
    \begin{gather}
    \tilde{P}^{(0)}_{\boldsymbol{\chi}}=K_{\boldsymbol{\chi}}(\psi,\psi) , \\ \tilde{P}^{(1)}_{\boldsymbol{\chi}}=K_{\boldsymbol{\chi}}(\psi,\sigma)+K_{\boldsymbol{\chi}}(\sigma,\psi) , \\ \tilde{P}^{(2)}_{\boldsymbol{\chi}}=K_{\boldsymbol{\chi}}(\psi,\varphi)+K_{\boldsymbol{\chi}}(\sigma,\sigma)+K_{\boldsymbol{\chi}}(\varphi,\psi) , \\ \tilde{P}^{(3)}_{\boldsymbol{\chi}}=K_{\boldsymbol{\chi}}(\varphi,\sigma)+K_{\boldsymbol{\chi}}(\sigma,\varphi) , \\ \tilde{P}^{(4)}_{\boldsymbol{\chi}}=K_{\boldsymbol{\chi}}(\varphi,\varphi) .
    \end{gather}
\end{subequations}
Consequently,
\begin{equation}
\begin{split}
    \tilde{\Xi}_2(\ket{\psi(\epsilon)})
    &\equiv \left(1+\epsilon^2\right)^4 \Xi_2(\ket{\psi(\epsilon)})= \left(1+\epsilon^2\right)^4 \sum_{\boldsymbol{\chi}} P_{\boldsymbol{\chi}}^2(\ket{\psi(\epsilon)})
     = \sum_{i=0}^8 \epsilon^i \tilde{\Xi}^{(i)}_{\boldsymbol{\chi}},
\end{split}
\end{equation}
where
\begin{subequations}
    \begin{align}
        \tilde{\Xi}^{(0)}_2 &= \sum_{\boldsymbol{\chi}} \left({\tilde{P}}^{(0)}_{\boldsymbol{\chi}}\right)^2, \\
        \tilde{\Xi}^{(1)}_2 &= 2 \sum_{\boldsymbol{\chi}} {{\tilde{P}}^{(0)}_{\boldsymbol{\chi}}} {{\tilde{P}}^{(1)}_{\boldsymbol{\chi}}}, \\
        \tilde{\Xi}^{(2)}_2 &= \sum_{\boldsymbol{\chi}}
        \left[
        2{{\tilde{P}}^{(0)}_{\boldsymbol{\chi}}} {{\tilde{P}}^{(2)}_{\boldsymbol{\chi}}}+
        \left({\tilde{P}}^{(1)}_{\boldsymbol{\chi}}\right)^2
        \right], \\
        \tilde{\Xi}^{(3)}_2 &= 2 \sum_{\boldsymbol{\chi}}
        \left[
        {{\tilde{P}}^{(0)}_{\boldsymbol{\chi}}}
        {{\tilde{P}}^{(3)}_{\boldsymbol{\chi}}}
        +{{\tilde{P}}^{(1)}_{\boldsymbol{\chi}}}
        {{\tilde{P}}^{(2)}_{\boldsymbol{\chi}}}
        \right], \\
        \tilde{\Xi}^{(4)}_2 &= \sum_{\boldsymbol{\chi}}
        \left[
        2{{\tilde{P}}^{(0)}_{\boldsymbol{\chi}}}
        {{\tilde{P}}^{(4)}_{\boldsymbol{\chi}}}
        +2{{\tilde{P}}^{(1)}_{\boldsymbol{\chi}}}
        {{\tilde{P}}^{(3)}_{\boldsymbol{\chi}}}
        +\left({\tilde{P}}^{(2)}_{\boldsymbol{\chi}}\right)^2
        \right], \\
        \tilde{\Xi}^{(5)}_2 &= 2 \sum_{\boldsymbol{\chi}}
        \left[
        {{\tilde{P}}^{(1)}_{\boldsymbol{\chi}}}
        {{\tilde{P}}^{(4)}_{\boldsymbol{\chi}}}
        +{{\tilde{P}}^{(2)}_{\boldsymbol{\chi}}}
        {{\tilde{P}}^{(3)}_{\boldsymbol{\chi}}}
        \right], \\
        \tilde{\Xi}^{(6)}_2 &= \sum_{\boldsymbol{\chi}}
        \left[
        2{{\tilde{P}}^{(2)}_{\boldsymbol{\chi}}} {{\tilde{P}}^{(4)}_{\boldsymbol{\chi}}}+
        \left({\tilde{P}}^{(3)}_{\boldsymbol{\chi}}\right)^2
        \right], \\
        \tilde{\Xi}^{(7)}_2 &= 2 \sum_{\boldsymbol{\chi}} {{\tilde{P}}^{(3)}_{\boldsymbol{\chi}}}
        {{\tilde{P}}^{(4)}_{\boldsymbol{\chi}}}, \\
        \tilde{\Xi}^{(8)}_2 &= \sum_{\boldsymbol{\chi}} \left({\tilde{P}}^{(4)}_{\boldsymbol{\chi}}\right)^2 .
    \end{align}
\end{subequations}

Theorem~\ref{theo:MainTheorem} guarantees that $\tilde{\Xi}_2^{(1)} = 0$.
Using the series
\[
\frac{1}{\left(1+\epsilon^2\right)^4} = \sum_{n=0}^\infty f_n\epsilon^{2n}=\sum_{n=0}^\infty \binom{n+3}{3}(-\epsilon^2)^n
\]
we can write
\begin{equation}
    \Xi_2(\ket{\psi(\epsilon)})
    =\frac{\tilde{\Xi}_2(\ket{\psi(\epsilon)})}{\left(1+\epsilon^2\right)^4} =  \sum_{n=0}^\infty \epsilon^n \Xi_2^{(n)},
\end{equation}
where
\begin{equation}
    \Xi_2^{(n)}=\sum_{i=0}^{\left\lfloor{n/2}\right\rfloor} f_i \tilde{\Xi}_2^{(n-2i)}.
\end{equation}
In particular, the first coefficients are
\begin{subequations}
\label{eq:Xi-explicit-up-to-8}
\begin{align}
\Xi_2^{(0)} &= \tilde{\Xi}_2^{(0)},\\
\Xi_2^{(1)} &= \tilde{\Xi}_2^{(1)}=0 \quad \text{(by Theorem~\ref{theo:MainTheorem})},\\
\Xi_2^{(2)} &= \tilde{\Xi}_2^{(2)} - 4 \tilde{\Xi}_2^{(0)},\\
\Xi_2^{(3)} &= \tilde{\Xi}_2^{(3)},\\
\Xi_2^{(4)} &= \tilde{\Xi}_2^{(4)} - 4 \tilde{\Xi}_2^{(2)} + 10 \tilde{\Xi}_2^{(0)},\\
\Xi_2^{(5)} &= \tilde{\Xi}_2^{(5)} - 4 \tilde{\Xi}_2^{(3)},\\
\Xi_2^{(6)} &= \tilde{\Xi}_2^{(6)} - 4 \tilde{\Xi}_2^{(4)} + 10 \tilde{\Xi}_2^{(2)} - 20 \tilde{\Xi}_2^{(0)},\\
\Xi_2^{(7)} &= \tilde{\Xi}_2^{(7)} - 4 \tilde{\Xi}_2^{(5)} + 10 \tilde{\Xi}_2^{(3)},\\
\Xi_2^{(8)} &= \tilde{\Xi}_2^{(8)} - 4 \tilde{\Xi}_2^{(6)} + 10 \tilde{\Xi}_2^{(4)} - 20 \tilde{\Xi}_2^{(2)} + 35 \tilde{\Xi}_2^{(0)}.
\end{align}
\end{subequations}

Note that if $\tilde{\Xi}_2^{(i)}=0$ for all $i=2,\dots,8$, this ``completely'' determines $\tilde{\Xi}_2(\ket{\psi(\epsilon)})$:
\begin{equation}
    \tilde{\Xi}_2^{(1)} = 0,
    \tilde{\Xi}_2^{(2)} = 4\tilde{\Xi}_2^{(0)},
    \tilde{\Xi}_2^{(3)} = 0,
    \tilde{\Xi}_2^{(4)} = 6\tilde{\Xi}_2^{(0)},
    \tilde{\Xi}_2^{(5)} = 0,
    \tilde{\Xi}_2^{(6)} = 4\tilde{\Xi}_2^{(0)},
    \tilde{\Xi}_2^{(7)} = 0,
    \tilde{\Xi}_2^{(8)} = \tilde{\Xi}_2^{(0)}.
\end{equation}
Consequently,
\begin{equation}
    \Xi_2(\ket{\psi(\epsilon)})=\Xi_2(\ket{\psi}),
\end{equation}
meaning that the function $\Xi_2$ remains constant along the considered path $\ket{\psi(\epsilon)}$.
Otherwise, let 
\[
    m = \min \left\{ i \in \{2,...,8\}  \; \middle| \;  \Xi_2^{(i)} \neq 0 \right\}.
\]
Then, the nature of the stationary point at $\ket{\psi}$ is determined by the parity and sign of $\Xi_2^{(m)}$:
\begin{enumerate}[label=(\roman*)]
    \item If $m$ is even:
    \begin{enumerate}[label=\arabic*.]
        \item If $\tilde{\Xi}_2^{(m)} > 0$, then $\ket{\psi}$ is a local \emph{minimum} along the path.
        \item If $\tilde{\Xi}_2^{(m)} < 0$, then $\ket{\psi}$ is a local \emph{maximum} along the path.
    \end{enumerate}
    \item If $m$ is odd, $\ket{\psi}$ is an \emph{inflection point} along the path.
\end{enumerate}

\noindent
\emph{Remark:} It is important to note that the critical behavior of the function at a given state is not necessarily preserved across different values of~$\alpha$. 
This stems from the fact that the ordering of Rényi entropies between two probability distributions is, in general, not invariant under changes of~$\alpha$. 
In other words, $M_{\alpha}(p) > M_{\alpha}(q)$ for some value of~$\alpha$ does not imply $M_{\alpha'}(p) > M_{\alpha'}(q)$ for all~$\alpha'$. 
Such a monotonic ordering across all $\alpha>0$ holds if and only if the distribution~$p$ \emph{majorizes}~$q$. As an illustrative example, consider varying the qutrit critical state $\ket{\mathbb{T}_0}$ along the path
\[
\ket{\varphi} = e^{i\phi_1}\cos\theta \ket{\mathbb{T}_1} 
               + e^{i\phi_2}\sin\theta \ket{\mathbb{T}_2},
\]
with $\theta=\frac{\pi}{4}$ and $\phi_2=\phi-\phi_1$. 
In this case, one finds
\[
\Xi_2^{(2)} = \frac{2(1-\cos\phi)}{9},
\qquad
\Xi_3^{(2)} = \frac{2(4-5\cos\phi)}{81},
\]
which can acquire opposite signs for the same value of~$\phi$. 
Hence, along this path the same state can correspond to a minimum for one Rényi measure and a maximum for another.

\section{Examples of Non-degenerate Clifford Eigenstates of Single Qudits}
\label{sec:Examples of Non-degenerate Clifford Eigenstates of Qudits}

As follows from the preceding definitions, a \emph{Clifford-stabilized space} is the stabilized subspace of some subgroup of the Clifford group acting on $\mH$. 
For example, the entire Hilbert space $\mH$ is trivially a Clifford-stabilized space, since it is stabilized by the trivial subgroup.
A \emph{Clifford-stabilizer state} is a normalized pure state whose density operator is the projector onto a Clifford-stabilized subspace. 
Equivalently, it is the unique normalized state (up to an overall global phase) that spans a one-dimensional Clifford-stabilized space. 

All Pauli-stabilizer states in $\mH$ are, in particular, Clifford-stabilizer states, as they are stabilized by subgroups of the Pauli group (which is itself a subgroup of the Clifford group).
Similarly, all non-degenerate eigenstates of Clifford unitaries on $\mH$ are Clifford-stabilizer states. 
The so-called ``\emph{Clifford magic states}'', as defined in Ref.~\cite{Bravyi2019simulationofquantum}, constitute a specific subclass of Clifford-stabilizer states. 
The following result is therefore immediate.
\begin{corollary}
    Each Clifford-stabilizer state extremizes the mana of magic, the stabilizer Rényi entropies and the stabilizer fidelity (and hence the min-relative entropy of magic) on pure states.
\end{corollary}
\noindent
This follows from the main Theorem~\ref{theo:MainTheorem}. Here we will give examples of non-degenerate eigenstates of Cliffords, which are the simplest case of Clifford-stabilizer states.

In this section, we explicitly enumerate the eigenstates of single-qutrit and single-ququint Clifford operators and analyze their behavior as critical points of both the stabilizer fidelity and the mana of magic.
We also give some examples for the $2$-SRE.
Since Clifford transformations can be trivially incorporated into any distillation or state-injection routine without altering its performance, it suffices to study one representative from each class of states that are related by a single-qudit Clifford unitary, which we refer to as \emph{Clifford-equivalent}. 
Indeed, any magic-state distillation or injection protocol constructed for a given state $\ket{\psi}$ can be straightforwardly adapted to operate on $C\ket{\psi}$ for any Clifford unitary $C$. 
Accordingly, in what follows, we enumerate all \textbf{Clifford-inequivalent non-degenerate eigenstates} of Clifford operators for qutrits and ququints, and analyze their key properties.

\subsection{Charachtarizing the criticality of the stabilizer fidelity\\and the mana of magic: $\mL$ and $\mW$ Matrices}

If $\{\ket{\psi_i}\}_{i=1}^d$ is an orthonormal basis of the Hilbert space of a single qudit of $d$-level system ($d$ is prime) where $\ket{\psi_1}$ is a non-degenerate eigenstate of a Clifford operator, then we define the matrix $\mL(\ket{\psi_1},...,\ket{\psi_d})$ as the $ |S_{\text{nearest}}(\ket{\psi_1})| \times (d-1)$ matrix with entries
\begin{equation}
    \mL_{i,j}(\ket{\psi_1},...,\ket{\psi_d}) = \ell_{\ket{\psi_1}}\Bigl(\ket{s_{\ket{\psi_1},i}}; \ket{\psi_j}\Bigr)  ,
\end{equation}
where $\ket{s_{\ket{\psi},i}}$ in the $i^{\text{th}}$ element in $S_{\text{nearest}}(\ket{\psi})$, $j$ runs from $2$ to $d$, and $\ell$ as defined in Eq.~(\ref{eq:L matrix}):
\[
    \bra{s}\sigma(\ket{\psi},\ket{\varphi})\ket{s} = 2\operatorname{Re}  \left(\langle \varphi|s\rangle\langle s|\psi\rangle \right)
=: 2\operatorname{Re} \ell_{\ket{\psi}}\Bigl(\ket{s}; \ket{\varphi}\Bigr)  .
\]
Note that, as dictated by Eq.~(\ref{eq:averageOfOverlaps}), the sum  $\sum_i \mL_{i,j}(\ket{\psi_1}, \ldots, \ket{\psi_d})$- that is, the sum of any column of the matrix $\mL$- vanishes. This reflects the fact that the first-order term in the average stabilizer fidelity of $\ket{\psi(\epsilon)}$, taken over all nearest stabilizer states to the slightly perturbed Clifford-stabilizer state, is zero, as previously shown.

With this definition, given any state $\ket{\varphi}=\sum_{i=2}^{d} \alpha_i\ket{\psi_i}$:
\begin{equation}
    \ell_{\ket{\psi_1}}\Bigl(\ket{s_{\ket{\psi_1},i}}; \ket{\varphi}\Bigr)=\sum_{j=1}^d \mL_{i,j}(\ket{\psi_1},...,\ket{\psi_d}) \alpha_j   ,
\end{equation}
helping us investigating the behavior of the stationary point of the stabilizer fidelity.
If $\operatorname{Re} \ell_{\ket{\psi_1}}\Bigl( \ket{s_{\ket{\psi_1},i}}; \ket{\varphi} \Bigr) \neq 0$ for some $i$, then $\ket{\psi_1}$ is a sharp local minimum along the direction defined by $\ket{\psi}$. Otherwise, namely if $\operatorname{Re} \ell_{\ket{\psi_1}}\Bigl( \ket{s_{\ket{\psi_1},i}}; \ket{\varphi} \Bigr) = 0
 \text{ for all } i,$ one should investigate the second-order term.

For odd dimensions, if in addition $\ket{\psi_1}$ has non-vanishing Wigner function everywhere, the first order term in Eq.~(\ref{eq: variation of Wigner norm}) vanishes. Therefore, supposing that $\{\ket{\psi_i}\}_{i=1}^d$ is a set of non-degenerate eigenstates of  a Clifford operator for a single qudit, Eq.~(\ref{eq:varied Wigner trace norm when the linear term vanishes}) simplifies to
\small{
\begin{equation}
\begin{split}
\| \rho(|\psi_i (\epsilon) \rangle) & \|_W - \| \rho(|\psi_i\rangle) \|_W   =  \frac{\epsilon^2}{1+\epsilon^2} \cdot \left( \sum_{\boldsymbol{\chi}\in\mathbb{V}_{1,d}} \sign [W_{\boldsymbol{\chi}}(\psi_i)] \cdot W_{\boldsymbol{\chi}}(\mu) - \| \psi_i \|_W \right)
 \\ & = \frac{\epsilon^2}{1+\epsilon^2} \cdot \left( \sum_{j=1}^d \left| \langle\varphi|\psi_j\rangle \right|^2 \sum_{\boldsymbol{\chi}\in\mathbb{V}_{1,d}} \sign [W_{\boldsymbol{\chi}}(\psi_i)] \cdot  W_{\boldsymbol{\chi}} (\psi_j) - \| \psi_i \|_W \right)   .
\end{split}
\label{eq:varied Wigner trace norm for a state having nonvanishing Wigner function everywhere and all eigenstates are nondegenerate}
\end{equation}
}
\normalsize{}
Therefore, all one needs to find is the $d\times d$ matrix $\mW(\ket{\psi_1},...,\ket{\psi_d})$, defined by
\begin{equation}
    \mW_{i,j}(\ket{\psi_1},...,\ket{\psi_d}) = \sum_{\boldsymbol{\chi}\in\mathbb{V}_{1,d}} \sign [W_{\boldsymbol{\chi}}(\psi_i)]\cdot W_{\boldsymbol{\chi}} (\psi_j),
\end{equation}
in order to investigate the behavior of the mana in the stationary points. Observing that the diagonal entry in each row of $\mW$ exceeds all other entries, one can deduce from Eq.~(\ref{eq:varied Wigner trace norm for a state having nonvanishing Wigner function everywhere and all eigenstates are nondegenerate}) that these states are local maxima of the mana.

\subsection{Single Qubit Examples}
For single qubits, the non-degenerate eigenstates of Clifford operations include six stabilizer states, eight $T$-states, and twelve $H$-states.
Figure~\ref{fig:Stabilizer fidelity for all single-qubit states} displays the stabilizer fidelity function for all qubits, where each state is plotted in its Bloch sphere representation, but instead of a unit radius, the radial coordinate corresponds to the stabilizer fidelity value.
The $T$-states appear as sharp minima in this function, while the $H$-states also form sharp minima in all but one direction.
In that particular direction, the $H$-states instead create a smooth maximum.

\begin{figure}[h!]
    \centering
    \includegraphics[trim=3.1cm 3cm 3.1cm 3cm,clip,width=0.6\linewidth]{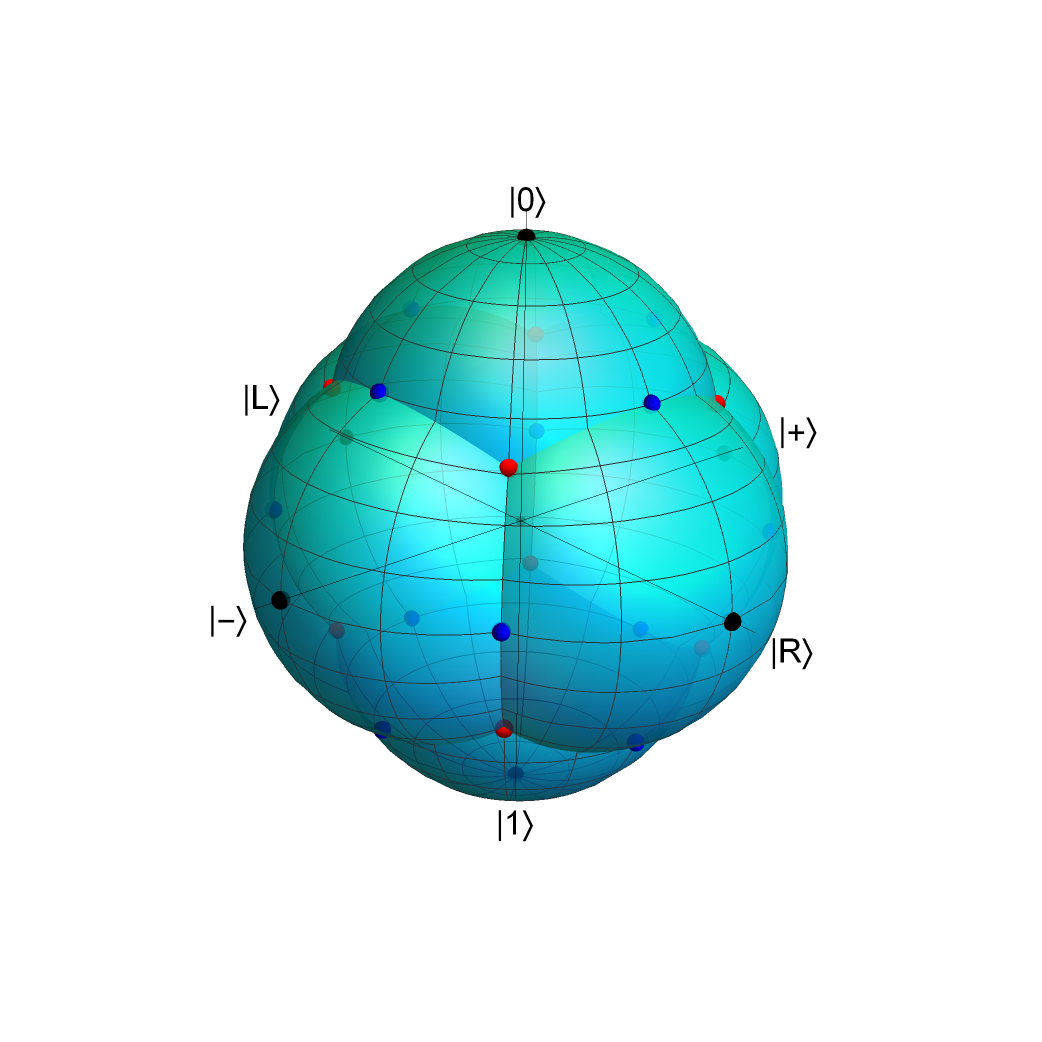}
    \caption{
        Stabilizer fidelity for all single-qubit states. 
        Each state is plotted in its Bloch sphere representation, but instead of a unit radius, the radial coordinate corresponds to the stabilizer fidelity value. 
        The black points represent the stabilizer states.
        The red points indicate the $T$-states, while the blue points correspond to the $H$-states. 
        The plot clearly shows that the $T$-states form sharp minima in the stabilizer fidelity function. Similarly, the $H$-states also create sharp minima in all but one direction. 
        In that particular direction, however, the $H$-states give rise to a smooth maximum instead.
    }
    \label{fig:Stabilizer fidelity for all single-qubit states}
\end{figure}

Now we calculate the corresponding $\mL$ matrices. For the $T$-states we work with the basis
\begin{equation}
    \ket{T_0}\equiv \sqrt{\frac{3+\sqrt{3}}{6}}\ket{0}+e^{i \frac{\pi}{4}}\sqrt{\frac{3-\sqrt{3}}{6}}\ket{1} \quad , \quad \ket{T_1}\equiv -\sqrt{\frac{3-\sqrt{3}}{6}}\ket{0}+e^{i \frac{\pi}{4}}\sqrt{\frac{3+\sqrt{3}}{6}}\ket{1}     ,
\end{equation}
and for the $H$-states we work with the basis
\begin{equation}
    \ket{H_0}\equiv \sqrt{\frac{2+\sqrt{2}}{4}}\ket{0}+\sqrt{\frac{2-\sqrt{2}}{4}}\ket{1} \quad , \quad \ket{H_1}\equiv -\sqrt{\frac{2-\sqrt{2}}{4}}\ket{0}+\sqrt{\frac{2+\sqrt{2}}{4}}\ket{1}     .
\end{equation}
The nearest stabilizer states to $\ket{T_0}$ are $\ket{0},\ket{+},\ket{R}$, and the nearest stabilizer states to $\ket{H_0}$ are $\ket{0},\ket{+}$, so that
\begin{subequations}
    \begin{gather}
            \mL\left( \ket{{T_0}}, \ket{T_1} \right)  = \left(
            \begin{array}{c}
              -\frac{1}{\sqrt{6}} \\
              \frac{e^{-\frac{i \pi }{3}}}{\sqrt{6}}
               \\
              \frac{e^{\frac{i \pi }{3}}}{\sqrt{6}}
               \\
            \end{array}
            \right)  \approx \left(
            \begin{array}{c}
              -0.408 \\
              0.204  -0.354 i \\
              0.204  +0.354 i \\
            \end{array}
            \right)   ,
        \\
        \mL\left( \ket{H_0}, \ket{H_1} \right)= \left(
        \begin{array}{c}
          -\frac{1}{2 \sqrt{2}} \\
          \frac{1}{2 \sqrt{2}} \\
        \end{array}
        \right) 
        \approx  \left(
        \begin{array}{c}
          -0.354 \\
          0.354 \\
        \end{array}
        \right)   .
    \end{gather}
    \label{eq:L matrices for qubits}
\end{subequations}
As expected, based on the above proofs, the sum of each column in these matrices vanishes. Notably, by varying $\ket{H_0}$ in a very specific direction, namely along $i\ket{H_1}$, the linear term in $\epsilon$ is nullified (recall Eq.~(\ref{eq:L matrix})). This is the only direction in which $\ket{H_0}$ is a smooth extremum along the path; in all other directions, it forms a sharp minimum.
For the $\ket{T_0}$ state, as evident from the $\mL$-matrix and Eq.~(\ref{eq:L matrix}), it exhibits a sharp minimum in every direction. This follows from the fact that, for any complex $\alpha\in\mathbb{C}$, the three quantities $\operatorname{Re}[-0.408 \alpha]$, $\operatorname{Re}[(0.204-0.354i) \alpha]$, and $\operatorname{Re}[(0.204+0.354i) \alpha]$ cannot simultaneously vanish.
These characteristics of the states are also clearly illustrated in Figure~\ref{fig:Stabilizer fidelity for all single-qubit states}.

For the behavior of the $2$-SRE
\footnote{
In Appendix~\ref{app:alpha-SREforsinglequbits}, we derive a closed-form expression for the $\alpha$-SRE of an arbitrary single-qubit pure state $\ket{\psi} = \cos\theta \ket{0} + e^{i\phi}\sin\theta \ket{1}$ in terms of the Bloch-sphere angles $(\theta,\phi)$. We then evaluate this general formula for the two canonical single-qubit magic states, $\ket{T}$ and $\ket{H}$, showing explicitly that $\ket{T}$ attains the larger stabilizer $\alpha$-SRE and in fact saturates the general upper bound from Eq.~(\ref{eq: Bound on SRE}). For the particularly relevant case $\alpha=2$, we obtain the concise values $M_2(\ket{T}) = \log(3/2)$ and $M_2(\ket{H}) = \log(4/3)$, which serve as convenient benchmarks for single-qubit magic in the main text.
}, consider first the $\ket{T_0}$ state, varied along the direction of $e^{i\phi}\cos\theta \ket{T_1}$:
\begin{equation}
    \ket{T_0(\epsilon)} = 
    \frac{\ket{T_0} + \epsilon e^{i\phi}\cos\theta \ket{T_1}}
    {\sqrt{1+\epsilon^{2}}}.
\end{equation}
The resulting quantity $\tilde{\Xi}_2(\ket{T_0(\epsilon)})$ is a polynomial in $\epsilon$ (but not in $\epsilon^2$), explicitly given by
\begin{equation}
    \Xi_2  \left(\ket{T_0(\epsilon)}\right)
    = 
    \frac{
        1 + 8\epsilon^{2} + 8\epsilon^{6} + \epsilon^{8} 
        + 4\sqrt{2} \epsilon^{3}(-1 + \epsilon^{2})\cos(3\phi)
    }{
        3(1 + \epsilon^{2})^{4}
    }.
\end{equation}
The expansion reveals that $\ket{T_0}$ is a \emph{minimum} of $\Xi_2$, as 
\[
\Xi_2^{(2)} = \frac{4}{3} > 0
\]
for all values of $\phi$.

Next, consider the variation of the $\ket{H_0}$ state along the direction $e^{i\phi}\cos\theta \ket{H_1}$:
\begin{equation}
    \ket{H_0(\epsilon)} = 
    \frac{\ket{H_0} + \epsilon e^{i\phi}\cos\theta \ket{H_1}}
    {\sqrt{1+\epsilon^{2}}}.
\end{equation}
Here too, $\tilde{\Xi}_2(\ket{H_0(\epsilon)})$ is polynomial in $\epsilon^2$, with
\begin{equation}
    \Xi_2  \left(\ket{H_0(\epsilon)}\right)
    = 
    \frac{3}{8}  \left[
        1
        + \frac{
            4\epsilon^{2}  \left(1 - 3\epsilon^{2} + \epsilon^{4}
            + \epsilon^{2}\cos(2\phi)\right)
        }{
            3(1+\epsilon^{2})^{4}
        }
        \big(1 + 3\cos(2\phi)\big)
    \right].
\end{equation}
In this case, the nature of the stationary point depends on the direction parameter~$\phi$:  
$\ket{H_0}$ is a local \emph{maximum} or \emph{minimum} depending on the sign of $1 + 3\cos(2\phi)$.  
When $3\cos(2\phi) = -1$, the variation becomes \emph{flat} along that direction, indicating a neutral curvature in the landscape of~$\Xi_2$.

\subsection{Single Qutrit Examples}

As can be checked in \cite{Veitch_2014,WangThaumaBounds,QutritQuquint} and Appendix~\ref{app:Wigner function for the single-qutrit states}, each non-degenerate eigenstate of a Clifford operator for single qutrits has non-vanishing Wigner function everywhere and is an eigenstate of a Clifford operation with simple spectrum (its eigenvalues are all distinct). The Clifford-inequivalent states and the Clifford operators, which have a simple spectrum and for which these states are eigenstates, are:
\begin{enumerate}
    \item The starnge state
    \begin{equation}
        \ket{\mathbb{S}}\equiv\frac{\ket{1}-\ket{2}}{\sqrt{2}}
    \end{equation}
    and the ``H plus'' state
    \begin{equation}
        \ket{H_+}=\frac{(1+\sqrt{3})\ket{0}+\ket{1}+\ket{2}}{\sqrt{2(3+\sqrt{3})}}
    \end{equation}
    state are non-degenerate eigenstates of the Hadamard gate $H=V_{\hat{H}}$. The third non-degenerate eigenstate, the ``H minus'' state, $\ket{H_-}=\frac{(1-\sqrt{3})\ket{0}+\ket{1}+\ket{2}}{\sqrt{2(3-\sqrt{3})}}$ is Clifford-equivalent to $\ket{H_+}$.
    \item The Norell state
    \begin{equation}
        \ket{\mathbb{N}}\equiv\frac{-\ket{0}+2\ket{1}-\ket{2}}{6}
    \end{equation}
    is a non-degenerate eigenstate of $(XH)V_{\hat{N}}(XH)^{-1}$. The other two non-degenerate eigenstates of this Clifford $\frac{\ket{0}-\ket{2}}{\sqrt{2}}$ and $\frac{\ket{0}+\ket{1}+\ket{2}}{\sqrt{3}}$ are Clifford-equivalent to $\ket{\mathbb{S}}$ and $\ket{0}$ respectively via $XH$.
    \item The $\mathbb{T}$-state
    \begin{equation}
        \ket{\mathbb{T}}\equiv\ket{\mathbb{T}_0}=\frac{\xi\ket{0}+\ket{1}+\xi^{-1}\ket{2}}{\sqrt{3}}
    \end{equation}
    where $\xi=e^{2\pi i/9}$ is a non-degenerate eigenstate of $(XH^2)N_{\hat{N}\hat{S}}(XH^2)^{-1}$. The other two non-degenerate eigenstates of this Clifford, $\ket{\mathbb{T}_1}=\frac{\xi^{-2}\ket{0}+\ket{1}+\xi^{2}\ket{2}}{\sqrt{3}}$ and $\ket{\mathbb{T}_2}=\frac{\xi^4\ket{0}+\ket{1}+\xi^{-4}\ket{2}}{\sqrt{3}}$, are Clifford equivalent to $\ket{\mathbb{T}}$ via $Z$ and $Z^2$ respectively.
\end{enumerate}

\subsubsection{$\mW$ matrices and Mana}
For each of the Clifford operators- namely,
\[
H = V_{\hat{H}}, \quad (XH)V_{\hat{N}}(XH)^{-1}, \quad \text{and} \quad (XH^2)N_{\hat{N}\hat{S}}(XH^2)^{-1},
\]
we compute the matrix $\mW$ for their non-degenerate eigenstates. The results are presented below:
\small{
\begin{subequations}
    \begin{gather}
        \mW\left( \ket{\mathbb{S}}, \ket{H_+}, \ket{H_-} \right)=\left(
\begin{array}{ccc}
 \frac{5}{3} & \frac{1}{3} & \frac{1}{3} \\
 -\frac{1}{3} & \frac{1}{3}+\frac{2}{\sqrt{3}} &
   \frac{1}{3}-\frac{2}{\sqrt{3}} \\
 -\frac{1}{3} & \frac{1}{3}-\frac{2}{\sqrt{3}} &
   \frac{1}{3}+\frac{2}{\sqrt{3}} \\
\end{array}
\right) \approx \left(
\begin{array}{ccc}
 1.667 & 0.333 & 0.333 \\
 -0.333 & 1.488 & -0.821 \\
 -0.333 & -0.821 & 1.488 \\
\end{array}
\right)   , \\
        \mW\left( \ket{\mathbb{N}}, \frac{\ket{0}-\ket{2}}{\sqrt{2}}, \frac{\ket{0}+\ket{1}+\ket{2}}{\sqrt{3}} \right)=\left(
\begin{array}{ccc}
 \frac{5}{3} & \frac{1}{3} & -\frac{1}{3} \\
 \frac{1}{3} & \frac{5}{3} & \frac{1}{3} \\
 0 & 0 & 1 \\
\end{array}
\right) \approx  \left(
\begin{array}{ccc}
 1.667 & 0.333 & -0.333 \\
 0.333 & 1.667 & 0.333 \\
 0& 0& 1. \\
\end{array}
\right)   , \\
        \begin{split}
                    \mW & \left( \ket{\mathbb{T}}, \frac{\xi^{-2}\ket{0}+\ket{1}+\xi^{2}\ket{2}}{\sqrt{3}}, \frac{\xi^4\ket{0}+\ket{1}+\xi^{-4}\ket{2}}{\sqrt{3}} \right) \\
                    & =\frac{1}{3}\left(
\begin{array}{ccc}
 1+4 \cos \left(\frac{\pi }{9}\right) & 1-2 \cos \left(\frac{\pi
   }{9}\right)-2 \sqrt{3} \sin \left(\frac{\pi }{9}\right) & 1-2 \cos
   \left(\frac{\pi }{9}\right)+2 \sqrt{3} \sin \left(\frac{\pi }{9}\right)
   \\
 1-2 \cos \left(\frac{\pi }{9}\right)+2 \sqrt{3} \sin \left(\frac{\pi
   }{9}\right) & 1+4 \cos \left(\frac{\pi }{9}\right) & 1-2 \cos
   \left(\frac{\pi }{9}\right)-2 \sqrt{3} \sin \left(\frac{\pi }{9}\right)
   \\
 1-2 \cos \left(\frac{\pi }{9}\right)-2 \sqrt{3} \sin \left(\frac{\pi
   }{9}\right) & 1-2 \cos \left(\frac{\pi }{9}\right)+2 \sqrt{3} \sin
   \left(\frac{\pi }{9}\right) & 1+4 \cos \left(\frac{\pi }{9}\right) \\
\end{array}
\right) \\
& \approx \left(
\begin{array}{ccc}
 1.586 & -0.688 & 0.102 \\
 0.102 & 1.586 & -0.688 \\
 -0.688 & 0.102 & 1.586 \\
\end{array}
\right)   .
        \end{split}
    \end{gather}
    \label{eq:SW matrices for qutrits}
\end{subequations}
}
\normalsize{}
Indeed, the diagonal entry in each row exceeds all other entries, and we deduce from Eq.~(\ref{eq:varied Wigner trace norm for a state having nonvanishing Wigner function everywhere and all eigenstates are nondegenerate}) that these states are local maxima of the mana.

\subsubsection{$\mL$ matrices and Stabilizer Fidelity}

After numbering the nearest stabilizer states corresponding to the four nonstabilizer states 
$\ket{\mathbb{S}}$, $\ket{\mathbb{N}}$, $\ket{H_+}$, and $\ket{\mathbb{T}}$ 
in the order presented in Table~\ref{tab:Single-qutrit nonstabilizer Clifford-inequivalent Clifford nondegenerate eigenstates and Wigner function} 
and Table~\ref{tab:Single-qutrit nonstabilizer Clifford-inequivalent Clifford nondegenerate eigenstates and nearest SS} 
in Appendix~\ref{app:Wigner function for the single-qutrit states}, 
we computed the corresponding $\mL$ matrices. 
The results are summarized in Table~\ref{tab:L matrices for single qutrits}.

As expected, the entries in each column of these matrices sum to zero. 
While the $\ket{\mathbb{S}}$ state exhibits a sharp minimum of the stabilizer fidelity in all directions, the states $\ket{\mathbb{N}}$, $\ket{H_+}$, and $\ket{\mathbb{T}}$ 
display a more intricate structure, suggesting richer geometric behavior in their neighborhoods. 
In the following, we analyze each of these four states in detail.

\begin{table}[h!]
\centering
\label{tab:L matrices qutrits}
\renewcommand{\arraystretch}{1.5}
\resizebox{\textwidth}{!}{
\begin{tabular}{|c|c|c|}
\hline
 \textbf{Basis} & \textbf{Exact $\mL$ Matrix} & \textbf{Approximate $\mL$ Matrix} \\
\hline

$\left\{ \ket{\mathbb{S}}, \ket{H_+}, \ket{H_-} \right\}$ &
$\begin{pmatrix}
\frac{1}{2\sqrt{3+\sqrt{3}}} & \frac{1}{2\sqrt{3-\sqrt{3}}} \\
-\frac{1}{2\sqrt{3+\sqrt{3}}} & -\frac{1}{2\sqrt{3-\sqrt{3}}} \\
-\frac{(1+\sqrt{3})e^{i\pi/4}}{2\sqrt{2(3+\sqrt{3})}} & \frac{(-1+\sqrt{3})e^{-i\pi/4}}{2\sqrt{2(3-\sqrt{3})}} \\
-\frac{(1+\sqrt{3})e^{-i\pi/4}}{2\sqrt{2(3+\sqrt{3})}} & \frac{(-1+\sqrt{3})e^{i\pi/4}}{2\sqrt{2(3-\sqrt{3})}} \\
\frac{(1+\sqrt{3})e^{-i\pi/4}}{2\sqrt{2(3+\sqrt{3})}} & -\frac{(-1+\sqrt{3})e^{i\pi/4}}{2\sqrt{2(3-\sqrt{3})}} \\
-\frac{i}{2\sqrt{3+\sqrt{3}}} & \frac{i}{2\sqrt{3-\sqrt{3}}} \\
\frac{i}{2\sqrt{3+\sqrt{3}}} & -\frac{i}{2\sqrt{3-\sqrt{3}}} \\
\frac{(1+\sqrt{3})e^{i\pi/4}}{2\sqrt{2(3+\sqrt{3})}} & -\frac{(-1+\sqrt{3})e^{-i\pi/4}}{2\sqrt{2(3-\sqrt{3})}}
\end{pmatrix}$ &
$\begin{pmatrix}
0.23 & 0.444 \\
-0.23 & -0.444 \\
-0.314-0.314i & 0.163-0.163i \\
-0.314+0.314i & 0.163+0.163i \\
0.314-0.314i & -0.163-0.163i \\
-0.23i & 0.444i \\
0.23i & -0.444i \\
0.314+0.314i & -0.163+0.163i
\end{pmatrix}$ \\
\hline

$\left\{ \ket{H_+}, \ket{H_-}, \ket{\mathbb{S}} \right\}$ &
$\begin{pmatrix}
-\frac{1+\sqrt{3}}{\sqrt{2}(3+\sqrt{3})} & 0 \\
-\frac{1}{6}(\sqrt{3}-3)\sqrt{2+\sqrt{3}} & 0
\end{pmatrix}$ &
$\begin{pmatrix}
-0.408 & 0 \\
0.408 & 0
\end{pmatrix}$ \\
\hline

$\left\{ \ket{\mathbb{N}}, \frac{\ket{0}-\ket{2}}{\sqrt{2}}, \frac{\ket{0}+\ket{1}+\ket{2}}{\sqrt{3}} \right\}$ &
$\begin{pmatrix}
 0 & \frac{\sqrt{2}}{3} \\
 0 & \frac{(\sqrt{3}-3i)^2}{18\sqrt{2}} \\
 0 & \frac{i(\sqrt{3}+i)}{3\sqrt{2}}
\end{pmatrix}$ &
$\begin{pmatrix}
 0 & 0.471 \\
 0 & -0.236-0.408i \\
 0 & -0.236+0.408i
\end{pmatrix}$ \\
\hline

$\left\{ \ket{\mathbb{T}}, \ket{\mathbb{T}_1} , \ket{\mathbb{T}_2} \right\}$ &
$\begin{pmatrix}
\frac{2}{9}(2\cos\frac{2\pi}{9} + \sin\frac{\pi}{18}) & \frac{2}{9}(-2\cos\frac{\pi}{9} + \cos\frac{2\pi}{9}) \\
\frac{2}{9}e^{-2\pi i/3}(2\cos\frac{2\pi}{9} + \sin\frac{\pi}{18}) & \frac{1}{9}e^{2\pi i/3}(-\sqrt{3}\cos\frac{\pi}{18} - 3\sin\frac{\pi}{18}) \\
\frac{2}{9}e^{2\pi i/3}(2\cos\frac{2\pi}{9} + \sin\frac{\pi}{18}) & \frac{1}{9}e^{-2\pi i/3}(-\sqrt{3}\cos\frac{\pi}{18} - 3\sin\frac{\pi}{18})
\end{pmatrix}$ &
$\begin{pmatrix}
0.379 & -0.247 \\
-0.19-0.328i & 0.124-0.214i \\
-0.19+0.328i & 0.124+0.214i
\end{pmatrix}$ \\
\hline
\end{tabular}}
\caption{
Exact and approximate forms of the $\mL$ matrices for the Clifford-inequivalent non-stabilizer non-degenerate Clifford eigenstates for single-qutrits. Each matrix is computed with respect to a basis of nearest stabilizer states as identified in Table~\ref{tab:Single-qutrit nonstabilizer Clifford-inequivalent Clifford nondegenerate eigenstates and nearest SS}. The matrices quantify the first-order variation in the squared overlap amplitude under general perturbations, and serve as a diagnostic tool for analyzing stabilizer fidelity landscapes around each non-stabilizer state.
}
\label{tab:L matrices for single qutrits}
\end{table}

\subsubsection*{The $\ket{\mathbb{S}}$ state}
It can be observed from the first row of Table~\ref{tab:L matrices for single qutrits} that the strange state, $\ket{\mathbb{S}}$, constitutes a sharp minimum along each considered direction. 
In particular, there exist no complex amplitudes $\alpha_2,\alpha_3 \in \mathbb{C}$ (except for the trivial case $\alpha_2=\alpha_3=0$) such that the superposition
\begin{equation}
    \ket{\varphi} \equiv \alpha_2 \ket{H_+} + \alpha_3 \ket{H_-}
\end{equation}
yields to
\begin{equation}
    \operatorname{Re} \ell_{\ket{\mathbb{S}}}  \left( \ket{s_{\ket{\mathbb{S}},i}}; \ket{\varphi} \right)
= 
\operatorname{Re}  \left[
  \alpha_2 \mL_{i,2}(\ket{\mathbb{S}},\ket{H_+},\ket{H_-})
 + \alpha_3 \mL_{i,3}(\ket{\mathbb{S}},\ket{H_+},\ket{H_-})
\right]
\end{equation}
that vanishes simultaneously for $i = 1, 2, 6, 7$.
This confirms that $\ket{\mathbb{S}}$ indeed corresponds to a local extremum that is a strict and sharp minimum in all relevant directions.

\subsubsection*{The $\ket{\mathbb{N}}$ state}
We observe from row 3 in Table~\ref{tab:L matrices for single qutrits} that for the state $\ket{\mathbb{N}}$, there exists a family of directions in the orthogonal complement of its spanned subspace, namely, the directions defined by $e^{i\phi}\frac{\ket{0}-\ket{2}}{\sqrt{2}}$ for all $\phi$, such that any continuous, finite (but bounded) perturbation of $\ket{\mathbb{N}}$ along these directions does not induce a first-order change in the stabilizer fidelity. 
Calculating $\braket{s|\mu|s}$ for the nearest stabilizer states yields a constant value of $-\frac{2}{3}$ in all cases, indicating that $\ket{\mathbb{N}}$ is a local maximum of the stabilizer fidelity on this two-dimensional surface
\footnote{
For qutrits, varying a state while keeping it normalized and disregarding its overall global phase ultimately results in four degrees of freedom. Consequently, these variations take place within a four-dimensional manifold, unlike the two-dimensional case for a single qubit, which is the so-called projective Hilbert space. 
The two-dimensional surface referred to here consists of the set of states given by $\frac{\ket{\mathbb{N}}+\epsilon e^{i\phi}\frac{\ket{0}-\ket{2}}{\sqrt{2}}}{\sqrt{1+\epsilon^2}}$.
}
:
\begin{equation}
    F\left( \frac{\ket{\mathbb{N}}+\epsilon e^{i\phi}\frac{\ket{0}-\ket{2}}{\sqrt{2}}}{\sqrt{1+\epsilon^2}} \right)=\frac{2}{3}-\frac{2}{3}\frac{\epsilon^2}{1+\epsilon^2}=\frac{2}{3(1+\epsilon^2)}
\end{equation}
Once we vary $\ket{\mathbb{N}}$ outside this surface, we find that this point is a sharp local minimum. This behavior is analogous to the $\ket{H_0}$ state in two dimensions, which acts as a smooth maximum when variations are restricted along a single line. However, when variations extend beyond this line, $\ket{H_0}$ becomes a sharp local minimum along the new path.

\subsubsection*{The $\ket{H_+}$ state}
From row 2 in Table~\ref{tab:L matrices for single qutrits}, we observe that $\ket{H_+}$ exhibits a distinct behavior. For example, it is immediate to see that there exists a family of directions in the orthogonal complement of its spanned subspace, namely, the directions defined by $e^{i\phi}\ket{\mathbb{S}}$ for all $\phi$, such that any perturbation of $\ket{H_+}$ along these directions does not induce a first-order change in the stabilizer fidelity. 
Computing $\braket{s|\mu|s}$ for the two nearest stabilizer states yields a value of $-\frac{3+\sqrt{3}}{6}$ for both, confirming that $\ket{H_+}$ is a local maximum on this two-dimensional surface.
When varying outside this surface, but not in the specific direction of $i \ket{H_-}$, we find that this point behaves as a sharp local minimum. However, when the variation occurs along the $i \ket{H_-}$ direction, there is no linear term, and calculating $\braket{s|\mu|s}$ yields $-\frac{1}{\sqrt{3}}$ for each of the two nearest stabilizer states.
Now, let us take the $\ket{\varphi}$ along which we vary to be
\begin{equation}
    \ket{\varphi}=i \cos\alpha \ket{H_-}+e^{i\phi} \sin \alpha \ket{\mathbb{S}}   .
\end{equation}
Calculating the second term in Eq.~(\ref{eq:variation in fidelity when the linear sharp term vanishes}) for the two nearest stabilizer states to $\ket{H_+}$,
\begin{subequations}
    \begin{gather}
        \braket{0|\mu|0}= -\frac{1 + (-2 + \sqrt{3}) \cos^2\alpha}{3 - \sqrt{3}} \\ 
        \frac{\bra{0}+\bra{1}+\bra{2}}{\sqrt{3}}\mu\frac{\ket{0}+\ket{1}+\ket{2}}{\sqrt{3}}=-\frac{1 + (-2 + \sqrt{3}) \cos^2\alpha}{3 - \sqrt{3}}
    \end{gather}
\end{subequations}
we find that
\begin{equation}
    \label{eq:fidelity of varied norell in the hypersurface}
    F\left( \frac{\ket{H_+}+\epsilon (i \cos\alpha \ket{H_-}+e^{i\phi} \sin \alpha \ket{\mathbb{S}}) }{\sqrt{1+\epsilon^2}} \right)
    = F\left(\ket{H_+}\right) - \frac{\epsilon^2}{1+\epsilon^2} \frac{1 + (-2 + \sqrt{3}) \cos^2\alpha}{3 - \sqrt{3}},
\end{equation}
indicating that the state $\ket{H_+}$ is a smooth maximum within this three-dimensional hypersurface. If we move outside this hypersurface, however, the state becomes a sharp local minimum along the path.

\subsubsection*{The $\ket{\mathbb{T}}$ state}

From the fourth row of Table~\ref{tab:L matrices for single qutrits}, 
one can directly solve and find that the condition
\begin{equation}
    \operatorname{Re} \ell_{\ket{\mathbb{T}}}  \left( \ket{s_{\ket{\mathbb{T}},i}}; \ket{\varphi} \right) = 0
\end{equation}
for all $i=1,2,3$ is satisfied only for variations of the form
\begin{equation}
    \ket{\varphi}
 = \cos\gamma  e^{i\phi} \ket{\mathbb{T}_1}
 + \sin\gamma  e^{-i\phi} \ket{\mathbb{T}_2},
\end{equation}
where
\begin{equation}
    \cos\gamma = \sqrt{\frac{1}{1 + 4\cos^2  \left(\tfrac{2\pi}{9}\right)}},
\end{equation}
and $\phi \in [0,2\pi)$ is a free phase parameter.
For any variation outside this two-dimensional manifold, 
the state $\ket{\mathbb{T}}$ constitutes a sharp minimum of the stabilizer fidelity.

We now calculate the second term in Eq.~(\ref{eq:variation in fidelity when the linear sharp term vanishes}) for the three nearest stabilizer states to $\ket{\mathbb{T}}$. The results are
\begin{subequations}
    \begin{gather}
        \frac{\bra{0}+\bra{1}+\bra{2}}{\sqrt{3}}\mu\frac{\ket{0}+\ket{1}+\ket{2}}{\sqrt{3}}=-A - B \cos(2\phi)\\
        \frac{e^{-2\pi i/3}\bra{0}+\bra{1}+\bra{2}}{\sqrt{3}}\mu\frac{e^{2\pi i/3}\ket{0}+\ket{1}+\ket{2}}{\sqrt{3}}=-A + \frac{B}{2} \cos(2\phi) - C \sin(2\phi)\\
        \frac{e^{-2\pi i/3}\bra{0}+e^{-2\pi i/3}\bra{1}+\bra{2}}{\sqrt{3}}\mu\frac{e^{2\pi i/3}\ket{0}+e^{2\pi i/3}\ket{1}+\ket{2}}{\sqrt{3}}=-A + \frac{B}{2} \cos(2\phi) + C \sin(2\phi),
    \end{gather}
\end{subequations}
where
\begin{subequations}
    \begin{gather}
        A=\frac{28 \cos  \frac{\pi}{9} - 18 \sin  \frac{\pi}{18} + 7}{51}\approx0.591877,\\
        B=\frac{30 - 16 \cos  \frac{\pi}{9} + 20 \sin  \frac{\pi}{18}}{153}\approx0.120509,\\
        C=\frac{\sqrt{3} - 2 \sin  \frac{\pi}{9}}{3  \left(3 + 2 \sin  \frac{\pi}{18}\right)}\approx0.104364.
    \end{gather}
\end{subequations}
Therefore, the $\ket{\mathbb{T}}$ state is a smooth maximum in this two-dimensional surface.

\subsubsection{The SRE}

In Appendix~\ref{app:SREforSingleQutrits}, we analyze in detail the SRE for single-qutrit states. We first derive the full discrete phase-space probability distribution $P_{\boldsymbol{\chi}}$ associated with a general single-qutrit pure state parametrized by four real angles, and from it obtain an explicit expression for the quantity $\Xi_2$ and hence for the $2$-SRE $M_2=-\log\bigl[3\Xi_2\bigr]$. We then evaluate the resulting $\alpha$-SRE for the four Clifford-inequivalent nonstabilizer eigenstates of single-qutrit Clifford operations, $\ket{\mathbb{S}},\ket{\mathbb{N}},\ket{H_+}$, and $\ket{\mathbb{T}}$, providing closed-form expressions valid for general $\alpha$ and simple logarithmic values when $\alpha=2$. This analysis reveals, in particular, that the Strange and Norell states saturate the same upper bound on SRE as in Eq.~(\ref{eq: Bound on SRE}), thereby identifying them as maximally magic within this family of qutrit states. Here we show the behavior near these four extremal points.

\subsubsection*{The Strange State $\ket{\mathbb{S}}$}
Consider varying the strange state $\ket{\mathbb{S}}$ along the general orthogonal direction
\begin{equation}
    \ket{\varphi}
    = e^{i\phi_1}\cos\theta \ket{0}
    + e^{i\phi_2}\sin\theta \frac{\ket{1}+\ket{2}}{\sqrt{2}}.
\end{equation}
This yields a function $\tilde{\Xi}_2(\ket{\mathbb{S}(\epsilon)})$ that is polynomial in $\epsilon^2$, with
\begin{equation}
    \Xi_2^{(2)} = \frac{2}{3}\sin^2\theta.
\end{equation}
Hence, $\ket{\mathbb{S}}$ is a local minimum of $\Xi_2(\ket{\mathbb{S}(\epsilon)})$ for all $\theta\neq0,\pi$.  
For $\theta=0$ or $\pi$, one finds
\begin{equation}
    \Xi_2(\ket{\mathbb{S}(\epsilon)})
    = \frac{1}{6} + \frac{\epsilon^4}{(1+\epsilon^2)^4},
\end{equation}
which confirms that $\ket{\mathbb{S}}$ is also a local minimum in these directions.  
Consequently, the strange state $\ket{\mathbb{S}}$ represents a \emph{maximum} of the magic monotone $M_\alpha(\ket{\mathbb{S}(\epsilon)})$ in all directions.

\subsubsection*{The Norrell State $\ket{\mathbb{N}}$}
Next, consider varying the Norrell state $\ket{\mathbb{N}}$ by the general orthogonal state
\begin{equation}
    \ket{\varphi}
    = e^{i\phi_1}\cos\theta \frac{\ket{0}-\ket{2}}{\sqrt{2}}
    + e^{i\phi_2}\sin\theta \frac{\ket{0}+\ket{1}+\ket{2}}{\sqrt{3}}.
\end{equation}
In this case, $\tilde{\Xi}_2(\ket{\mathbb{N}(\epsilon)})$ is polynomial in $\epsilon$ (but not in $\epsilon^2$), and
\begin{equation}
    \Xi_2^{(2)} = \frac{2}{3}\sin^2\theta.
\end{equation}
Therefore, $\ket{\mathbb{N}}$ is a local minimum of $\Xi_2(\ket{\mathbb{N}(\epsilon)})$ for all $\theta\neq0,\pi$.  
For $\theta=0$ or $\pi$, one finds
\begin{equation}
    \Xi_2(\ket{\mathbb{N}(\epsilon)})
    = \frac{1}{6} + \frac{8\sin^4\phi_1 \epsilon^4}{3(1+\epsilon^2)^4},
\end{equation}
showing that $\ket{\mathbb{N}}$ remains a minimum provided $\phi_1\neq0,\pi$.  
When $\phi_1=0$ or $\pi$, the function becomes constant.  
Thus, $\ket{\mathbb{N}}$ is likewise a \emph{maximum} of $M_\alpha(\ket{\mathbb{N}(\epsilon)})$ in all directions.

\subsubsection*{The $\ket{\mathbb{T}}$ State}
Next, consider varying the $\ket{\mathbb{T}}$ state by the general orthogonal state
\begin{equation}
    \ket{\varphi}
    = e^{i\phi_1}\cos\theta \ket{\mathbb{T}_1}
    + e^{i\phi_2}\sin\theta \ket{\mathbb{T}_2}.
\end{equation}
The resulting function $\tilde{\Xi}_2(\ket{\mathbb{T}(\epsilon)})$ is polynomial in $\epsilon$ (but not in $\epsilon^2$), and its low-order expansion coefficients are
\begin{subequations}
    \begin{gather}
        \Xi_2^{(2)}
        = \frac{4}{9}\cos\theta
          \left[\cos\theta - 2\sin\theta\cos(\phi_1+\phi_2)\right], \\[4pt]
        \begin{split}
            \Xi_2^{(3)} = \frac{8}{27}\big[
                &2\cos^3\theta \cos(3\phi_1)
                - 6\cos^2\theta \cos(2\phi_1 - \phi_2)\sin\theta  \\
                &+ 3\cos\theta \cos(\phi_1 - 2\phi_2)\sin^2\theta
                - \cos(3\phi_2)\sin^3\theta
            \big],
        \end{split} \\[4pt]
        \begin{split}
            \Xi_2^{(4)} = \frac{1}{18}\big[
                &-21 - 16\cos(2\theta) - 3\cos(4\theta)
                + 4\cos  \big(2(\phi_1+\phi_2)\big)\sin^2(2\theta) \\
                &\quad - 6\cos(\phi_1+\phi_2)\big(-6\sin(2\theta) + \sin(4\theta)\big)
            \big].
        \end{split}
    \end{gather}
\end{subequations}

From $\Xi_2^{(2)}$, we deduce that $\ket{\mathbb{T}}$ is a local minimum of $\Xi_2(\ket{\mathbb{T}(\epsilon)})$ for all $\theta\neq\pm\frac{\pi}{2}$, provided that $\cot\theta \neq 2\cos(\phi_1+\phi_2)$.  
For $\theta=\pm\frac{\pi}{2}$, we obtain
\begin{equation}
    \Xi_2^{(3)} = \mp\frac{8}{27}\cos(3\phi_2),
    \qquad
    \Xi_2^{(4)} = -\frac{4}{9}.
\end{equation}
Thus, when $\cos(3\phi_2)\neq0$, the $\ket{\mathbb{T}}$ state is an \emph{inflection point} along the path $\ket{\mathbb{T}(\epsilon)}$.  
When $\cos(3\phi_2)=0$, it becomes a \emph{maximum}, since $\Xi_2^{(4)}<0$.  
We leave the special case $\cot\theta = 2\cos(\phi_1+\phi_2)$, but note that if $\Xi_2^{(2)}=\Xi_2^{(3)}=0$, then necessarily $\Xi_2^{(4)}\neq0$, indicating that the fourth-order expansion is sufficient to capture the local behavior.

\subsubsection*{The $\ket{H_+}$ State}

Finally, consider varying the $\ket{H_+}$ state along
\begin{equation}
    \ket{\varphi}
    = e^{i\phi_1}\cos\theta \ket{H_-}
    + e^{i\phi_2}\sin\theta \ket{\mathbb{S}}.
\end{equation}
The resulting function $\tilde{\Xi}_2(\ket{H_+(\epsilon)})$ is polynomial in $\epsilon$ (but not in $\epsilon^2$), and its low-order expansion coefficients are
\begin{subequations}
    \begin{gather}
        \Xi_2^{(2)}
        = \frac{1}{2} \left( \cos(2\theta_1) + \sqrt{3}\cos^2  \theta_1 \cos(2\phi_1) \right), \\[4pt]
        \Xi_2^{(3)}
        = 0, \\[4pt]
        \begin{split}
            \Xi_2^{(4)} = \frac{1}{48} & \left( -9 - 132\cos(2\theta_1) + 9\cos(4\theta_1) + 48\sqrt{3}\cos^2  \theta_1 (-3 + \cos(2\theta_1))\cos(2\phi_1) \right. \\
            & \quad \left. + 12\cos^4  \theta_1\cos(4\phi_1) - 8\sqrt{3}\cos(4\phi_2)\sin^4  \theta_1 \right)
        \end{split} \\[4pt]
        \Xi_2^{(5)} = \frac{2}{3}\sqrt{2}  \cos\theta_1  \big( -3\cos(\phi_1 - 2\phi_2) + \sqrt{3}\cos(\phi_1 + 2\phi_2) \big) \sin^4  \theta_1.
    \end{gather}
\end{subequations}
We leave the investigation here for the reader, but note that if $\Xi_2^{(2)}=\Xi_2^{(4)}=0$, then necessarily $\Xi_2^{(5)}\neq0$, indicating that the fifth-order expansion is sufficient to capture the local behavior.

\subsection{Single Ququint Examples}

Since there may be typographical errors in \cite{QutritQuquint}, we present the states explicitly for clarity.
Defining the following constants
\begin{subequations}
    \begin{gather}
        \chi = \sqrt{\frac{5 + \sqrt{5}}{10}}   , \\
        \eta_\pm = \mp\sqrt{30 - 6\sqrt{5}} + \sqrt{5} - 3   , \\
        \kappa_\pm = \frac{1}{2} \left( \pm \sqrt{6(5 + \sqrt{5})} - \sqrt{5} - 3 \right)   ,
    \end{gather}
\end{subequations}
and labeling each eigenstate by its corresponding eigenoperator (i.e., the operator for which it is an eigenstate) and eigenvalue, the unnormalized nonstabilizer Clifford-inequivalent non-degenerate eigenstates of Clifford operations on single ququints are \textbf{a subset} of
\begin{subequations}
    \begin{gather}
        \ket{H,\pm i} \propto \sqrt{1\pm\chi}(\ket{1}-\ket{4}) 
        \pm \sqrt{1\mp \chi}(\ket{2} -\ket{3})   , \\
        \ket{H,-1} \propto (1-\sqrt{5})\ket{0} 
        + \ket{1} 
        +\ket{2} 
        + \ket{3} 
        + \ket{4}   , \\
        \ket{X V_{\hat{S}},1} \propto \ket{0} + \ket{1} + \omega^3\ket{2} + \ket{3} + \omega^2\ket{4}   , \\
        \ket{B^\prime,-1} \propto (3+\sqrt{5})\ket{0} 
        - 2(\ket{1} + \ket{2} + \ket{3} + \ket{4})   , \\
        \ket{B^\prime,-e^{\pm \frac{2\pi i}{3}}} \propto 4 \kappa_\pm\ket{0} 
        - \kappa_\pm^2(\ket{1}+\ket{4}) +4(\ket{2} +\ket{3})   , \\
        \ket{B^\prime,e^{\pm\frac{2\pi i}{3}}} \propto \eta_\pm(\ket{1}-\ket{4} )
        + 4(\ket{2} - \ket{3} )   , \\
        \ket{A_{\pm\omega_5^2}} \propto \ket{2} \pm \ket{3}   ,
    \end{gather}
\end{subequations}
where $A=V_{\hat{S}}H^2$ , $B=H^3V_{\hat{S}}$ and $B^\prime=V_{\hat{K}} B V_{\hat{K}}^{-1}$, when $\hat{K}=\begin{pmatrix} 1 & 2 \\ 2 & 0 \end{pmatrix}$.
Among these states, $\ket{H,\pm i}$ are Clifford-equivalent, $\ket{B',e^{\pm \frac{2\pi i}{3}}}$ are Clifford-equivalent, and $\ket{B',-e^{\pm \frac{2\pi i}{3}}}$ are also Clifford-equivalent. However, $\ket{A_{\pm\omega_5^2}}$ are not Clifford-equivalent. This can be seen immediately by evaluating the set of overlap amplitudes with all stabilizer states of both states.

The Clifford operations $X V_{\hat{S}}$, $B'$, and $A$ each have a simple spectrum, meaning that all their eigenstates are non-degenerate.
For the Hadamard operation, however, the eigenvalue $1$ is degenerate, and its eigenspace is spanned by
\begin{equation}
\ket{H,1;1} \propto (1+\sqrt{5})\ket{0} + 2\big( \ket{1} + \ket{4} \big),
\end{equation}
and
\begin{equation}
\ket{H,1;2} \propto (1+\sqrt{5})\ket{0} + 2\big( \ket{2} + \ket{3} \big).
\end{equation}
For completeness, we present in Appendix~\ref{app:The bases worked with for each single-ququint Clifford} the bases of orthonormal eigenstates we work with for each Clifford.

\subsubsection{$\mL$ Matrices and Stabilizer Fidelity}

Following the ordering of the nearest stabilizer states as given in Table~\ref{tab:Single-ququint nonstabilizer Clifford-inequivalent Clifford nondegenerate eigenstates and Wigner function} and Table~\ref{tab:Single-ququint nonstabilizer Clifford-inequivalent Clifford nondegenerate eigenstates and nearest SS}, we computed the corresponding $\mL$ matrices. The full results are provided in Appendix~\ref{app:L matrices and Stabilizer Fidelity for Single Ququint Clifford-stabilizer States}.

To assess the local stabilizer fidelity behavior, we consider a general variation $\ket{\varphi} = \sum_{j=1}^5 (a_j + i b_j) \ket{\psi_j}$ with real coefficients $a_j, b_j$ satisfying $\sum_{j=1}^5 (a_j^2 + b_j^2) = 1$ and $\ket{\psi_j}$ running over the five eigenstates of the relevant operator. We compute the first-order variation in $\epsilon$ of the squared amplitude overlap with the nearest stabilizer states for each investigated Clifford-inequivalent, non-degenerate Clifford eigenstate.

Our findings, detailed in Appendix~\ref{app:L matrices and Stabilizer Fidelity for Single Ququint Clifford-stabilizer States}, indicate that none of these states constitutes a sharp minimum of the stabilizer fidelity in all directions. A more comprehensive analysis of their directional behavior is left for the reader.

Moreover, numerical evidence from extensive random sampling suggests the presence of states whose stabilizer fidelity is strictly smaller than that of any non-degenerate eigenstate of the single-ququint Clifford group.
This reveals that the global minimum of stabilizer fidelity dis not achieved by any of these eigenstates.

\subsubsection{$\mW$ matrices and Mana}

\subsubsection*{Smooth Local Maxima}

All eigenstates of $B^\prime$ are non-degenerate and possess Wigner functions that are nowhere vanishing.
Therefore, Eq.~(\ref{eq:varied Wigner trace norm for a state having nonvanishing Wigner function everywhere and all eigenstates are nondegenerate}) applies to them.
We compute the corresponding $\mW$ matrix and obtain:
\begin{equation}
\label{eq:W matrices for ququints - Btag}
 \begin{split}
             \mW_{B^\prime}:=\mW & \left( \ket{B^\prime,-1}, \ket{B^\prime,  -e^{\frac{2\pi i}{3}}}, \ket{B^\prime,-e^{-\frac{2\pi i}{3}}} , \ket{B^\prime, e^{\frac{2\pi i}{3}}} , \ket{B^\prime,e^{-\frac{2\pi i}{3}}} \right)
             \\
             & \approx 
\left(
\begin{array}{ccccc}
 1.98885 & -0.694427 & -0.694427 & -0.2 & -0.2 \\
 -0.294427 & 2.11335 & -0.0189273 & 0.651682 & 0.148318 \\
 -0.294427 & -0.0189273 & 2.11335 & 0.148318 & 0.651682 \\
 1.09443 & -0.498895 & 0.00446812 & 1.86614 & -0.266141 \\
 1.09443 & 0.00446812 & -0.498895 & -0.266141 & 1.86614 \\
\end{array}
\right)
          .
 \end{split}
\end{equation}
In each row, the diagonal element dominates, and, in accordance with Eq.~(\ref{eq:varied Wigner trace norm for a state having nonvanishing Wigner function everywhere and all eigenstates are nondegenerate}), each of these states is a smooth local maximum of the mana of magic in all directions.

All non-degenerate eigenstates of $H$ possess Wigner functions that are nowhere vanishing too. Computing the $\mW$ matrix for the $H$ operator eigenstates yields:
\begin{equation}
\label{eq:W matrices for ququints - H}
    \begin{split}
        \mW_H := \mW & \left( \ket{H,+i}, \ket{H,-i}, \ket{H,-1} , \ket{H,1;1} , \ket{H,1;2} \right) \\ & \approx 
        \left(
\begin{array}{ccccc}
 1.83107 & -0.631073 & -0.6 & 0.385871 & -0.0929356 \\
 -0.631073 & 1.83107 & -0.6 & -0.0929356 & 0.385871 \\
 0.2 & 0.2 & 1.8 & 0.307064 & 0.307064 \\
 1.07023 & 0.129772 & -0.0472136 & 1.90706 & 0.653532 \\
 0.129772 & 1.07023 & -0.0472136 & 0.653532 & 1.90706 \\
\end{array}
\right)
  .     
\end{split}
\end{equation}
Also here the diagonal element dominates in every row. However, Eq.~(\ref{eq:varied Wigner trace norm for a state having nonvanishing Wigner function everywhere and all eigenstates are nondegenerate}) does not directly apply in this case, as there are two degenerate eigenstates. Twirling $\mu$ in this situation preserves the entire density matrix supported on the subspace spanned by these two states. Therefore, in accordance with Eq.~(\ref{eq:varied Wigner trace norm when the linear term vanishes}), we must also compute:
\begin{equation}
\begin{split}
        \sum_{\boldsymbol{\chi}\in\mathbb{V}_{N,d}} & s_{\rho(\ket{H,i})}(\boldsymbol{\chi})\cdot W_{\rho(\cos{\alpha}\ket{H,1;1}+e^{i\phi}\sin{\alpha}\ket{H,1;2})}=\cos^2{\alpha} {\mW_H}_{1,4}+\sin^2{\alpha} {\mW_H}_{1,5}
        \\ & + \frac{1}{2}\sin{2\alpha} \sum_{\boldsymbol{\chi}\in\mathbb{V}_{N,d}} s_{\rho(\ket{H,i})}(\boldsymbol{\chi})\cdot (W_{\ket{H,1;1}\bra{H,1;2}}+W_{\ket{H,1;2}\bra{H,1;1}})
        \\
        & \approx 0.385871\cos^2{\alpha} - 0.440895 \cos{\alpha}\sin{\alpha} - 0.0929356 \sin^2{\alpha} < {\mW_H}_{1,1} ,
\end{split}
\end{equation}
and hence $\ket{H,i}$ is a local maximum of the mana in all directions. Similarly,
\begin{equation}
\begin{split}
        \sum_{\boldsymbol{\chi}\in\mathbb{V}_{N,d}} & s_{\rho(\ket{H,-1})}(\boldsymbol{\chi})\cdot W_{\rho(\cos{\alpha}\ket{H,1;1}+e^{i\phi}\sin{\alpha}\ket{H,1;2})}=\cos^2{\alpha} {\mW_H}_{3,4}+\sin^2{\alpha} {\mW_H}_{3,5}
        \\ & + \frac{1}{2}\sin{2\alpha} \sum_{\boldsymbol{\chi}\in\mathbb{V}_{N,d}} s_{\rho(\ket{H,-1})}(\boldsymbol{\chi})\cdot (W_{\ket{H,1;1}\bra{H,1;2}}+W_{\ket{H,1;2}\bra{H,1;1}})
        \\
        & \approx 0.307064 + 1.56209 \cos{\alpha}\sin{\alpha} < {\mW_H}_{3,3} ,
\end{split}
\end{equation}
and hence $\ket{H,-1}$ is a local maximum of the mana in all directions.

\subsubsection*{Anisotropic Saddle States}

The Wigner function of $\ket{A,-\omega^2=e^{-\frac{\pi i}{5}}}\propto \ket{2}-\ket{3}$ does vanish in some phase-space points. Therefore, one has to compute the linear term in Eq.~(\ref{eq: variation of Wigner norm}). For that, we define
\begin{equation}
    \bar{\sigma}(\ket{\psi_i},\ket{\psi_j}):= \ket{\psi_i} \bra{\psi_j} - \braket{\psi_i|\psi_j}  \ket{\psi_i} \bra{\psi_i}
\end{equation}
such that
\begin{equation}
    \sigma(\ket{\psi_i},\ket{\psi_j})=\bar{\sigma}(\ket{\psi_i},\ket{\psi_j})+{\bar{\sigma}}^\dagger (\ket{\psi_i},\ket{\psi_j})   ,
\end{equation}
and hence,
\begin{equation}
    W_{\sigma(\ket{\psi_i},\ket{\psi_j})}=W_{\bar{\sigma}(\ket{\psi_i},\ket{\psi_j})}+W^*_{\bar{\sigma}(\ket{\psi_i},\ket{\psi_j})}
\end{equation}
and $\bar{\sigma}(\ket{\psi_i},\ket{\psi_j})$ is linear in $\ket{\varphi}$. Then we compute:
\tiny
\begin{subequations}
\begin{gather}
        W_{\bar{\sigma}( \ket{A,e^{-\frac{\pi i}{5}}}, \ket{A,e^{\frac{\pi i}{5}}})} \approx \left(
\begin{array}{ccccc}
 0 & 0 & 0 & 0 & 0 \\
 -0.1 & -0.0309+0.0951 i & 0.0809  +0.0588 i & 0.0809  -0.0588 i & -0.0309-0.0951 i
   \\
 0.1 & -0.0809+0.0588 i & 0.0309  -0.0951 i & 0.0309  +0.0951 i & -0.0809-0.0588 i \\
 0.1 & -0.0809-0.0588 i & 0.0309  +0.0951 i & 0.0309  -0.0951 i & -0.0809+0.0588 i \\
 -0.1 & -0.0309-0.0951 i & 0.0809  -0.0588 i & 0.0809  +0.0588 i & -0.0309+0.0951 i
   \\
\end{array}
\right)   ,
        \\
        W_{\bar{\sigma}( \ket{A,e^{-\frac{\pi i}{5}}}, \ket{A,e^{\frac{4\pi i}{5}}})} \approx \left(
\begin{array}{ccccc}
 0 & -0.1902 i & -0.1176 i & 0.1176 i & 0.1902 i \\
 0 & 0 & 0 & 0 & 0 \\
 0.1 & 0.1 & 0.1 & 0.1 & 0.1 \\
 -0.1 & -0.1 & -0.1 & -0.1 & -0.1 \\
 0 & 0 & 0 & 0 & 0 \\
\end{array}
\right)   ,
        \\
        W_{\bar{\sigma}( \ket{A,e^{-\frac{\pi i}{5}}}, \ket{A,e^{-\frac{4\pi i}{5}}})} \approx \left(
\begin{array}{ccccc}
 0 & 0 & 0 & 0 & 0 \\
 -0.1 & -0.0309+0.0951 i & 0.0809  +0.0588 i & 0.0809  -0.0588 i & -0.0309-0.0951 i
   \\
 -0.1 & 0.0809  -0.0588 i & -0.0309+0.0951 i & -0.0309-0.0951 i & 0.0809  +0.0588 i
   \\
 0.1 & -0.0809-0.0588 i & 0.0309  +0.0951 i & 0.0309  -0.0951 i & -0.0809+0.0588 i \\
 0.1 & 0.0309  +0.0951 i & -0.0809+0.0588 i & -0.0809-0.0588 i & 0.0309  -0.0951 i \\
\end{array}
\right)   ,
        \\
        W_{\bar{\sigma}( \ket{A,e^{-\frac{\pi i}{5}}}, \ket{A,1})} \approx \left(
\begin{array}{ccccc}
 0 & 0 & 0 & 0 & 0 \\
 0.1414 & -0.1144+0.0831 i & 0.0437  -0.1345 i & 0.0437  +0.1345 i & -0.1144-0.0831 i
   \\
 0 & 0 & 0 & 0 & 0 \\
 0 & 0 & 0 & 0 & 0 \\
 -0.1414 & 0.1144  +0.0831 i & -0.0437-0.1345 i & -0.0437+0.1345 i & 0.1144  -0.0831
   i \\
\end{array}
\right)   .
\end{gather}
\end{subequations}
\normalsize{}
As calculated in Table~\ref{tab:Single-ququint nonstabilizer Clifford-inequivalent Clifford nondegenerate eigenstates and Wigner function}:
\[
W_{\ket{A,e^{\frac{4\pi i}{5}}}\bra{A,e^{\frac{4\pi i}{5}}}} \approx  \begin{pmatrix}
-0.2 & -0.0618 & 0.1618 & 0.1618 & -0.0618 \\
0& 0& 0& 0& 0\\
0.1 & 0.1 & 0.1 & 0.1 & 0.1 \\
0.1 & 0.1 & 0.1 & 0.1 & 0.1 \\
0& 0& 0& 0& 0.
\end{pmatrix}   .
\]

With these results, it is straightforward to verify that for any variation
\[
\ket{\varphi} \in \text{Span}^\perp\left( \ket{A, e^{-\frac{\pi i}{5}}} \right),
\]
the discrete Wigner function $W_{\sigma\left( \ket{A, e^{-\frac{\pi i}{5}}}, \ket{\varphi} \right)}$ vanishes at all entries where the Wigner function of the pure state 
$W_{\ket{A, e^{-\frac{\pi i}{5}}}\bra{A, e^{-\frac{\pi i}{5}}}}$ vanishes if and only if
\[
\ket{\varphi} = e^{i\phi} \ket{A, e^{\frac{4\pi i}{5}}}
\quad \text{for some } \phi \in \mathbb{R}.
\]
Along any path outside this two-dimensional surface, the state $\ket{A, e^{-\frac{\pi i}{5}}}$ is a sharp minima of the mana of magic.

To characterize the behavior of the mana in the vicinity of $\ket{A, e^{-\frac{\pi i}{5}}}$ when the variation is given by 
$\ket{\varphi} = e^{i\phi} \ket{A, e^{\frac{4\pi i}{5}}}$, we invoke Eq.~\eqref{eq:varied Wigner trace norm when the linear term vanishes}. Using the Wigner functions
\[
W_{\ket{A, e^{-\frac{\pi i}{5}}}\bra{A, e^{-\frac{\pi i}{5}}}} 
\quad \text{and} \quad 
W_{\ket{\varphi}\bra{\varphi}} = 
W_{\ket{A, e^{\frac{4\pi i}{5}}}\bra{A, e^{\frac{4\pi i}{5}}}},
\]
as tabulated in Table~\ref{tab:Single-ququint nonstabilizer Clifford-inequivalent Clifford nondegenerate eigenstates and Wigner function}, we find:
\begin{equation}
    \left\| \rho\left( \ket{A, e^{-\frac{\pi i}{5}}}(\epsilon) \right) \right\|_W 
    - \left\| \rho\left( \ket{A, e^{-\frac{\pi i}{5}}} \right) \right\|_W 
    \approx -1.2944   \frac{\epsilon^2}{1+\epsilon^2}   ,
\end{equation}
which indicates that $\ket{A, e^{-\frac{\pi i}{5}}}$ is a smooth local maximum of the mana along this two-dimensional submanifold.

The eigenstate $\ket{A,\omega^2 = e^{\frac{4\pi i}{5}}} \propto \ket{2}+\ket{3}$ exhibits a similar behavior. We find
\tiny
\begin{subequations}
\begin{gather}
        W_{\bar{\sigma}( \ket{A,e^{\frac{4\pi i}{5}}}, \ket{A,e^{\frac{-4\pi i}{5}}}} \approx \left(
\begin{array}{ccccc}
 0 & 0 & 0 & 0 & 0 \\
 0.1 & 0.0309  -0.0951 i & -0.0809-0.0588 i & -0.0809+0.0588 i & 0.0309  +0.0951 i \\
 0.1 & -0.0809+0.0588 i & 0.0309  -0.0951 i & 0.0309  +0.0951 i & -0.0809-0.0588 i \\
 0.1 & -0.0809-0.0588 i & 0.0309  +0.0951 i & 0.0309  -0.0951 i & -0.0809+0.0588 i \\
 0.1 & 0.0309  +0.0951 i & -0.0809+0.0588 i & -0.0809-0.0588 i & 0.0309  -0.0951 i \\
\end{array}
\right)   ,
        \\
        W_{\bar{\sigma}( \ket{A,e^{\frac{4\pi i}{5}}}, \ket{A,e^{\frac{\pi i}{5}}}} \approx \left(
\begin{array}{ccccc}
 0 & 0 & 0 & 0 & 0 \\
 0.1 & 0.0309  -0.0951 i & -0.0809-0.0588 i & -0.0809+0.0588 i & 0.0309  +0.0951 i \\
 -0.1 & 0.0809  -0.0588 i & -0.0309+0.0951 i & -0.0309-0.0951 i & 0.0809  +0.0588 i
   \\
 0.1 & -0.0809-0.0588 i & 0.0309  +0.0951 i & 0.0309  -0.0951 i & -0.0809+0.0588 i \\
 -0.1 & -0.0309-0.0951 i & 0.0809  -0.0588 i & 0.0809  +0.0588 i & -0.0309+0.0951 i
   \\
\end{array}
\right)   ,
        \\
        W_{\bar{\sigma}( \ket{A,e^{\frac{4\pi i}{5}}}, \ket{A,e^{-\frac{\pi i}{5}}}} \approx \left(
\begin{array}{ccccc}
 0 & 0.1902 i & 0.1176 i & -0.1176 i & -0.1902 i \\
 0 & 0 & 0 & 0 & 0 \\
 0.1 & 0.1 & 0.1 & 0.1 & 0.1 \\
 -0.1 & -0.1 & -0.1 & -0.1 & -0.1 \\
 0 & 0 & 0 & 0 & 0 \\
\end{array}
\right)   ,
        \\
        W_{\bar{\sigma}( \ket{A,e^{\frac{4\pi i}{5}}}, \ket{A,1})} \approx \left(
\begin{array}{ccccc}
 0 & 0 & 0 & 0 & 0 \\
 0.1414 & -0.1144+0.0831 i & 0.0437  -0.1345 i & 0.0437  +0.1345 i & -0.1144-0.0831 i
   \\
 0 & 0 & 0 & 0 & 0 \\
 0 & 0 & 0 & 0 & 0 \\
 0.1414 & -0.1144-0.0831 i & 0.0437  +0.1345 i & 0.0437  -0.1345 i & -0.1144+0.0831 i
   \\
\end{array}
\right)   .
\end{gather}
\end{subequations}
\normalsize{}
As calculated in Table~\ref{tab:Single-ququint nonstabilizer Clifford-inequivalent Clifford nondegenerate eigenstates and Wigner function}:
\[
W_{\ket{A,e^{\frac{4\pi i}{5}}}\bra{A,e^{\frac{4\pi i}{5}}}} \approx  \begin{pmatrix}
0.2 & 0.0618 & -0.1618 & -0.1618 & 0.0618 \\
0& 0& 0& 0& 0\\
0.1 & 0.1 & 0.1 & 0.1 & 0.1 \\
0.1 & 0.1 & 0.1 & 0.1 & 0.1 \\
0& 0& 0& 0& 0.
\end{pmatrix}   .
\]
And again, with these results, we confirm that for any variation
\[
\ket{\varphi} \in \text{Span}^\perp\left( \ket{A, e^{\frac{4\pi i}{5}}} \right),
\]
the Wigner function $W_{\sigma\left( \ket{A, e^{\frac{4\pi i}{5}}}, \ket{\varphi} \right)}$ vanishes at all entries where $W_{\ket{A, e^{\frac{4\pi i}{5}}}\bra{A, e^{\frac{4\pi i}{5}}}}$ vanishes if and only if
\[
\ket{\varphi} = e^{i\phi} \ket{A, e^{-\frac{\pi i}{5}}}
\quad \text{for some } \phi \in \mathbb{R}.
\]
Along any path outside this two-dimensional surface, the state $\ket{A, e^{\frac{\pi i}{5}}}$ is a sharp minima of the mana of magic.

Therefore, we focus on the behavior of the mana in the two-dimensional variation subspace spanned by $\ket{A, e^{-\frac{\pi i}{5}}}$. Using Eq.~\eqref{eq:varied Wigner trace norm when the linear term vanishes}, and the fact that $W_{\ket{A, e^{\frac{4\pi i}{5}}}\bra{A, e^{\frac{4\pi i}{5}}}}$ and $W_{\ket{A, e^{-\frac{\pi i}{5}}}\bra{A, e^{-\frac{\pi i}{5}}}}$ are known, we obtain:
\begin{equation}
    \left\| \rho\left( \ket{A, e^{\frac{4\pi i}{5}}}(\epsilon) \right) \right\|_W 
    - \left\| \rho\left( \ket{A, e^{\frac{4\pi i}{5}}} \right) \right\|_W 
    \approx -1.2944   \frac{\epsilon^2}{1+\epsilon^2}   .
\end{equation}

This confirms that $\ket{A, e^{\frac{4\pi i}{5}}}$, like $\ket{A, e^{-\frac{\pi i}{5}}}$, constitutes a smooth local \emph{maximum} of the mana in its own uniquely aligned two-dimensional variation submanifold. The symmetry observed here between the two Clifford eigenstates is a hallmark of their shared structural features in the Wigner representation.

\subsubsection*{Last Critical Point}

Similarly, for $\ket{XV_{\hat{S}},1}$ we find
\tiny
\begin{subequations}
\begin{gather}
        W_{\bar{\sigma}(\ket{XV_{\hat{S}},1},\ket{XV_{\hat{S}},\omega)}} \approx \left(
\begin{array}{ccccc}
 0& 0.0171  -0.0526 i & -0.0276+0.0851 i & 0.0276  -0.0851 i & 0.0447  -0.1376 i
   \\
 0& -0.0447-0.0325 i & 0.0724  +0.0526 i & -0.0724-0.0526 i & -0.1171-0.0851 i \\
 -0.1171+0.0851 i & 0& -0.0447+0.0325 i & 0.0724  -0.0526 i & -0.0724+0.0526 i \\
 -0.0276-0.0851 i & 0.0276  +0.0851 i & 0.0447  +0.1376 i & 0& 0.0171  +0.0526 i
   \\
 0.1447 & 0& 0.0553 & -0.0894 & 0.0894 \\
\end{array}
\right)   ,
        \\
        W_{\bar{\sigma}(\ket{XV_{\hat{S}},1},\ket{XV_{\hat{S}},\omega^{-1})}} \approx \left(
\begin{array}{ccccc}
 0.0447  +0.1376 i & 0& 0.0171  +0.0526 i & -0.0276-0.0851 i & 0.0276  +0.0851 i
   \\
 -0.1171+0.0851 i & 0& -0.0447+0.0325 i & 0.0724  -0.0526 i & -0.0724+0.0526 i \\
 -0.0724-0.0526 i & -0.1171-0.0851 i & 0& -0.0447-0.0325 i & 0.0724  +0.0526 i \\
 0.0171  -0.0526 i & -0.0276+0.0851 i & 0.0276  -0.0851 i & 0.0447  -0.1376 i & 0.
   \\
 0.0894 & 0.1447 & 0& 0.0553 & -0.0894 \\
\end{array}
\right)   ,
        \\
        W_{\bar{\sigma}(\ket{XV_{\hat{S}},1},\ket{XV_{\hat{S}},\omega^2)}} \approx \left(
\begin{array}{ccccc}
 -0.0724-0.0526 i & -0.1171-0.0851 i & 0& -0.0447-0.0325 i & 0.0724  +0.0526 i \\
 0.0276  +0.0851 i & 0.0447  +0.1376 i & 0& 0.0171  +0.0526 i & -0.0276-0.0851 i
   \\
 -0.0276+0.0851 i & 0.0276  -0.0851 i & 0.0447  -0.1376 i & 0& 0.0171  -0.0526 i
   \\
 0& -0.0447+0.0325 i & 0.0724  -0.0526 i & -0.0724+0.0526 i & -0.1171+0.0851 i \\
 -0.0894 & 0.0894 & 0.1447 & 0& 0.0553 \\
\end{array}
\right)   ,
        \\
        W_{\bar{\sigma}(\ket{XV_{\hat{S}},1},\ket{XV_{\hat{S}},\omega^{-2})}} \approx \left(
\begin{array}{ccccc}
 -0.0447+0.0325 i & 0.0724  -0.0526 i & -0.0724+0.0526 i & -0.1171+0.0851 i & 0\\
 0.0171  -0.0526 i & -0.0276+0.0851 i & 0.0276  -0.0851 i & 0.0447  -0.1376 i & 0.
   \\
 0& 0.0171  +0.0526 i & -0.0276-0.0851 i & 0.0276  +0.0851 i & 0.0447  +0.1376 i
   \\
 -0.0724-0.0526 i & -0.1171-0.0851 i & 0& -0.0447-0.0325 i & 0.0724  +0.0526 i \\
 0& 0.0553 & -0.0894 & 0.0894 & 0.1447 \\
\end{array}
\right)   .
\end{gather}
\end{subequations}
\normalsize{}
As calculated in Table~\ref{tab:Single-ququint nonstabilizer Clifford-inequivalent Clifford nondegenerate eigenstates and Wigner function}:
\[
W_{\ket{XV_{\hat{S}},1}\bra{XV_{\hat{S}},1}} \approx \begin{pmatrix}
-0.0894 & 0.0894 & 0.1447 & 0& 0.0553 \\
-0.0894 & 0.0894 & 0.1447 & 0& 0.0553 \\
0.0553 & -0.0894 & 0.0894 & 0.1447 & 0\\
0.1447 & 0& 0.0553 & -0.0894 & 0.0894 \\
0.0553 & -0.0894 & 0.0894 & 0.1447 & 0.
\end{pmatrix}   .
\]
A straightforward calculation yields that for any variation
\[
\ket{\varphi} \in \text{Span}^\perp\left( \ket{XV_{\hat{S}},1} \right),
\]
the Wigner function $W_{\sigma\left( \ket{XV_{\hat{S}},1}, \ket{\varphi} \right)}$ vanishes at all entries where $W_{\ket{XV_{\hat{S}},1}\bra{XV_{\hat{S}},1}}$ vanishes if and only if
\begin{equation}
    \begin{split}
        \ket{\varphi} = &
    e^{i\phi_1}\cos\alpha \ket{XV_{\hat{S}},\omega}
    + e^{-i\phi_1}\cos\alpha \ket{XV_{\hat{S}},\omega^{-1}}
    \\ &  + \sqrt{\frac{3 - \sqrt{5}}{6}} \left(
    e^{i\phi_2} \sin\alpha \ket{XV_{\hat{S}},\omega^{-2}}
    - \frac{3+\sqrt{5}}{2} e^{-i\phi_2} \sin\alpha \ket{XV_{\hat{S}},\omega^{2}}
    \right) \\
     & \quad \quad \quad \quad \text{for some } \alpha,\phi_1,\phi_2 \in \mathbb{R}   .
    \end{split}
\end{equation}
The behavior of the mana for such variations is significantly more involved, and we leave its detailed analysis to the reader.

\subsubsection*{Summary}

To summarize the findings of this section, we present in Table~\ref{tab:Summary of mana for ququints} an overview of the qualitative behavior of the mana for the investigated Clifford inequivalent non-degenerate non-stabilizer Clifford eigenstates for single ququints. The table classifies each state according to whether its Wigner function is fully supported on all entries or vanishes on some entries, and describes the resulting nature of its mana landscape, such as smooth local maxima, subspace-restricted maxima, or critical points whose detailed behavior remains unresolved.
\begin{table}[h!]
\centering
\begin{tabular}{|c|c|c|}
\hline
\textbf{Eigenstate} & \textbf{Wigner Function Support} & \textbf{Mana Behavior} \\
\hline
$\ket{B',-1}$ & Non-vanishing & Smooth Local Max \\
$\ket{B',-e^{\frac{2\pi i}{3}}}$ & Non-vanishing & Smooth Local Max \\
$\ket{B',e^{\frac{2\pi i}{3}}}$ & Non-vanishing & Smooth Local Max \\
$\ket{H,+i}$ & Non-vanishing & Smooth Local Max \\
$\ket{H,-1}$ & Non-vanishing & Smooth Local Max \\
$\ket{A,e^{-\frac{\pi i}{5}}}$ & Partially vanishing & Smooth Local Max in 2D subspace \\
$\ket{A,e^{\frac{4\pi i}{5}}}$ & Partially vanishing & Smooth Local Max in 2D subspace \\
$\ket{XV_{\hat{S}},1}$ & Partially vanishing & Undetermined / Critical Point \\
\hline
\end{tabular}
\caption{Summary of mana behavior for all Clifford-inequivalent non-degenerate non-stabilizer Clifford eigenstates of single ququints. Each eigenstate is labeled by its corresponding eigenvalue, with the Wigner function support indicating whether the state's Wigner representation vanishes at any phase-space point. The mana behavior is qualitatively characterized based on whether the state exhibits a smooth local maximum in all directions, a restricted maximum within a subspace, or an unresolved critical point.}
\label{tab:Summary of mana for ququints}
\end{table}

\subsubsection{On the $\alpha$-SRE}

In Appendix~\ref{app:SREforSingleQuquints}, we extend the computation of the $\alpha$-SRE to all single-ququint non-stabilizer, Clifford-inequivalent, nondegenerate Clifford eigenstates considered here in the main text.
For each such eigenstate, we present an explicit closed-form expression $M_\alpha(\ket{\psi})$, and we evaluate their values at the special point $\alpha=2$ to obtain compact logarithmic formulas.
This systematic survey shows that, in contrast to the qubit and qutrit cases, none of the investigated ququint eigenstates saturates the general upper bound in Eq.~(\ref{eq: Bound on SRE}), a fact that we relate to the observation that none of them simultaneously attains the minimal stabilizer fidelity.
Together, these results provide a comprehensive characterization of SRE in prime dimensions $d=2,3,5$.
However, a detailed analysis of the type of extremum they represent is left to the reader.

\subsection{SIC-POVM Fiducial States}

Symmetric informationally complete (SIC) positive operator-valued measure (POVM) fiducial states stand at the intersection of quantum state geometry, symmetry, and resource theories.
Originally introduced in the context of optimal quantum state tomography, they have since emerged as extremal objects in a wide range of seemingly unrelated settings, including discrete phase-space formulations, frame theory, and measures of non-stabilizerness.
In the present work, SIC fiducials arise naturally as extremizers of Clifford-covariant magic functionals, placing them conceptually alongside Clifford-stabilizer states as highly symmetric and structurally distinguished points in state space.
This subsection explores this connection in detail, progressing from quantitative extremality properties to a symmetry-based structural conjecture.

\subsubsection{Fiducial States, and Extremizing $p$-Norms and SREs}

We begin by recalling a striking extremality property of SIC-POVM fiducial states with respect to a broad family of magic measures.
In particular, for a continuous range of parameters, fiducial states arise as the unique maximizers or minimizers of $L^p$-norm magic and, equivalently, of stabilizer R\'enyi entropies.
These results provide a purely quantitative characterization of SIC fiducials, independent of any explicit symmetry assumptions.
As such, they offer a natural bridge between operational notions of quantum magic and the highly constrained geometric structure exhibited by SIC states.

\begin{definition}[SIC-POVM Fiducial State]
A SIC \emph{fiducial state} is a normalized state
$\ket{\psi} \in \mH_{1,d}$ such that
\[
    \left| \braket{\psi | T_{\boldsymbol{\chi}} | \psi} \right|
    = \frac{1}{\sqrt{d+1}}
    \qquad
    \forall\, \boldsymbol{\chi} \in \mathbb{V}_{1,d} \setminus \{0\}.
\]
\end{definition}
\noindent
Given such a fiducial state, one defines the orbit states
\begin{equation}
    \ket{\psi_{\boldsymbol{\chi}}}
    \;=\;
    T_{\boldsymbol{\chi}} \ket{\psi},
\end{equation}
and the associated rank-one POVM elements
\begin{equation}
    E_{\boldsymbol{\chi}}
    \;=\;
    \frac{1}{d}\,
    \ket{\psi_{\boldsymbol{\chi}}}\bra{\psi_{\boldsymbol{\chi}}}.
\end{equation}
The collection
$\{ E_{\boldsymbol{\chi}} \}_{\boldsymbol{\chi} \in \mathbb{V}_{1,d}}$
then forms a SIC-POVM on $\mH_{1,d}$.
In fact, the vast majority of SIC-POVMs constructed to date are \emph{group-covariant} under the action of the Weyl-Heisenberg (WH) group, with the fiducial state generating the entire POVM via its group orbit.

It was shown in Ref.~\cite{Feng2022} that for \( p \in [1,2] \), the \( p \)-norm magic, defined in Eq.~\eqref{eq:LpNorm}, satisfies
\begin{equation}
\label{eq:BoundLpA}
    d^{N/p}
    \;\leq\;
    L_p(\rho)
    \;\leq\;
    \bigl[ 1 + (d^N-1)(d^N+1)^{1 - p/2} \bigr]^{1/p},
\end{equation}
where, in the this regime, the lower bound is attained if and only if \( \rho \) is a stabilizer state, while the upper bound (equivalently, the maximal value of \( L_p \) whenever a fiducial state exists) is attained if and only if \( \rho \) is a fiducial state.
For \( p > 2 \), the ordering of the bounds is reversed:
\begin{equation}
\label{eq:BoundLpB}
    \bigl[ 1 + (d^N-1)(d^N+1)^{1 - p/2} \bigr]^{1/p}
    \;\leq\;
    L_p(\rho)
    \;\leq\;
    d^{N/p},
\end{equation}
where, in this case, the upper bound is attained if and only if \( \rho \) is a stabilizer state, whereas the lower bound (corresponding to the minimal value of \( L_p \) when a fiducial state exists) is attained if and only if \( \rho \) is a fiducial state.

\noindent
\emph{Remark.}
In view of the relation~\eqref{eq:SREandLpNorm} between SREs and $L^p$ norms, the bounds in Eqs.~\eqref{eq: Bound on SRE} and~\eqref{eq:BoundLpB} are merely two equivalent formulations of the same underlying constraint.
Consequently, a state saturates this bound (equivalently, maximizes all stabilizer SREs) if and only if it is a SIC-POVM fiducial state.
At present, however, this characterization does not guarantee the existence of such states in arbitrary dimensions.

\subsubsection{Fiducial States as Clifford Eigenstates}

Over the years, a substantial body of evidence has accumulated in the literature
supporting the view that SIC fiducial states possess distinguished symmetries
with respect to the Clifford group.
Starting from Zauner’s original thesis \cite{Zauner1999}, and continuing through the works of Appleby~\cite{Appleby2005,Appleby2012Imprimitivity,Appleby2012Galois}, Flammia~\cite{Flammia2006}, and later by Bengtsson, Appleby, Flammia, as well as by Len Bos and Shayne Waldron~\cite{BengtssonApplebyFlammia2018}, it has been consistently observed that every \emph{known} WH covariant SIC fiducial state is an eigenstate of a Clifford unitary of order three.
Canonical examples include the single-qubit $\ket{T}$ state (and its Clifford orbit), as well as the single-qutrit strange state $\ket{\mathbb{S}}$ and the Norell state (together with their respective Clifford orbits)~\cite{SIC2004,Feng2022}.

While this structural picture is strongly supported by numerical and theoretical evidence, it remains a conjecture rather than a fully proven statement in general dimensions.
Nevertheless, if the conjecture that every SIC-POVM fiducial state is a non-degenerate eigenstate of a Clifford operation holds, then SIC-POVM fiducial states indeed fall into the class of Clifford-stabilizer states considered in our work.
In that sense, the extremality results established for $L^p$-norm magic~\cite{Feng2022} can be viewed as a special, and very interesting, instance of the more general group-covariant extremality framework we developed.

\subsubsection{The Possibility of a Converse to our Main Theorem and SIC-POVM Fiducial States}

A direct corollary of our main theorem is the following one-directional statement.
\begin{corollary}
Let $ G \subset \mathrm{U}(\mH) $ be a finite subgroup, and let \( f : \mathrm{Herm}(\mH) \to \mathbb{R} \) be an \emph{analytic} functional that is invariant under $G$-conjugation, i.e.,
\[
    f(O) = f\!\left(g^\dagger O g\right)
    \qquad
    \forall\, g \in G.
\]
Then every $G$-stabilizer state is a critical point of the restriction of $f$ to the manifold of normalized pure states.
\end{corollary}
\noindent
This result establishes that group-stabilized pure states necessarily give rise to stationary points of any analytic $G$-covariant functional. Importantly, however, the statement holds only in this direction.

Motivated by the empirical observation that all currently known pure states which are extremal or critical for standard measures of quantum magic are Clifford-stabilizer states, it is natural to ask whether some form of a converse statement might hold.  
A tempting, but ultimately incorrect, guess is the following.

\noindent\textbf{Na\"ive Converse Claim.}\\
Let \( G \subset \mathrm{U}(\mH) \) be a finite subgroup, and let
\[
    f : \mathrm{Herm}(\mH) \longrightarrow \mathbb{R}
\]
be an \emph{analytic} functional that is invariant under conjugation by \( G \), i.e.,
\[
    f(O) = f\!\left(g^\dagger O g\right)
    \qquad
    \forall\, g \in G .
\]
Suppose that \( \ket{\psi} \in \mH_{1,d} \) is a normalized pure state such that the rank-one projector
\( \ket{\psi}\!\bra{\psi} \) is a critical point of the restriction of \( f \) to the manifold of normalized pure states.
Then \( \ket{\psi} \) must be a \( G \)-stabilizer state.

Despite its superficial plausibility, this converse statement is \emph{false in general}.
The obstruction is conceptual rather than technical: invariance under a finite symmetry group is simply too weak a constraint to enforce stabilizer structure.

Indeed, Given any functional that is invariant under conjugation by a finite group \( G \), it is, by definition, invariant under conjugation by any subgroup $H \subset G$.
Such functionals can admit critical points that are $G$-stabilizer states but not $H$ stabilizer states.
An example of such a construction for single qubits is the functional
\[
    f(\ket{\psi})
    \;=\;
    \sum_{\mu=0}^{3} a_\mu
    \bigl( \bra{\psi}\,\sigma_\mu\,\ket{\psi} \bigr)^2 ,
\]
where \( \{\sigma_\mu\}_{\mu=0}^3 = \{ \mathbb{I}, \sigma_x, \sigma_y, \sigma_z \} \) denotes the Pauli operators, and
\( \{a_\mu\}_{\mu=0}^3 \subset \mathbb{R} \) are pairwise distinct, nonzero coefficients.
This functional is manifestly invariant under conjugation by the Pauli group, since Pauli conjugation merely permutes (up to signs) the expectation values
\( \bra{\psi}\sigma_\mu\ket{\psi} \).
However, the explicit asymmetry in the coefficients \( a_\mu \) breaks invariance under any strictly larger subgroup of \( \mathrm{U}(2) \).
In particular, \( f \) is not invariant under any non-Pauli Clifford operation.
It is easy to verify this through explicit calculation.

Therefore, one might think that the claim should take the maximal group of symmetry of the functional under conjugation.
However, the claim will still not hold and will be wrong. A mere looking at Fig.~\ref{fig:Stabilizer fidelity for all single-qubit states}, would convince the reader that simple changes can be made in the function such that it stills have the Clifford-conjugation symmetry, but it can have minima or maxima that is not a Clifford-stabilizer state.

Despite these considerations, the results of the present work, together with the influential observations of Flammia and others, strongly motivate the following conjecture.
\begin{conjecture}[Stabilizer nature of SIC fiducials]
Every SIC-POVM fiducial state is a Clifford-stabilizer state.
\end{conjecture}
If true, this conjecture would imply that the extremality of SIC fiducials for $L^p$-norm magic and stabilizer R\'enyi entropies is not accidental, but rather a direct consequence of an underlying stabilizer symmetry.
From this perspective, SIC fiducials would represent a distinguished boundary case of stabilizer geometry, simultaneously maximizing non-stabilizerness while remaining fixed points of finite Clifford subgroups.

\section{Non-Degenerate Eigenstates of Cliffords for Two Qubits}
\label{sec:Non-Degenerate Eigenstates of Cliffords for Two Qubits}

In the single-qubit setting, the well-known states $\ket{H}$ and $\ket{T}$ are not only non-stabilizer states but also eigenstates of specific single-qubit Clifford unitaries: the Hadamard gate for $\ket{H}$ and a particular Clifford operator for $\ket{T}$. Their capability to be distilled into higher-fidelity versions using only Clifford operations and stabilizer measurements underlies many magic-state distillation protocols. This property is precisely why the authors of Ref.~\cite{KitaevBravyi} designated them as ``magic.''

Analogously, in higher dimensions, eigenstates of Clifford unitaries emerge as natural candidates for magic states, offering promising pathways for the fault-tolerant implementation of non-Clifford operations. By classifying two-qubit states up to Clifford equivalence, we can systematically identify distinct families of candidate magic states, much as the twenty single-qubit magic states can be reduced, under Clifford transformations, to two canonical representatives corresponding to $\ket{H}$ and $\ket{T}$.

In this section, we extend this framework and perspective to the two-qubit domain. Specifically, we aim to identify and analyze the eigenstates of representative elements from the conjugacy classes of the two-qubit Clifford group. This analysis might provide a foundation for prospective applications in multi-qubit magic-state distillation and the realization of non-Clifford gates within fault-tolerant architectures.

With the aid of Table~1 in Ref.~\cite{CliffordClassesForTwoQubits}, which classifies the two-qubit Clifford group into conjugacy classes, we identify all Clifford-inequivalent non-degenerate eigenstates of two-qubit Clifford operations.
The results are summarized in Table~\ref{tab:Clifford-inequivalent non-degenerate eigenstates of Clifford operations for two qubits}.
Further details are provided in Appendix~\ref{app:Finding the eigenstates of 2 qubits}.
It is worth noting, as shown in the table, that for these states the stabilizer fidelity is multiplicative. This aligns with two key facts: first, the states in question are Clifford-stabilizer states, for which the stabilizer fidelity equals the inverse of the stabilizer extent (see Appendix~\ref{app:group-stabilizer extent} or Theorem 4 in~\cite{Bravyi2019simulationofquantum}); and second, the stabilizer extent is known to be multiplicative for systems of up to three qubits (see Proposition 1 in \cite{Bravyi2019simulationofquantum}).

\begin{table}[h]
    \centering
    \resizebox{\linewidth}{!}{
    \begin{tabular}{|c|c|c|c|}
        \hline
        \textbf{State} & \textbf{\footnotesize{\makecell{Stabilizer\\Fidelity}}} & \textbf{\footnotesize{\makecell{Number of\\Nearest SS}}} & \textbf{\small{\makecell{Clifford\\equivalent}}} \\
        \hline
        \normalsize{}
        $|00\rangle$ &  1  &  1  & $\ket{00}$ \\ \hline
        $|H 0\rangle$ &  $\frac{2+\sqrt{2}}{4}\approx0.853553$  &  2  & $\ket{H0}$  \\ \hline
        $|T 0\rangle$  &  $\frac{3+\sqrt{3}}{6}\approx0.788675$  &  3  & $\ket{T0}$  \\ \hline   
        $\ket{HH}$  &  $\frac{3+2\sqrt{2}}{8}\approx0.728553$  &  4  & $\ket{HH}$  \\ \hline
        $|T H\rangle$  &  $\frac{(2+\sqrt{2})(3+\sqrt{3})}{24}\approx0.673176$  &  6  & $\ket{TH}$  \\ \hline
        $|TT\rangle$  &  $\frac{2+\sqrt{3}}{6}\approx0.622008$  &  9  & $\ket{TT}$ \\ \hline
        $\ket{G_{4},2}=\dfrac{2 \ket{00}+\ket{01}+\ket{10}}{\sqrt{6}}$  & $\frac{3}{4}=0.75$   &  2  & $\frac{\ket{T_0T_0}-\ket{T_1T_1}}{\sqrt{2}}$ \\ \hline
        $\ket{G_{16},1}$ , $\ket{G_{16},2}$ , $\ket{G_{16},3}$ , $\ket{G_{16},4}$  &  $\frac{5+\sqrt{5}+2\sqrt{5+2\sqrt{5}}}{20} \approx 0.669572$  &  5  &   \\ \hline
        $\ket{G_{20},1}$ , $\ket{G_{20},2}$ , $\ket{G_{20},3}$ , $\ket{G_{20},4}$  &  $\frac{5}{8}=0.625$  &  8  &  \\ \hline
    \end{tabular}
    }
    \caption{Clifford-inequivalent non-degenerate eigenstates of Clifford operations for two qubits with their stabilizer fidelity and number of nearest stabilizer states.}
    \label{tab:Clifford-inequivalent non-degenerate eigenstates of Clifford operations for two qubits}
\end{table}

The last three states in Table~\ref{tab:Clifford-inequivalent non-degenerate eigenstates of Clifford operations for two qubits} are newly identified. The first among them is Clifford-equivalent to the state $\frac{\ket{T_0T_0}-\ket{T_1T_1}}{\sqrt{2}}$, and we propose a distillation protocol based on the five-qubit code, originally designed for distilling the single-qubit $\ket{T}$ state, that leverages this structural similarity. 
We also note that, during the course of this work, an independent study was published on the subject of maximal magic in two-qubit systems, using the SRE of magic as the quantifier~\cite{Liu2025}. In that work, the states exhibiting maximal magic (for $\alpha=2$) are in fact Clifford-equivalent to the states denoted here by $\ket{G_{20},1;2;3;4}$.
For instance, the state $\ket{G_{20},1}$ considered in this paper is related to
\[ \ket{\psi_{\text{max},SRE,\alpha=2}}\equiv\frac{\ket{00}+i\ket{01}+i\ket{10}+i\ket{11}}{2}
\]
through the Clifford operation
\[
\ket{G_{20},1}= e^{i\frac{\pi}{4}} H_2\cdot S_1^3\cdot \mathrm{CZ} \cdot H_1 \cdot S_1 \cdot \mathrm{CZ}   \ket{\psi_{\text{max},SRE,\alpha=2}}.
\]

For completeness, we present below the $\alpha$-SRE values for all these Clifford-inequivalent, nondegenerate eigenstates of two-qubit Clifford operations. 
The entropy values are given as follows:
\begin{subequations}
    \begin{gather}
        M_\alpha\left( \ket{00} \right) = 0 ,\\
        M_\alpha\left( \ket{H_00} \right) = \frac{1}{1-\alpha} \log \left[ \tfrac{1}{2}\left( 1 + 2^{1-\alpha} \right) \right] ,\\
        M_\alpha\left( \ket{T_00} \right) = \frac{1}{1-\alpha} \log \left[ \tfrac{1}{2}\left( 1 + 3^{1-\alpha} \right)  \right] ,\\
        M_\alpha\left( \ket{H_0H_0} \right) = \frac{1}{1-\alpha} \log \left[ \tfrac{1}{4}\left( 1 + 2^{1-\alpha} \right)^2 \right] ,\\
        M_\alpha\left( \ket{T_0H_0} \right) = \frac{1}{1-\alpha} \log \left[ \tfrac{1}{4}\left( 1 + 3^{1-\alpha} \right) \left( 1 + 2^{1-\alpha} \right) \right] ,\\
        M_\alpha\left( \ket{T_0T_0} \right) = \frac{1}{1-\alpha} \log \left[ \tfrac{1}{4}\left( 1 + 3^{1-\alpha} \right)^2 \right] ,\\
        M_\alpha\left( \frac{\ket{T_0T_0}-\ket{T_1T_1}}{\sqrt{2}} \right) = \frac{1}{1-\alpha} \log \left[ \frac{1}{4} \left( 1 + 3 \cdot 9^{-\alpha}  + 4 \cdot \left( \frac{3}{2} \right) ^{1-2\alpha} \right) \right] ,\\
        M_\alpha\left( \ket{G_{16},1} \right) = \frac{1}{1-\alpha} \log \left[ \frac{1}{4}\Bigl( 1 + 5^{1 - \alpha} + 5^{1 - 2\alpha} \cdot \frac{5^{\alpha} + (5 + 2\sqrt{5})^{2\alpha}}{(5 + 2\sqrt{5})^{\alpha}} \Bigr) \right] ,\\
        M_\alpha\left( \ket{G_{20},1} \right) = M_\alpha\left( \frac{\ket{00}+i\ket{01}+i\ket{10}+i\ket{11}}{2} \right) = \frac{1}{1-\alpha} \log \left[ \frac{1}{4} \left( 1 + 3 \cdot 4^{1-\alpha}\right) \right].
    \end{gather}
\end{subequations}
As observed, the $\ket{G_{20},1}$ state exhibit the maximal SRE according to this measure, thereby saturating the upper bound in Eq.~(\ref{eq: Bound on SRE}).
In the special case $\alpha = 2$, the expressions above for the three newly introduced states simplify to
\begin{equation}
    M_2 \left( \frac{\ket{T_0T_0}-\ket{T_1T_1}}{\sqrt{2}} \right) = \log \frac{9}{5}
    \quad , \quad
    M_2 \left( \ket{G_{16},1} \right) = \log \frac{25}{12}
    \quad , \quad
    M_2 \left( \ket{G_{20},1} \right) = \log \frac{16}{7}.
\end{equation}

We now investigate the behavior of $ M_{2} $ in a specific example, focusing on a neighborhood of any state in the Clifford orbit of the state $ \ket{G_{20},1} $.
Consider the following orthonormal basis of two-qubit states, each of which attains the maximal value of $M_2$:~\cite{Liu2025}
\begin{subequations}
    \begin{align}
        \ket{\psi_{\text{max},0}} &\equiv
        \tfrac{1}{2}\bigl(\ket{0} + i\ket{1} + i\ket{2} + i\ket{3}\bigr),\\[4pt]
        \ket{\psi_{\text{max},1}} &\equiv
        \tfrac{1}{2}\bigl(i\ket{0} + \ket{1} + \ket{2} - \ket{3}\bigr),\\[4pt]
        \ket{\psi_{\text{max},2}} &\equiv
        \tfrac{1}{2}\bigl(i\ket{0} + \ket{1} - \ket{2} + \ket{3}\bigr),\\[4pt]
        \ket{\psi_{\text{max},3}} &\equiv
        \tfrac{1}{2}\bigl(-i\ket{0} + \ket{1} - \ket{2} - \ket{3}\bigr).
    \end{align}
\end{subequations}
Each of these basis states satisfies
\[
    M_2 = \log  \left(\tfrac{16}{7}\right),
\]
and thus represents a maximally non-stabilizer configuration in the two-qubit Hilbert space.

To examine the local behavior of $\Xi_2$ around one of these states, we vary $\ket{\psi_{\text{max},0}}$ along a general orthogonal direction parameterized as
\begin{equation}
    \ket{\varphi} =
    e^{i\phi_1}\cos\theta_1 \ket{\psi_{\text{max},0}}
    + e^{i\phi_2}\sin\theta_1\cos\theta_2 \ket{\psi_{\text{max},1}}
    + e^{i\phi_3}\sin\theta_1\sin\theta_2 \ket{\psi_{\text{max},2}}.
\end{equation}
The resulting expansion of $\Xi_2(\ket{\psi_{\text{max},0}(\epsilon)})$ yields the second-order coefficient
\begin{equation}
    \begin{split}
        \Xi_2^{(2)}
        &= \frac{1}{16}
        \Bigl(
        -7
        + 3\cos^2  \theta_1 [4 + \cos(2\phi_1)]
        \\
        &\quad
        + 3\sin^2  \theta_1\bigl[
            \cos^2  \theta_2 [4 + \cos(2\phi_2)]
            + \sin^2  \theta_2 [4 + \cos(2\phi_3)]
        \bigr]
        \Bigr).
    \end{split}
\end{equation}
Since $\Xi_2^{(2)} > \tfrac{1}{8}$ in all directions, it follows that $\ket{\psi_{\text{max},0}}$ is a local \emph{minimum} of $\Xi_2$ and hence a maximizer of the $M_2$ measure. In fact, it was shown in~\cite{Liu2025} that this is the global maximum of $\Xi_2$ for two-qubit states.

\section{Notable Three-Qubit Examples}
\label{sec:Three-Qubit States}

During the course of our calculations, we encountered several interesting three-qubit states that extremize the stabilizer fidelity. Among them, we highlight the following:

\begin{enumerate}
    \item The state
    \begin{equation}
        \frac{\ket{100} + \ket{010} + \ket{001}}{\sqrt{3}}
    \end{equation}
    has two nearest stabilizer states among all three-qubit states, with a stabilizer fidelity of $\frac{3}{4}$. This observation suggests that it may be Clifford-equivalent to a tensor product of $\ket{0}$ and the two-qubit state $\ket{G_4,2}$ from Table~\ref{tab:Clifford-inequivalent non-degenerate eigenstates of Clifford operations for two qubits}. A randomized Clifford search confirms this conjecture, yielding:
    \begin{equation}
        C_1 \left( \frac{\ket{100} + \ket{010} + \ket{001}}{\sqrt{3}} \right) = i \ket{0} \ket{G_4,2}   ,
    \end{equation}
    where the Clifford operator $C_1$ is given by:
    \begin{equation}
        C_1 = H_2 S_1  \text{CZ}_{23}  H_1 H_2 H_3  S_3 S_2 S_1  \text{CZ}_{12} \text{CZ}_{13}  S_1  H_1 H_3 H_3  \text{CZ}_{23}  S_1   .
    \end{equation}

    \item Similarly, the state
    \begin{equation}
        \frac{\ket{100} + \ket{010} + \ket{001} + i\ket{111}}{2}
    \end{equation}
    has eight nearest stabilizer states and a stabilizer fidelity of $\frac{5}{8}$. This suggests a possible Clifford equivalence to $\ket{0} \ket{G_{20},1}$, as listed in Table~\ref{tab:Clifford-inequivalent non-degenerate eigenstates of Clifford operations for two qubits}. A randomized search verifies this hypothesis by establishing:
    \begin{equation}
        \frac{i - 2}{\sqrt{5}} \cdot \frac{\ket{100} + \ket{010} + \ket{001} + i\ket{111}}{2} = C_2 \ket{0} \ket{G_{20},4}   ,
    \end{equation}
    where $C_2$ is the Clifford operator
    \begin{equation}
        C_2 = \text{CZ}_{13} S_3  \text{CZ}_{13} S_3  H_1  \text{CZ}_{13} \text{CZ}_{12}  S_1 S_1  H_3 H_2 H_1  \text{CZ}_{12} S_2 H_2  \text{CZ}_{12} \text{CZ}_{23} \text{CZ}_{13}   .
    \end{equation}
    As noted in the previous section, the state $\ket{G_{20},4}$ attains the maximal SRE (for $\alpha = 2$) among all two-qubit states.

    \item On the other hand, the Toffoli state
    \begin{equation}
        \ket{\text{TOF}} = \frac{\ket{000} + \ket{010} + \ket{100} + \ket{111}}{2},
    \end{equation}
    which is Clifford-equivalent to the $\text{CCZ}$ state,
    \begin{equation}
        \ket{\text{CCZ}} = \text{CCZ} \ket{+++} = H_3 \ket{\text{TOF}} = \frac{1}{\sqrt{8}} \sum_{x,y,z \in \{0,1\}} (-1)^{xyz} \ket{xyz}   ,
    \end{equation}
    has eight nearest stabilizer states with a stabilizer fidelity of $\frac{9}{16}$. Unlike the previous examples, it is not evident whether the CCZ state is Clifford-equivalent to a tensor product involving lower-dimensional states. Moreover, an exhaustive randomized search indicates that it is not a non-degenerate eigenstate of any Clifford operator on three qubits. Nonetheless, the Toffoli and CCZ states exhibit significant symmetry, possess known distillation protocols using Clifford operations \cite{Toffoli1,Toffoli2,Gupta2024}, and extremize stabilizer fidelity. These features, in fact, constitute part of the motivation underlying the present work, and these states might be stabilizer states for some subgroup of the Clifford group of three qubits.
\end{enumerate}

\newpage
\section{An Inefficient Distillation Protocol Demonstrating Magic}
\label{sec:An Inefficient Distillation Protocol Demonstrating Magic}

The following three observations suggest the possibility of a new process of magic state distillation:
\begin{itemize}
    \item The gate $\hat{T} \equiv e^{i\pi/4} SH$ is transversal in the perfect five-qubit code, playing a central role in which its eigenstate $\ket{T_0}$ is distillable via this code.
    \item The state $\frac{\ket{T_0T_0} - \ket{T_1T_1}}{\sqrt{2}}$ is a Clifford-stabilizer state, as shown in Table~\ref{tab:Clifford-inequivalent non-degenerate eigenstates of Clifford operations for two qubits}.
    \item The state $\frac{\ket{T_0T_0} - \ket{T_1T_1}}{\sqrt{2}}$ is an eigenstate of $\hat{T} \otimes \hat{T}^{-1}$.
\end{itemize}
These observations lead to the conjecture that the state $\frac{\ket{T_0T_0} - \ket{T_1T_1}}{\sqrt{2}}$ might be distillable via ``a direct product of the perfect code with itself''. Since this state is degenerate, it is expected to mix with the other orthogonal eigenstate:
\[
\frac{\ket{T_0T_0} + \ket{T_1T_1}}{\sqrt{2}} = \frac{\ket{00} + i \ket{11}}{\sqrt{2}},
\]
which is a stabilizer state. The key question is whether distillation remains effective under these conditions.

Indeed, defining the states:
\begin{subequations}
    \begin{gather}
        \ket{\psi_{0}}=\ket{\psi_{00}} = \frac{\ket{T_0 T_0} - \ket{T_1 T_1}}{\sqrt{2}}   , \\
        \ket{\psi_{1}}=\ket{\psi_{01}} = \ket{T_0 T_1}   , \\
        \ket{\psi_{2}}=\ket{\psi_{10}} = \ket{T_1 T_0}   , \\
        \ket{\psi_{3}}=\ket{\psi_{11}} = \frac{\ket{T_0 T_0} + \ket{T_1 T_1}}{\sqrt{2}}   ,
    \end{gather}
\end{subequations}

and encoding two logical qubits by encoding each one using the perfect five-qubit code, we consider an initial state of five copies of a two-qubit system given by:
\begin{equation}
    \rho = (1 - \epsilon_1 - \epsilon_2 - \epsilon_3) \ket{\psi_{00}} \bra{\psi_{00}} + \epsilon_1 \ket{\psi_{01}} \bra{\psi_{01}} + \epsilon_2 \ket{\psi_{10}} \bra{\psi_{10}} + \epsilon_3 \ket{\psi_{11}} \bra{\psi_{11}}   .
    \label{eq:faulty rho}
\end{equation}
This state can be obtained by a dephasing process where the operations $I$, $T \otimes T^{-1}$, and $T^2 \otimes T^{-2}$ are applied with equal probability of $\frac{1}{3}$ each, followed by the application of $I$ or the entangling Clifford gate
\[\begin{pmatrix}
0 & 0 & 0 & i \\
0 & 1 & 0 & 0 \\
0 & 0 & 1 & 0 \\
- i & 0 & 0 & 0
\end{pmatrix},\]
each with probability $\frac{1}{2}$, demonstrating the fact that these are Clifford-stabilizer states.
Notably, a slower distillation process can be performed, yielding a higher fidelity of $0.75$, compared to $0.622$ when distilling two $\ket{T_0}$ states solely using the five-qubit code. This improvement is illustrated in Table~\ref{tab:Clifford-inequivalent non-degenerate eigenstates of Clifford operations for two qubits}. Now we want to show that this distillation occurs.

First, we outline the procedure. Consider two systems, denoted $ A $ and $ B $, each comprising five qubits. The Hilbert space of the $ i $-th qubit ($ i = 1, \dots, 5 $) in system $ s \in \{A, B\} $ is denoted by $ \mH_{s,i} $. The total Hilbert space of the ten-qubit composite system is then
\[
\mH = \bigoplus_{i=1}^{5} \left( \mH_{A,i} \oplus \mH_{B,i} \right).
\]
We refer to each two-qubit subsystem $ \mH_i = \mH_{A,i} \oplus \mH_{B,i} $ as the $ i $-th pair. Each pair is initially prepared in the pure state $ \ket{\psi_{00}} $; however, due to imperfect preparation and subsequent dephasing, the actual state of each pair is the mixed state $ \rho $ defined in Eq.~(\ref{eq:faulty rho}) or (\ref{eq:rho_initial}).
Next, we measure the stabilizer generators of the well-known five-qubit perfect code on both system $ A $ and system $ B $ independently (see Fig.~\ref{fig:new_distillation} for an illustration). We regard the distillation process as successful when all stabilizer measurements yield the $ +1 $ outcome- that is, when the trivial syndrome is obtained and each subsystem is projected onto the corresponding code space of its five-qubit code. Conditioned on this success, the two logical qubits encoded in systems $ A $ and $ B $ are significantly closer- in trace distance or fidelity- to the ideal state $ \ket{\psi_{00}}\bra{\psi_{00}} $ than the initial (ensemble of the) five noisy pairs. The following proposition starts the proof of this statement.

\begin{figure}[h!]
    \centering
    \includegraphics[width=0.99\linewidth,page=3]{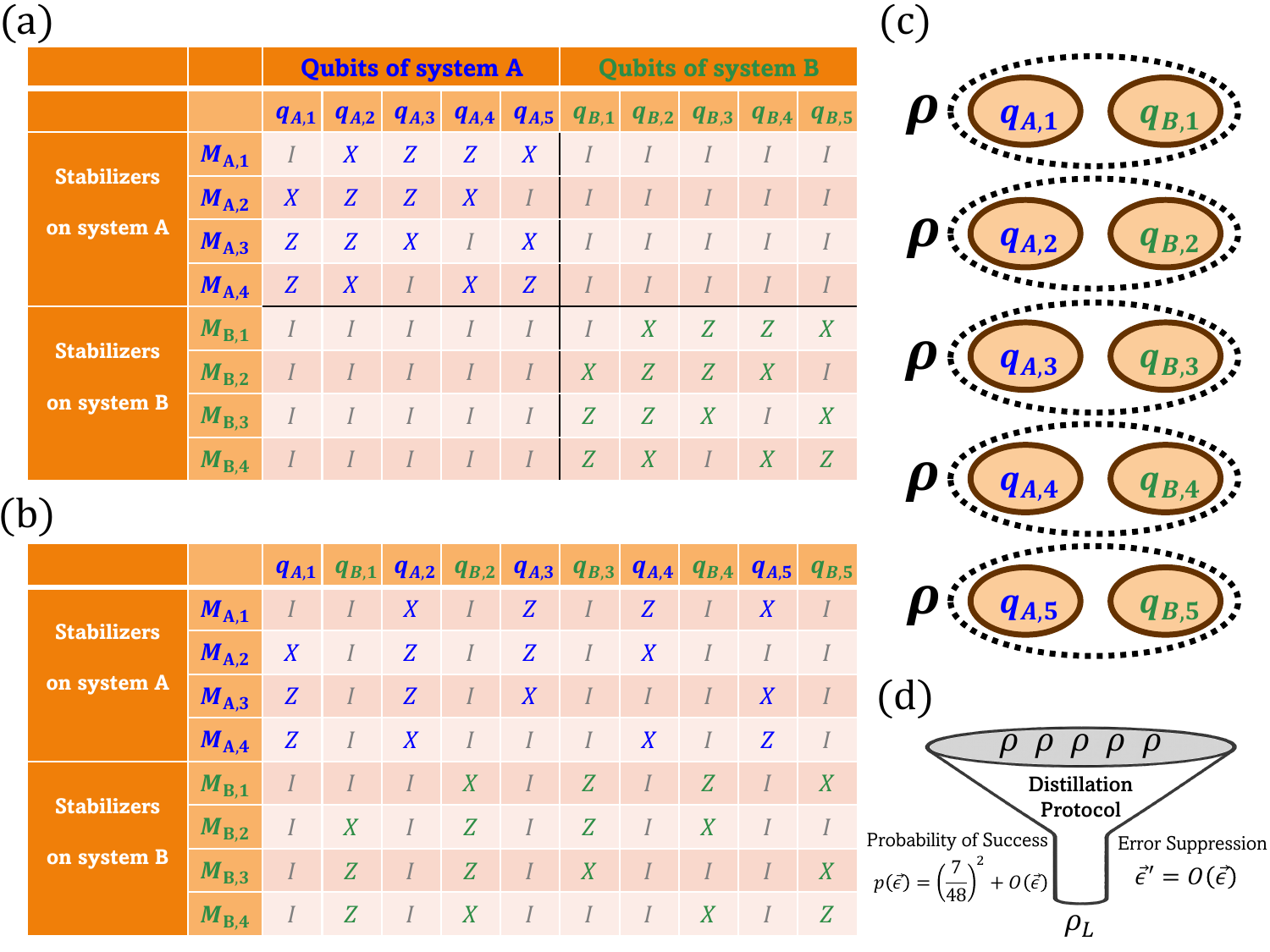}
\caption{
Schematic illustration of the distillation protocol for the state $\ket{\psi_{00}}$. 
(a)~The eight stabilizers of the full ten-qubit system, showing systems $A$ (blue) and $B$ (green) separately, each encoding a logical qubit via the five-qubit perfect code. 
(b)~The same stabilizers represented after reordering the physical qubits according to the convention described in the main text. 
(c)~Preparation stage: each pair $\mH_i$ is initialized in the imperfect state $\rho$ defined in Eq.~(\ref{eq:faulty rho}) or (\ref{eq:rho_initial}), representing a dephased noisy version of $\ket{\psi_{00}}$. 
(d)~Conceptual illustration of the distillation step: simultaneous stabilizer measurements project both subsystems onto their respective code spaces whenever all syndromes are trivial (i.e., all stabilizer outcomes are $+1$). 
Conditioned on this outcome, the resulting logical qubits in $\mH$ exhibit a reduced error rate, with the dominant contribution suppressed linearly in the original preparation error.
}
    \label{fig:new_distillation}
\end{figure}

\begin{proposition}
\label{prop:preserving the structure}
    Consider the mixed state
    \begin{equation}
    \label{eq:rho_initial}
            \rho = (1 - \epsilon_1 - \epsilon_2 - \epsilon_3) \ket{\psi_{0}} \bra{\psi_{0}}
             + (a + b i)\ket{\psi_0}\bra{\psi_3} + (a - b i)\ket{\psi_3}\bra{\psi_0}+ \sum_{i=1}^3 \epsilon_i \ket{\psi_{i}} \bra{\psi_{i}},
    \end{equation}
    where $\epsilon_1$, $\epsilon_2$, $\epsilon_3$, $a$, and $b$ are real parameters constrained such that $\rho$ is a valid density operator (i.e., positive semidefinite and $\mathrm{Tr} \rho = 1$).
    We consider five independent copies of this state, each characterized by its own set of parameters $\{\epsilon_1, \epsilon_2, \epsilon_3, a, b\}$.
    Upon measuring the stabilizer generators of the code described above and obtaining the trivial syndrome (all $+1$ outcomes), the post-selected logical state retains the same algebraic form as Eq.~(\ref{eq:rho_initial}), but with renormalized parameters that are nonlinear functions of the initial ones.
\end{proposition}

\begin{proof}[Proof of Proposition~\ref{prop:preserving the structure}]
    Note that $\hat{T}\otimes\hat{T}^{-1}\ket{\psi_n} = e^{i \frac{2\pi n}{3}}\ket{\psi_n}$ for each pair. Therefore, denoting $\ket{\psi_{\vec{n}}}\equiv\bigotimes_{i=1}^5 \ket{\psi_{n_i}}_i$ and applying this logically (transversally in this code) on the logical qubits, we get:
    \begin{equation}
        \hat{T}_{A,L}\hat{T}_{B,L}^{-1}\ket{\psi_{\vec{n}}} = e^{i \frac{2\pi}{3}\sum_{i=1}^5 n_i} \ket{\psi_{\vec{n}}}.
    \end{equation}
    Since $\hat{T}_{A,L}\hat{T}_{B,L}^{-1}$ commutes with the projector onto the code space, $\Pi$, we conclude that:
    \begin{subequations}
        \begin{gather}
            \Pi \ket{\psi_{\vec{n}}} \propto \Pi \ket{\psi_{(2,2,2,2,2)}} \quad \text{if} \quad \sum_{i=1}^5 n_i \mod 3 = 1, \\
            \Pi \ket{\psi_{\vec{n}}} \propto \Pi \ket{\psi_{(1,1,1,1,1)}} \quad \text{if} \quad \sum_{i=1}^5 n_i \mod 3 = 2, \\
            \Pi \ket{\psi_{\vec{n}}} = \alpha_{\vec{n}} \Pi \ket{\psi_{(0,0,0,0,0)}} + \beta_{\vec{n}} \Pi \ket{\psi_{(3,3,3,3,3)}} \quad \text{if} \quad \sum_{i=1}^5 n_i \mod 3 = 0.
        \end{gather}
    \end{subequations}
    The density matrix of a system consisting of five instances of the state in Eq.~(\ref{eq:rho_initial}), not necessarily with identical parameter values, is a linear combination of terms of the form $\ket{\psi_{\vec{n}}} \bra{\psi_{\vec{m}}}$, where the indices satisfy the condition 
    \begin{equation}
        \sum_{i=1}^{5} n_i \mod 3 = \sum_{i=1}^{5} m_i \mod 3.
    \end{equation}
    Thus, with projecting onto the code space, the proposition follows.
\end{proof}

A direct calculation for the case of five exact copies leads to the updated parameters to the linear order as follows:
\begin{subequations}
    \begin{gather}
        \epsilon_1^\prime=\frac{45}{49} \epsilon_1   , \\
        \epsilon_2^\prime=\frac{45}{49} \epsilon_2   , \\
        \epsilon_3^\prime=\frac{5}{49} \epsilon_3   , \\
        a^\prime+b^\prime i = -\frac{5}{7}(a+bi)   .
    \end{gather}
\end{subequations}
with a probability of $p=\left(\frac{7}{48}\right)^2+O(\epsilon_1,\epsilon_2,\epsilon_3,a,b)$ for measuring the trivial syndrome. Starting with $a=b=\epsilon_1=\epsilon_2=0,\epsilon_3=\epsilon$, the probability of the right syndrome is $p=\frac{49 - 240 \epsilon + 600 \epsilon^2 - 640 \epsilon^3 + 240 \epsilon^4}{2304}$ and the updated error is $\epsilon^\prime=\frac{\epsilon \left(5 + 100 \epsilon - 240 \epsilon^2 + 160 \epsilon^3 - 16 \epsilon^4\right)}
    {49 - 240 \epsilon + 600 \epsilon^2 - 640 \epsilon^3 + 240 \epsilon^4}$.

Although the distillation process is very slow, as $\epsilon^\prime$ scales linearly with $\epsilon$, we have demonstrated that this new state can indeed be distilled using stabilizer codes.
Moreover, this non-stabilizer state can be used as a magic state for state-injection into a stabilizer-code-based quantum computation, thereby enabling universal quantum computation.
Thus, it qualifies as a magic state, in the same sense as $T$-states, $H$-states, $\text{CZZ}$-states, $\text{TOF}$-states, and other single-qudit magic states, such as the strange states in single qutrits.
Its stabilizer fidelity exceeds that of other known magic two-qubit states like $\ket{TT}$, $\ket{TH}$, and $\ket{HH}$, leaving open the possibility that a more efficient distillation protocol exists, one that could suppress errors more rapidly and enhance the practicality of this state as a magic resource.

\section{$L^{p}$-Norm Generalizations of Magic Quantifiers}
\label{sec: Candidate Measures to investigate}

Motivated by the role of $L^{p}$ norms of characteristic functions in phase-space formulations, as well as by the structure of stabilizer Rényi entropies (SREs), we propose a natural $L^{p}$-norm generalization of two standard quantifiers of non-stabilizerness: the stabilizer fidelity and the mana.
Specifically, we introduce the following two families of functionals on pure states:
\begin{subequations}
\begin{gather}
    F_p(\ket{\psi})
    :=
    \left( \sum_{\ket{s} \in \mSS_{N,d}}
        \bigl| \langle \psi | s \rangle \bigr|^{p} \right)^{1/p} , \\
    M_p(\ket{\psi})
    :=
    \left( \sum_{\boldsymbol{\chi} \in \mathbb{V}_{N,d}}
        \bigl| W_{\boldsymbol{\chi}} (\ket{\psi}\bra{\psi}) \bigr|^{p} \right)^{1/p} .
\end{gather}
\end{subequations}
For $p=\infty$, $F_p$ reduces to the square root of the stabilizer fidelity, while $M_1$ coincides with the mana.
It would be of interest to investigate the mathematical properties of these functionals, their behavior under Clifford operations and stabilizer-preserving maps, and assess the extent to which they qualify as valid magic measures within the stabilizer resource theory.

\section{Conclusion}
\label{sec:Conclusion}

The framework of group-covariant functionals developed in this work provides a unifying geometric perspective on the extremal behavior of magic measures and similar quantities.
As an important specialization of this general theory, we obtain a clear structural understanding of quantum states that extremize standard measures of non-stabilizerness or magic.
Specifically, we showed that pure states stabilized by finite subgroups of the Clifford group, which belong to some Clifford-stabilized space in our terminology, serve as extremizers or critical states for the stabilizer fidelity, the stabilizer Rényi entropies, the generalized stabilizer entropies and the mana of magic when variations are restricted to directions orthogonal to their stabilized spaces. This result identifies these states as special geometric points in the magic landscape, which are critical under local perturbations, and hence fundamental to the geometry of non-stabilizerness.

Importantly, our analysis is not confined to the conventional Pauli-Clifford framework.
Instead, it reveals a more general connection between the extremality of suitably symmetric functionals and underlying group-theoretic stabilization principles.
This perspective naturally extends to a broad class of resource theories, including measures of fermionic magic and non-Gaussianity in finite-dimensional Hilbert spaces \cite{FM1,FM2,FM3,FM4,FM5}.
In such settings, the relevant group is the group of Gaussian unitary transformations rather than the Clifford group.
Within this generalized framework, our results encompass recently proposed commutant-based measures of non-Gaussianity too~\cite{FermionicMagic2025}, placing them within a unified structural picture.
Furthermore, we introduced the notion of the \emph{group-stabilizer extent} (Appendix~\ref{app:group-stabilizer extent}), which extends the standard stabilizer extent beyond the Clifford setting to arbitrary finite unitary groups.
This construction provides a refined and flexible quantitative tool for characterizing the resource content of structured quantum states, and highlights the central role of symmetry and covariance in shaping extremal behavior across diverse quantum resource theories.

Through explicit examples, we classified and analyzed the critical behavior of non-degenerate Clifford eigenstates in several important systems.
For single qudits, we characterized the extremizing behavior of these measures in dimensions $d=2,3,5$, revealing that the Clifford-stabilizer framework naturally captures all known single-qudit magic candidates and exposes new ones.
For $d=5$, we completely characterized the extremizing behavior of the mana of magic at all these states but one, though we left the behavior of the stabilizer fidelity and Rényi entropies for the reader.
For two qubits, we identified all Clifford-inequivalent non-degenerate eigenstates of Clifford operations, including three previously unreported ones. Among these, we highlighted a new two-qubit magic state that, up to Clifford equivalence, coincides with the state recently shown to maximize the SRE of magic~\cite{Liu2025}. 

Furthermore, we proposed an explicit (albeit inefficient) distillation protocol that produces this state, whose stabilizer fidelity surpasses known benchmarks such as $\ket{TT}$ and $\ket{TH}$. This provides operational evidence that the extremizing structures we identify are not merely mathematical curiosities but also physically relevant distillation targets.

Our numerical results also support a broader conjecture: \emph{any Clifford-stabilizer state with a non-vanishing Wigner function locally maximizes the mana of magic under perturbations orthogonal to its stabilized subspace}. 
If true, this conjecture would establish a deep connection between phase-space positivity and local extremality in the geometry of quantum magic.

Another central open question emerging from our analysis concerns the structural nature of SIC-POVM fiducial states.
Motivated by the observed extremality of SIC fiducials for $L^p$-norm magic and stabilizer R\'enyi entropies, together with numerical and theoretical evidence indicating that all known fiducials are eigenstates of finite-order Clifford unitaries, we have put forward the conjecture that every SIC-POVM fiducial state is a Clifford-stabilizer state.
If true, this conjecture would imply that the distinguished role of SIC fiducials in quantum information is not merely a consequence of their informational completeness, but rather reflects an underlying stabilizer symmetry that places them at a boundary between stabilizer geometry and maximal non-stabilizerness.

Several other open questions remain.
Chief among them is whether Clifford-stabilizer states constitute the only pure states that locally extremize magic monotones.
Another important direction concerns the design of efficient distillation protocols for the candidate states identified in this work, especially the high-fidelity new two-qubit state for which we proposed an inefficient distillation protocol.
A particularly accessible and concrete direction for future work is the complete classification of Clifford-stabilizer states for two qubits. This task is now within reach, especially in light of recent enumerations of all subgroups of the two-qubit Clifford group~\cite{Grier2022classificationof,kubischta2024classificationsubgroupstwoqubitclifford}.
A systematic investigation of these states, their extremizing properties, and their potential utility in magic distillation could yield valuable insights and possibly lead to new, practically relevant protocols.
Finally, pursuing a deeper understanding of the behavior of other magic measures under perturbations orthogonal to stabilized subspaces naturally presents itself as a promising avenue for revealing new geometric principles that govern the structure of quantum resources.

\section*{Acknowledgments}

Our work has been supported by the Israel Science Foundation (ISF) and the Directorate for Defense Research and Development (DDR\&D) Grant No. 3427/21, the ISF Grant No. 1113/23, and the US-Israel Binational Science Foundation (BSF) Grants No. 2020072 and 2024140.
We are grateful to X. Turkeshi and S. Luo for valuable comments regarding the role of $L^p$-based measures, closely related to stabilizer Rényi entropies, and for insightful discussions and suggestions.

\printbibliography

\appendix

\newpage
\section{Group-Stabilizer Extent
\label{app:group-stabilizer extent}}

In this appendix, we introduce and study the concept of \emph{group-stabilizer extent}, a natural generalization of stabilizer extent to quantum states stabilized by finite subgroups of the Clifford group. Although not directly essential to the main results of this work, this quantity fits naturally into the broader framework developed here and offers a deeper perspective on resource quantification within group-stabilized subspaces. In a similar spirit to our generalization of stabilizer fidelity, the group-stabilizer extent exemplifies how resource measures can be extended to more structured settings, potentially leading to new insights and applications.

\subsection*{The Stabilizer Extent}

We begin by recalling the \emph{stabilizer extent}~\cite{Bravyi2019simulationofquantum}, a central quantity in the resource theory of magic:
\begin{equation}
    \xi(\ket{\psi})
    :=
    \min
    \left\{
        \left( \sum_{\ket{s} \in \mSS_{N,d}} |c_{\ket{s}}| \right)^{2}
        \; \middle| \;
        \ket{\psi} = \sum_{\ket{s} \in \mSS_{N,d}} c_{\ket{s}} \ket{s}
    \right\}.
\end{equation}
The stabilizer extent measures the ``best'' decomposition of a given pure state $\ket{\psi}$ into stabilizer states in terms of the minimal possible $\ell_{1}$-norm of the coefficient vector.
Operationally, it quantifies the minimal overhead required to simulate $\ket{\psi}$ using stabilizer-state expansions and plays a key role in classical simulation algorithms: States with small extent admit efficient simulation, whereas states with large extent generally do not.  
Among all stabilizer decompositions of $\ket{\psi}$, the optimal one is sometimes called the \emph{minimum stabilizer decomposition}.

For pure states, the stabilizer extent is directly related to the Max-relative entropy of magic:
\[
    \mD_{\mathrm{max}}(\ket{\psi}) = \log \xi(\ket{\psi}).
\]
The Max-relative entropy of magic (equivalently, the logarithm of the generalized robustness) is defined for arbitrary states by~\cite{LiuBuTakagi2019,Regula2021}
\begin{equation}
    \mD_{\mathrm{max}}(\rho)
    :=
    \log\bigl( 1 + R_{g}(\rho) \bigr),
\end{equation}
where the generalized robustness satisfies
$
    1 + R_{g}(\rho)
    =
    \min \{ \lambda \; \mid \; \rho \le \lambda \sigma,\ \sigma \in \mathrm{STAB}_{N,d} \}.
$
Equivalently, the max-relative entropy between two states is
\[
    D_{\mathrm{max}}(\rho\|\sigma)
    :=
    \log \min \bigl\{ \lambda \ \; \big| \;\ \rho \le \lambda \sigma \bigr\},
\]
where the matrix inequality $\rho \le \lambda \sigma$ means that
$\lambda \sigma - \rho$ is positive semidefinite.

Thus, $\mD_{\mathrm{max}}$ captures the smallest multiplicative factor by which a stabilizer state must be ``stretched'' to dominate $\rho$, and for pure states this quantity exactly reproduces the stabilizer extent.
This connection highlights the stabilizer extent as a geometric measure of how far a state lies outside the stabilizer polytope and provides the bridge between stabilizer decompositions and robustness-based notions of magic.

\subsection*{Group-Stabilizer Extent}

We now introduce the notion of \emph{$G$-stabilizer extent}.
In the following, we show that this generalized extent satisfies a structural property directly paralleling a well-known result for the ordinary stabilizer extent.

\begin{definition}
    Given a finite subgroup $ G \subset \mathrm{U}(\mH) $, the $\boldsymbol{G}$\textbf{-stabilizer extent} of a normalized state $\ket{\psi}\in\mH$ for which $F_G(\ket{\psi})>0$, denoted by $\xi_G(\ket{\psi})$, is
    \begin{equation}
        \xi_G(\ket{\psi}) := \min \left\{ \left(\sum_{\ket{s}\in\mSS_G} |c_{\ket{s}}|\right)^2      :     P_{\Span(\mS_G)}\ket{\psi} = \sum_{\ket{s}\in\mSS_G} c_{\ket{s}} \ket{s} \right\}  .
    \end{equation}
\end{definition}

\begin{theorem}
\label{theo:G-stabilizer extent}
    Let $G$ be a finite subgroup of $\mathrm{U}(\mH)$. Then, for every normalized $\ket{\psi}\in\mH$ with $F_G(\ket{\psi})>0$, it holds that
    \[\xi_G(\ket{\psi}) \geq \frac{\braket{\psi|P_{\Span(\mS_G)}|\psi}^2}{F_G(\ket{\psi})}  .\]
    If in addition $\ket{\psi}$ is a $\mQ$-stabilizer state where $\mQ$ is a finite subgroup of $\mathrm{N}_{\mathrm{U(\mH)}}(G)$, then
    \[P_{\Span(\mS_G)}\ket{\psi} =\ket{\psi} \quad \text{and} \quad \xi_G(\ket{\psi})=\frac{1}{F_G(\ket{\psi})}   .\]
\end{theorem}
The proof of this theorem follows exactly the same reasoning as the proof of Proposition 2 in \cite{Bravyi2019simulationofquantum}.  
Notably, the authors extend their result at the conclusion of their proof to encompass Clifford-stabilizer states, as defined here, rather than restricting themselves to the "Clifford magic states" introduced in their paper.  
The only substantive difference in our proof, which remains otherwise equivalent, is that we do not impose the restriction that $ \ket{s} $ (denoted as $ \ket{\phi_0} $ in \cite{Bravyi2019simulationofquantum}) must be stabilized by only a single operator in $ \mQ $, specifically the identity operator.
Apart from this, our proof is more general: it applies beyond the Pauli and Clifford groups, does not assume that the group-stabilizer states span the entire Hilbert space, and accommodates additional generalizations. Nonetheless, it remains nearly identical to the original proof.
\begin{proof}[Proof of Theorem~\ref{theo:G-stabilizer extent}]
    For any $G$-stabilizer decomposition of $P_{\Span(\mS_G)}\ket{\psi}$
    \[
    P_{\Span(\mS_G)}\ket{\psi}=\sum_{i=1}^{r} c_i \ket{s_i}   ,
    \]
    where $\ket{s_i}\in\mSS_G   \forall i$, we have
    \[
    \braket{\psi|P_{\Span(\mS_G)}|\psi}=\left| \sum_{i=1}^{r} c_i \braket{\psi|s_i} \right| \leq  \sum_{i=1}^{r} \left|c_i\right| \left|\braket{\psi|s_i}\right| \leq \sqrt{F_G(\ket{\psi})} \sum_{i=1}^{r} \left|c_i\right|   .
    \]
    Therefore,
    \[
    \sqrt{\xi_G(\psi)} \geq \frac{\braket{\psi|P_{\Span(\mS_G)}|\psi}}{\sqrt{F_G(\ket{\psi})}}   .
    \]

    Now, let $\mQ$ be a finite subgroup of $\mathrm{N}_{\mathrm{U}(\mH)}(G)$ such that
    \[\ket{\psi}\bra{\psi}=\frac{1}{|\mQ|}\sum_{q\in\mQ}q  .
    \]
    Denote
    \[
    \ket{s}=\argmax_{\ket{s^\prime}\in\mSS_G}|\langle s^\prime \ket{\psi}|^2   .
    \]
    $|\braket{\psi|s}|>0$ since $F_G(\ket{\psi})>0$. Then
    \[
    \ket{\psi}=\frac{\ket{\psi}\braket{\psi|s}}{\braket{\psi|s}}=\frac{1}{|\mQ|\braket{\psi|s}}\sum_{q\in\mQ}q\ket{s} ,
    \]
    such that $\ket{\psi}\in\Span(\mS_G)$ and then $P_{\Span(\mS_G)}\ket{\psi} =\ket{\psi}$.
    Now, each $q\ket{s}$ in the sum is a stabilizer state since $q\in\mathrm{N}_{\mathrm{U}(\mH)}(G)$. If $q_1\ket{s}=e^{i\phi}q_2\ket{s}$ for some $q_1,q_2\in\mQ$, then $\phi=0$ since $0\neq\braket{\psi|s}=\bra{\psi}q\ket{s}$ for every $q\in\mQ$. Therefore,
    \[
    \ket{\psi}=\frac{1}{|\mQ|\braket{\psi|s}}\sum_{i=1}^p n_i\ket{s_i} ,
    \]
    where the set $\{ \ket{s_i} \}_{i=1}^p$ is a set of distinct stabilizer states, and $\{ n_i \}_{i=1}^p$ are natural numbers such that $\sum_{i=1}^p n_i=|\mQ|$. Therefore,
    \[
    \sqrt{\xi_G(\ket{\psi})} \leq \sum_{i=1}^p \left| \frac{n_i}{|\mQ|\braket{\psi|s}}\right|=\frac{1}{\left| \braket{\psi|s} \right|}=\frac{1}{\sqrt{F_G(\ket{\psi})}}   .
    \]
\end{proof}

Another theorem than can be proved is the following.
\begin{theorem}
	\label{theo:witness}
        Let $G$ be a finite subgroup of $\mathrm{U}(\mH)$. Then, for every normalized state $\ket{\psi}$ with $F_G(\ket{\psi})>0$, it holds that
	\begin{equation}
	\xi_G(\ket{\psi})=\max_{\ket{\omega}} \frac{\left|\braket{\psi|\omega}\right|^2}{F_G(\ket{\omega})},
	\label{eq:witness}
	\end{equation}
	where the maximum is over all states $\ket{\omega}$ in the Hilbert space. 
\end{theorem}
\noindent
 Thus any state $\ket{\omega}$ in the Hilbert space can act as a witness to provide a lower bound on $\xi_G$ and, furthermore, there exists at least one optimal witness state $\ket{\omega_{\star}}$ which achieves the maximum in Eq.~\eqref{eq:witness}. For example, choosing $\ket{\omega}=\ket{\psi}$, we get the lower bound 
\begin{equation}
\xi_G(\ket{\psi}) \geq \frac{1}{F_G(\ket{\psi})}.
\label{eq:mirrorwitness}
\end{equation}
For $\mQ$-stabilizer states where $\mQ$ is a finite subgroup of $\mathrm{N}_{\mathrm{U(\mH)}}(G)$ this lower bound is tight as stated in Theorem~\ref{theo:G-stabilizer extent}. The proof closely parallels that of Theorem 4 in~\cite{Bravyi2019simulationofquantum}, and is therefore omitted here for brevity.

\newpage
\section{$\alpha$-SRE for single qubits}
\label{app:alpha-SREforsinglequbits}

The calculated $\alpha$-SRE for a general single qubit state,
\[
    \ket{\psi} = \cos\theta \ket{0} + e^{i\phi}\sin\theta \ket{1},
\]
is given by
\begin{equation}
        M_\alpha =\frac{1}{1-\alpha} \log  \Bigl( \tfrac{1}{2}\bigl( 1 + (\cos^2 2\theta)^{\alpha} + (\sin^2 2\theta)^{\alpha} \left[ (\cos^2\phi)^{\alpha} + (\sin^2\phi)^{\alpha} \right] \bigr) \Bigr).
\end{equation}
For the two canonical magic states, $\ket{T}$ and $\ket{H}$, the corresponding stabilizer $\alpha$-Rényi entropies read
\begin{subequations}
    \begin{gather}
        M_\alpha(\ket{T})= \frac{1}{1-\alpha} \log  \left[ \tfrac{1}{2}\left( 1 + 3^{1-\alpha} \right) \right],\\
        M_\alpha(\ket{H})= \frac{1}{1-\alpha} \log  \left[ \tfrac{1}{2}\left( 1 + 2^{1-\alpha} \right) \right].
    \end{gather}
\end{subequations}
As seen, the $\ket{T}$ state attains the largest SRE under this measure, thereby saturating the upper bound in Eq.~(\ref{eq: Bound on SRE}).
For the particularly relevant case $\alpha = 2$, these expressions yield
\begin{equation}
    M_2 \left( \ket{T} \right) = \log \frac{3}{2}
    \quad , \quad
    M_2 \left( \ket{H} \right) = \log \frac{4}{3}.
\end{equation}

\section{Wigner function for the single-qutrit states}
\label{app:Wigner function for the single-qutrit states}

Here, we present detailed information about the Wigner functions, stabilizer fidelities, and nearest stabilizer states for single-qutrit nonstabilizer Clifford-inequivalent non-degenerate eigenstates.

\begin{table}[h!]
    \centering
    \resizebox{\linewidth}{!}{
    \begin{tabular}{|c|c|c|c|c|}
        \hline
        \textbf{State} &
        \textbf{\makecell{Wigner Function Matrix}} &
        \textbf{\makecell{Wigner Trace\\Norm}} \\ 
        \hline
        
        $\ket{\mathbb{S}}=\frac{\ket{1}-\ket{2}}{\sqrt{2}}$  & 
        $\begin{pmatrix}
-\frac{1}{3} & \frac{1}{6} & \frac{1}{6} \\
\frac{1}{6}  & \frac{1}{6} & \frac{1}{6} \\
\frac{1}{6}  & \frac{1}{6} & \frac{1}{6}
\end{pmatrix}$ & $\frac{5}{3}\approx1.66667$ \\ 
        \hline
        
        $\ket{\mathbb{N}}=\frac{-\ket{0}+2\ket{1}-\ket{2}}{\sqrt{6}}$ & 
        $\begin{pmatrix}
-\frac{1}{6} & \frac{1}{6} & \frac{1}{6} \\
\frac{1}{3}  & \frac{1}{6} & \frac{1}{6} \\
-\frac{1}{6} & \frac{1}{6} & \frac{1}{6}
\end{pmatrix}$ & $\frac{5}{3}\approx1.66667$ \\ 
        \hline
        
        $\ket{H_+}=\frac{(1+\sqrt{3})\ket{0}+\ket{1}+\ket{2}}{\sqrt{2(3+\sqrt{3})}}$  & 
        $\begin{pmatrix}
\frac{1}{3} & \frac{1}{12}\bigl(1+\sqrt{3}\bigr) & \frac{1}{12}\bigl(1+\sqrt{3}\bigr) \\
\frac{1}{12}\bigl(1+\sqrt{3}\bigr) & \frac{1}{12}\bigl(1-\sqrt{3}\bigr) & \frac{1}{12}\bigl(1-\sqrt{3}\bigr) \\
\frac{1}{12}\bigl(1+\sqrt{3}\bigr) & \frac{1}{12}\bigl(1-\sqrt{3}\bigr) & \frac{1}{12}\bigl(1-\sqrt{3}\bigr)
\end{pmatrix}$ & $\frac{1}{3} + \frac{2}{\sqrt{3}}\approx1.48803$ \\ 
        \hline
        
        $\ket{\mathbb{T}}=\frac{e^{2\pi i/9}\ket{0}+\ket{1}+e^{-2\pi i/9}\ket{2}}{\sqrt{3}}$ &  
        $\begin{pmatrix}
\frac{1}{9}\Bigl(1+2\cos\Bigl(\frac{2\pi}{9}\Bigr)\Bigr) & \frac{1}{9}\Bigl(1-2\cos\Bigl(\frac{\pi}{9}\Bigr)\Bigr) & \frac{1}{9}\Bigl(1+2\sin\Bigl(\frac{\pi}{18}\Bigr)\Bigr) \\
\frac{1}{9}\Bigl(1+2\sin\Bigl(\frac{\pi}{18}\Bigr)\Bigr) & \frac{1}{9}\Bigl(1+2\cos\Bigl(\frac{2\pi}{9}\Bigr)\Bigr) & \frac{1}{9}\Bigl(1-2\cos\Bigl(\frac{\pi}{9}\Bigr)\Bigr) \\
\frac{1}{9}\Bigl(1+2\cos\Bigl(\frac{2\pi}{9}\Bigr)\Bigr) & \frac{1}{9}\Bigl(1-2\cos\Bigl(\frac{\pi}{9}\Bigr)\Bigr) & \frac{1}{9}\Bigl(1+2\sin\Bigl(\frac{\pi}{18}\Bigr)\Bigr)
\end{pmatrix}$ & $\frac{1}{3}\Bigl( 1 + 4\cos\Bigl(\frac{\pi}{9}\Bigr)\Bigr)\approx1.58626$ \\ 
        \hline
    \end{tabular}
    }
    \caption{Single-qutrit nonstabilizer Clifford-inequivalent Clifford nondegenerate eigenstates with their corresponding Wigner function matrices, and Wigner trace norms.}
    \label{tab:Single-qutrit nonstabilizer Clifford-inequivalent Clifford nondegenerate eigenstates and Wigner function}
\end{table}

\begin{table}[h!]
    \centering
    \resizebox{\linewidth}{!}{
    \begin{tabular}{|c|c|c|c|c|}
        \hline
        \textbf{State} & 
        \textbf{\makecell{Stabilizer\\Fidelity}} & 
        \textbf{\makecell{Number of\\Nearest\\Stabilizer States}} & 
        \textbf{\makecell{Nearest\\Stabilizer States}} \\ 
        \hline
        
        $\ket{\mathbb{S}}=\frac{\ket{1}-\ket{2}}{\sqrt{2}}$ & $\frac{1}{2}$ & $8$ & 
        \makecell{
        $\ket{1}$ ,
        $\ket{2}$ , \\
        $\frac{\ket{0}+e^{\pm 2\pi i/3}\ket{1}+\ket{2}}{\sqrt{3}}$ , \\
        $\frac{e^{2\pi i/3}\ket{0}+e^{\pm 2\pi i/3}\ket{1}+\ket{2}}{\sqrt{3}}$ , \\
        $\frac{e^{-2\pi i/3}\ket{0}+e^{\pm 2\pi i/3}\ket{1}+\ket{2}}{\sqrt{3}}$
        }
        \\
        \hline
        
        $\ket{\mathbb{N}}=\frac{-\ket{0}+2\ket{1}-\ket{2}}{\sqrt{6}}$ & $\frac{2}{3}$ & $3$ & 
        \makecell{$\ket{1}$ , \\ $\frac{\ket{0}+e^{\pm 2\pi i/3}\ket{1}+\ket{2}}{\sqrt{3}}$}
         \\ 
        \hline
        
        $\ket{H_+}=\frac{(1+\sqrt{3})\ket{0}+\ket{1}+\ket{2}}{\sqrt{2(3+\sqrt{3})}}$ & $\frac{3+\sqrt{3}}{6}$ & $2$ & 
        \makecell{$\ket{0}$ , \\ $\frac{\ket{0}+\ket{1}+\ket{2}}{\sqrt{3}}$} \\ 
        \hline
        
        $\ket{\mathbb{T}}=\frac{e^{2\pi i/9}\ket{0}+\ket{1}+e^{-2\pi i/9}\ket{2}}{\sqrt{3}}$ & $\frac{1}{9} \left(1 + 2\cos\left(\frac{2\pi}{9}\right)\right)^2$ & $3$ & 
        \makecell{$\frac{\ket{0}+\ket{1}+\ket{2}}{\sqrt{3}}$ , \\ $\frac{e^{2\pi i/3}\ket{0}+\ket{1}+\ket{2}}{\sqrt{3}}$ , \\ $\frac{e^{2\pi i/3}\ket{0}+e^{2\pi i/3}\ket{1}+\ket{2}}{\sqrt{3}}$} \\ 
        \hline
    \end{tabular}
    }
    \caption{Single-qutrit nonstabilizer Clifford-inequivalent Clifford nondegenerate eigenstates with their stabilizer fidelities, number of nearest stabilizer states, and nearest stabilizer states.}
    \label{tab:Single-qutrit nonstabilizer Clifford-inequivalent Clifford nondegenerate eigenstates and nearest SS}
\end{table}

\begin{figure}[h!]
    \centering
    \begin{subfigure}[b]{0.45\linewidth}
        \centering
        \includegraphics[width=\linewidth]{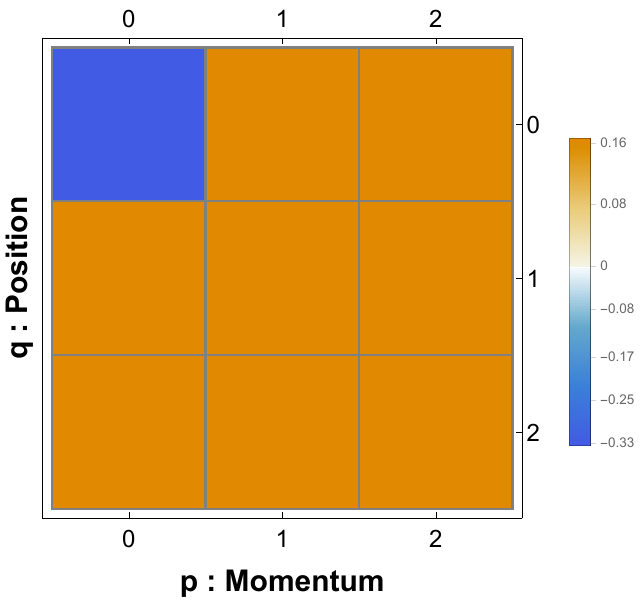}
        \caption{$\ket{\mathbb{S}}$}
        \label{fig:Strange}
    \end{subfigure}
    \hfill
    \begin{subfigure}[b]{0.45\linewidth}
        \centering
        \includegraphics[width=\linewidth]{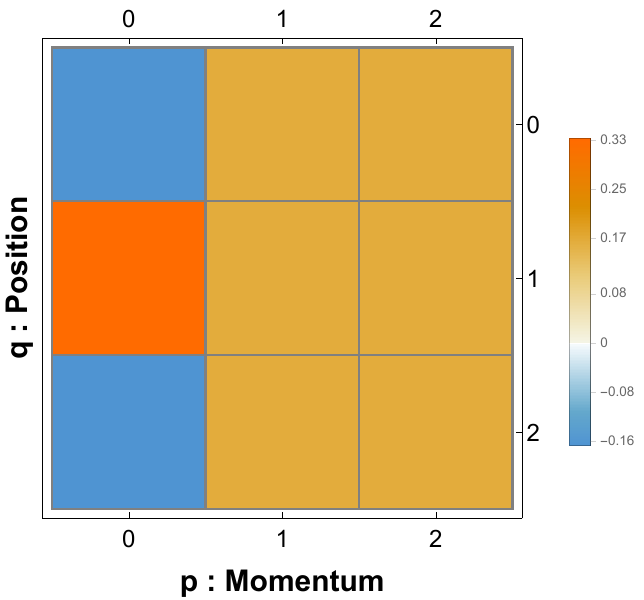}
        \caption{$\ket{\mathbb{N}}$}
        \label{fig:Norrell}
    \end{subfigure}
    
    \vspace{0.5cm}
    
    \begin{subfigure}[b]{0.45\linewidth}
        \centering
        \includegraphics[width=\linewidth]{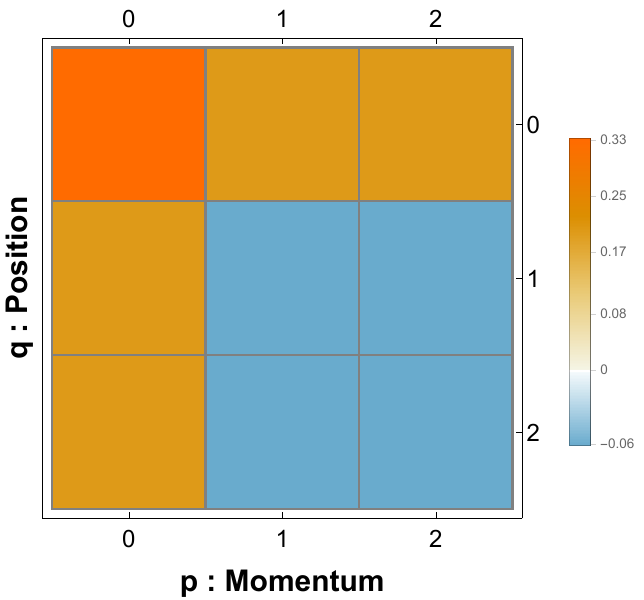}
        \caption{$\ket{H_+}$}
        \label{fig:Hplus}
    \end{subfigure}
    \hfill
    \begin{subfigure}[b]{0.45\linewidth}
        \centering
        \includegraphics[width=\linewidth]{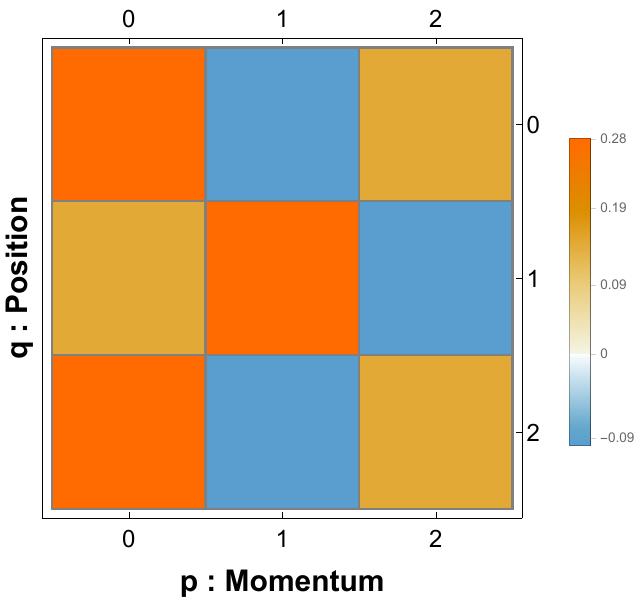}
        \caption{$\ket{\mathbb{T}}$}
        \label{fig:Tstate}
    \end{subfigure}
    
    \caption{Wigner function heatmap visualizations for all nonstabilizer Clifford-inequivalent non-degenerate eigenstates of Clifford operations of single qutrits.}
    \label{fig:qutrit Wigner functions visualizations}
\end{figure}

These data concretely illustrate the behavior of certain single-qutrit states with respect to magic measures, supporting the theoretical discussions and examples of the main text.

\newpage
\section{SRE for Single qutrits}
\label{app:SREforSingleQutrits}

\newpage
\subsection*{General Form of the probability distribution and of $2$-SRE}
The calculated phase-space probability distribution 
$P_{\boldsymbol{\chi}}$ for a general single qutrit state,
\[
    \ket{\psi} =
    \cos\theta_1 \ket{0}
    + e^{i\phi_1}\sin\theta_1\cos\theta_2 \ket{1}
    + e^{i\phi_2}\sin\theta_1\sin\theta_2 \ket{2},
\]
is given by
\begin{subequations}
    \begin{gather}
        3P_{(0,0)} = 1, \\
        3P_{(0,1)} = 3P_{(0,2)} = \cos^4\theta_1 - \cos^2\theta_1 \sin^2\theta_1 + \tfrac{1}{8}(5 + 3\cos 4\theta_2)\sin^4\theta_1, \\
        \begin{split}
            48P_{(1,0)} = 48P_{(2,0)} = & \sin^2\theta_1 \bigl( 9 + 7\cos 2\theta_1 - 2\cos 4\theta_2 \sin^2\theta_1 \\
             & + 8\sin2\theta_1 \sin2\theta_2 (\cos\theta_2 \cos(2\phi_1-\phi_2) + \cos(\phi_1 - 2\phi_2)\sin\theta_2) \bigr)\\
             & + 4\cos(\phi_1 + \phi_2)\sin^2 2\theta_1 \sin 2\theta_2
        \end{split}, \\
        \begin{split}
            96P_{(1,1)} = & \sin^2\theta_1 \bigl[ 18 + 14\cos 2\theta_1 + \cos 2(\theta_1 - 2\theta_2) - 2\cos 4\theta_2 + \cos 2(\theta_1 + 2\theta_2)\\
            & - 16 \sin 2\theta_1 \sin2\theta_2 (\sin\theta_2
            \sin( 2\phi_2 - \phi_1+\frac{\pi}{6}) + \cos\theta_2
            \sin(2\phi_1 - \phi_2+\frac{\pi}{6})) \bigr] \\
            & + 8\sin^2 2\theta_1 \sin 2\theta_2 \sin(\phi_1 + \phi_2-\frac{\pi}{6}) ,
        \end{split} \\
        \begin{split}
            3P_{(1,2)} = 3P_{(2,1)} = & \sin^2\theta_1  \Bigl[ \cos\theta_2 \sin\theta_2
            \bigl( -\sin 2\theta_1 \sin\theta_2
            \sin(\phi_1 - 2\phi_2+\frac{\pi}{6}) \\
            & + \cos\theta_2 (\sin^2\theta_1 \sin\theta_2 + \sin 2\theta_1
            \sin(2\phi_1 - \phi_2-\frac{\pi}{6})) \bigr) \\
            & + \cos^2\theta_1 (1 - \sin 2\theta_2
            \sin(\phi_1 + \phi_2+\frac{\pi}{6}) ) \Bigr] ,
        \end{split}
        \\
        \begin{split}
            96P_{(2,2)} = & \sin^2\theta_1 \bigl[ 18 + 14\cos 2\theta_1 + \cos 2(\theta_1 - 2\theta_2) - 2\cos 4\theta_2 + \cos 2(\theta_1 + 2\theta_2) \\
            & - 16\sin 2\theta_1 \sin2\theta_2 (\sin\theta_2
            \sin(2\phi_2-\phi_1+\frac{\pi}{6}) + \cos\theta_2\sin(2\phi_1 - \phi_2+\frac{\pi}{6})) \bigr] \\
            & + 8\sin^2 2\theta_1 \sin 2\theta_2
            \sin(\phi_1 + \phi_2-\frac{\pi}{6}) .
        \end{split}
    \end{gather}
\end{subequations}
Hence, the corresponding $\Xi_2$ function evaluates to
\begin{equation}
    \begin{split}
        4608 \Xi_2 = 
        & 512 + 16 \bigl( 8\cos^4\theta_1 + (5 + 3\cos 4\theta_2)\sin^4\theta_1 - 2\sin^2 2\theta_1 \bigr)^2 \\
        & \quad + 4\Bigl( 4\cos(\phi_1 + \phi_2)\sin^2 2\theta_1 \sin 2\theta_2 + \sin^2\theta_1 \bigl(9 + 7\cos 2\theta_1 
        - 2\cos 4\theta_2 \sin^2\theta_1 \\
        & \quad + 8\sin 2\theta_1 (\cos\theta_2 \cos(2\phi_1 - \phi_2) + \cos(\phi_1 - 2\phi_2)\sin\theta_2)\sin 2\theta_2\bigr) \Bigr)^2 \\
        & \quad + 1024\sin^4\theta_1 \Bigl( \cos\theta_2 \sin\theta_2 \bigl( -\sin 2\theta_1 \sin\theta_2 \sin  \left(\tfrac{\pi}{6} + \phi_1 - 2\phi_2\right) \\
        & \quad + \cos\theta_2 (\sin^2\theta_1 \sin\theta_2 - \sin 2\theta_1 \sin  \left(\tfrac{\pi}{6} - 2\phi_1 + \phi_2\right)) \bigr) \\
        & \quad + \cos^2\theta_1 \bigl(1 - \sin 2\theta_2 \sin  \left(\tfrac{\pi}{6} + \phi_1 + \phi_2\right)\bigr)
        \Bigr)^2 \\
        & \quad + \Bigl( \sin^2\theta_1 \bigl(
        18 + 14\cos 2\theta_1 + \cos 2(\theta_1 - 2\theta_2)
        - 2\cos 4\theta_2 + \cos 2(\theta_1 + 2\theta_2) \\
        & \quad - 16\sin 2\theta_1 \sin 2\theta_2 
        \bigl(\cos\theta_2 \sin  \left(\tfrac{\pi}{6} + 2\phi_1 - \phi_2\right) + \sin\theta_2 \sin  \left(\tfrac{\pi}{6} - \phi_1 + 2\phi_2\right)\bigr) \bigr) \\
        & \quad - 8\sin^2 2\theta_1 \sin 2\theta_2 \sin  \left(\tfrac{\pi}{6} - (\phi_1 + \phi_2)\right)
        \Bigr)^2 ,
    \end{split}
\end{equation}
while the $2$-SRE has the form $M_2=-\log \left[3\Xi_2\right]$.

\subsection*{Examples of $\alpha$-SRE}
Calculating the SRE for the four Clifford-inequivalent eigenstates of single-qudit Clifford operators yields:
\begin{subequations}
    \begin{equation}
        M_{\alpha}(\ket{\mathbb{S}})=M_{\alpha}(\ket{\mathbb{N}})=\frac{1}{1-\alpha}\log{\frac{8+4^\alpha}{3\cdot4^\alpha}}
    \end{equation}
    \begin{equation}
        M_{\alpha}(\ket{H_+})=\frac{1}{1-\alpha}\log{\frac{8^\alpha+4(2-\sqrt{3})^\alpha+4(2+\sqrt{3})^\alpha}{3\cdot8^\alpha}}
    \end{equation}
    \begin{equation}
        M_{\alpha}(\ket{\mathbb{T}})=\frac{1}{1-\alpha}\log{\frac{6+3^\alpha}{3^{\alpha+1}}}
    \end{equation}
\end{subequations}
Note: The result for the $\ket{H_+}$ state differs from that reported in Ref.~\cite{SRE2}!
As observed, the Strange and Norell states exhibit the maximal SRE according to this measure, thereby saturating the upper bound in Eq.~(\ref{eq: Bound on SRE}).
For the particularly relevant case $\alpha = 2$, these expressions yield
\begin{equation}
    M_2 \left( \ket{\mathbb{S}} \right) =
    M_2 \left( \ket{\mathbb{N}} \right) = \log 2
    \quad , \quad
    M_2 \left( \ket{H_+} \right) = \log \frac{8}{5}
    \quad , \quad
    M_2 \left( \ket{\mathbb{T}} \right) = \log \frac{9}{5}.
\end{equation}

\newpage
\section{The bases worked with for each single-ququint Clifford}
\label{app:The bases worked with for each single-ququint Clifford}

This appendix lists the explicit eigenbases employed for analyzing various single-ququint Clifford operations, including the normalization constants and parametrizations used throughout our calculations.
Recall that
\begin{subequations}
    \begin{gather}
        \chi = \sqrt{\frac{5 + \sqrt{5}}{10}}   , \\
        \eta_\pm = \mp\sqrt{30 - 6\sqrt{5}} + \sqrt{5} - 3   , \\
        \kappa_\pm = \frac{1}{2} \left( \pm \sqrt{6(5 + \sqrt{5})} - \sqrt{5} - 3 \right)   .
    \end{gather}
\end{subequations}
The vectors employed in this work are the following, each normalized by an appropriate positive real constant.
\begin{enumerate}
    \item The unnormalized orthonormal eigenstates of $H$ we work with are:
\begin{subequations}
    \begin{gather}
        \ket{H,\pm i} \propto \sqrt{1\pm\chi}(\ket{1}-\ket{4}) 
        \pm \sqrt{1\mp \chi}(\ket{2} -\ket{3})   , \\
        \ket{H,-1} \propto (1-\sqrt{5})\ket{0} 
        + \ket{1} 
        +\ket{2} 
        + \ket{3} 
        + \ket{4}   , \\
        \ket{H,1;1} \propto (1+\sqrt{5})\ket{0} + 2\big( \ket{1} + \ket{4} \big), \\
        \ket{H,1;2} \propto (1+\sqrt{5})\ket{0} + 2\big( \ket{2} + \ket{3} \big).
    \end{gather}
\end{subequations}

\item The unnormalized orthonormal eigenstates of $XV_{\hat{S}}$ we work with are:
\begin{subequations}
    \begin{gather}
\ket{X V_{\hat{S}},1} \propto \ket{0} + \ket{1} + \omega^3\ket{2} + \ket{3} + \omega^2\ket{4}   , \\
\ket{X V_{\hat{S}},\omega} \propto \omega^{4} \ket{0} + \omega^{3} \ket{1} + \ket{2} + \omega \ket{3} + \omega^{2} \ket{4}   , \\
\ket{X V_{\hat{S}},\omega^{-1}} \propto \omega \ket{0} + \omega^{2} \ket{1} + \omega \ket{2} + \omega^{4} \ket{3} + \omega^{2} \ket{4}   , \\
\ket{X V_{\hat{S}},\omega^2} \propto \omega^{3} \ket{0} + \omega \ket{1} + \omega^{2} \ket{2} + \omega^{2} \ket{3} + \omega^{2} \ket{4}   , \\
\ket{X V_{\hat{S}},\omega^{-2}} \propto \omega^{2} \ket{0} + \omega^{4} \ket{1} + \omega^{4} \ket{2} + \omega^{3} \ket{3} + \omega^{2} \ket{4}   .
    \end{gather}
\end{subequations}

\item The unnormalized orthonormal eigenstates of $B^\prime=V_{\hat{K}} B V_{\hat{K}}^{-1}$ (where $B=H^3V_{\hat{S}}$ and $\hat{K}=\begin{pmatrix} 1 & 2 \\ 2 & 0 \end{pmatrix}$) we work with are:
\begin{subequations}
    \begin{gather}
        \ket{B^\prime,-1} \propto (3+\sqrt{5})\ket{0} 
        - 2(\ket{1} + \ket{2} + \ket{3} + \ket{4})   , \\
        \ket{B^\prime,-e^{\pm \frac{2\pi i}{3}}} \propto 4 \kappa_\pm\ket{0} 
        - \kappa_\pm^2(\ket{1}+\ket{4}) +4(\ket{2} +\ket{3})   , \\
        \ket{B^\prime,e^{\pm\frac{2\pi i}{3}}} \propto \eta_\pm(\ket{1}-\ket{4} )
        + 4(\ket{2} - \ket{3} )   .
    \end{gather}
\end{subequations}

\item The unnormalized orthonormal eigenstates of $A=V_{\hat{S}}H^2$ we work with are:
\begin{subequations}
    \begin{gather}
    \ket{A,e^{\frac{4\pi i}{5}}} \propto \ket{2} + \ket{3} \\
    \ket{A,e^{-\frac{4\pi i}{5}}} \propto \ket{1}+\ket{4}\\
    \ket{A,e^{\frac{\pi i}{5}}} \propto \ket{4}-\ket{1} \\
    \ket{A,e^{-\frac{\pi i}{5}}} \propto \ket{2}-\ket{3} \\
    \ket{A,1} = \ket{0}   .
    \end{gather}
\end{subequations}
\end{enumerate}

Having these eigenbases explicitly available facilitates reproducibility and provides clarity in extending or verifying the results concerning ququint systems.

\newpage
\section{Wigner function for the single-ququint states}

We now provide the Wigner function matrices, stabilizer fidelities, and nearest stabilizer states for the relevant nonstabilizer Clifford-inequivalent non-degenerate eigenstates of single-ququint Clifford operations.

\begin{table}[h!]
    \centering
    \resizebox{\linewidth}{!}{
    \begin{tabular}{|c|c|c|}
        \hline
        \textbf{Unnormalized State} & 
        \textbf{\makecell{Wigner Function Matrix}} &
        \textbf{\makecell{Wigner Trace\\Norm}} \\ \hline
        $\displaystyle \ket{H,i} 	\propto \sqrt{1+\chi}(\ket{1}-\ket{4}) 
        + \sqrt{1-\chi}(\ket{2} -\ket{3})$ 
        & $\begin{pmatrix}
-0.2 & 0.1451 & -0.0451 & -0.0451 & 0.1451 \\
0.1451 & 0.1088 & 0.05 & 0.05 & 0.1088 \\
-0.0451 & 0.05 & -0.0088 & -0.0088 & 0.05 \\
-0.0451 & 0.05 & -0.0088 & -0.0088 & 0.05 \\
0.1451 & 0.1088 & 0.05 & 0.05 & 0.1088
\end{pmatrix}$ & \makecell{$\frac{3+2\sqrt{5+2\sqrt{5}}}{5}$\\$\approx1.8311$} \\ \hline
        $\displaystyle \ket{H,-1} 	\propto (1-\sqrt{5})\ket{0} 
        + \ket{1} 
        +\ket{2} 
        + \ket{3} 
        + \ket{4}$ 
        & $\begin{pmatrix}
0.2 & 0.0191 & 0.0191 & 0.0191 & 0.0191 \\
0.0191 & 0.1309 & -0.05 & -0.05 & 0.1309 \\
0.0191 & -0.05 & 0.1309 & 0.1309 & -0.05 \\
0.0191 & -0.05 & 0.1309 & 0.1309 & -0.05 \\
0.0191 & 0.1309 & -0.05 & -0.05 & 0.1309
\end{pmatrix}$ & $\frac{9}{5}=1.8$ \\ \hline
        $\displaystyle \ket{X_{V\hat{S},1}} 	\propto \ket{0} + \ket{1} + \omega^3\ket{2} + \ket{3} + \omega^2\ket{4}$ 
        & $\begin{pmatrix}
-0.0894 & 0.0894 & 0.1447 & 0& 0.0553 \\
-0.0894 & 0.0894 & 0.1447 & 0& 0.0553 \\
0.0553 & -0.0894 & 0.0894 & 0.1447 & 0\\
0.1447 & 0& 0.0553 & -0.0894 & 0.0894 \\
0.0553 & -0.0894 & 0.0894 & 0.1447 & 0.
\end{pmatrix}$ & \makecell{$1+\frac{2}{\sqrt{5}}$\\$\approx1.8944$} \\ \hline
        $\displaystyle \ket{B^\prime,-1} 	\propto (3+\sqrt{5})\ket{0} 
        - 2(\ket{1} + \ket{2} + \ket{3} + \ket{4})$ 
        & $\begin{pmatrix}
0.2 & 0.1079 & 0.1079 & 0.1079 & 0.1079 \\
-0.0412 & 0.1079 & -0.0412 & -0.0412 & 0.1079 \\
-0.0412 & -0.0412 & 0.1079 & 0.1079 & -0.0412 \\
-0.0412 & -0.0412 & 0.1079 & 0.1079 & -0.0412 \\
-0.0412 & 0.1079 & -0.0412 & -0.0412 & 0.1079
\end{pmatrix}$ & \makecell{$1+\frac{4}{\sqrt{5}}$\\$\approx1.9889$} \\ \hline
        $\displaystyle \ket{B^\prime,-\omega_3} 	\propto 4 \kappa_+\ket{0} 
        - \kappa_+^2(\ket{1}+\ket{4}) +4\ket{2} +4\ket{3}$ 
        & $\begin{pmatrix}
0.2 & 0.0849 & -0.0928 & -0.0928 & 0.0849 \\
0.0916 & -0.0928 & 0.0496 & 0.0496 & -0.0928 \\
0.0496 & 0.0916 & 0.0849 & 0.0849 & 0.0916 \\
0.0496 & 0.0916 & 0.0849 & 0.0849 & 0.0916 \\
0.0916 & -0.0928 & 0.0496 & 0.0496 & -0.0928
\end{pmatrix}$ & \makecell{$\frac{3+\sqrt{5}+\sqrt{15+6\sqrt{5}}}{5}$\\$\approx2.1134$} \\ \hline
        $\displaystyle \ket{B^\prime_{\omega_3}} 	\propto \eta_+(\ket{1}-\ket{4})
        + 4(\ket{2} - \ket{3})$ 
        & $\begin{pmatrix}
-0.2 & 0.071 & 0.029 & 0.029 & 0.071 \\
-0.0388 & 0.029 & 0.1388 & 0.1388 & 0.029 \\
0.1388 & -0.0388 & 0.071 & 0.071 & -0.0388 \\
0.1388 & -0.0388 & 0.071 & 0.071 & -0.0388 \\
-0.0388 & 0.029 & 0.1388 & 0.1388 & 0.029
\end{pmatrix}$ & \makecell{$\frac{4+\sqrt{15+6\sqrt{5}}}{5}$\\$\approx1.8661$} \\ \hline
        $\displaystyle \ket{A_{-\omega_5^2}} 	\propto \ket{2} - \ket{3}$ 
        & $\begin{pmatrix}
-0.2 & -0.0618 & 0.1618 & 0.1618 & -0.0618 \\
0& 0& 0& 0& 0\\
0.1 & 0.1 & 0.1 & 0.1 & 0.1 \\
0.1 & 0.1 & 0.1 & 0.1 & 0.1 \\
0& 0& 0& 0& 0.
\end{pmatrix}$ & \makecell{$\frac{6}{5}+\frac{1}{\sqrt{5}}$\\$\approx1.6472$} \\ \hline
        $\displaystyle \ket{A_{\omega_5^2}} 	\propto \ket{2}+\ket{3}$ 
        & $\begin{pmatrix}
0.2 & 0.0618 & -0.1618 & -0.1618 & 0.0618 \\
0& 0& 0& 0& 0\\
0.1 & 0.1 & 0.1 & 0.1 & 0.1 \\
0.1 & 0.1 & 0.1 & 0.1 & 0.1 \\
0& 0& 0& 0& 0.
\end{pmatrix}$ & \makecell{$\frac{6}{5}+\frac{1}{\sqrt{5}}$\\$\approx1.6472$} \\ \hline
    \end{tabular}
    }
    \caption{
    Unnormalized Single-ququint nonstabilizer Clifford-inequivalent Clifford nondegenerate eigenstates, their corresponding Wigner function matrices, and Wigner trace norms.}
    \label{tab:Single-ququint nonstabilizer Clifford-inequivalent Clifford nondegenerate eigenstates and Wigner function}
\end{table}

\begin{table}[h!]
    \centering
    \resizebox{\linewidth}{!}{
    \begin{tabular}{|c|c|c|c|}
        \hline
        \textbf{Unnormalized State} & 
        \textbf{\makecell{Stabilizer\\Fidelity}} & 
        \textbf{\makecell{Number of Nearest\\Stabilizer States}} &
        \textbf{\makecell{Nearest\\Stabilizer States}} \\ \hline
        $\displaystyle \ket{H,i} 	\propto \sqrt{1+\chi}(\ket{1}-\ket{4}) 
        + \sqrt{1-\chi}(\ket{2} -\ket{3})$ 
        & $0.4627$ & $4$ &
        \makecell{
        $\ket{1}$ ,
        $\ket{4}$ , \\
        $\frac{e^{\pm 2\pi i/5}\ket{0}+e^{\pm 4\pi i/5}\ket{1}+e^{\mp 4\pi i/5}\ket{2}+e^{\mp 2\pi i/5}\ket{3}+\ket{4}}{\sqrt{5}}$
        }
        \\ \hline
        
        $\displaystyle \ket{H,-1} 	\propto (1-\sqrt{5})\ket{0} 
        + \ket{1} 
        +\ket{2} 
        + \ket{3} 
        + \ket{4}$ 
        & $0.7236$ & $2$ &
        \makecell{
        $\frac{e^{\pm 4\pi i/5}\ket{0}+\ket{1}+e^{\mp 2\pi i/5}\ket{2}+e^{\mp 2\pi i/5}\ket{3}+\ket{4}}{\sqrt{5}}$
        }
        \\ \hline
        
        $\displaystyle \ket{X_{V\hat{S},1}} 	\propto \ket{0} + \ket{1} + \omega^3\ket{2} + \ket{3} + \omega^2\ket{4}$ 
        & $0.5236$ & $5$ & 
        \makecell{
        $\frac{e^{4\pi i/5}\ket{0}+e^{-2\pi i/5}\ket{1}+e^{2\pi i/5}\ket{2}+e^{-4\pi i/5}\ket{3}+\ket{4}}{\sqrt{5}}$ , \\
        $\frac{e^{-4 \pi i/5}\ket{0}+e^{4\pi i/5}\ket{1}+e^{4\pi i/5}\ket{2}+e^{-4\pi i/5}\ket{3}+\ket{4}}{\sqrt{5}}$ , \\
        $\frac{e^{\pm 4\pi i/5}\ket{0}+e^{\pm 4\pi i/5}\ket{1}+\ket{2}+e^{\pm2\pi i/5}\ket{3}+\ket{4}}{\sqrt{5}}$ , \\
        $\frac{\ket{0}+e^{-2\pi i/5}\ket{1}+e^{4\pi i/5}\ket{2}+e^{-2\pi i/5}\ket{3}+\ket{4}}{\sqrt{5}}$
        }
         \\ \hline
        
        $\displaystyle \ket{B^\prime,-1} 	\propto (3+\sqrt{5})\ket{0} 
        - 2(\ket{1} + \ket{2} + \ket{3} + \ket{4})$ 
        & $0.6315$ & $3$ & 
        \makecell{
        $\ket{0}$ , $\frac{e^{\pm 4\pi i/5}\ket{0}+\ket{1}+e^{\mp 2\pi i/5}\ket{2}+e^{\mp 2\pi i/5}\ket{3}+\ket{4}}{\sqrt{5}}$
        }
        \\ \hline
        
        $\displaystyle \ket{B^\prime,-\omega_3} 	\propto 4 \kappa_+\ket{0} 
        - \kappa_+^2(\ket{1}+\ket{4}) +4\ket{2} +4\ket{3}$ 
        & $0.4824$ & $3$ & 
        \makecell{
        $\frac{\ket{0}+\ket{1}+\ket{2}+\ket{3}+\ket{4}}{\sqrt{5}}$ , \\
        $\frac{e^{\pm 2\pi i/5}\ket{0}+\ket{1}+e^{\pm 4\pi i/5}\ket{2}+e^{\pm 4\pi i/5}\ket{3}+\ket{4}}{\sqrt{5}}$
        }
        \\ \hline
        
        $\displaystyle \ket{B^\prime_{\omega_3}} 	\propto \eta_+(\ket{1}-\ket{4})
        + 4(\ket{2} - \ket{3})$ 
        & $0.4487$ & $6$ & 
        \makecell{
        $\frac{e^{\pm 4\pi i/5}\ket{0}+e^{\mp 2\pi i/5}\ket{1}+e^{\pm 2\pi i/5}\ket{2}+e^{\mp 4\pi i/5}\ket{3}+\ket{4}}{\sqrt{5}}$ , \\
        $\frac{\ket{0}+e^{\pm 4\pi i/5}\ket{1}+e^{\pm 2\pi i/5}\ket{2}+e^{\pm 4\pi i/5}\ket{3}+\ket{4}}{\sqrt{5}}$ , \\
        $\frac{e^{\pm 4\pi i/5}\ket{0}+e^{\pm 4\pi i/5}\ket{1}+\ket{2}+e^{\pm 2\pi i/5}\ket{3}+\ket{4}}{\sqrt{5}}$
        }
        \\ \hline
        
        $\displaystyle \ket{A_{-\omega_5^2}} 	\propto \ket{2} - \ket{3}$ 
        & $0.5$ & $2$ & 
        \makecell{
        $\ket{2}$ , $\ket{3}$
        }
        \\ \hline
        
        $\displaystyle \ket{A_{\omega_5^2}} 	\propto \ket{2}+\ket{3}$ 
        & $0.5$ & $2$ & 
        \makecell{
        $\ket{2}$ , $\ket{3}$
        }
        \\ \hline
    \end{tabular}
    }
    \caption{
    Unnormalized Single-ququint nonstabilizer Clifford-inequivalent Clifford nondegenerate eigenstates, their stabilizer fidelities, number of nearest stabilizer states, and nearest stabilizer states.}
    \label{tab:Single-ququint nonstabilizer Clifford-inequivalent Clifford nondegenerate eigenstates and nearest SS}
\end{table}

\begin{figure}[h!]
    \centering
    \begin{subfigure}[b]{0.3\linewidth}
        \centering
        \includegraphics[width=\linewidth]{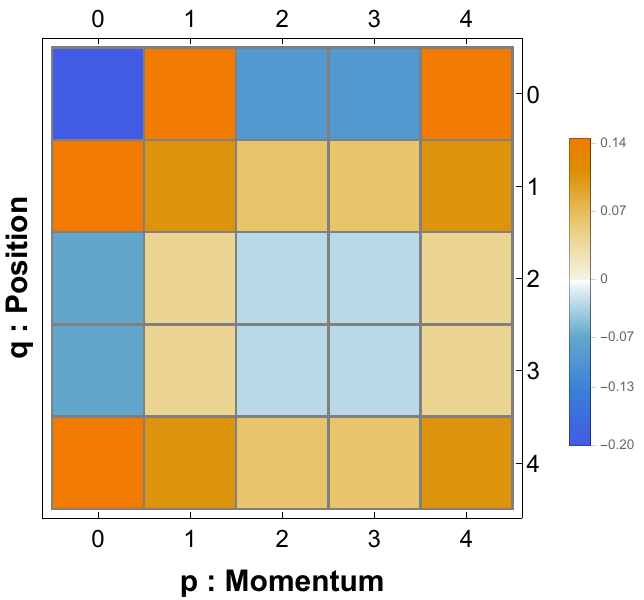}
        \caption{$\ket{H,i}$}
        \label{fig:H_i}
    \end{subfigure}
    \hfill
    \begin{subfigure}[b]{0.3\linewidth}
        \centering
        \includegraphics[width=\linewidth]{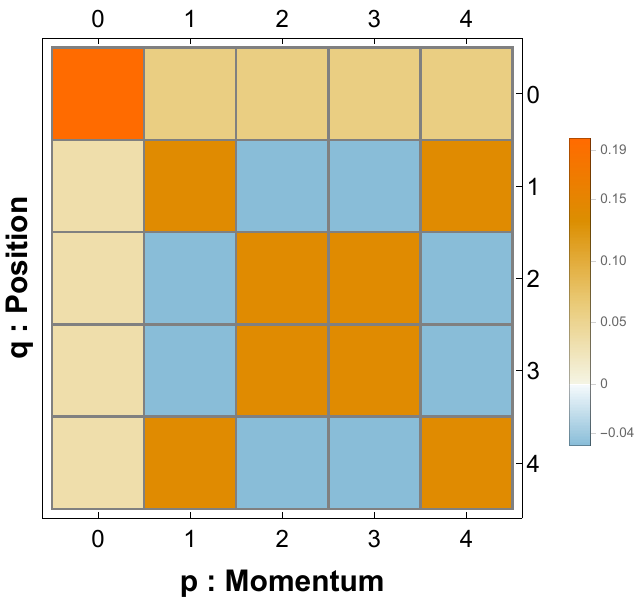}
        \caption{$\ket{H,-1}$}
        \label{fig:H_minus1}
    \end{subfigure}
    \hfill
    \begin{subfigure}[b]{0.3\linewidth}
        \centering
        \includegraphics[width=\linewidth]{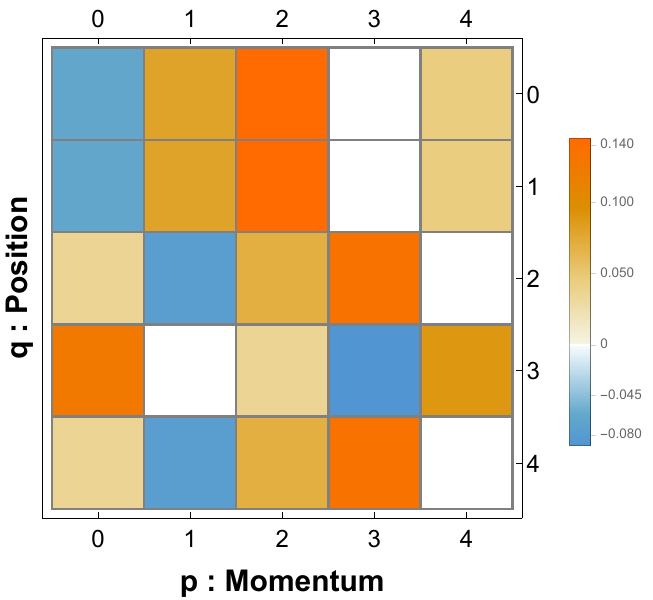}
        \caption{$\ket{X_{V\hat{S},1}}$}
        \label{fig:X_V_hatS_1}
    \end{subfigure}
    
    \vspace{0.5cm}
    
    \begin{subfigure}[b]{0.3\linewidth}
        \centering
        \includegraphics[width=\linewidth]{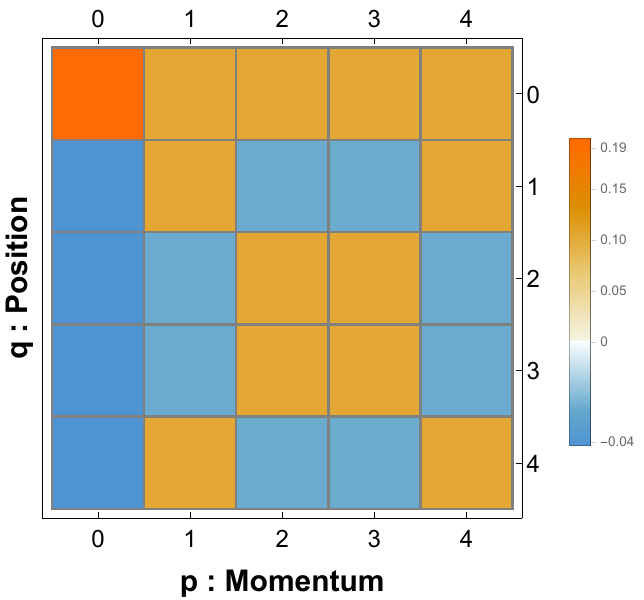}
        \caption{$\ket{B^\prime,-1}$}
        \label{fig:Bprime_minus1}
    \end{subfigure}
    \hfill
    \begin{subfigure}[b]{0.3\linewidth}
        \centering
        \includegraphics[width=\linewidth]{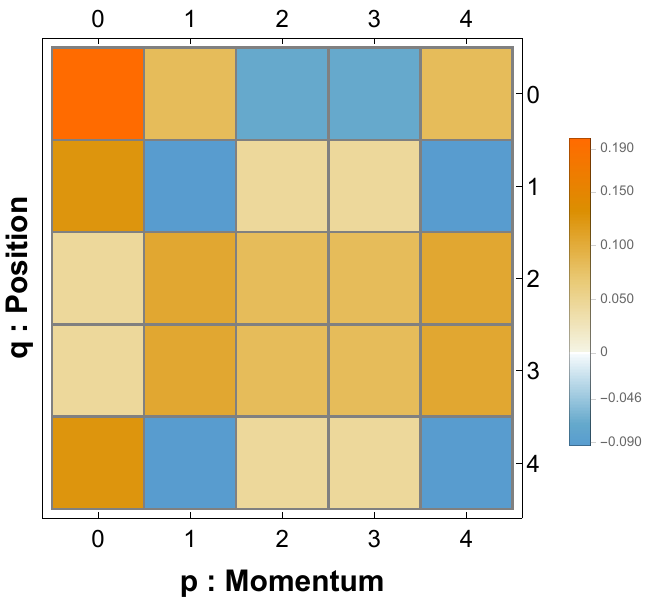}
        \caption{$\ket{B^\prime,-\omega_3}$}
        \label{fig:Bprime_minus_omega3}
    \end{subfigure}
    \hfill
    \begin{subfigure}[b]{0.3\linewidth}
        \centering
        \includegraphics[width=\linewidth]{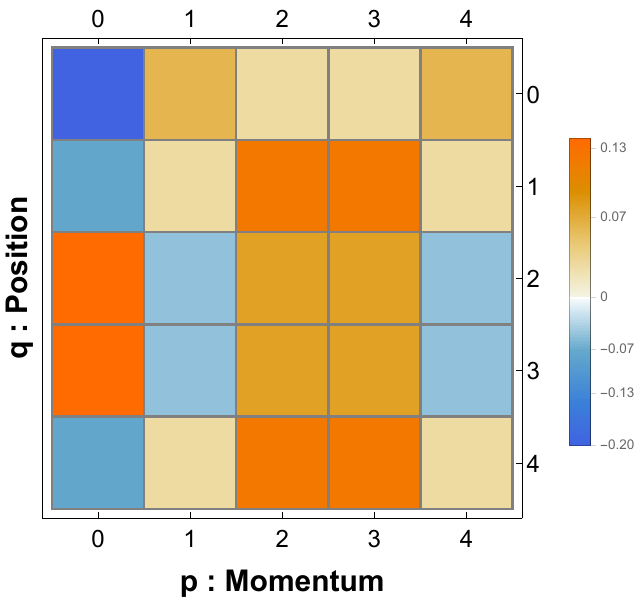}
        \caption{$\ket{B^\prime_{\omega_3}}$}
        \label{fig:Bprime_omega3}
    \end{subfigure}
    
    \vspace{0.5cm}
    
    \begin{subfigure}[b]{0.45\linewidth}
        \centering
        \includegraphics[width=\linewidth]{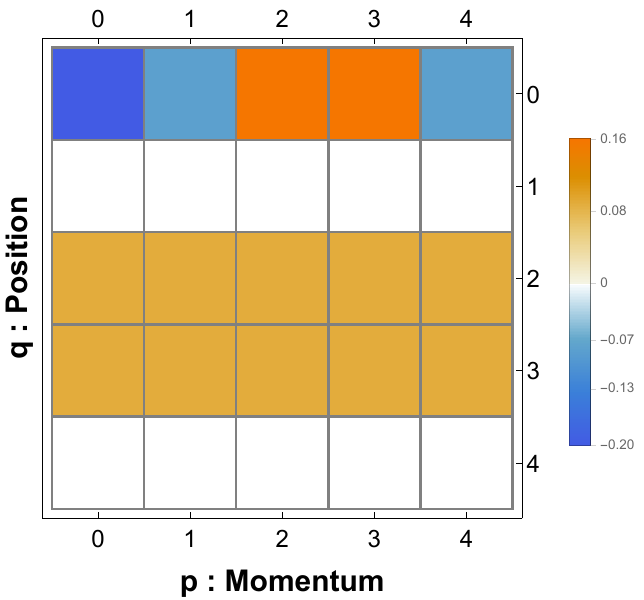}
        \caption{$\ket{A_{-\omega_5^2}}$}
        \label{fig:A_minus_omega5sq2}
    \end{subfigure}
    \hfill
    \begin{subfigure}[b]{0.45\linewidth}
        \centering
        \includegraphics[width=\linewidth]{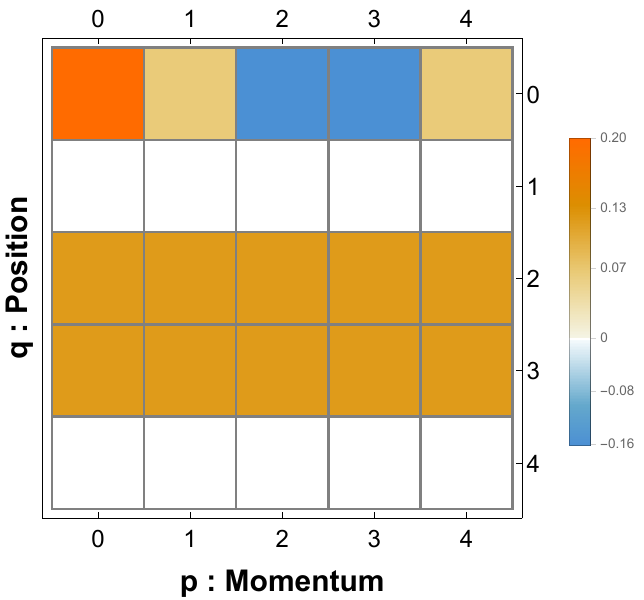}
        \caption{$\ket{A_{\omega_5^2}}$}
        \label{fig:A_omega5sq2}
    \end{subfigure}
    
    \caption{Wigner function heatmap visualizations for all nonstabilizer Clifford-inequivalent non-degenerate eigenstates of Clifford operations of single ququints.}
    \label{fig:ququint Wigner functions visualizations}
\end{figure}

These results highlight the structure and properties of ququint eigenstates, supporting the theoretical discussions and examples of the main text.

\section{$\mL$ matrices and Stabilizer Fidelity for the single ququint states}
\label{app:L matrices and Stabilizer Fidelity for Single Ququint Clifford-stabilizer States}
After numbering the nearest stabilizer states as in the order presented in Table~\ref{tab:Single-ququint nonstabilizer Clifford-inequivalent Clifford nondegenerate eigenstates and Wigner function} and Table~\ref{tab:Single-ququint nonstabilizer Clifford-inequivalent Clifford nondegenerate eigenstates and nearest SS}, we computed the corresponding $\mL$ matrices. The results are shown below:
\newpage
\small
\begin{subequations}
\label{eq:L matrices for ququints}
    \begin{gather}
        \mL \left( \ket{H,+i}, \ket{H,-i}, \ket{H,-1} , \ket{H,1;1} , \ket{H,1;2} \right) \approx 
        \left(
\begin{array}{cccc}
0.1314 & 0.2893 & 0.3165 & 0\\
0.1314 & -0.2893 & -0.3165 & 0\\
-0.1314 & 0.2893 i & -0.3165 i & 0\\
-0.1314 & -0.2893 i & 0.3165 i & 0\\
\end{array}
\right)
  ,  
\\
\mL  \left( \ket{H,-1}, \ket{H,+i}, \ket{H,-i} , \ket{H,1;1} , \ket{H,1;2} \right) 
        \approx 
        \left(
\begin{array}{cccc}
0& 0& 0.2081 i & -0.2081 i \\
0& 0& -0.2081 i & 0.2081 i \\
\end{array}
\right)
  ,
 \\
 \begin{split}
             \mL & \left( \ket{XV_{\hat{S}},1},  \ket{XV_{\hat{S}},e^{\frac{2\pi i}{5}}},  \ket{XV_{\hat{S}},e^{-\frac{2\pi i}{5}}} , \ket{XV_{\hat{S}},e^{\frac{4\pi i}{5}}} , \ket{XV_{\hat{S}},e^{\frac{-4\pi i}{5}}} \right) \\ &  \approx 
        \left(
\begin{array}{cccc}
0.1  -0.3078 i & 0& 0.2618  +0.1902 i & -0.1618+0.1176 i \\
0.3236 & 0& -0.3236 & 0.2 \\
-0.2618+0.1902 i & 0& -0.1+0.3078 i & 0.0618  +0.1902 i \\
0.1  +0.3078 i & 0& 0.2618  -0.1902 i & -0.1618-0.1176 i \\
-0.2618-0.1902 i & 0& -0.1-0.3078 i & 0.0618  -0.1902 i \\
\end{array}
\right)
  ,
 \end{split}
 \\
 \begin{split}
             \mL & \left( \ket{B^\prime,-1}, \ket{B^\prime,  -e^{\frac{2\pi i}{3}}}, \ket{B^\prime,-e^{-\frac{2\pi i}{3}}} , \ket{B^\prime, e^{\frac{2\pi i}{3}}} , \ket{B^\prime,e^{-\frac{2\pi i}{3}}} \right)
             \\
             & \approx 
        \left( \begin{array}{cccc}
             0.3411 & -0.3411 & 0& 0\\
             -0.1706+0.2954 i & 0.1706  +0.2954 i & 0& 0\\
             -0.1706-0.2954 i & 0.1706  -0.2954 i & 0& 0\\
            \end{array} \right)
          ,
 \end{split}
 \\
 \begin{split}
             \mL & \left( \ket{B^\prime,-e^{\frac{2\pi i}{3}}}, \ket{B^\prime,-e^{-\frac{2\pi i}{3}}} ,  \ket{B^\prime, e^{\frac{2\pi i}{3}}} , \ket{B^\prime,e^{-\frac{2\pi i}{3}}} , \ket{B^\prime,-1} \right) \\  & \approx \left(
\begin{array}{cccc}
 -0.4824 & 0& 0& -0.1303 \\
 0.2412  +0.4178 i & 0& 0& 0.0651  -0.1128 i \\
 0.2412  -0.4178 i & 0& 0& 0.0651  +0.1128 i \\
\end{array}
\right)   ,
 \end{split}
 \\
 \begin{split}
             \mL & \left( \ket{B^\prime, e^{\frac{2\pi i}{3}}} ,  \ket{B^\prime,e^{-\frac{2\pi i}{3}}} , \ket{B^\prime,-1} , \ket{B^\prime,-e^{\frac{2\pi i}{3}}}, \ket{B^\prime,-e^{-\frac{2\pi i}{3}}} \right) \\ & \approx \left(
\begin{array}{cccc}
 0.1518 & 0.329 i & 0.2812 i & 0.1924 i \\
 0.1518 & -0.329 i & -0.2812 i & -0.1924 i \\
 -0.0759+0.1314 i & 0.2849  -0.1645 i & 0.2812 i & -0.1666-0.0962 i \\
 -0.0759-0.1314 i & 0.2849  +0.1645 i & -0.2812 i & -0.1666+0.0962 i \\
 -0.0759-0.1314 i & -0.2849-0.1645 i & 0.2812 i & 0.1666  -0.0962 i \\
 -0.0759+0.1314 i & -0.2849+0.1645 i & -0.2812 i & 0.1666 +0.0962 i \\
\end{array}
\right)    ,
 \end{split}
 \\
          \mL \left( \ket{A,e^{-\frac{\pi i}{5}}}, \ket{A,e^{\frac{\pi i}{5}}}, \ket{A,e^{\frac{4\pi i}{5}}} , \ket{A,e^{-\frac{4\pi i}{5}}} , \ket{A,1} \right)   \approx \left(
\begin{array}{cccc}
 0& 0.5 & 0& 0\\
 0& -0.5 & 0& 0\\
\end{array}
\right)     ,
 \\
             \mL \left( \ket{A,e^{\frac{4\pi i}{5}}},  \ket{A,e^{-\frac{4\pi i}{5}}}, \ket{A,e^{\frac{\pi i}{5}}} , \ket{A,e^{-\frac{\pi i}{5}}} , \ket{A,1} \right) \approx \left(
\begin{array}{cccc}
 0& 0& 0.5 & 0\\
 0& 0& -0.5 & 0\\
\end{array}
\right)   .
    \end{gather}
\end{subequations}
\normalsize{}

Taking a general variation $\ket{\varphi} = \sum_{j=2}^5 (a_j + i b_j) \ket{\psi_j}$ with real coefficients $a_j, b_j$ satisfying $\sum_{j=2}^5 (a_j^2 + b_j^2) = 1$ and $\ket{\psi_j}$ running over the four eigenstates of the relevant operator, orthogonal to $\ket{\psi_1}$, we compute the first order in $\epsilon$ of the variation in the squared overlap amplitude with the nearest stabilizer states for each of the investigated Clifford-equivalent, non-degenerate Clifford eigenstates, as follows:
\small
\begin{subequations}
    \begin{gather}
\text{for } \ket{H,+i} \text{ : }
\begin{pmatrix}
 0.262866 a_2+0.578607 a_3+0.633044 a_4 \\
 0.262866 a_2-0.578607 a_3-0.633044 a_4 \\
 -0.262866 a_2-0.578607 b_3+0.633044 b_4 \\
 -0.262866 a_2+0.578607 b_3-0.633044 b_4 \\
\end{pmatrix}
\\
\text{for } \ket{H,-1} \text{ : }
\left(
\begin{array}{c}
 0.416214 b_5-0.416214 b_4 \\
 0.416214 b_4-0.416214 b_5 \\
\end{array}
\right)
\\
\text{for } \ket{XV_{\hat{S}},1} \text{ : }
\left(
\begin{array}{c}
 0.2 a_2+0.523607 a_4-0.323607 a_5+0.615537 b_2-0.380423 b_4-0.235114 b_5 \\
 0.647214 a_2-0.647214 a_4+0.4 a_5 \\
 -0.523607 a_2-0.2 a_4+0.123607 a_5-0.380423 b_2-0.615537 b_4-0.380423 b_5 \\
 0.2 a_2+0.523607 a_4-0.323607 a_5-0.615537 b_2+0.380423 b_4+0.235114 b_5 \\
 -0.523607 a_2-0.2 a_4+0.123607 a_5+0.380423 b_2+0.615537 b_4+0.380423 b_5 \\
\end{array}
\right)
\\
\text{for } \ket{B^\prime,-1} \text{ : }
\left(
\begin{array}{c}
 0.682223 a_2-0.682223 a_3 \\
 -0.341112 a_2+0.341112 a_3-0.590822 b_2-0.590822 b_3 \\
 -0.341112 a_2+0.341112 a_3+0.590822 b_2+0.590822 b_3 \\
\end{array}
\right)
\\
\text{for } \ket{B^\prime,-e^\frac{2\pi i}{3}} \text{ : }
\left(
\begin{array}{c}
 -0.964809 a_2-0.260586 a_5 \\
 0.482405 a_2+0.130293 a_5-0.835549 b_2+0.225674 b_5 \\
 0.482405 a_2+0.130293 a_5+0.835549 b_2-0.225674 b_5 \\
\end{array}
\right)
\\
\begin{split}
& \text{for } \ket{B^\prime,e^\frac{2\pi i}{3}} \text{ : }
\\ & \left(
\begin{array}{c}
 0.303531 a_2-0.657932 b_3-0.562487 b_4-0.384786 b_5 \\
 0.303531 a_2+0.657932 b_3+0.562487 b_4+0.384786 b_5 \\
 -0.151765 a_2+0.569786 a_3-0.333235 a_5-0.262866 b_2+0.328966 b_3-0.562487
   b_4+0.192393 b_5 \\
 -0.151765 a_2+0.569786 a_3-0.333235 a_5+0.262866 b_2-0.328966 b_3+0.562487
   b_4-0.192393 b_5 \\
 -0.151765 a_2-0.569786 a_3+0.333235 a_5+0.262866 b_2+0.328966 b_3-0.562487
   b_4+0.192393 b_5 \\
 -0.151765 a_2-0.569786 a_3+0.333235 a_5-0.262866 b_2-0.328966 b_3+0.562487
   b_4-0.192393 b_5 \\
\end{array}
\right)
\end{split}
\\
\text{for } \ket{A,e^{-\frac{\pi i}{5}}} \text{ : }
\left(
\begin{array}{c}
 a_3 \\
 - a_3 \\
\end{array}
\right)
\\
\text{for } \ket{A,e^\frac{4\pi i}{5}} \text{ : }
\left(
\begin{array}{c}
  a_4 \\
 - a_4 \\
\end{array}
\right)
    \end{gather}
\end{subequations}
\normalsize{}
It is evident that none of these states can be a sharp minimum in all directions. A more detailed investigation of their behavior is left to the reader.

\newpage
\section{Examples of $\alpha$-SRE for single ququints}
\label{app:SREforSingleQuquints}

Here we present the computed values of the $\alpha$-SRE for all single-ququint, non-stabilizer, Clifford-inequivalent, nondegenerate Clifford eigenstates.

\begin{subequations}
\begin{gather}
\begin{split}
M_\alpha(\ket{H,\pm i}) 
&= \frac{1}{1-\alpha}\log  \left[\frac{1}{5} 2^{-5\alpha}  \left(
2^{3+\alpha} + 32^{\alpha} \right.\right. \\
&\quad \left.\left.
+ 4(7 - \sqrt{5} - 2\sqrt{10 - 2\sqrt{5}})^{\alpha} 
+ 4(7 - \sqrt{5} + 2\sqrt{10 - 2\sqrt{5}})^{\alpha}
\right.\right.\\
&\quad \left.\left.
+ 4(7 + \sqrt{5} - 2\sqrt{2(5 + \sqrt{5})})^{\alpha}
+ 4(7 + \sqrt{5} + 2\sqrt{2(5 + \sqrt{5})})^{\alpha}
\right)\right],
\end{split} \\
M_\alpha(\ket{H,-1}) = \frac{1}{1-\alpha}\log{\left[\frac{1}{5} 2^{-5\alpha}  \left(2^{3+\alpha} + 32^\alpha 
+ 8(7 - 3\sqrt{5})^\alpha 
+ 8(7 + 3\sqrt{5})^\alpha\right)\right]}, \\
M_\alpha(\ket{X V_{\hat{S}},1}) = \frac{1}{1-\alpha}\log{\left[\frac{1}{5} + 4\cdot 5^{-\alpha}\right]}, \\
M_\alpha(\ket{B^\prime,-1}) = \frac{1}{1-\alpha}\log{\left[\frac{1}{5} 18^{-\alpha}  \left(18^\alpha + 12(3 - \sqrt{5})^\alpha + 12(3 + \sqrt{5})^\alpha\right)\right]}, \\
\begin{split}
M_\alpha(\ket{B',-e^{\pm \frac{2\pi i}{3}}}) 
&= \frac{1}{1-\alpha}\log  \left[\frac{1}{5} 144^{-\alpha}  \left(
3\cdot 2^{1+\alpha}(12 + \sqrt{5} - \sqrt{15 - 6\sqrt{5}})^{\alpha}
\right.\right.\\
&\quad \left.\left.
+ 2(2 + \sqrt{5} + \sqrt{15 - 6\sqrt{5}})^{2\alpha} \right.\right.\\
&\quad \left.\left.
+ 2^{\alpha}  \Big(
72^{\alpha}
+ 6(12 + \sqrt{5} + \sqrt{15 - 6\sqrt{5}})^{\alpha}
+ 6(12 - \sqrt{5} - \sqrt{15 + 6\sqrt{5}})^{\alpha}  \right.\right. \\
& \quad \left.\left. + 4(12 - \sqrt{5} + \sqrt{15 + 6\sqrt{5}})^{\alpha}
\Big)
\right)\right],
\end{split} \\
\begin{split}
M_\alpha(\ket{B',e^{\pm\frac{2\pi i}{3}}}) 
&= \frac{1}{1-\alpha}\log  \left[\frac{1}{5} 24^{-\alpha}  \left(
24^{\alpha}
\right.\right.\\
&\quad \left.\left.
+ 6(4 - \sqrt{5} - \sqrt{15 - 6\sqrt{5}})^{\alpha}
+ 6(4 - \sqrt{5} + \sqrt{15 - 6\sqrt{5}})^{\alpha}
\right.\right.\\
&\quad \left.\left.
+ 6(4 + \sqrt{5} - \sqrt{15 + 6\sqrt{5}})^{\alpha}
+ 6(4 + \sqrt{5} + \sqrt{15 + 6\sqrt{5}})^{\alpha}
\right)\right],
\end{split} \\
M_\alpha(\ket{A_{\pm\omega_5^2}}) = \frac{1}{1-\alpha}\log{\left[\frac{1}{5} 8^{-\alpha}  \left(2^\alpha(10 + 4^\alpha) 
+ 2(3 - \sqrt{5})^\alpha + 2(3 + \sqrt{5})^\alpha\right)\right]}. 
    \end{gather}
\end{subequations}
As can be seen, none of them achieve the upper bound in Eq.~(\ref{eq: Bound on SRE}). This discrepancy may be related to the fact that none of these states attains the minimum value of the stabilizer fidelity too.
For the particularly relevant case $\alpha = 2$, these expressions yield
\begin{equation}
    \begin{split}
        & M_2 \left( \ket{H_\pm} \right) =
        M_2 \left( \ket{H_{-1}} \right) =
        M_2 \left( \ket{B',e^{\pm\frac{2\pi i}{3}}} \right) =
        M_2 \left( \ket{A_{\pm\omega_5^2}} \right) = \log 2,\\
        & M_2 \left( \ket{X V_{\hat{S}},1} \right) = \log \frac{25}{9}
        \quad , \quad
        M_2 \left( \ket{B^\prime,-1} \right) =  \log \frac{27}{11}
        \quad , \quad
        M_2 \left( \ket{B',-e^{\pm \frac{2\pi i}{3}}} \right) =  \log \frac{54}{19}.
    \end{split}
\end{equation}

\newpage
\section{Clifford-inequivalent Non-degenerate Eigenstates of $\mC_{2,2}$}
\label{app:Finding the eigenstates of 2 qubits}

In this appendix, we systematically identify and classify all Clifford-inequivalent non-degenerate eigenstates of two-qubit Clifford operations, providing explicit constructions and equivalence arguments.

\subsection*{Eigenstates of representatives of the group classes}

We first find and identify the non-degenerate eigenstates of representatives of the conjugacy classes of the two-qubit Clifford group (see Table 1 in \cite{CliffordClassesForTwoQubits}). The results are summarized in Table~\ref{tab:eigenstates of Clifford classes representatives for two qubits}, where we define the following states here.

\begin{subequations}
    \begin{gather}
        \resizebox{\textwidth}{!}{$\ket{G_{16},1}=-\frac{\left(10+3\sqrt{5}-2\sqrt{5\bigl(5+2\sqrt{5}\bigr)}\right)^{\frac{1}{4}}}{\sqrt{10}}  e^{i\pi/4}\ket{00} + \frac{\left(10+3\sqrt{5}+2\sqrt{5\bigl(5+2\sqrt{5}\bigr)}\right)^{\frac{1}{4}}}{\sqrt{10}}  e^{i\pi/4}\ket{01} + \frac{i}{\sqrt{5(3+\sqrt{5})+\sqrt{250+110\sqrt{5}}}}\ket{10} + \frac{1}{\sqrt{5(3+\sqrt{5})-\sqrt{250+110\sqrt{5}}}}\ket{11}   $} \\
        \resizebox{\textwidth}{!}{$\ket{G_{16},2} =-\frac{\left(10-3\sqrt{5}-2\sqrt{5\bigl(5-2\sqrt{5}\bigr)}\right)^{\frac{1}{4}}}{\sqrt{10}}  e^{i\pi/4}\ket{00} - \frac{\left(10-3\sqrt{5}+2\sqrt{5\bigl(5-2\sqrt{5}\bigr)}\right)^{\frac{1}{4}}}{\sqrt{10}}  e^{i\pi/4}\ket{01} + \frac{i}{2}\sqrt{\frac{5+\sqrt{5}+\sqrt{2\bigl(5+\sqrt{5}\bigr)}}{5}} \ket{10} + \frac{1}{2}\sqrt{1+\frac{1}{\sqrt{5}}-\frac{1}{5}\sqrt{2\bigl(5+\sqrt{5}\bigr)}} \ket{11}   $} \\
        \resizebox{\textwidth}{!}{$\ket{G_{16},3}=\frac{\left(10-3\sqrt{5}+2\sqrt{5\bigl(5-2\sqrt{5}\bigr)}\right)^{\frac{1}{4}}}{\sqrt{10}}  e^{i\pi/4}\ket{00} - \frac{\left(10-3\sqrt{5}-2\sqrt{5\bigl(5-2\sqrt{5}\bigr)}\right)^{\frac{1}{4}}}{\sqrt{10}}  e^{i\pi/4}\ket{01} - \frac{i}{2}\sqrt{1+\frac{1}{\sqrt{5}}-\frac{1}{5}\sqrt{2\bigl(5+\sqrt{5}\bigr)}} \ket{10} + \frac{1}{2}\sqrt{\frac{5+\sqrt{5}+\sqrt{2\bigl(5+\sqrt{5}\bigr)}}{5}} \ket{11}   $} \\
        \resizebox{\textwidth}{!}{$\ket{G_{16},4}=-\frac{\left(10+3\sqrt{5}+2\sqrt{5\bigl(5+2\sqrt{5}\bigr)}\right)^{\frac{1}{4}}}{\sqrt{10}}  e^{i\pi/4}\ket{00} - \frac{\left(10+3\sqrt{5}-2\sqrt{5\bigl(5+2\sqrt{5}\bigr)}\right)^{\frac{1}{4}}}{\sqrt{10}}  e^{i\pi/4}\ket{01} - \frac{i}{2}\sqrt{\frac{5-\sqrt{5}+\sqrt{10-2\sqrt{5}}}{5}} \ket{10} + \frac{1}{\sqrt{5(3+\sqrt{5})+\sqrt{250+110\sqrt{5}}}} \ket{11}   . $}
    \end{gather}
\end{subequations}

\begin{subequations}
    \begin{gather}
        \ket{G_{18},1}=\frac{i\sqrt{2} \ket{00}+2i \ket{10}+\sqrt{2} \ket{11}}{2\sqrt{2}}   ,   \\
        \ket{G_{18},2}=\frac{-i\sqrt{2} \ket{00}-2i \ket{01}+\sqrt{2} \ket{11}}{2\sqrt{2}}   ,  \\
        \ket{G_{18},3}=\frac{-i\sqrt{2} \ket{00}+2i \ket{01}+\sqrt{2} \ket{11}}{2\sqrt{2}}   ,  \\       \ket{G_{18},4}=\frac{i\sqrt{2} \ket{00}-2i \ket{10}+\sqrt{2} \ket{11}}{2\sqrt{2}}   .
    \end{gather}
\end{subequations}

\begin{subequations}
    \begin{gather}
        \ket{G_{20},1}=\frac{i \ket{00}-i \ket{01}-(1+2i) \ket{10}+\ket{11}}{2\sqrt{2}}   ,   \\
        \ket{G_{20},2}=\frac{-(2-i) \ket{00}-(2-i) \ket{01}+(1+2i) \ket{10}+5 \ket{11}}{2\sqrt{10}}   ,  \\
        \ket{G_{20},3}=\frac{-i \ket{00}+(2+i) \ket{01}-\ket{10}+\ket{11}}{2\sqrt{2}}   ,  \\
        \ket{G_{20},4}=\frac{(2-i) \ket{00}-i \ket{01}+\ket{10}+\ket{11}}{2\sqrt{2}}   .
    \end{gather}
\end{subequations}

Table~\ref{tab:eigenstates of Clifford classes representatives for two qubits} presents the representatives of the conjugacy classes of the two-qubit Clifford group (see Table 1 in \cite{CliffordClassesForTwoQubits}) alongside their corresponding non-degenerate eigenstates. The results for non-entangling operations are straightforward. However, the table highlights in red the entangling operations, specifically, those whose representatives include the CZ gate.

\renewcommand{\arraystretch}{1.1}
\begin{table}[h!]
    \centering
    \begin{tabular}{|c|c|c|c|}
        \hline
        \textbf{Label} & \textbf{Size} & \textbf{Representative} & \textbf{Non-Degenerate Eigenstates} \\
        \hline
        No.1  & 1    & $I$ & none \\
        \hdashline
        No.2  & 640  & $S_2 S_1 H_2 H_1$ & $|T_0T_0\rangle$ , $|T_1T_1\rangle$ \\
        \hdashline
        No.3  & 60   & $C$ & $|11\rangle$ \\
        \hdashline
        {\color{red}No.4} & 1920 & $C H_2 H_1$ & \tiny{\makecell{$\frac{\ket{01}-\ket{10}}{\sqrt{2}}$ , $\ket{G_4,2}=\frac{2 \ket{00}+\ket{01}+\ket{10}}{\sqrt{6}}$ , \\ $\ket{G_4,3}=\frac{-i \ket{00}+i \ket{01}+i \ket{10}+\sqrt{3} \ket{11}}{\sqrt{6}}$ , \\  $\ket{G_4,4}=\frac{i \ket{00}-i \ket{01}-i \ket{10}+\sqrt{3} \ket{11}}{\sqrt{6}}$}} \\
        \hdashline
        No.5  & 15   & $S_1^2$ & none \\
        \hdashline
        No.6  & 180  & $S_2 S_1$ & $|00\rangle$ , $|11\rangle$ \\
        \hdashline
        {\color{red}No.7} & 720  & $C H_1$ & $|H_0 0\rangle$ , $|H_1 0\rangle$ , $|R 1\rangle$ , $|L 1\rangle$ \\
        \hdashline
        No.8  & 30   & $S_1$ & none \\
        \hdashline
        No.9  & 360  & $S_2 H_1$ & $|H_0 0\rangle$ , $|H_0 1\rangle$ , $|H_1 0\rangle$ , $|H_1 1\rangle$ \\
        \hdashline
        {\color{red}No.10} & 180  & $H_2 S_2^2 H_2 C$ & $|0+\rangle$ , $|0-\rangle$ , $|1R\rangle$ , $|1L\rangle$ \\
        \hdashline
        No.11 & 90   & $S_2 S_1^2$ & $|00\rangle$ , $|01\rangle$ , $|10\rangle$ , $|11\rangle$ \\
        \hdashline
        No.12  & 120  & $H_1$ & none \\
        \hdashline
        No.13  & 160  & $S_1 H_1$ & none \\
        \hdashline
        No.14  & 480  & $S_2^2 S_1 H_1$ & $|T_0 0\rangle$ , $|T_0 1\rangle$ , $|T_1 0\rangle$ , $|T_1 1\rangle$ \\
        \hdashline
        No.15  & 960  & $S_2 S_1 H_1$ & $|T_0 0\rangle$ , $|T_0 1\rangle$ , $|T_1 0\rangle$ , $|T_1 1\rangle$  \\
        \hdashline
        {\color{red}No.16} & 2304 & $C S_1 H_2 H_1$ & $\ket{G_{16},1}$ , $\ket{G_{16},2}$ , $\ket{G_{16},3}$ , $\ket{G_{16},4}$ \\
        \hdashline
        No.17  & 180  & $H_2 H_1$ & none \\
        \hdashline
        {\color{red}No.18} & 720  & $H_1 H_2 S_2 S_1^{-1} C H_1 H_2 C$ & $\ket{G_{18},1}$ , $\ket{G_{18},2}$ , $\ket{G_{18},3}$ , $\ket{G_{18},4}$ \\
        \hdashline
        No.19  & 960  & $S_1 H_2 H_1$ & $|T_0 H_0\rangle$ , $|T_0 H_1\rangle$ , $|T_1 H_0\rangle$ , $|T_1 H_1\rangle$  \\
        \hdashline
        {\color{red}No.20} & 720  & $H_1 C H_2 S_2 C$ & $\ket{G_{20},1}$ , $\ket{G_{20},2}$ , $\ket{G_{20},3}$ , $\ket{G_{20},4}$ \\
        \hdashline
        {\color{red}No.21} & 720  & $C S_2 H_1$ & $|H_0 0\rangle$ , $|H_1 0\rangle$ , $|R 1\rangle$ , $|L 1\rangle$ \\
        \hline
    \end{tabular}
    \caption{Representatives of conjugacy classes of the two-qubit Clifford group and their non-degenerate eigenstates. Here each $C$ is the CZ operation controlled by the first qubit.}
    \label{tab:eigenstates of Clifford classes representatives for two qubits}
\end{table}

\subsection*{Identifying the Clifford-inequivalent states}

Looking at the states in Table [1]:
\begin{enumerate}
    \item The states $|11\rangle$ , $|00\rangle$ , $|R 1\rangle$ , $|L 1\rangle$ , $\dfrac{\ket{01}-\ket{10}}{\sqrt{2}}$ , $|0+\rangle$ , $|0-\rangle$ , $|1R\rangle$ , $|1L\rangle$ , $|01\rangle$ , $|10\rangle$ are all stabilizer states and Clifford-equivalent to $\ket{00}$.
    \item The states $|H_0 0\rangle$ , $|H_1 0\rangle$ , $|H_0 1\rangle$ , $|H_1 1\rangle$ are all Clifford-equivalent to $\ket{H0}$.
    \item The states $|T_0 0\rangle$ , $|T_1 0\rangle$ , $|T_0 1\rangle$ , $|T_1 1\rangle$ are all Clifford-equivalent to $\ket{T0}$.
    \item The states $|T_0 H_0\rangle$ , $|T_1 H_0\rangle$ , $|T_0 H_1\rangle$ , $|T_1 H_1\rangle$ are all Clifford-equivalent to $\ket{TH}$.
    \item The states $|T_0 T_0\rangle$ , $|T_1 T_1\rangle$  are Clifford-equivalent.
    \item One can verify that
        \begin{equation}
        \ket{G_{4},2}=S_1^\dagger H_1\text{CZ}H_2\text{CZ} S_1^\dagger H_2 H_1 \frac{\ket{T_0T_0}-\ket{T_1T_1}}{\sqrt{2}}.
        \end{equation}
    \item One can verify that
        \begin{subequations}
        \begin{gather}
        \ket{G_{4},3}= -e^{i\frac{\pi}{4}}S_2^\dagger H_2H_1\text{CZ}H_1\text{CZ}  \ket{T_1T_0},\\[1ex]
        \ket{G_{4},4}=-\text{CZ} \ket{G_{4},3},
        \end{gather}
        \end{subequations}
    so that they both are Clifford-equivalent to $\ket{TT}$ (remember that all $T$-states are Clifford-equivalent; in particular $\ket{T_1}=-e^{-i\frac{\pi}{4}}S^3X\ket{T_0}$).
    \item One can verify that
        \begin{subequations}
        \begin{gather}
        \ket{G_{18},1}=i S_2 H_1 Z_1 \text{CNOT} \ket{H_0H_0},\\[1ex]        \ket{G_{18},2}=\text{SWAP}\cdot\text{CZ} \ket{G_{18},1},\\[1ex]
        \ket{G_{18},3}=Z_1Z_2 \ket{G_{18},2},\\[1ex]
        \ket{G_{18},4}=Z_1Z_2 \ket{G_{18},1},
        \end{gather}
        \end{subequations}
    so that these four states are Clifford-equivalent to $\ket{HH}$.
    \item One can verify that
        \begin{subequations}
        \begin{gather}       \ket{G_{20},2}=H_1S_1^\dagger H_1H_2 \ket{G_{20},1},\\[1ex]
        \ket{G_{20},3}=H_1S_1^\dagger H_1\text{CZ}H_2 \ket{G_{20},1},\\[1ex]
        \ket{G_{20},4}=H_2S_2^\dagger Z_1 H_1 S_1^\dagger \ket{G_{20},1},
        \end{gather}
        \end{subequations}
    so that these four states are Clifford-equivalent to $\ket{G_{20},1}$.
\end{enumerate}

\subsection*{Summarizing}

In Table~\ref{tab:Clifford-inequivalent non-degenerate eigenstates of Clifford operations for two qubits} in the maintext, we collect all different Clifford-inequivalent non-degenerate eigenstates of Clifford operations for two qubits. For each eigenstate, we also report the stabilizer fidelity and the number of nearest stabilizer states.

\section{Calculations for The Doubled five-qubit Stabilizer Code}

This appendix details the explicit calculations for the doubled five-qubit stabilizer code presented in Section~\ref{sec:An Inefficient Distillation Protocol Demonstrating Magic}, including projections of tensor-product magic states and the analysis of logical operations. The notation adopted here follows that of~\cite{KitaevBravyi}, to which we refer the reader for additional details.

\subsection*{The Five Qubit Code}

Using similar notations as in \cite{KitaevBravyi}, we define
\begin{align*}
    S_\pm(5) &\equiv \left\{ \pm\bigotimes_{i=1}^{5} \sigma_{\mu_i} \; \mid \; \mu_i \in \{0,1,2,3\} \right\}, \\
    G &\equiv \langle S_1, S_2, S_3, S_4 \rangle \subset S_+(5), \\
    \hat{T}_L &\equiv \hat{T}^{\otimes 5}, \quad \left|T_{\vec{x}}\right\rangle \equiv \bigotimes_{i=1}^{5} \left| T_{x_i} \right\rangle, \\
    \Pi &\equiv \frac{1}{16} \prod_{i=1}^{4} (I + S_i) = \frac{1}{16} \sum_{h \in G} h,
\end{align*}
where $\sigma_0, \sigma_1, \sigma_2, \sigma_3$ are the known Pauli matrices, and $S_1, S_2, S_3, S_4$ are the four generators of the stabilizer group of the 5-qubit code:
\begin{align*}
    S_1 &= XZZXI, \quad S_2 = IXZZX, \quad S_3 = XIXZZ, \quad S_4 = ZXIXZ.
\end{align*}
$\ket{T_0}, \ket{T_1}$ are the $T$-magic states with eigenvalues $e^{\pm i \pi/3}$ for the $\hat{T}$ gate:
\begin{align*}
    \hat{T} = e^{i \pi/4} S H = \frac{e^{i \pi/4}}{\sqrt{2}} \begin{bmatrix} 1 & 1 \\ i & -i \end{bmatrix}.
\end{align*}

It is easy to see that:
\begin{align*}
    \left|T_x\right\rangle \langle T_x| &= \frac{1}{2} \left( I + \frac{(-1)^x}{\sqrt{3}} (X+Y+Z) \right), \\
    \left|T_{\vec{x}}\right\rangle \langle T_{\vec{x}}| &= \frac{1}{2^5} \bigotimes_{i=1}^{5} \left( I + \frac{(-1)^{x_i}}{\sqrt{3}} (X+Y+Z) \right) \\
    &= \frac{1}{2^5} \sum_{g \in S_+(5)} \left( \frac{1}{\sqrt{3}} \right)^{|g|} (-1)^{\sum_{i=1}^{5} x_i I_{Pauli}(g,i)} g,
\end{align*}
where $I_{Pauli}(g,i)$ is the indicator for a Pauli matrix in $g$ at the $i$-th place.

Therefore,
\begin{align*}
    \bra{T_x} \Pi \ket{T_x} &= \bra{T_x} \sum_{\vec{y}} \ket{T_y} \bra{T_y} \Pi \ket{T_x} 
    = \sum_{\vec{y}} \bra{T_y} \Pi \ket{T_x} \bra{T_x} \ket{T_y} \\
    &= \text{tr} \left[ \Pi \ket{T_x} \bra{T_x} \right] 
    = \frac{1}{16} \sum_{h \in G} \text{tr} \left[ h \ket{T_x} \bra{T_x} \right] \\
    &= \frac{1}{16} \frac{1}{2^5} \sum_{h \in G} \sum_{g \in S_+(5)} \left(\frac{1}{\sqrt{3}}\right)^{|g|} (-1)^{\sum_{i=1}^{5} x_i I_{Pauli}(g,i)} \text{tr}[hg] \\
    &= \frac{1}{16} \frac{1}{2^5} \sum_{h \in G} \sum_{g \in S_+(5)} \left(\frac{1}{\sqrt{3}}\right)^{|g|} (-1)^{\sum_{i=1}^{5} x_i I_{Pauli}(g,i)} 2^5 \delta_{h,g} \\
    &= \frac{1}{16} \sum_{g \in G} \left( (-1)^{\sum_{i=1}^{5} x_i I_{Pauli,i}(g,i)} \left(\frac{1}{\sqrt{3}}\right)^{|g|} \right) \\
    &= \frac{1}{16} \left[ 1 + 3 \left(\frac{1}{\sqrt{3}}\right)^4 (-1)^{\sum_{i=1}^{5} x_i} \sum_{i=1}^{5} (-1)^{x_i} \right] \\
    &= \frac{1}{16} \left[ 1 + \frac{1}{3} (-1)^{\sum_{i=1}^{5} x_i} \sum_{i=1}^{5} (-1)^{x_i} \right] 
    = \begin{cases} 
        \frac{1}{6}, & |\vec{x}| = 0 \\
        0, & |\vec{x}| = 1 \\
        \frac{1}{12}, & |\vec{x}| = 2 \\
        \frac{1}{12}, & |\vec{x}| = 3 \\
        0, & |\vec{x}| = 4 \\
        \frac{1}{6}, & |\vec{x}| = 5
    \end{cases}
\end{align*}

So define:
\begin{align*}
    \ket{T_{0,L}} &= \sqrt{6} \Pi \ket{T_1}, \\
    \ket{T_{1,L}} &= \sqrt{6} \Pi \ket{T_0}.
\end{align*}
These states have the right corresponding eigenvalues for the logical transversal $\hat{T}_L$ gate composed of 5 $\hat{T}$ gates.
And then:
\begin{align*}
    \Pi \ket{T_x} = 
    \begin{cases} 
        \Pi \ket{T_0}, & |\vec{x}| = 0 \\
        0, & |\vec{x}| = 1 \\
        \frac{e^{i\phi_{\vec{x}}}}{\sqrt{2}} \Pi \ket{T_1}, & |\vec{x}| = 2 \\
        \frac{e^{i\phi_{\vec{x}}}}{\sqrt{2}} \Pi \ket{T_0}, & |\vec{x}| = 3 \\
        0, & |\vec{x}| = 4 \\
        \Pi \ket{T_1}, & |\vec{x}| = 5
    \end{cases} 
    = \begin{cases} 
        \frac{1}{\sqrt{6}} \ket{T_{1,L}}, & |\vec{x}| = 0 \\
        0, & |\vec{x}| = 1 \\
        \frac{e^{i\phi_{\vec{x}}}}{\sqrt{12}} \ket{T_{0,L}}, & |\vec{x}| = 2 \\
        \frac{e^{i\phi_{\vec{x}}}}{\sqrt{12}} \ket{T_{1,L}}, & |\vec{x}| = 3 \\
        0, & |\vec{x}| = 4 \\
        \frac{1}{\sqrt{6}} \ket{T_{0,L}}, & |\vec{x}| = 5
    \end{cases}
\end{align*}

And from explicit direct calculation:
\begin{align*}
    \phi_{\vec{x}} &=
    \begin{cases}
        -\frac{2\pi}{3}, & \vec{x} = 00111, 01110, 10011, 11001, 11100 \\
        \frac{2\pi}{3}, & \vec{x} = 01011, 01101, 10101, 10110, 11010
    \end{cases} \\
    \varphi_{\vec{x}} &=
    \begin{cases}
        -\frac{\pi}{3}, & \vec{x} = 00011, 00110, 01100, 10001, 11000 \\
        \frac{\pi}{3}, & \vec{x} = 00101, 01001, 01010, 10010, 10100
    \end{cases}
\end{align*}

\subsection*{The Doubled Five Qubit Code}

Note that
\begin{subequations}
    \begin{gather}
        \ket{\psi_{\vec{3}}}\equiv\ket{\psi_{(3,3,3,3,3)}}=\frac{1}{\sqrt{2^5}}\sum_{\vec{x}\in \left\{ 0,1 \right\}^5} \ket{T_{\vec{x}}T_{\vec{x}}}   , \\
        \ket{\psi_{\vec{0}}}\equiv\ket{\psi_{(0,0,0,0,0)}}=\frac{1}{\sqrt{2^5}}\sum_{\vec{x}\in \left\{ 0,1 \right\}^5} (-1)^{|\vec{x}|} \ket{T_{\vec{x}}T_{\vec{x}}}   . 
    \end{gather}
\end{subequations}
Therefore,
\tiny{
\begin{equation}
\begin{split}
    \Pi \ket{\psi_{(3,3,3,3,3)}} & = \frac{1}{\sqrt{2^5}} \sum_{\vec{x} \in \{0,1\}^5} \Pi \ket{T_{\vec{x}} T_{\vec{x}}} = \frac{1}{\sqrt{2^5}} \left[ \sum_{i=0}^{5} \sum_{\substack{\vec{x} \in \{0,1\}^5 \\ |\vec{x}| = i}} \Pi \ket{T_{\vec{x}} T_{\vec{x}}} \right] \\
    & = \frac{1}{\sqrt{2^5}} \left[ \Pi \ket{T_{\vec{0}} T_{\vec{0}}} + \sum_{\substack{\vec{x} \in \{0,1\}^5 \\ |\vec{x}| = 2}} \Pi \ket{T_{\vec{x}} T_{\vec{x}}} + \sum_{\substack{\vec{x} \in \{0,1\}^5 \\ |\vec{x}| = 3}} \Pi \ket{T_{\vec{x}} T_{\vec{x}}} + \Pi \ket{T_{\vec{1}} T_{\vec{1}}} \right] \\
    & =\frac{1}{\sqrt{2^5}} \left[ \frac{1}{6} \ket{T_{(1,L)} T_{(1,L)}} + \sum_{\substack{\vec{x} \in \{0,1\}^5 \\ |\vec{x}| = 2}} \frac{e^{2i\varphi_{\vec{x}}}}{12} \ket{T_{(0,L)} T_{(0,L)}} + \sum_{\substack{\vec{x} \in \{0,1\}^5 \\ |\vec{x}| = 3}} \frac{e^{2i\phi_{\vec{x}}}}{12} \ket{T_{(1,L)} T_{(1,L)}} + \frac{1}{6} \ket{T_{(0,L)} T_{(0,L)}} \right] \\
    & = \frac{1}{6} \frac{1}{\sqrt{2^5}} \left[ \left(1 + \sum_{\substack{\vec{x} \in \{0,1\}^5 \\ |\vec{x}| = 3}} \frac{e^{2i\phi_{\vec{x}}}}{2} \right) \ket{T_{(1,L)} T_{(1,L)}} + \left(1 + \sum_{\substack{\vec{x} \in \{0,1\}^5 \\ |\vec{x}| = 2}} \frac{e^{2i\varphi_{\vec{x}}}}{2} \right) \ket{T_{(0,L)} T_{(0,L)}} \right] \\
    & = \frac{1}{6} \frac{1}{\sqrt{2^5}} \left(1 + 5 \frac{e^{4i\pi/3} + e^{-4i\pi/3}}{2} \right) \ket{T_{(1,L)} T_{(1,L)}} + \frac{1}{6} \frac{1}{\sqrt{2^5}} \left(1 + 5 \frac{e^{2i\pi/3} + e^{-2i\pi/3}}{2} \right) \ket{T_{(0,L)} T_{(0,L)}} \\
    & = \frac{1}{6} \frac{1}{\sqrt{2^5}} \left(1 + 5 \cos\frac{2\pi}{3} \right) \left[ \ket{T_{(1,L)} T_{(1,L)}} + \ket{T_{(0,L)} T_{(0,L)}} \right] 
     = \frac{1}{6} \frac{1}{\sqrt{2^5}} \left(1 - \frac{5}{2} \right) \left[ \ket{T_{(1,L)} T_{(1,L)}} + \ket{T_{(0,L)} T_{(0,L)}} \right] \\
    & = -\frac{1}{4} \frac{1}{\sqrt{2^5}} \left[ \ket{T_{(1,L)} T_{(1,L)}} + \ket{T_{(0,L)} T_{(0,L)}} \right] = -\frac{1}{16} \frac{\left[ \ket{T_{(1,L)} T_{(1,L)}} + \ket{T_{(0,L)} T_{(0,L)}} \right]}{\sqrt{2}}
\end{split}
\end{equation}
}
\normalsize{}

Similarly,
\tiny{
\begin{equation}
\begin{split}
        \Pi \ket{\psi_{(0,0,0,0,0)}} & = \frac{1}{\sqrt{2^5}} \sum_{\vec{x} \in \{0,1\}^5} (-1)^{|\vec{x}|} \Pi \ket{T_{\vec{x}} T_{\vec{x}}} = \frac{1}{\sqrt{2^5}} \left[ \sum_{i=0}^{5} \sum_{\substack{\vec{x} \in \{0,1\}^5 \\ |\vec{x}| = i}} (-1)^{|\vec{x}|} \Pi \ket{T_{\vec{x}} T_{\vec{x}}} \right] = \\
        & = \frac{1}{\sqrt{2^5}} \left[ \Pi \ket{T_{\vec{0}} T_{\vec{0}}} + \sum_{\substack{\vec{x} \in \{0,1\}^5 \\ |\vec{x}| = 2}} \Pi \ket{T_{\vec{x}} T_{\vec{x}}} - \sum_{\substack{\vec{x} \in \{0,1\}^5 \\ |\vec{x}| = 3}} \Pi \ket{T_{\vec{x}} T_{\vec{x}}} - \Pi \ket{T_{\vec{1}} T_{\vec{1}}} \right] \\
        & = \frac{1}{\sqrt{2^5}} \left[ \frac{1}{6} \ket{T_{(1,L)} T_{(1,L)}} + \sum_{\substack{\vec{x} \in \{0,1\}^5 \\ |\vec{x}| = 2}} \frac{e^{2i\varphi_{\vec{x}}}}{12} \ket{T_{(0,L)} T_{(0,L)}} - \sum_{\substack{\vec{x} \in \{0,1\}^5 \\ |\vec{x}| = 3}} \frac{e^{2i\phi_{\vec{x}}}}{12} \ket{T_{(1,L)} T_{(1,L)}} - \frac{1}{6} \ket{T_{(0,L)} T_{(0,L)}} \right] \\
        & = \frac{1}{6} \frac{1}{\sqrt{2^5}} \left[ \left(1 - \sum_{\substack{\vec{x} \in \{0,1\}^5 \\ |\vec{x}| = 3}} \frac{e^{2i\phi_{\vec{x}}}}{2} \right) \ket{T_{(1,L)} T_{(1,L)}} - \left(1 - \sum_{\substack{\vec{x} \in \{0,1\}^5 \\ |\vec{x}| = 2}} \frac{e^{2i\varphi_{\vec{x}}}}{2} \right) \ket{T_{(0,L)} T_{(0,L)}} \right] \\
        & = \frac{1}{6} \frac{1}{\sqrt{2^5}} \left(1 - 5 \frac{e^{4i\pi/3} + e^{-4i\pi/3}}{2} \right) \ket{T_{(1,L)} T_{(1,L)}} - \frac{1}{6} \frac{1}{\sqrt{2^5}} \left(1 - 5 \frac{e^{2i\pi/3} + e^{-2i\pi/3}}{2} \right) \ket{T_{(0,L)} T_{(0,L)}} \\
        & = \frac{1}{6} \frac{1}{\sqrt{2^5}} \left(1 - 5 \cos\frac{2\pi}{3} \right) \left[ \ket{T_{(1,L)} T_{(1,L)}} - \ket{T_{(0,L)} T_{(0,L)}} \right] 
         = \frac{1}{6} \frac{1}{\sqrt{2^5}} \left(1 + \frac{5}{2} \right) \left[ \ket{T_{(1,L)} T_{(1,L)}} - \ket{T_{(0,L)} T_{(0,L)}} \right] \\
        & = \frac{7}{12} \frac{1}{\sqrt{2^5}} \left[ \ket{T_{(1,L)} T_{(1,L)}} - \ket{T_{(0,L)} T_{(0,L)}} \right] 
         = \frac{7}{48} \frac{\left[ \ket{T_{(1,L)} T_{(1,L)}} - \ket{T_{(0,L)} T_{(0,L)}} \right]}{\sqrt{2}}
\end{split}
\end{equation}
}
\normalsize{}

Thus, the projections satisfy:
\begin{subequations}
    \begin{gather}
        \Pi \ket{\psi_{\vec{1}}} = \frac{1}{6} \ket{\psi_2^L}, \\  \Pi \ket{\psi_{\vec{2}}} = \frac{1}{6} \ket{\psi_1^L}, \\
    \Pi \ket{\psi_{\vec{0}}} = \frac{7}{48} \ket{\psi_0^L}, \\  \Pi \ket{\psi_{\vec{3}}} = -\frac{1}{16} \ket{\psi_3^L}.
    \end{gather}
\end{subequations}

These explicit derivations validate the constructions discussed in in Section~\ref{sec:An Inefficient Distillation Protocol Demonstrating Magic}.

\end{document}